\documentclass[11pt]{amsbook}
\usepackage[utf8]{inputenc}
\usepackage[titletoc]{appendix}
\usepackage{empheq}

\usepackage[backend=biber, style=numeric, date=year, url=false, doi=false, isbn=false, eprint=true, firstinits=true, maxcitenames=5, maxbibnames=5, style=alphabetic]{biblatex}
\AtEveryBibitem{%
  \ifentrytype{online}{%
  }{%
    \clearfield{eprint}%
    \clearfield{urldate}%
  }%
}
\usepackage{blindtext}
\usepackage{amssymb}
\usepackage{amsfonts}
\usepackage{amsthm}
\usepackage{tikz}
\usepackage{tikz-cd}
\usepackage{caption}
\usepackage{graphics}
\usepackage{xcolor}
\usepackage{shadethm}
\usepackage{subfigure}
\usepackage{epigraph}
\usepackage[includeheadfoot,margin=1.2in]{geometry}
\usepackage{placeins}
\usetikzlibrary{snakes}
\usepackage[new]{old-arrows}
\usepackage{chngcntr}

\usepackage{enumerate}
\usepackage{verbatim}
\usetikzlibrary{trees}
\usetikzlibrary{decorations.markings}
\usetikzlibrary{arrows}
\usetikzlibrary{cd}
\usepackage{mathrsfs}
\usepackage{xparse}
\usetikzlibrary{positioning,arrows,patterns}
\usetikzlibrary{decorations.markings}
\usetikzlibrary{calc}
\tikzset{
  photon/.style={decorate, decoration={snake}, draw=black},
  fermion/.style={draw=black, postaction={decorate},decoration={markings,mark=at position .55 with {\arrow{>}}}},
  fermion2/.style={dashed, dash phase=0.1pt, draw=black, postaction={decorate},decoration={markings,mark=at position .55 with {\arrow{>}}}},
  vertex/.style={draw,shape=circle,fill=black,minimum size=5pt,inner sep=0pt},
particle/.style={thick,draw=black},
particle2/.style={thick,draw=blue},
avector/.style={thick,draw=black, postaction={decorate},
    decoration={markings,mark=at position 1 with {\arrow[black]{triangle 45}}}},
gluon/.style={decorate, draw=black,
    decoration={coil,aspect=0}}
 }
\NewDocumentCommand\semiloop{O{black}mmmO{}O{above}}
{%
\draw[#1] let \p1 = ($(#3)-(#2)$) in (#3) arc (#4:({#4+180}):({0.5*veclen(\x1,\y1)})node[midway, #6] {#5};)
}

\tikzcdset{scale cd/.style={every label/.append style={scale=#1},
    cells={nodes={scale=#1}}}}

\newcommand*\widefbox[1]{\fbox{\hspace{2em}#1\hspace{2em}}}
\usepackage{url}
\usepackage{hyperref}
\hypersetup{colorlinks=true, citecolor=purple, linktocpage, linkcolor=blue, urlcolor=cyan} 

\theoremstyle{plain}
\newtheorem{thm}{Theorem}[section]
\newtheorem{lem}[thm]{Lemma}
\newtheorem{prop}[thm]{Proposition}
\newtheorem{cor}[thm]{Corollary}
\theoremstyle{definition}
\newtheorem{exe}[thm]{Exercise}
\newtheorem{conj}[thm]{Conjecture}
\newtheorem{defn}[thm]{Definition}

\newtheorem*{thm*}{Theorem}
\newtheorem*{lem*}{Lemma}
\newtheorem*{prop*}{Proposition}
\newtheorem*{cor*}{Corollary}
\newtheorem*{exe*}{Exercise}
\newtheorem*{defn*}{Definition}
\newtheorem{rem}[thm]{Remark}

\newtheorem{ex}[thm]{Example}

\theoremstyle{remark}

\numberwithin{section}{chapter}
\numberwithin{subsection}{section}
\numberwithin{subsubsection}{subsection}

\newcommand{\R}{\mathbb{R}}
\newcommand{\Z}{\mathbb{Z}}

\newcommand{\N}{\mathbb{N}}

\newcommand{\calU}{\mathcal{U}}
\newcommand{\calD}{\mathcal{D}}

\numberwithin{equation}{section}
\counterwithin{figure}{section}

\newcommand{\dd}{{\mathrm{d}}}

\DeclareMathOperator{\GL}{GL}

\DeclareMathOperator{\ev}{ev}
\newcommand{\id}{\mathrm{id}}

\DeclareMathOperator{\tr}{Tr}

\DeclareMathOperator{\Div}{\textnormal{div}}

\newcommand{\C}{\mathbb{C}}

\DeclareMathOperator{\ad}{ad}

\DeclareMathOperator{\Ad}{Ad}

\DeclareMathOperator{\End}{End}
\DeclareMathOperator{\Hom}{Hom}

\newcommand{\de}{\partial}

\newcommand{\calA}{\mathcal{A}}

\newcommand{\calH}{\mathcal{H}}

\newcommand{\calC}{\mathcal{C}}

\newcommand{\calG}{\mathcal{G}}
\newcommand{\calI}{\mathcal{I}}
\newcommand{\calO}{\mathcal{O}}

\newcommand{\calM}{\mathcal{M}}

\newcommand{\calW}{\mathcal{W}}

\newcommand{\calN}{\mathcal{N}}
\newcommand{\supp}{\mathrm{supp}\,}

\newcommand{\calF}{\mathcal{F}}

\def\gpd{\,\lower1pt\hbox{$\longrightarrow$}\hskip-.24in\raise2pt
               \hbox{$\longrightarrow$}\,}

\newcommand{\im}{\mathrm{im}\,}

\newcommand{\I}{\mathrm{i}}

\newcommand{\calV}{\mathcal{V}}

\newcommand{\Sym}{\textnormal{Sym}}

\makeatletter
\newcommand{\upint}{\DOTSI\upintop\ilimits@}
\newcommand{\upoint}{\DOTSI\upointop\ilimits@}
\makeatother

\usepackage{todonotes}

\makeatletter
\providecommand\@dotsep{5}
\renewcommand{\listoftodos}[1][\@todonotes@todolistname]{%
  \@starttoc{tdo}{#1}}
\makeatother

\makeatletter

\pgfdeclareshape{genus}{
 \anchor{center}{\pgfpointorigin}
\backgroundpath{
    \begingroup
    \tikz@mode
    \iftikz@mode@fill

         \pgfpathmoveto{\pgfqpoint{-0.78cm}{-.17cm}}
     \pgfpathcurveto %
            {\pgfpoint{-0.35cm}{-.44cm}}
        {\pgfpoint{0.35cm}{-.44cm}}
        {\pgfpoint{.78cm}{-0.17cm}} 
     \pgfpathmoveto{\pgfqpoint{-0.78cm}{-0.17cm}}
     \pgfpathcurveto %
            {\pgfpoint{-0.25cm}{.25cm}}
        {\pgfpoint{.25cm}{.25cm}}
        {\pgfpoint{0.78cm}{-0.17cm}}
        \pgfusepath{fill}
        \fi

        \iftikz@mode@draw
     \pgfpathmoveto{\pgfqpoint{-1cm}{0cm}}
     \pgfpathcurveto %
            {\pgfpoint{-0.5cm}{-.5cm}}
        {\pgfpoint{0.5cm}{-.5cm}}
        {\pgfpoint{1cm}{0cm}}

     \pgfpathmoveto{\pgfqpoint{-0.75cm}{-0.15cm}}
     \pgfpathcurveto %
            {\pgfpoint{-0.25cm}{.25cm}}
        {\pgfpoint{.25cm}{.25cm}}
        {\pgfpoint{0.75cm}{-0.15cm}}
              \pgfusepath{stroke}
        \fi
        \endgroup
    }
    }

    \makeatother

\tikzset{residual/.style={draw, shape=circle, black,inner sep=1pt}}

\title{Lectures on Symplectic Geometry, Poisson Geometry, Deformation Quantization \\and Quantum Field Theory}
\author[N. Moshayedi]{Nima Moshayedi\vspace{14cm}}
\address{Institut f\"ur Mathematik\\ Universit\"at Z\"urich\\ 
Winterthurerstrasse 190
CH-8057 Z\"urich}
\email[N.~Moshayedi]{nima.moshayedi@math.uzh.ch}

\begin{document}

\maketitle

\begin{abstract}
    These are lecture notes for the course ``Poisson geometry and deformation quantization'' given by the author during the fall semester 2020 at the University of Zurich. The first chapter is an introduction to differential geometry, where we cover manifolds, tensor fields, integration on manifolds, Stokes' theorem, de Rham's theorem and Frobenius' theorem. The second chapter covers the most important notions of symplectic geometry such as Lagrangian submanifolds, Weinstein's tubular neighborhood theorem, Hamiltonian mechanics, moment maps and symplectic reduction. The third chapter gives an introduction to Poisson geometry where we also cover Courant structures, Dirac structures, the local splitting theorem, symplectic foliations and Poisson maps. The fourth chapter is about deformation quantization where we cover the Moyal product, $L_\infty$-algebras, Kontsevich's formality theorem, Kontsevich's star product construction through graphs, the globalization approach to Kontsevich's star product and the operadic approach to formality. The fifth chapter is about the quantum field theoretic approach to Kontsevich's deformation quantization where we cover functional integral methods, the Moyal product as a path integral quantization, the Faddeev--Popov and BRST method for gauge theories, infinite-dimensional extensions, the Poisson sigma model, the construction of Kontsevich's star product through a perturbative expansion of the functional integral quantization for the Poisson sigma model for affine Poisson structures and the general construction. Any comments and remarks can be send by email to the author.
\end{abstract}

\chapter*{Introduction}

Poisson geometry appears naturally in physics in the context of the dynamics for classical mechanical systems. Mathematically, it lies in the intersection of differential geometry and noncommutative algebra.
In physics, people are often interested in \emph{observables}, which mathematically can be described as elements of the algebra of smooth functions on some underlying manifold $M$ endowed with certain structure, and their dynamics. In fact, the algebra $C^\infty(M)$ can be endowed with an algebra structure given by a bracket $\{\enspace,\enspace\}$, called a \emph{Poisson bracket}, satisfying certain properties, similarly as in the case of a Lie algebra $\mathfrak{g}$ endowed with its Lie bracket $[\enspace,\enspace]$. Time evolution of an observable $O$ is described by the equation
\[
\frac{\dd O}{\dd t}=\{H,O\},
\]
where $H\in C^\infty(M)$ denotes the \emph{Hamiltonian} of the system.\\

This structure is actually motivated by the \emph{symplectic structure} appearing additionally within the natural manifold structure for the \emph{phase space} of a classical system. The phase space is a way of expressing the dynamics in a classical system where each state is represented by a unique point. The most simple phase space is given by $\R^2$ and in higher dimension the local structure looks like some $\R^{2n}$, for $n\geq 1$.
One can then induce a Poisson bracket $\{\enspace,\enspace\}$ from a closed, nondegenerate 2-form $\omega$, the \emph{symplectic form}. Such a 2-form also appears naturally in the structure of a phase space containing the information for the base coordinates $q^{i}$ and the corresponding fiber momentum-coordinates $p_i$ on the vector bundle given by the cotangent bundle $T^*M$ of the manifold $M$. Instead of looking at the structure sheaf of smooth functions $C
^\infty(M)$ on a manifold $M$ which is endowed with a Poisson bracket, one can also consider the geometric picture of the manifold endowed with a bivector field $\pi$ encoding the information of the bracket $\{\enspace,\enspace\}$ geometrically in an equivalent way.\\

A quantum system is described by a complex Hilbert space $\calH$ together with an operator $\widehat{H}$. A physical state of the system is then represented by an element in $\calH$ whereas the physical observables are given by self-adjoint operators on $\calH$. Denote this space by $L(\calH)$. Time evolution is then given by the \emph{Heisenberg equation}
\[
\frac{\dd \widehat{O}}{\dd t}=\frac{\I}{\hbar}\left[\widehat{H},\widehat{O}\right],
\]
where $[\enspace,\enspace]$ denotes the commutator of operators. The parameter $\hbar$ is called the \emph{reduced Planck constant} which is a physical constant naturally appearing at the quantum level.
The connection to classical mechanics is usually given by the introduction of position $\widehat{q_i}$ and momentum $\widehat{p_j}$ operators. They satisfy the following commutation relation:
\[
\left[\widehat{p_i},\widehat{q_j}\right]=\frac{\I}{\hbar}\delta_{ij}.
\]
In this way, one can obtain classical mechanics from the quantum theory in the limit where $\hbar\to 0$. This is a general concept. An important question is whether there is a precise mathematical formulation of such a quantization procedure in terms of a well-defined map  between classical objects and their quantum picture. There are different ways of quantizing a classical system. If one starts from the canonical quantization on $\R^{2n}$, one can consider the method of \emph{geometric quantization} (see e.g. \cite{Kir85,Wood97,Moshayedi2020_geomquant}). There the idea is the quantization of the classical phase space $\R^{2n}$ to the corresponding Hilbert space $\calH=L^2(\R^n)$ on which the \emph{Schr\"odinger equation} is defined. Another quantization approach focuses on the observables of the classical system. There one tries to capture the noncommutativity structure of the space of operators from the commutative structure of $C^\infty(\R^{2n})$. A result of \emph{Groenewold} \cite{Groenewold1946} states that it is impossible to quantize the Poisson algebra $C^\infty(\R^{2n})$ in a way where the Poisson bracket of two functions is sent onto the Lie bracket of the corresponding operators. To overcome this issue, one can instead consider a \emph{deformation} of the pointwise product on $C^\infty(\R^{2n})$ to a noncommutative product.\\

Deformation quantization originated from the work of \emph{Weyl} \cite{Weyl1931}, who gave an explicit formula for the operator $O_f$ on $L(\R^n)$ associated to a function $f\in C^\infty(\R^{2n})$:
\[
O_f:=\int_{\R^{2n}}\check{f}(\xi,\eta)\exp\left(\frac{\I}{\hbar}(P\xi+Q\eta)\right)\dd^n\xi\dd^n\eta,
\]
where $\check{f}$ denotes the inverse Fourier transform of $f$. Here $P=(P_i)$ and $Q=(Q_j)$ denote operators satisfying the canonical commutation relation. Moreover, the integral is considered in the weak sense. An inverse map was later found by \emph{Wigner} \cite{Wigner1932}, who gave a way to recover the corresponding classical observable by taking the \emph{symbol} of the operator. \emph{Moyal} \cite{Moy} interpreted the symbol of the commutator of two operators corresponding to the functions $f$ and $g$ as what is today called the \emph{Moyal bracket}:
\[
\calM(f,g):=\frac{\sinh(\epsilon P)}{\epsilon}(f,g)=\sum_{k=0}^\infty\frac{\epsilon^{2k}}{(2k+1)!}P^{2k+1}(f,g),
\]
where $\epsilon:=\frac{\I\hbar}{2}$ and $P^k$ is the $k$-th power of the Poisson bracket on $C^\infty(\R^{2n})$. Already Groenewold had a similar formula for the symbol of a product $O_fO_g$ which today can be interpreted as the first appearance of the Moyal star product $\star$. The Moyal bracket can then be rewritten in terms of this star product as
\[
\calM(f,g)=\frac{1}{2\epsilon}(f\star g-g\star f).
\]
It was \emph{Flato} who recognized this star product as a deformation of the commutative pointwise product on $C^\infty(\R^{2n})$. This was the beginning of \emph{deformation quantization}. He conjectured the problem of giving a general recipe to deform the pointwise product on $C^\infty(M)$ in such a way that $\frac{1}{2\epsilon}(f\star g-g\star f)$ still remains a deformation of the Poisson structure on $M$. Following this conjecture, a first way of formulating quantum mechanics as such a deformation of classical mechanics had been discovered. The work of \emph{Bayen}, \emph{Flato}, \emph{Fronsdal}, \emph{Lichnerowicz} and \emph{Sternheimer} was essential for the formulation of the deformation problem for symplectic spaces and the physical applications \cite{BFFLS1,BFFLS2}.
\emph{DeWilde} and \emph{Lecomte} \cite{DeWildeLecomte1983} have proven the existence of a star product on a generic symplectic manifold by using \emph{Darboux's theorem} which tells that locally any symplectic manifold of dimension $2n$ can be identified with $\R^{2n}$ by a choice of \emph{Darboux charts}. Using cohomological arguments, one can construct such a star product by a correct gluing of the locally defined Moyal product. Independently of the previous result, \emph{Fedosov} \cite{Fedosov1994} gave an explicit construction of a star product on a symplectic manifold. 
The generalization to any Poisson manifold was given through the formality theorem of \emph{Kontsevich} \cite{K}. He derived an explicit formula for the product on $\R^d$ by using special graphs and configuration space integrals, which can be used to define it locally on any Poisson manifold $M$. The globalization procedure was described by \emph{Cattaneo, Felder} and \emph{Tomassini} \cite{CFT,CattaneoFelderTomassini2002} by extending Fedosov's construction to the Poisson case where they used notions of \emph{formal geometry} developed by \emph{Gelfand--Fuks} \cite{GelfandFuks1969,GelfandFuks1970}, \emph{Gelfand--Kazhdan} \cite{GK} and \emph{Bott} \cite{B}, which completed the program of Flato proposed thirty years before. Another approach to Kontsevich's result was given by \emph{Tamarkin} \cite{Tamarkin1998,Tamarkin2003} who used the notion of \emph{operads} in order to prove the formality theorem.\\

Moreover, \emph{Cattaneo} and \emph{Felder} \cite{CF1} have shown that Kontsevich's formula can actually be formulated in terms of a perturbative expression of the functional integral of a topological field theory. This theory is called the \emph{Poisson sigma model}. The Poisson sigma model was discovered independently by \emph{Ikeda} \cite{I} and \emph{Schaller--Strobl} \cite{SS1,SS2} by an attempt to combine 2-dimensional gravity with Yang--Mills theories.
In fact, one can show that the graphs, which have been constructed by Kontsevich, arise naturally in this context as the Feynman diagrams subject to the expectation of certain observables. Recently, \emph{Cattaneo, Moshayedi} and \emph{Wernli} gave a field-theoretic picture to obtain a global deformation quantization for constant Poisson structures (Moyal product) by cutting and gluing techniques for certain worldsheet manifolds \cite{CMW2}. This is done by using a globalized version of the Poisson sigma model and methods for gauge theories on manifolds with boundary developed in \cite{CMW3}. It is expected that the cutting and gluing construction can be extended for general Poisson structures.\\

An approach regarding higher gauge theories is encoded in the setting of \emph{higher shifted symplectic and Poisson structures}. A way of dealing with these concepts and constructing a higher shifted deformation quantization in the setting of \emph{derived algebraic geometry} was formulated by \emph{Calaque, Pantev, To\"en, Vaqui\'e} and \emph{Vezzosi} \cite{PantevToenVaquieVezzosi2013,CalaquePantevToenVaquieVezzosi2017}.

\vspace{2cm}

\textbf{Acknowledgements}
I would like to thank A. S. Cattaneo for comments on these notes. This research was supported by the NCCR SwissMAP, funded by the Swiss National Science Foundation, and by the SNF grant No. 200020\_192080.

\tableofcontents

\chapter{Foundations of Differential Geometry}
We want to start by introducing the main concepts and notions of differential geometry which are needed in order to be able to understand the discussions in the following chapters. The experienced reader might skip this chapter, but we will still refer to some parts of it in other chapters and one can always come back in case there are any uncertainties. In this chapter we cover differentiable manifolds, vector bundles, (multi)vector fields, differential forms, general tensor fields, integration on manifolds, Stokes' theorem for manifolds, Stokes' theorem for chains, de Rham cohomology, singular homology, de Rham's theorem, distributions and Frobenius' theorem.  
This chapter is mainly based on \cite{BottTu,Cattaneo_Manifolds,Lee}.

\section{Differentiable manifolds}

\subsection{Charts and atlases}

\begin{defn}[Chart]
A \emph{chart} on a set $M$ is a pair $(U,\phi)$ where $U\subset M$ is a subset and $\phi$ is an injective map $U\to \R^n$ for some $n$.
\end{defn}

We call $\phi$ a \emph{chart map} or a \emph{coordinate map}. Sometime we also refer only to $\phi$ as a chart since $U$ is already contained inside the definition of $\phi$. Let $(U,\phi_U)$ and $(V,\phi_V)$ be charts on $M$. Then we can compose the bijections $(\phi_U)\vert_{U\cap V}\colon U\cap V\to \phi_U(U\cap V)$ and $(\phi_V)\vert_{U\cap V}\colon U\cap V\to \phi_V(U\cap V)$ to the bijection 
\[
\phi_{U,V}\colon (\phi_U)\vert_{U\cap V}\circ ((\phi_V)\vert_{U\cap V})^{-1}\colon \phi_{U}(U\cap V)\to \phi_V(U\cap V).
\]
We call this the \emph{transition map} between the charts $(U,\phi_U)$ and $(V,\phi_V)$. Moreover, we refer to $n$ as being the \emph{dimension} of $M$ (we will give another definition for the dimension later on). See Figure \ref{fig:charts} for a visualization.

\begin{defn}[Atlas]
An atlas on a set $M$ is a collection of charts $\{(U_\alpha,\phi_\alpha)\}_{\alpha\in I}$, where $I$ is an index set such that $\bigcup_{\alpha\in I}U_\alpha=M$.
\end{defn}

\begin{rem}
We will denote the transition maps between two charts $(U_\alpha,\phi_\alpha)$ and $(U_\beta,\phi_\beta)$ simply by $\phi_{\alpha\beta}$.
\end{rem}

If $\phi_\alpha(U_\alpha)$ is open for all $\alpha\in I$, then the atlas $\calA=\{(U_\alpha,\phi_\alpha)\}_{\alpha\in I}$ induces a topology on $M$. This topology is given by
\[
\calO_\calA(M):=\{V\subset M\mid \phi_{\alpha}(V\cap U_\alpha)\text{ is open }\forall \alpha\in I\}.
\]

\begin{defn}[Open atlas]
An atlas is said to be \emph{open} if $\phi_{\alpha}(U_\alpha\cap U_\beta)$ is open for all $\alpha,\beta\in I$.
\end{defn}

\begin{defn}[Differentiable atlas]
An atlas is said to be \emph{differentiable} if it is open and all transition functions are $C^k$-maps for $k=0,1,\ldots,\infty$.
\end{defn}

\begin{defn}[Smooth atlas]
An atlas is said to be \emph{smooth} if it is open and all transition functions are $C^\infty$-maps.
\end{defn}

\begin{defn}[$C^k$-equivalence]
Two $C^k$-atlases on the same set are said to be \emph{$C^k$-equivalent} if their union is a $C^k$-atlas for $k=0,1,\ldots,\infty$.
\end{defn}

\begin{figure}
\begin{center}
\begingroup%
  \makeatletter%
  \providecommand\color[2][]{%
    \errmessage{(Inkscape) Color is used for the text in Inkscape, but the package 'color.sty' is not loaded}%
    \renewcommand\color[2][]{}%
  }%
  \providecommand\transparent[1]{%
    \errmessage{(Inkscape) Transparency is used (non-zero) for the text in Inkscape, but the package 'transparent.sty' is not loaded}%
    \renewcommand\transparent[1]{}%
  }%
  \providecommand\rotatebox[2]{#2}%
  \newcommand*\fsize{\dimexpr\f@size pt\relax}%
  \newcommand*\lineheight[1]{\fontsize{\fsize}{#1\fsize}\selectfont}%
  \ifx\svgwidth\undefined%
    \setlength{\unitlength}{200bp}%
    \ifx\svgscale\undefined%
      \relax%
    \else%
      \setlength{\unitlength}{\unitlength * \real{\svgscale}}%
    \fi%
  \else%
    \setlength{\unitlength}{\svgwidth}%
  \fi%
  \global\let\svgwidth\undefined%
  \global\let\svgscale\undefined%
  \makeatother%
  \begin{picture}(1,1.41428571)%
    \lineheight{1}%
    \setlength\tabcolsep{0pt}%
    \put(0,0){\includegraphics[width=9cm]{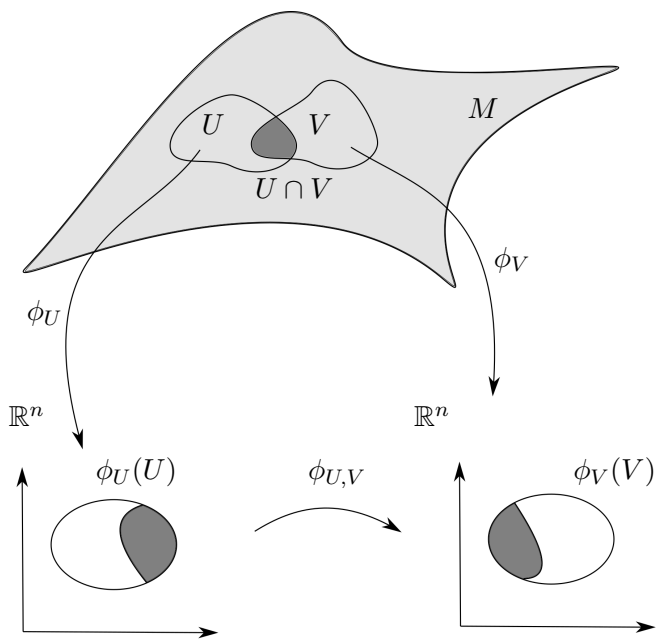}}%
    \put(0.9,0.98){\color[rgb]{0,0,0}\makebox(0,0)[lt]{\lineheight{1.25}\smash{\begin{tabular}[t]{l}$M$\end{tabular}}}}%
    \put(0.03601619,0.4){\color[rgb]{0,0,0}\makebox(0,0)[lt]{\lineheight{1.25}\smash{\begin{tabular}[t]{l}$\mathbb{R}^n$\end{tabular}}}}%
    \put(0.8,0.4){\color[rgb]{0,0,0}\makebox(0,0)[lt]{\lineheight{1.25}\smash{\begin{tabular}[t]{l}$\mathbb{R}^n$\end{tabular}}}}%
    \put(0.95,0.7){\color[rgb]{0,0,0}\makebox(0,0)[lt]{\lineheight{1.25}\smash{\begin{tabular}[t]{l}$\phi_V$\end{tabular}}}}%
    \put(0.07,0.6){\color[rgb]{0,0,0}\makebox(0,0)[lt]{\lineheight{1.25}\smash{\begin{tabular}[t]{l}$\phi_U$\end{tabular}}}}%
    \put(0.4,0.95){\color[rgb]{0,0,0}\makebox(0,0)[lt]{\lineheight{1.25}\smash{\begin{tabular}[t]{l}$U$\end{tabular}}}}%
    \put(0.6,0.95){\color[rgb]{0,0,0}\makebox(0,0)[lt]{\lineheight{1.25}\smash{\begin{tabular}[t]{l}$V$\end{tabular}}}}%
    \put(0.2,0.3){\color[rgb]{0,0,0}\makebox(0,0)[lt]{\lineheight{1.25}\smash{\begin{tabular}[t]{l}$\phi_U(U)$\end{tabular}}}}%
    \put(1.1,0.3){\color[rgb]{0,0,0}\makebox(0,0)[lt]{\lineheight{1.25}\smash{\begin{tabular}[t]{l}$\phi_V(V)$\end{tabular}}}}%
    \put(0.5,0.83){\color[rgb]{0,0,0}\makebox(0,0)[lt]{\lineheight{1.25}\smash{\begin{tabular}[t]{l}$U\cap V$\end{tabular}}}}%
    \put(0.6,0.3){\color[rgb]{0,0,0}\makebox(0,0)[lt]{\lineheight{1.25}\smash{\begin{tabular}[t]{l}$\phi_{U,V}$\end{tabular}}}}%
  \end{picture}%
\endgroup%
\end{center}
\caption{Example of charts and transition map on an $n$-dimensional manifold $M$.}
\label{fig:charts}
\end{figure}

\begin{defn}[$C^k$-manifold]
A \emph{$C^k$-manifold} is an equivalence class of $C^k$-atlases for $k=0,1,\ldots,\infty$.
\end{defn}

\begin{defn}[Smooth manifold]
A \emph{smooth manifold} is an equivalence class of $C^\infty$-atlases.
\end{defn}

\begin{rem}
From now on, if we call a set $M$ a manifold, we will always mean a smooth manifold. Moreover, all maps between manifolds will be regarded as smooth maps. We will call a map of manifolds a \emph{diffeomorphism}, if it is invertible with smooth inverse.
\end{rem}

\begin{defn}[Submanifold]
Let $N$ be an $n$-dimensional manifold. A $k$-dimensional \emph{submanifold}, with $k\leq n$, is a subset $M$ of $N$ such that there is an atlas $\{(U_\alpha,\phi_\alpha)\}_{\alpha\in I}$ of $N$ with the property that for all $\alpha$ with $U_\alpha\cap M\not=\varnothing$ we have $\phi_\alpha(U_\alpha\cap M)=W_\alpha\times\{x\}$ with $W_\alpha\subset \R^k$ an open subset and $x\in \R^{n-k}$.    
\end{defn}

\begin{rem}
Any chart with this property is called an \emph{adapted chart} and an atlas consisting of adapted charts is called an \emph{adapted atlas}. Moreover, by a diffeomorphism of $\R^n$ we can always assume that $x=0$.
\end{rem}

\begin{ex}[Graphs]
Let $F$ be a smooth map from an open subset $V\subset \R^k$ to $\R^{n-k}$ and consider its graph 
\[
M=\{(x,y)\in V\times \R^{n-k}\mid y=F(x)\}.
\]
Then $M$ is a submanifold of $N:=V\times \R^{n-k}$. As an adapted atlas we may take the one consisting of a single chart $(N,\iota)$, where $\iota\colon N\hookrightarrow \R^n$ denotes the inclusion map. 
\end{ex}

\subsection{Pullback and pushforward}
Let $M$ and $N$ be two manifolds and consider a map $F\colon M\to N$. 

\begin{defn}[Pullback]
The $\R$-linear map 
\begin{align*}
F^*\colon C^\infty(N)&\to C^\infty(M),\\
f&\mapsto f\circ F.
\end{align*}
is called \emph{pullback} by $F$. 
\end{defn}

\begin{exe}
Show that for $f,g\in C^\infty(N)$ we have 
\[
F^*(fg)=F^*(f)F^*(g).
\]
Moreover, if $G\colon N\to Z$ is a map between manifolds $N$ and $Z$, show that 
\[
(G\circ F)^*=F^*\circ G^*.
\]
\end{exe}

\begin{defn}[Pushforward]
Using $F$ as before, we define the \emph{pushforward} to be the inverse of the pullback $F^*$ which we denote by $F_*$. In fact, we get 
\begin{align*}
F_*\colon C^\infty(M)&\to C^\infty(N),\\
f&\mapsto f\circ F^{-1}.
\end{align*}
\end{defn}

\begin{exe}
Show that 
\[
F_*(fg)=F_*(f)F_*(g),
\]
and 
\[
(G\circ F)_*=G_*\circ F_*.
\]

\end{exe}

\subsection{Tangent space}

Let $M$ be a manifold. 

\begin{defn}[coordinatized tangent vector]
A \emph{coordinatized tangent vector} at $q\in M$ is a triple $(U,\phi_U,v)$ where $(U,\phi_U)$ is a chart with $U\ni q$ and $v$ is an element of $\R^n$. 
\end{defn}

We say that two coordinatized tangent vectors $(U,\phi_U,v)$ and $(V,\phi_V,w)$ are equivalent if 
\[
w=\dd_{\phi_U(q)}\phi_{U,V}v.
\]

\begin{defn}[Tangent vector]
A \emph{tangent vector} at $q\in M$ is an equivalence class of coordinatized tangent vectors at $q$.
\end{defn}

\begin{defn}[Tangent space]
The \emph{tangent space} of $M$ at $q\in M$ is given by the set of all tangent vectors at $q$.
\end{defn}

Note that each chart $(U,\phi_U)$ at $q\in M$ defines a bijection of sets
\begin{align*}
    \Phi_{q,U}\colon T_qM&\to \R^n,\\
    [(U,\phi_U,v)]&\mapsto v.
\end{align*}

We will also just write $\Phi_U$ when the base point $q\in M$ is understood. This bijection allows us to transfer the vector space structure of $\R^n$ to $T_qM$ and gives $\Phi_{q,U}$ the structure of a linear isomorphism. 

\begin{lem}
$T_qM$ has a canonical vector space structure for which $\Phi_{q,U}$ is a linear isomorphism for every chart $(U,\phi_U)$ containing $q$.
\end{lem}

\begin{proof}
Given a chart $(U,\phi_U)$, the bijection $\Phi_U$ defines the linear structure 
\begin{align*}
    \lambda\cdot_U[(U,\phi_U,v)]&=[(U,\phi_U,\lambda v)],\qquad \forall \lambda\in \R,\\
    [(U,\phi_U,v)]+_U[(U,\phi_U,v')]&=[(U,\phi_U,v+v')],\qquad \forall v,v'\in \R^n 
\end{align*}
If $(V,\phi_V)$ is another chart, we have 
\begin{align*}
    \lambda\cdot_U[(U,\phi_U,v)]&= [(U,\phi_U,\lambda v)]\\
    &=[(V,\phi_V,\dd_{\phi_U(q)}\phi_{U,V}\lambda v)]\\
    &=[(V,\phi_V,\lambda\dd_{\phi_U(q)}\phi_{U,V} v)]\\
    &=\lambda\cdot_V[(V,\phi_V,\dd_{\phi_U(q)}\phi_{U,V} v)]\\
    &=\lambda\cdot_V[(U,\phi_U,v)],
\end{align*}
and so $\cdot_U=\cdot_V$. Similarly, we have
\begin{align*}
    [(U,\phi_U,v)]+_U[(U,\phi_U,v')]&=[(U,\phi_U,v+v')]\\
    &=[(V,\phi_V,\dd_{\phi_U(q)}\phi_{U,V} (v+v')]\\
    &=[(V,\phi_V,\dd_{\phi_U(q)}\phi_{U,V}v+\dd_{\phi_U(q)}\phi_{U,V}v']\\
    &=[(V,\phi_V,\dd_{\phi_U(q)}\phi_{U,V}v)]+_V[(V,\phi_V,\dd_{\phi_U(q)}\phi_{U,V}v')]\\
    &=[(U,\phi_U,v)]+_V[(U,\phi_U,v')],
\end{align*}
and so $+_U=+_V$.
\end{proof}

\begin{defn}[Dimension of a manifold]
We define the \emph{dimension} of a manifold $M$ by
\[
\dim M:=\dim T_qM,\quad q\in M.
\]
\end{defn}

Let $F\colon M\to N$ be a differentiable map between two manifolds $M$ and $N$. Given a chart $(U,\phi_U)$ of $M$ containing $q\in M$ and a chart $(V,\psi_V)$ of $N$ containing $F(q)\in N$, we have a linear map 
\[
\dd_q^{U,V}F:= \Phi^{-1}_{F(q),V}\dd_{\phi_U(q)}F_{U,V}\Phi_{q,U}\colon T_qM\to T_{F(q)}N.
\]

\begin{lem}[Differential/Tangent map]
The map $\dd_q^{U,V}F$ does not depend on the choice of charts, so we get a canonically defined linear map 
\[
\dd_qF\colon T_qM\to T_{F(q)}N,
\]
called the \emph{differential} (or \emph{tangent map}) of $F$ at $q\in M$.
\end{lem}

\begin{proof}
Let $(U',\phi_{U'})$ be another chart of $M$ containing $q\in M$ and $(V',\psi_{V'})$ another chart of $N$ containing $F(q)\in N$. Then 
\begin{align*}
    \dd_q^{U,V}F[(U,\phi_U,v)]&=[(V,\psi_V,\dd_{\phi_U(q)}F_{U,V}v)]\\
    &=[(V',\psi_{V'},\dd_{\psi_{V'}(F(q))}\psi_{V,V'}\dd_{\phi_U(q)}F_{U,V}v)]\\
    &=[(V',\psi_{V'},\dd_{\phi_{U'}(q)}F_{U',V'}(\dd_{\phi_U(q)}\phi_{U,U'})^{-1}v)]\\
    &=\dd^{U',V'}_qF[(U',\phi_{U'},(\dd_{\phi_U(q)}\phi_{U,U'})^{-1}v)]\\
    &=\dd^{U',V'}_qF[(U,\phi_{U},v)],
\end{align*}
so we get $\dd^{U,V}_qF=\dd^{U',V'}_qF$. 
\end{proof}

We immediately get the following lemma.

\begin{lem}
Let $F\colon M\to N$ and $G\colon N\to Z$ be maps between manifolds $M,N,Z$. Then 
\[
\dd_q(G\circ F)=\dd_{F(q)}G\circ \dd_qF,\quad\forall q\in M.
\]
\end{lem}

\section{Vector fields and differential $1$-forms}
\label{sec:vector_fields_and_differential_1-forms}

\subsection{Tangent bundle}

We can glue all the tangent spaces of a manifold $M$ together and obtain the following definition.

\begin{defn}[Tangent bundle]
The \emph{tangent bundle} of a manifold $M$ is given by 
\[
TM:=\bigsqcup_{q\in M}T_qM.
\]
\end{defn}

An element of $TM$ is of the form $(q,v)$, where $q\in M$ and $v\in T_qM$. Consider the surjective map $\pi\colon TM\to M$, $(q,v)\mapsto q$. Then the \emph{fiber} $T_qM$ can be denoted by $\pi^{-1}(q)$.

\begin{prop}
$TM$ is a manifold.
\end{prop}

\begin{proof}
Let $\{(U_\alpha,\phi_\alpha)\}_{\alpha\in I}$ be an atlas in the equivalence class defining $M$. We set $\hat{U}_\alpha :=\pi^{-1}(U_\alpha)$ and 
\begin{align}
\begin{split}
    \hat{\phi}_\alpha\colon \hat{U}_\alpha&\to \R^n\times \R^n,\\
    (q,v)&\mapsto (\phi_\alpha(q),\Phi_{q,U_\alpha}v).
\end{split}
\end{align}
Note that the chart maps are linear in the fibers. The transition maps are then given by 
\[
\hat{\phi}_{\alpha\beta}(x,w)=(\phi_{\alpha\beta}(x),\dd_x\phi_{\alpha\beta}w).
\]
We can then define the tangent bundle of $M$ to be the equivalence class of the atlas $\{(\hat{U}_\alpha,\hat{\phi}_\alpha)\}_{\alpha\in I}$.
\end{proof}

\begin{rem}
Note that $\dim TM=2\dim M$.
\end{rem}

\subsection{Vector bundles}

\begin{defn}[Vector bundle]
A \emph{vector bundle} of rank $r$ over a manifold $M$ of dimension $n$ is a manifold $E$ together with a surjection $\pi\colon E\to M$ such that
\begin{enumerate}
    \item $E_q:=\pi^{-1}(q)$ is an $r$-dimensional vector space for all $q\in M$.
    \item $E$ possesses an atlas of the form $\{(\tilde{U}_\alpha,\tilde{\phi}_\alpha)\}_{\alpha\in I}$ with $\tilde{U}_\alpha=\pi^{-1}(U_\alpha)$ for an atlas $\{(U_\alpha,\phi_\alpha)\}_{\alpha\in I}$ of $M$ and 
    \begin{align}
    \begin{split}
        \tilde{\phi}_\alpha\colon \tilde{U}_\alpha&\to \R^n\times\R^r\\
        (q,v\in E_q)&\mapsto (\phi_\alpha(q),A_\alpha(q)v),
    \end{split}
    \end{align}
    where $A_\alpha(q)$ is a linear isomorphism for all $q\in U_\alpha$.
\end{enumerate}

\end{defn}



\begin{rem}
We usually call $E$ the \emph{total space} and $M$ the \emph{base space}. Moreover, we usually call $E_q:=\pi^{-1}(q)$ the \emph{fiber} at $q\in M$. 
\end{rem}

The maps
\begin{align}
\begin{split}
A_{\alpha\beta}\colon U_{\alpha}\cap U_\beta&\to \End(\R^r),\\
q&\mapsto A_{\alpha\beta}(q):=A_\beta(q)A_\alpha(q)^{-1}\colon \R^r\to \R^r.
\end{split}
\end{align}
are smooth for all $\alpha,\beta\in I$. The transition maps
\[
\tilde{\phi}_{\alpha\beta}(x,u)=(\phi_{\alpha\beta}(x),A_{\alpha\beta}(\phi^{-1}_\alpha(x))u)
\]
are linear in the second factor $\R^r$.

\begin{ex}[Tangent bundle]
The \emph{tangent bundle} $TM$ of a manifold $M$ is an example of a vector bundle. 
\end{ex}

\begin{ex}[Trivial bundle]
Let $M$ be a manifold. Then we can define a vector bundle of rank $n$ with total space $E=M\times \R^n$ and base space $M$. Note that $\pi\colon E\to M$ is the projection onto the first factor. We call this a \emph{trivial bundle} over $M$.
\end{ex}

\begin{ex}[M\"obius band]
\label{ex:Mobius_band}
We can view the M\"obius band $E$ as a vector bundle of rank 1 over the circle $S^1$. Each fiber at a point $x\in S^1$ is given in the form $U\times \R$, where $U\subset S^1$ is an open arc on the circle including $x$. Note that the total space $E$ is not given by the trivial bundle $S^1\times \R$ which would be the cylinder.
\end{ex}

\begin{defn}[Line bundle]
We call a vector bundle a \emph{line bundle} if it has rank 1.
\end{defn}

\begin{rem}
Example \ref{ex:Mobius_band} is an example of a line bundle.
\end{rem}

\begin{defn}[Morphism of vector bundles]
Let $(E_1,M_1,\pi_1)$ and $(E_2,M_2,\pi_2)$ be two vector bundles. A \emph{morphism} between $(E_1,M_1,\pi_1)$ and $(E_2,M_2,\pi_2)$ is given by a pair of continuous maps $f\colon E_1\to E_2$ and $g\colon M_1\to M_2$ such that $g\circ \pi_1=\pi_2\circ f$ and for all $q\in M_1$, the map $\pi_1^{-1}(q)\to \pi_2^{-1}(g(q))$ induced by $f$ is a homomorphism of vector spaces.
\end{defn}

\begin{defn}[Section]
Let $\pi\colon E\to M$ be a vector bundle. A \emph{section} on an open subset $U\subset M$ is a continuous map $\sigma\colon U\to E$ such that $\pi\circ \sigma=\id_U$.
\end{defn}

\begin{rem}
We denote the space of sections of a vector bundle $E\to M$ by $\Gamma(E)$.
\end{rem}

\subsection{Vector fields}

\begin{defn}[Vector field]
Let $M$ be a manifold. A section $X\in \Gamma(TM)$ of the tangent bundle $TM$ is called a \emph{vector field}. 
\end{defn}

In an atlas $\{(U_\alpha,\phi_\alpha)\}_{\alpha\in I}$ of $M$ and the corresponding atlas $\{(\tilde{U}_\alpha,\tilde{\phi}_\alpha)\}_{\alpha\in I}$ of $TM$, a vector field $X$ is represented by a collection of smooth maps $X_\alpha\colon \phi_\alpha(U_\alpha)\to \R^n$. All these maps are related by 
\[
X_\beta(\phi_{\alpha\beta}(x))=\dd_x\phi_{\alpha\beta}X_\alpha(x),\quad \forall \alpha,\beta\in I,\forall x\in \phi_\alpha(U_\alpha\cap U_\beta).
\]

\begin{rem}
The vector at $q\in M$ defined by the vector field $X$ is usually denoted by $X_q$ as well as by $X(q)$. The latter notation is often avoided as one may apply a vector field to a function $f$, and in this case the standard notation is $X(f)$. One also uses $X_\alpha$ to denote the representation of $X$ in the chart with index $\alpha$, but this should not create confusion with the notation $X_q$ for $X$ at the point $q$.
\end{rem}

We can add and multiply vector fields. Let $X,Y\in \Gamma(TM)$ be two vector fields on $M$, $\lambda\in \R$, and $f\in C^\infty(M)$. Then 
\begin{align}
(X+Y)_q&:=X_q+Y_q,\quad \forall q\in M\\
(\lambda X)_q&:=\lambda X_q,\quad \forall q\in M\\
(fX)_q&:=f(q)X_q,\quad \forall q\in M.
\end{align}

\begin{rem}
We denote the space of vector fields on a manifold $M$ by $\mathfrak{X}(M):=\Gamma(TM)$.
\end{rem}

Let $F\colon M\to N$ be a diffeomorphism between two manifolds. If $X$ is a vector field on $M$, then $\dd_q FX_q$ is a vector in $T_{F(q)}N$ for each $q\in M$. If $F$ is a diffeomorphism, we can perform this construction for each $y\in N$, by setting $q=F^{-1}(y)$, and define a vector field, denoted by $F_*X$ on $N$:
\[
(F_*X)_{F(q)}:=\dd_q FX_q,\quad \forall q\in M,
\]
or, equivalently, 
\[
(F_*X)_y=\dd_{F^{-1}(y)}FX_{F^{-1}(y)},\quad\forall y\in N.
\]

\begin{defn}[Push-forward of a vector field]
For a map of manifolds $F\colon M\to N$, the $\R$-linear map $F_*\colon \mathfrak{X}(M)\to \mathfrak{X}(N)$ is called the \emph{push-forward} of vector fields.
\end{defn}

\begin{rem}
Note that if $G\colon N\to Z$ is another diffeomorphism between manifolds, we immediately have 
\[
(G\circ F)_*=G_*F_*.
\]
Moreover, we have $(F_*)^{-1}=(F^{-1})_*$.
\end{rem}

If $U\subset \R^n$ is an open subset, $X$ a smooth vector field and $f$ a smooth map, then we can define 
\[
X(f)=\sum_{i=1}^nX^i\frac{\de f}{\de x^i},
\]
where on the right hand side we regard $X$ as a map $U\to \R^n$. Each $X^i$ is given by a map in $C^\infty(M)$.
Note that the map $C^\infty(M)\to C^\infty(M)$, $f\mapsto X(f)$, is $\R$-linear and satisfies the Leibniz rule
\[
X(fg)=X(f)g+fX(g),\quad\forall f,g\in C^\infty(M).
\]
This means, choosing local coordinates $(x^i)$ on $M$, we can write a vector field as 
\[
X=\sum_{i=1}^nX^i\frac{\de}{\de x^i}\in \mathfrak{X}(M).
\]

\begin{defn}[Commutator of vector fields]
\label{defn:commutator}
One can define an anti-symmetric $\R$-bilinear map on $\mathfrak{X}(M)$ by
\begin{align}
\begin{split}
    [\enspace,\enspace]\colon \mathfrak{X}(M)\times \mathfrak{X}(M)&\to \mathfrak{X}(M)\\
    (X,Y)&\mapsto[X,Y]:=XY-YX,
\end{split}
\end{align}
called the \emph{commutator} of vector fields.
\end{defn}

\begin{exe}[Jacobi identity]
Show that for all $X,Y,Z\in \mathfrak{X}(M)$
\[
[X,[Y,Z]]+[Y,[Z,X]]+[Z,[X,Y]]=0.
\]
\end{exe}

\begin{exe}[Derivation]
Show that $[\enspace,\enspace]$ is a derivation, i.e. for $X,Y\in \mathfrak{X}(M)$ and $f\in C^\infty(M)$ we have 
\[
[X,fY]=X(f)Y+f[X,Y].
\]
\end{exe}

\subsection{Flow of a vector field}
To a vector field $X\in \mathfrak{X}(M)$ we associate the ODE 
\begin{equation}
\label{eq:ODE}
\dot{q}=X(q).
\end{equation}

\begin{defn}[Integral curve]
A solution of \eqref{eq:ODE} is called an \emph{integral curve}. It is given by a path $q\colon I\to M$ such that $\dot{q}(t)=X(q(t))\in T_{q(t)}M$ for all $t\in I$. Here, $I:=[a,b]\subset \R$ denotes some interval.  
\end{defn}

\begin{defn}[Maximal integral curve]
An integral curve is called \emph{maximal} if it is not the restriction of a solution to a proper subset of its domain.
\end{defn}

\begin{defn}[Flow]
\label{defn:flow}
Let $M$ be a manifold and let $X\in\mathfrak{X}(M)$ be a vector field. The \emph{flow} $\Phi^X_t$ of $X$ is given as follows: For $q\in M$ and $t$ in a neighborhood of $0$, $\Phi^X_t(q)$ is the unique solution at time $t$ to the Cauchy problem with initial condition at $q\in M$. Explicitly, 
\begin{subequations}
\begin{empheq}[box=\widefbox]{align}
\frac{\de}{\de t}\Phi^X_t(q)&=X(\Phi_t^X(q))\nonumber\\
\Phi^X_0(q)&=q\nonumber
\end{empheq}
\end{subequations}
\end{defn}

\begin{rem}
One can then show that 
\begin{equation}
\label{eq:flow}
\Phi^X_{t+s}(q)=\Phi_t^X(\Phi^X_s(q)),\quad \forall q\in M
\end{equation}
and for all $t,s\in \R$ such that the flow is defined.
\end{rem}

\begin{defn}[Complete vector field]
A vector field is called \emph{complete} if all its integral curves exist for all $t\in\R$.
\end{defn}

\begin{defn}[Global flow]
The flow of a complete vector field $X\in \mathfrak{X}(M)$ is a diffeomorphism $\Phi^X_t\colon M\to M$ for all $t\in \R$. We then call $\Phi^X_t$ the \emph{global flow} of $X$.
\end{defn}

\begin{rem}
The condition \eqref{eq:flow} for a global flow can then be written more compactly as 
\[
\Phi^X_{t+s}=\Phi^X_t\circ \Phi^X_s
\]
\end{rem}

\subsection{Cotangent bundle}

If $E$ is a vector bundle over $M$, then the union of the dual spaces $E_q^*$ is also a vector bundle, called the \emph{dual bundle} of $E$. Namely, let $E^*=\bigsqcup_{q\in M}E_q^*$. We denote an element of $E^*$ as a pair $(q,\omega)$ with $\omega\in E_q^*$. We let $\tilde{\pi}\colon E^*\to M$ to be such that $\tilde{\pi}(q,\omega)=q$ (the projection onto the first factor). To an atlas $\{(\tilde{U}_\alpha,\tilde{\phi}_\alpha)\}_{\alpha\in I}$ of $E$ we associate the atlas $\{(\hat{U}_\alpha,\hat{\phi}_\alpha)\}_{\alpha\in I}$ of $E^*$ with $\hat{U}_\alpha=\tilde{\pi}^{-1}(U_\alpha)=\bigsqcup_{q\in M}E^*_q$ and 
\begin{align}
\begin{split}
    \hat{\phi}_\alpha\colon \hat{U}_\alpha&\to \R^n\times (\R^r)^*\\
    (q,\omega\in E^*_q)&\mapsto (\phi_\alpha(q),(A_\alpha(q)^*)^{-1}\omega),
\end{split}
\end{align}
where we regard $(\R^r)^*$ as the manifold $\R^r$ with its standard structure. 

\begin{defn}[Cotangent bundle]
The dual bundle of the tangent bundle $TM$ of a manifold $M$, denoted by $T^*M$ is called the \emph{cotangent bundle} of $M$.
\end{defn}

\subsection{Differential 1-forms}
Similarly as we have defined vector fields as sections of the tangent bundle, we can construct a \emph{differential 1-form} to be a section of the cotangent bundle. We denote the space of 1-forms on a manifold $M$ by 
\[
\Omega^1(M):=\Gamma(T^*M).
\]

\begin{defn}[de Rham differential]
The \emph{de Rham differential} is defined as the map
\begin{align}
\begin{split}
    \dd\colon C^\infty(M)&\to \Omega^1(M)\\
    f&\mapsto \dd f
\end{split}
\end{align}
which is $\R$-linear and satisfies the Leibniz rule
\[
\dd(fg)=\dd fg+f\dd g,\quad \forall f,g\in C^\infty(M).
\]
\end{defn}

Note that a 1-form $\omega\in \Gamma(T^*M)$ is dual to a vector field $X\in \Gamma(TM)$, so there is a pairing between them. The pairing is often denoted by $\iota_X\omega$. The operation $\iota$ is often called the \emph{contraction}. 
If $V\subset \R^n$ is open, we can consider the differentials $\dd x^i$ of the coordinate functions $x^i$. Note that we have 
\[
\iota_{\de_j}\dd x^i=\frac{\de x^i}{\de x^j}=\delta^i_j,
\]
where $\de_j:=\frac{\de}{\de x^j}$. Thus we can write a 1-form as 
\[
\omega=\sum_{i=1}^n\omega_i\dd x^i,
\]
where $\omega_i\in C^\infty(M)$ are uniquely determined functions for $1\leq i\leq n$. If $f\in C^\infty(M)$, then 
\[
\dd f=\sum_{i=1}^n\de_if\dd x^i.
\]
One can define a new notion of derivative by using the notion of a flow of a vector field. We define the \emph{Lie derivative} of a 1-form by 
\[
L_X\omega:=\frac{\de}{\de t}\bigg|_{t=0}(\Phi^X_{-t})_*\omega=\frac{\de}{\de t}\bigg|_{t=0}(\Phi^X_t)^*\omega,\quad \forall X\in\mathfrak{X}(M),\forall \omega\in \Omega^1(M).
\]

\section{Tensor fields}

\subsection{Tensor bundle}
Let $V$ be a vector space over a field $K$. We define the $k$-th \emph{tensor power} by 
\[
V^{\otimes k}:=\underbrace{V\otimes\dotsm\otimes V}_{\text{$k$ times}}.
\]
By convention we have $V^{\otimes 0}:=K$. Moreover, note that $\dim V^{\otimes k}=(\dim V)^k$. 

\begin{defn}[Tensor]
We call an element of $V^{\otimes k}$ a tensor of order $k$. 
\end{defn}

Let $(e_i)_{i\in I}$ be a basis of $V$. Then $(e_{i_1}\otimes\dotsm \otimes e_{i_k})_{i_1,\ldots,i_k\in I}$ is a basis of $V^{\otimes k}$ and a tensor $T$ of order $k$ can be uniquely written by 
\[
T=\sum_{i_1,\ldots,i_k\in I}T^{i_1\dotsm i_k}e_{i_1}\otimes \dotsm \otimes e_{i_k},\quad T^{i_1\dotsm i_k}\in K,\forall i_1,\ldots, i_k\in I.
\]

\begin{defn}[Tensor algebra]
The \emph{tensor algebra} of a vector space $V$ is defined as 
\[
T(V):=\bigoplus_{k=0}^\infty V^{\otimes k}.
\]
\end{defn}

\begin{defn}[Tensor of type $(k,s)$]
We define a \emph{tensor of type $(k,s)$} for a vector space $V$ as an element of 
\[
T^k_s(V):=V^{\otimes k}\otimes (V^*)^{\otimes s}.
\]
\end{defn}

\begin{rem}
Tensors of type $(0,s)$ are called \emph{covariant} tensors of order $s$ and tensors of type $(k,0)$ are called \emph{contravariant} of order $k$.
\end{rem}

If we pick a basis $(e_i)_{i\in I}$ of $V$ and consider the dual basis $(e^{j})_{j\in I}$ of $V^*$. Then we get a basis of $T^k_s(V)$ as
\[
(e_{i_1}\otimes\dotsm \otimes e_{i_k}\otimes e^{j_1}\otimes \dotsm \otimes e^{j_s})_{i_1,\ldots,i_k,j_1,\ldots,j_s\in I}. 
\]
A tensor of type $(k,s)$ can then be uniquely written as 
\[
T=\sum_{i_1,\ldots,i_k,j_1,\ldots,j_s\in I}T^{i_1\dotsm i_k}_{j_1\dots j_s}e_{i_1}\otimes \dotsm \otimes e_{i_k}\otimes e^{j_1}\otimes\dotsm \otimes e^{j_s}.
\]

\begin{defn}[Tensor bundle]
If $E$ is a vector bundle over $M$, we define $T_s^k(E)$ as the vector bundle whose fiber at $q$ is $T^k_s(E_q)$. Namely, to an adapted atlas $\{(\tilde{U}_\alpha,\tilde{\phi}_\alpha)\}_{\alpha\in I}$ of $E$ over the trivializing atlas $\{(U_\alpha,\phi_\alpha)\}_{\alpha\in I}$ of $M$, we associate the atlas $\{(\hat{U}_\alpha,\hat{\phi}_\alpha)\}_{\alpha\in I}$ of $T^k_s(E)$ with $\hat{U}_\alpha=\tilde{\pi}^{-1}(U_\alpha)=\bigsqcup_{q\in U_\alpha}T^k_s(E_q)$ and 
\begin{align}
\begin{split}
    \hat{\phi}_\alpha\colon \hat{U}_\alpha&\to \R^n\times T^k_s(\R^r)\\
    (q,\omega\in T^k_s(E_q))&\mapsto (\phi_\alpha(q),(A_\alpha(q))^k_s\omega),
\end{split}
\end{align}
where we identify $T^k_s(\R^r)$ with $\R^{r(k+s)}$.
\end{defn}

We have the transition maps
\[
\hat{\phi}_{\alpha\beta}(x,u)=(\phi_{\alpha\beta}(x),(A_{\alpha\beta}(\phi^{-1}_\alpha(x)))^k_su).
\]

\subsection{Multivector fields and differential $s$-forms}
\label{subsec:Multivectorfields_differentialforms}
The important case is when $E=TM$ is the tangent bundle of a $n$-dimensional manifold $M$. Denote by $\mathrm{Alt}^k_s$ the tensor bundle induced by the \emph{alternating tensor product} (wedge product) denoted by $\land$.
Then a contravariant tensor field of order $k$ is also called a \emph{multivector field} and a covariant tensor field of order $s$ is called a \emph{differential $s$-form}. In particular, the space of multivector fields of order $k$ on an $n$-dimensional manifold $M$ is given by 
\[
\mathfrak{X}^k(M):=\Gamma(\mathrm{Alt}^k_0(TM))=\Gamma\left(\bigwedge^kTM\right).
\]
Choosing local coordinates $(x^i)$ on $M$ we can represent an element $X\in \mathfrak{X}^k(M)$ as 
\[
X=\sum_{1\leq i_1<\dotsm <i_k\leq n}X^{i_1\dotsm i_k}\de_{i_1}\land\dotsm \land \de_{i_k}.
\]
The space of differential $s$-forms is then given by 
\[
\Omega^s(M):=\Gamma(\mathrm{Alt}^0_s(TM))=\Gamma\left(\bigwedge^s T^*M\right).
\]
Choosing local coordinates $(x^i)$ on $M$ we can represent an element $\omega\in \Omega^s(M)$ as 
\[
\omega=\sum_{1\leq i_1<\dotsm <i_s\leq n}\omega_{i_1\dotsm i_s}\dd x^{i_1}\land\dotsm \land \dd x^{i_s}.
\]
We can define a product $\land$ between differential forms. For $\omega\in\Omega^s(M),\eta\in \Omega^\ell(M)$ given by 
\[
\omega\land \eta=\sum_{1\leq i_1<\dotsm <i_s\leq n}\,\, \sum_{1\leq j_1<\dotsm <j_\ell\leq n} \omega_{i_1\dotsm i_s}\eta_{j_1\dotsm j_\ell}\dd x^{i_1}\land\dotsm \land\dd x^{i_s}\land \dd x^{j_1}\land\dotsm \land \dd x^{j_\ell}.
\]
Note that is $s+\ell>0$, then $\omega\land \eta=0$.
\begin{exe}
Show that if $\omega\in \Omega^s(M)$ and $\eta\in \Omega^\ell(M)$, then 
\[
\omega\land \eta=(-1)^{s+\ell}\eta\land \omega.
\]
\end{exe}

One can extend the de Rham differential to general $s$-forms 
\begin{align}
\begin{split}
\dd\colon \Omega^s(M)&\to \Omega^{s+1}(M)\\
\omega&\mapsto \dd\omega
\end{split}
\end{align}
and obtain a sequence
\begin{equation}
    0\to C^\infty(M)\xrightarrow{\dd}\Omega^1(M)\xrightarrow{\dd}\Omega^2(M)\xrightarrow{\dd}\dotsm \xrightarrow{\dd}\Omega^s(M)\xrightarrow{\dd}\Omega^{s+1}(M)\xrightarrow{\dd}\dotsm\xrightarrow{\dd}\Omega^n(M)\to 0
\end{equation}
This sequence is called \emph{de Rham complex}. Note that by construction $\Omega^0(M)=C^\infty(M)$.
In particular, for a differential form $\omega=\sum_{1\leq i_1<\dotsm <i_s\leq n}\omega_{i_1\dotsm i_s}\dd x^{i_1}\land\dotsm \land \dd x^{i_s}\in \Omega^s(M)$ we get 
\[
\dd\omega=\sum_{1\leq i_1<\dotsm <i_s\leq n}\,\,\sum_{j=1}^n\de_j\omega_{i_1\dotsm i_s}\dd x^j\land\dd x^{i_1}\land\dotsm \land \dd x^{i_s}.
\]

\begin{exe}
Show that $\dd^2=0$. \emph{Hint: use that derivatives commute, i.e. $\frac{\de}{\de x^i}\frac{\de}{\de x^j}=\frac{\de}{\de x^j}\frac{\de}{\de x^i}$.}
\end{exe}

\begin{exe}
Show that if $\omega\in \Omega^s(M)$ and $\eta\in \Omega^\ell(M)$ then 
\[
\dd(\omega\land \eta)=\dd\omega\land \eta+(-1)^s\omega\land \dd\eta.
\]
\end{exe}

We can extend the Lie derivative and the contraction to any differential $s$-forms. 
\begin{defn}[Lie derivative]
We define the \emph{Lie derivative} of an $s$-form $\omega$ with respect to a vector field $X$ as 
\[
L_X\omega:=\lim_{t\to 0}\frac{(\Phi^X_t)^*\omega-\omega}{t},
\]
where $\Phi^X_t$ denotes the flow of $X$ at time $t$. Explicitly, we have
\begin{multline}
    L_X\omega=\sum_{1\leq i_1<\ldots<i_s\leq n} X(\omega_{i_1\dotsm i_s})\dd x^{i_1}\land\dotsm \land \dd x^{i_s}+\\ +\sum_{1\leq i_1<\ldots<i_s\leq n}\,\,\sum_{1\leq k,r\leq n}(-1)^{k-1}\omega_{i_1\dotsm i_s}\de_rX^{i_k}\dd x^r\land \dd x^{i_1}\land\dotsm \land \widehat{\dd x^{i_k}}\land \dotsm\land \dd x^{i_s},
\end{multline}
where $\enspace\widehat{\enspace}\enspace$ means that this element is omitted.
\end{defn}

\begin{exe}[Properties for Lie derivative]
Let $X,Y\in \mathfrak{X}(M)$. Show that
\begin{itemize}
    \item $L_Xf=X(f)$, $\forall f\in C^\infty(M)$,
    \item $L_X(\omega\land \eta)=L_X\omega\land \eta+\omega\land L_X\eta$, $\forall \omega\in \Omega^s(M),\forall \eta\in \Omega^\ell(M)$,
    \item $L_X\dd\omega=\dd L_X\omega$, $\forall \omega\in \Omega^s(M)$,
    \item $L_XL_Y\omega-L_YL_X\omega=L_{[X,Y]}\omega$, $\forall \omega\in \Omega^s(M)$,
\end{itemize}
\end{exe}

\begin{defn}[Contraction]
The \emph{contraction} of a vector field $X$ with a differential $s$-form $\omega$ is the differential $(s-1)$-form given by 
\[
\iota_X\omega=\sum_{1\leq i_1<\ldots< i_s\leq n}\,\, \sum_{k=1}^n(-1)^k\omega_{i_1\ldots i_s}X^{i_k}\dd x^{i_1}\land\dotsm \land \widehat{\dd x^{i_k}}\land\dotsm \land \dd x^{i_s}.
\]
\end{defn}

\begin{rem}
If $\omega\in \Omega^0(M)$, i.e. $s=0$, then for any $X\in \mathfrak{X}(M)$ we get that $\iota_X\omega$ is automatically zero.
\end{rem}

\begin{exe}
Let $X\in \mathfrak{X}(M)$ and let $\omega\in \Omega^s(M)$, $\eta\in \Omega^\ell(M)$. Show that 
\[
\iota_X(\omega\land \eta)=\iota_X\omega\land \eta+(-1)^s\omega\land \iota_X\eta.
\]
Moreover, show that if $Y\in \mathfrak{X}(M)$ is another vector field, we have
\[
\iota_X\iota_Y\omega=-\iota_Y\iota_X\omega,
\]
and 
\[
\iota_XL_Y\omega-L_Y\iota_X\omega=\iota_{[X,Y]}\omega.
\]
\end{exe}

\begin{thm}[Cartan's magic formula]
\label{thm:Cartan_magic_formula}
Let $M$ be a manifold. Let $\omega\in \Omega^s(M)$ and let $X\in \mathfrak{X}(M)$. Then
\begin{equation}
\label{eq:Cartan_magic_formula}
L_X\omega=\iota_X\dd\omega+\dd\iota_X\omega.
\end{equation}
\end{thm}

\begin{exe}
Prove Theorem \ref{thm:Cartan_magic_formula}.
\end{exe}

\section{Integration on manifolds and Stokes' theorem}

\subsection{Integration of densities}
Let $M$ be a manifold endowed with an atlas $\{(U_\alpha,\phi_\alpha)\}_{\alpha\in I}$. The differential $\dd_x\phi_{\alpha\beta}$ is a linear map
\[
T_x\phi_\alpha(U_\alpha\cap U_\beta)\to T_{\phi_{\alpha\beta}(x)}\phi_{\beta}(U_\alpha\cap U_\beta).
\]
Since the image of the charts are open subsets of $\R^n$ we can identify the tangent spaces with $\R^n$. Hence the linear map $\dd_x\phi_{\alpha\beta}$ is canonically given by an $n\times n$ matrix. Let $s\in \R$ and 
\[
A_{\alpha\beta}(q)=\vert\det\dd_{\phi_\alpha(q)}\phi_{\alpha\beta}\vert^{-s}.
\]
\begin{exe}
Show that this defines a line bundle $\vert\Lambda M\vert^s$ over $M$.
\end{exe}

\begin{defn}[$s$-density]
A section of $\vert\Lambda M\vert^s$ is called an \emph{$s$-density}.
\end{defn}

Let $\sigma\in \Gamma(\vert\Lambda M\vert^s)$ be an $s$-density. We represent it in the chart $(U_\alpha,\phi_\alpha)$ by the smooth function $\sigma_\alpha$. It satisfies the transition rule
\begin{equation}
\label{eq:transformation_density}
\sigma_\beta(\phi_{\alpha\beta}(x))=\vert\det\dd_q\phi_{\alpha\beta}\vert^{-s}\sigma_\alpha(x),\quad \forall \alpha,\beta\in I,\forall x\in \phi_\alpha(U_\alpha\cap U_\beta)
\end{equation}
For $s=1$ one simply speaks of a \emph{density}. Densities are the natural objects to integrate on a manifold. Let $\sigma$ be a density on $M$ and let $\{(U_\alpha,\phi_\alpha)\}_{\alpha\in I}$ be an atlas on $M$. One can show that there always exists a finite partition of unity $\{\rho_j\}_{j\in J}$ subordinate to $\{U_\alpha\}_{\alpha\in I}$, i.e. for each $j\in J$ we have an $\alpha_j$ with $\supp \rho_j\subset U_{\alpha_j}$. In fact, $\supp \rho_j$ is compact. Since $\phi_{\alpha_j}$ is a homeomorphism, also $\phi_{\alpha_j}(\supp \rho_j)$ is compact. The representation $(\rho_j\sigma)_{\alpha_j}$ of the density $\rho_j\sigma$ in the chart $\{(U_\alpha,\phi_\alpha)\}_{\alpha\in I}$ is smooth in $\phi_{\alpha_j}(U_{\alpha_j})$, so it is integrable on $\phi_{\alpha_j}(\supp \rho_j)$. We define
\begin{equation}
\label{eq:integral1}
\int_{M;\{(U_\alpha,\phi_\alpha)\};\{\rho_j\}}\sigma:=\sum_{j\in J}\int_{\phi_{\alpha_j}(\supp \rho_j)}(\rho_j\sigma)_{\alpha_j}\dd^nx,
\end{equation}
where $\dd^nx$ denotes the Lebesgue measure on $\R^n$ (also denoted by $\dd x^1\dotsm \dd x^n$). 

\begin{lem}
The integral defined in \eqref{eq:integral1} does not depend on the choice of atlas and partition of unity. 
\end{lem}

\begin{proof}
Consider an atlas $\{(\bar U_{\bar \alpha},\bar \phi_{\bar \alpha})\}_{\bar \alpha\in \bar I}$ and a finite partition of unity $\{\bar \rho_{\bar j}\}_{\bar j\in \bar J}$ subordinate to it. From $\sum_{\bar j\in \bar J}\bar \rho_{\bar j}=1$, it follows that $\sigma=\sum_{\bar j\in \bar J}\bar\rho_{\bar j}\sigma$, so we have 
\begin{multline}
    \int_{M;\{(U_\alpha,\phi_\alpha)\};\{\rho_j\}}\sigma=\sum_{j\in J}\int_{\phi_{\alpha_j}(\supp \rho_j)}(\rho_j\sigma)_{\alpha_j}\dd^nx=\sum_{\bar j\in \bar J}\sum_{j\in J}\int_{\phi_{\alpha_{j}}(\supp \rho_j)}(\rho_j\rho_{\bar j}\sigma)_{\alpha_j}\dd^nx,
\end{multline}
where we have taken out the finite sum $\sum_{\bar j\in \bar J}$. Observe that 
\begin{multline}
    S_{j\bar j}:=\int_{\phi_{\alpha_j}(\supp \rho_j)}(\rho_j\bar\rho_{\bar j}\sigma)_{\alpha_j}\dd^nx\\=\int_{\phi_{\alpha_j}(\supp\rho_j\cap \supp\bar\rho_{\bar j})}(\rho_j\bar\rho_{\bar j}\sigma)_{\alpha_j}\dd^nx\\=\int_{\phi^{-1}_{\bar\alpha_{\bar j}\alpha_j}(\bar\phi_{\bar\alpha_{\bar j}}(\supp\rho_j\cap \supp\bar\rho_{\bar j}))}(\rho_j\bar\rho_{\bar j}\sigma)_{\alpha_j}\dd^nx,
\end{multline}
where $\phi_{\bar\alpha_{\bar j}\alpha_j}$ denotes the transition map $\phi_{\alpha_j}(U_{\alpha_j})\to \bar\phi_{\bar\alpha_{\bar j}}(\bar U_{\bar\alpha_{\bar j}})$. Since $\rho_j\bar\rho_{\bar j}\sigma$ is a density, we have 
\[
(\rho_j\bar\rho_{\bar j}\sigma)_{\alpha_j}(x)=\vert \det\dd_x\phi_{\bar\alpha_{\bar j}\alpha_j}\vert(\rho_j\bar\rho_{\bar j}\sigma)_{\bar \alpha_{\bar j}}(\bar x),
\]
with $\bar x:=\phi_{\bar\alpha_{\bar j}\alpha_j}(x)$ and $x\in \phi_{\alpha_j}(\supp \rho_j\cap\supp \bar\rho_{\bar j})$. By change-of-variables we get 
\[
S_{j\bar j}=\int_{\bar\phi_{\bar\alpha_{\bar j}}(\supp\rho_j\cap\supp\bar\rho_{\bar j})}(\rho_j\bar\rho_{\bar j}\sigma)_{\bar \alpha_{\bar j}}\dd^n\bar x=\int_{\bar \phi_{\bar\alpha_{\bar j}}(\supp\bar\rho_{\bar j})}(\rho_j\bar\rho_{\bar j}\sigma)_{\bar\alpha_{\bar j}}\dd^n\bar x.
\]
Hence 
\[
\sum_{j\in J}S_{j\bar j}=\int_{\bar\phi_{\bar\alpha_{\bar j}}(\supp\bar\rho_{\bar j})}(\bar \rho_{\bar j}\sigma)_{\bar \alpha_{\bar j}}\dd^n\bar x
\]
and 
\[
\int_{M;\{(U_\alpha,\phi_\alpha)\};\{\rho_j\}}\sigma=\sum_{\bar j\in \bar J}\sum_{j\in J}S_{j\bar j}=\sum_{\bar j\in \bar J}\int_{\bar\phi_{\bar\alpha_{\bar j}}(\supp\bar\rho_{\bar j})}(\bar \rho_{\bar j}\sigma)_{\bar \alpha_{\bar j}}\dd^n\bar x=\int_{M;\{(\bar U_{\bar \alpha},\bar\phi_{\bar \alpha})\};\{\bar\rho_{\bar j}\}}\sigma.
\]
\end{proof}

We can drop the choice of atlas and partition of unity and simply define
\begin{equation}
\label{eq:integral2}
    \boxed{
    \int_M\sigma:=\sum_{j\in J}\int_{\phi_{\alpha_j}(\supp\rho_j)}(\rho_j\sigma)_{\alpha_j}\dd^nx
    }
\end{equation}

\begin{exe}
Let $M_1$ and $M_2$ be disjoint open subsets of a manifold $M$ such that $M\setminus (M_1\cup M_2)$ has measure zero. Show that
\[
\int_M\sigma=\int_{M_1}\sigma+\int_{M_2}\sigma.
\]
\end{exe}

We can also pullback densities by diffeomorphisms $F\colon M\to N$. We define the pullback by
\[
(F^*\sigma)_q:=(\dd_q F)^*\sigma_{F(q)},\quad \forall q\in M.
\]
We can also formulate it in terms of coordinates. Let $\{(U_\alpha,\phi_\alpha)\}_{\alpha\in I}$ be an atlas of $M$ and let $\{(V_j,\psi_j)\}_{j\in J}$ be an atlas of $N$. Let $\{\sigma_j\}$ denote the representation of the density $\sigma$ in the atlas of $N$ and $\{F_{\alpha j}\}$ the representation of $F$ with respect to the two atlases. Then we have
\[
(F^*\sigma)_\alpha(x)=\vert\det\dd_x F_{\alpha j}\vert^s\sigma_j(F_{\alpha j}(x)),\quad \forall x\in \phi_\alpha(U_\alpha).
\]
\begin{exe}
Show that $F^*$ is linear and that 
\[
F^*(\sigma_1\sigma_2)=F^*\sigma_1F^*\sigma_2.
\]
\end{exe}

\begin{exe}
Show that if $F\colon M\to N$ is a diffeomorphism and $\sigma$ is a density on $M$, then 
\[
\int_M\sigma=\int_NF_*\sigma.
\]
\end{exe}

\subsection{Integration of differential forms}
Similarly as for densities we can define the integral of a differential form for some manifold $M$ as in \eqref{eq:integral2}. The difference is that we want to consider the notion of orientation on $M$. This corresponds to the notion of an \emph{oriented atlas}. 

\begin{defn}
A diffeomorphism $F\colon U\to V$ of open subsets of $\R^n$ is \emph{orientation preserving} (orientation reversing) with respect to the standard orientation if and only if $\det F>0$ ($\det F<0$). 
\end{defn}

\begin{defn}
We call an atlas \emph{oriented} if all its transition maps are orientation preserving.
\end{defn}

At first, one defines the notion of an \emph{orientable manifold}. 

\begin{defn}[Top form]
A differential $n$-form on an $n$-dimensional manifold is called a \emph{top form}. We will denote the space of top forms on a manifold $M$ by $\Omega^\mathrm{top}(M)$ without mentioning the dimension of $M$.
\end{defn}

\begin{defn}[Volume form]
A \emph{volume form} on a manifold $M$ is a nowhere vanishing top form.
\end{defn}

\begin{defn}[Orientable manifold]
A manifold $M$ is called \emph{orientable} if it admits a volume form.
\end{defn}

Let $\{(U_\alpha,\phi_\alpha)\}_{\alpha\in I}$ be an atlas of a manifold $M$. Consider a volume form $v\in \Omega^\mathrm{top}(M)$ whose representation in the given atlas is denoted by $v_\alpha$. Let $v_\alpha=\underline{v_\alpha}\dd x^1\land \dotsm \land\dd x^n$. Then the functions $\underline{v_\alpha}$ transform as 
\[
\underline{v_\alpha}(x)=\det \dd_x\phi_{\alpha\beta}\underline{v_\beta}(\phi_{\alpha\beta}(x)),\quad \forall \alpha,\beta\in I,\forall x\in \phi_\alpha(U_\alpha\cap U_\beta).
\]
This is almost the same as the transformation rule \eqref{eq:transformation_density} for densities. To fix this, we can take absolute values of the representations $v_\alpha$. We define the absolute value $\vert v\vert$ as the density with representations $\vert v_\alpha\vert$. Moreover, we want to restrict everything to top forms that do not change sign (at least locally). Choosing a volume form $v$ on $M$, we can define a $C^\infty(M)$-linear isomorphism 
\[
\phi_v\colon\Omega^\mathrm{top}(M)\to \mathrm{Dens}(M)
\]
as follows: since $v$ is nowhere vanishing, for every top form $\omega$ there is a uniquely defined function $f$ such that $\omega=fv$; the corresponding density is then defined to be $f\vert v\vert$. Formally, we may write
\[
\phi_v\omega=\omega\frac{\vert v\vert}{v}.
\]
Two volume forms $v_1$ and $v_2$ yield the same isomorphism, i.e. $\phi_{v_1}=\phi_{v_2}$, if and only if there is a positive function $g$ such that $v_1=gv_2$. 

\begin{exe}
Show that this defines an equivalence relation on the set of volume forms on $M$.
\end{exe}

\begin{defn}[Orientation]
An equivalence class of volume forms on an orientable manifold $M$ is called an \emph{orientation}. An orientable manifold with a choice of orientation is called \emph{oriented}.
\end{defn}

We denote by $[v]$ an orientation, by $(M,[v])$ the corresponding oriented manifold and by $\phi_{[v]}$ the isomorphism given by $\phi_v$ for any $v\in [v]$. If $M$ admits a partition of unity, we can define the integral of a top form $\omega\in \Omega^\mathrm{top}(M)$ by 
\[
\int_{(M,[v])}\omega:=\int_M\phi_{[v]}\omega,
\]
where we use the already defined integration of densities. More explicitly, for an atlas $\{(U_\alpha,\phi_\alpha)\}_{\alpha\in I}$ we have the representations $\omega_\alpha=\underline{\omega_\alpha}\dd^nx$, for uniquely defined maps $\underline{\omega_\alpha}$, and we get 
\[
(\phi_{[v]}\omega)_\alpha=\underline{\omega_\alpha}\vert \dd^n x\vert.
\]
Hence, the integral of $\omega$ on an oriented manifold $M$ is given by 
\begin{equation}
    \int_{(M,[v])}\omega=\sum_{j\in J}\int_{\phi_{\alpha_j}(\supp\rho_j)}(\rho_j)_{\alpha_j}\underline{\omega_{\alpha_j}}\dd^nx,
\end{equation}
where $\{(U_\alpha,\phi_\alpha)\}_{\alpha\in I}$ is an oriented atlas corresponding\footnote{This means that it is an atlas in which any $v\in[v]$ is represented by a positive volume form.} to $[v]$ and $\{\rho_j\}_{j\in J}$ is a partition of unity subordinate to $\{U_\alpha\}_{\alpha\in I}$. If the identification of differential forms and densities on $\phi_{\alpha_j}(U_{\alpha_j})$ is understood, then we can write
\[
\int_{(M,[v])}\omega=\sum_{j\in J}\int_{\phi_{\alpha_j}(\supp\rho_j)}(\rho_j\omega)_{\alpha_j}.
\]

\begin{lem}
\label{lem:localization}
Let $\{(U_\alpha,\phi_\alpha)\}_{\alpha\in I}$ be an oriented atlas of $(M,[v])$ and $\{\rho_j\}_{j\in J}$ a partition of unity subordinate to it. Let $\omega\in \Omega^\mathrm{top}(M)$ with $\supp \omega\subset U_{\alpha_k}$ for some $k\in J$. Then 
\[
\int_{(M,[v])}\omega=\int_{\phi_{\alpha_k}(U_{\alpha_k})}\omega_{\alpha_k}.
\]
\end{lem}

\begin{proof}
We have 
\begin{multline*}
\int_{(M,[v])}\omega=\sum_{j\in J}\int_{\phi_{\alpha_j}(\supp\rho_j)}(\rho_j\omega)_{\alpha_j}=\sum_{j\in J}\int_{\phi_{\alpha_j}(\supp\rho_j\cap\supp\omega)}(\rho_j\omega)_{\alpha_j}\\
=\sum_{j\in J}\int_{\phi_{\alpha_j}(U_{\alpha_j}\cap U_{\alpha_k})}(\rho_j\omega)_{\alpha_j}=\sum_{j\in J}\int_{\phi_{\alpha_j}(U_{\alpha_j}\cap U_{\alpha_k})}(\rho_j\omega)_{\alpha_k}\\
=\sum_{j\in J}\int_{\phi_{\alpha_k}(U_{\alpha_k})}(\rho_j\omega)_{\alpha_k}=\int_{\phi_{\alpha_k}(U_{\alpha_k})}\sum_{j\in J}(\rho_j\omega)_{\alpha_k}=\int_{\phi_{\alpha_k}(U_{\alpha_k})}\omega_{\alpha_k}.
\end{multline*}
\end{proof}

\begin{lem}
\label{lem:orientation}
A connected oriented manifold admits two orientations.
\end{lem}

\begin{exe}
Prove Lemma \ref{lem:orientation}.
\end{exe}

\begin{defn}[Orientation preserving/reversing]
A diffeomorphism $F$ of connected oriented manifolds $(M,[v_M])$ and $(N,[v_N])$ is called \emph{orientation preserving} if $F^*[v_N]=[v_M]$ and \emph{orientation reversing} if $F^*[v_N]=-[v_M]$.
\end{defn}

\begin{prop}[Change of variables]
\label{prop:change_of_variables}
let $(M,[v_M])$ and $(N,[v_N])$ be connected oriented manifolds, $F\colon M\to N$ a diffeomorphism and $\omega$ a top form on $N$. Then 
\[
\int_{(M,[v_M])}F^*\omega=\pm \int_{(N,[v_N])}\omega,
\]
with plus sign if $F$ is orientation preserving and the minus sign if $F$ is orientation reversing.
\end{prop}

\begin{exe}
Prove Proposition \ref{prop:change_of_variables}.
\end{exe}

Typically, the chosen orientation is understood, so one simply  writes
\[
\int_M\omega.
\]

\subsection{Stokes' theorem}
We will denote the space of compactly supported differential $s$-forms on a manifold $M$ by $\Omega^s_c(M)$. Moreover we define the \emph{$n$-dimensional upper half-space} by 
\begin{equation}
\label{eq:upper_half-space}
\mathbb{H}^n:=\{(x^1,\ldots,x^n)\in \R^n\mid x^n\geq 0\}.
\end{equation}
The boundary of the upper half-space is given by 
\[
\de\mathbb{H}^n=\{(x^1,\ldots,x^n)\in \R^n\mid x^n=0\}.
\]
On the boundary, we can take the orientation induced by the \emph{outward pointing vector field} $-\de_n$, i.e. 
\[
[i^*\iota_{-\de_n}\dd^nx]=(-1)^n[\dd x^1\land\dotsm \land \dd x^{n-1}],
\]
where $i\colon\de\mathbb{H}^n\hookrightarrow \mathbb{H}^n$ denotes the inclusion. 

\begin{lem}
\label{lem:local_Stokes}
Let $\omega\in \Omega^{n-1}_c(\mathbb{H}^n)$. Then, using the orientations defined as before, we get 
\[
\int_{\mathbb{H}^n}\dd\omega=\int_{\de\mathbb{H}^n}\omega.
\]
\end{lem}

\begin{proof}
We write 
\[
\omega=\sum_{j=1}^n(-1)^{j-1}\omega^j\dd x^1\land\dotsm\land \widehat{\dd x^j}\land\dotsm \land \dd x^n.
\]
The components $\omega^j$ are then related to the components $\omega_{i_1\dotsm i_n}$ by a sign. Then \[
\dd\omega=\sum_{j=1}^n\de_j\omega^j\dd^nx. 
\]
Using the standard orientation, we get 
\[
\int_{\mathbb{H}^n}\dd\omega=\sum_{j=1}^n\int_{\mathbb{H}^n}\de_j\omega^j\dd^nx.
\]
We use \emph{Fubini's theorem} to integrate the $j$-th term along the $j$-th axis. Then, since $\omega$ has compact support, we get 
\[
\int_{-\infty}^{+\infty}\de_j\omega^j\dd x^j=0,\quad \forall j<n,
\]
but for $j=n$ we have
\[
\int_0^{+\infty}\de_n\omega^n\dd x^n=-\omega^n\big|_{x^n=0}.
\]
Hence, we get 
\[
\int_{\mathbb{H}^n}\dd\omega=-\int_{\de\mathbb{H}^n}\omega^n\dd^{n-1}x.
\]
On the other hand, we have 
\[
i^*\omega=(-1)^{n-1}\omega^n\big|_{x^n=0}\dd x^1\land \dotsm \land \dd x^{n-1}.
\]
Thus, using the orientation of $\de\mathbb{H}^n$ as before, we get 
\[
\int_{\de\mathbb{H}^n}\omega=-\int_{\de\mathbb{H}^n}\omega^n\dd^{n-1}x,
\]
which concludes the proof.
\end{proof}

Similarly, we get the following lemma.

\begin{lem}
\label{lem:integral_identity_1}
Let $\omega\in \Omega_c^{n-1}(\R^n)$. Then 
\[
\int_{\R^n}\dd\omega=0.
\]
\end{lem}

\begin{defn}[Manifold with boundary]
An $n$-dimensional \emph{manifold with boundary} is an equivalence class of atlases whose charts take values in $\mathbb{H}^n$.
\end{defn}

\begin{rem}
One considers $\mathbb{H}^n$ as a topological space with topology induced from $\R^n$.
\end{rem}

\begin{rem}
Let $M$ be a manifold with boundary.
For any $q\in M$ one can show that if there is a chart map sending $q$ to an interior point of $\mathbb{H}^n$, then any chart map will send it to an interior point of $\mathbb{H}^n$. On the other hand, if there is a chart map sending $q$ to a boundary point of $\mathbb{H}^n$, then any chart map will send $q$ to a boundary point of $\mathbb{H}^n$.
Hence, one can induce a manifold structure on the interior points $\mathring{M}$ and boundary points $\de M$ of $M$ out of the manifold structure of $M$, such that $\dim \mathring{M}=\dim M=\dim \de M+1$, by restricting atlases of $M$. The manifold $\de M$ is called the \emph{boundary} of $M$.
\end{rem}

\begin{thm}[Stokes]
\label{thm:Stokes}
Let $M$ be an $n$-dimensional oriented manifold with boundary and let $\omega\in\Omega^{n-1}_c(M)$. Then 
\begin{equation}
\label{eq:Stokes}
\boxed{
\int_M\dd\omega=\int_{\de M}\omega
}
\end{equation}
where we use the induced orientation on $\de M$.
\end{thm}

\begin{rem}
If $M$ has no boundary, we get $\int_M\dd\omega=0$.
\end{rem}

\begin{rem}
Theorem \ref{thm:Stokes} was never officially proved by Stokes but appeared (in some version) in Maxwell's book on electrodynamics from 1873 \cite{Maxwell1873} where he mentions in a footnote that the idea comes from Stokes who used this theorem in the Smith's Prize Examination of 1854. This is the reason why we call it Stokes' theorem today. A first proof of this theorem was given by Hermann in 1861 \cite{Hermann1861} who does not mention Stokes at all. See also \cite{Katz1979} for some more historical facts on Stokes' theorem.
\end{rem}

\begin{proof}[Proof of Theorem \ref{thm:Stokes}]
Let $\{(U_\alpha,\phi_\alpha)\}_{\alpha\in I}$ be an orientable atlas of $M$ corresponding to the given orientation and let $\{\rho_j\}_{j\in J}$ be a partition of unity subordinate to it. First, observe that 
\[
\dd\omega=\dd\left(\sum_{j\in J}\rho_j\omega\right)=\sum_{j\in J}\dd(\rho_j\omega).
\]
Note that $\supp(\rho_j\omega)\subset U_{\alpha_j}$ and, by Lemma \ref{lem:localization}, we have 
\[
\int_M\dd(\rho_j\omega)=\int_{\phi_{\alpha_j}(U_{\alpha_j})}\dd(\rho_j\omega)_{\alpha_j}.
\]
If $\phi_{\alpha_j}(U_{\alpha_j})$ is contained in the interior of $\mathbb{H}^n$, then we regard $\dd(\rho_j\omega)_{\alpha_j}$ as a compactly supported top form on $\R^n$ by extending it by zero outside of its support. Hence, by Lemma \ref{lem:integral_identity_1}, we get 
\[
\int_{\phi_{\alpha_j}(U_{\alpha_j})}\dd(\rho_j\omega)_{\alpha_j}=0.
\]
Otherwise, we regard $\dd(\rho_j\omega)_{\alpha_j}$ as a compactly supported top form on $\mathbb{H}^n$ by again extending by zero outside of its support. Hence, by Lemma \ref{lem:local_Stokes}, we get 
\[
\int_{\phi_{\alpha_j}(U_{\alpha_j})}\dd(\rho_j\omega)_{\alpha_j}=\int_{\de(\phi_{\alpha_j}(U_{\alpha_j}))}(\rho_j\omega)_{\alpha_j}. 
\]
Note that $\de(\phi_{\alpha_j}(U_{\alpha_j}))=\phi_{\alpha_j}(\de U_{\alpha_j})$ by definition and both are oriented by outward pointing vectors. Thus, again by Lemma \ref{lem:localization}, we get
\[
\int_M\dd(\rho_j\omega)=\int_{\de M}\rho_j\omega.
\]
Summing over $j$ yields the result.
\end{proof}

\section{de Rham's theorem}

\subsection{Singular homology}

\begin{defn}[$p$-simplex]
The \emph{standard $p$-simplex} is the closed subset given by 
\[
\Delta^p:=\left\{(x^1,\ldots,x^p)\in \R^p\,\bigg|\,\sum_{i=1}^px^i\leq 1, x^i\geq 0\, \forall i\right\}\subset \R^p.
\]
\end{defn}

The interior of $\Delta^p$ is a $p$-dimensional manifold. A smooth differential form on $\Delta^p$ is by definition the restriction to $\Delta^p$ of a smooth differential form defined on an open neighborhood of $\Delta^p$ in $\R^p$. Let $\omega\in \Omega^{p-1}(\Delta^p)$. Explicitly, let 
\[
\omega=\sum_{j=1}^p\omega^j\dd x^1\land\dotsm \land\widehat{\dd x^j}\land\dotsm \land \dd x^p.
\]
Then 
\[
\dd\omega=\sum_{j=1}^p(-1)^{j+1}\de_j\omega^j\dd^px
\]
and 
\[
\int_{\Delta^p}\dd\omega=\sum_{j=1}^p(-1)^{j+1}\int_{\Delta^p}\de_j\omega^j\dd^px.
\]
Using \emph{Fubini's theorem} and the \emph{fundamental theorem of analysis}, we get 
\[
\int_{\Delta^p}\de_j\omega^j\dd x^j=\omega^j\big|_{x^j=1-\sum_{i=1\atop i\not=j}x^i}-\omega^j\big|_{x^j=0}.
\]
Hence
\begin{multline}
\label{eq:integral3}
    \int_{\Delta^p}\dd \omega=\sum_{j=1}^p(-1)^{j+1}\int_{\Delta^p\cap\left\{\sum_{i=1}^px^i=1\right\}}\omega^j\dd x^1\land \dotsm \land \widehat{\dd x^j}\land\dotsm \land \dd x^p+\\
    +\sum_{j=1}^p(-1)^{j}\int_{\Delta^p\cap\left\{x^j=0\right\}}\omega^j\dd x^1\land \dotsm \land \widehat{\dd x^j}\land \dotsm \land \dd x^p.
\end{multline}
We can rewrite this in another way if we regard the faces on which we integrate as images of $(p-1)$-simplices. Namely, for $i=0,\ldots,p$, we define smooth maps
\[
k_i^{p-1}\colon \Delta^{p-1}\to \Delta^p,
\]
by 
\[
k_0^{p-1}(a^1,\ldots, a^{p-1})=\left(1-\sum_{i=1}^{p-1}a^i,a^1,\ldots,a^{p-1}\right)
\]
and 
\[
k_j^{p-1}(a^1,\ldots,a^{p-1})=(a^1,\ldots,a^{j-1},0,a^j,\ldots,a^{p-1}),\quad \forall j>0.
\]
The $j$-th integral in the second line of \eqref{eq:integral3} is just the integral on $\Delta^{p-1}$ of the pullback of $\omega$ by $k_j^{p-1}$. In fact, 
\begin{multline}
    (k^{p-1}_j)^*\omega=(k^{p-1}_j)^*\sum_{i=1}^p\omega^i\dd x^1\land\dotsm \land \widehat{\dd x^i}\land\dotsm \land \dd x^p\\=\omega^j(a^1,\ldots,a^{j-1},0,a^j,\ldots, a^{p-1})\dd^{p-1}a
\end{multline}
We integrate over $\Delta^{p-1}$ with the standard orientation and rename variables $x^i=a^i$ for $i<j$ and $x^i=a^{i+1}$ for $i>j$. Note that the $j$-th integral is given by the integral over $\Delta^{p-1}$ of the pullback of $(-1)^{j+1}\omega^j\dd x^1\land \dotsm \land \widehat{\dd x^j}\land\dotsm \land \dd x^p$ by $k_0^{p-1}$. In particular, 
\begin{multline}
    (k_0^{p-1})^*\omega^j\dd x^1\land \dotsm \land \widehat{\dd x^j}\land\dotsm \land \dd x^p\\=-\omega^j\left(1-\sum_{i}a^i,a^1,\ldots,a^{p-1}\right)\sum_i\dd a^i\land \dd a^1\land\dotsm\land \widehat{\dd a^{j-1}}\land\dotsm\land \dd a^{p-1}\\= (-1)^{j+1}\omega^j\left(1-\sum_ia^i,a^1,\ldots,a^{p-1}\right)\dd^{p-1}a.
\end{multline}
Summing everything up, we get \emph{Stokes' theorem for a simplex}:
\begin{equation}
\label{eq:integral4}
\int_{\Delta^p}\dd\omega=\sum_{j=0}^p(-1)^j\int_{\Delta^{p-1}}(k^{p-1}_j)^*\omega,
\end{equation}
where the $j=0$ term corresponds to the whole sum in the first line of \eqref{eq:integral3} and each other term corresponds to a term in the second line. Consider a map 
\[
\sigma\colon \Delta^p\to M,
\]
where $M$ is a manifold. For a $p$-form $\omega$ on $M$, we define
\[
\int_\sigma\omega:=\int_{\Delta^{p}}\sigma^*\omega.
\]
If we define $\sigma^j:=\sigma\circ k_j^{p-1}\colon \Delta^{p-1}\to M$, then, by \eqref{eq:integral4}, we get 
\[
\int_\sigma\dd \omega=\sum_{j=0}^p(-1)^j\int_{\sigma^j}\omega.
\]
\begin{defn}[$p$-chains]
A \emph{$p$-chain} with real coefficients in a manifold $M$ is a finite linear combination $\sum_ka_k\sigma_k$, where $a_k\in\R$ for all $k$, of maps $\sigma_k\colon \Delta^p\to M$. If $\omega$ is a $p$-form on $M$, we define 
\begin{equation}
    \label{eq:integral5}
    \int_{\sum_ka_k\sigma_k}\omega:=\sum_ka_k\int_{\sigma_k}\omega.
\end{equation}
\end{defn}

\begin{thm}[Stokes' theorem for chains]
We have 
\[
\boxed{\int_\sigma\dd\omega=\int_{\de\sigma}\omega}
\]
where 
\[
\de\sigma:=\sum_{j=0}^p(-1)^j\sigma^j.
\]
\end{thm}

Let $\Omega_p(M,\R)$ denote the vector space of $p$-chains in $M$ with real coefficients and extend $\de$ to it by linearity. Then we get that $\de$ is an endomorphism of degree $-1$ for the graded vector space 
\[
\Omega_\bullet(M,\R):=\bigoplus_{j=0}^\infty\Omega_j(M,\R),
\]
i.e. a map for all $1\leq p\leq \dim M$ we have
\[
\de\colon\Omega_p(M,\R)\to \Omega_{p-1}(M,\R)
\]

\begin{exe}
Show that $\de^2:=\de\circ \de=0$.
\end{exe}

For $\sigma\in \Omega_p(M,\R)$ and $\omega\in \Omega^p(M)$ we can define
\begin{equation}
\label{eq:pairing}
\langle\sigma,\omega\rangle:=\int_\sigma\omega.
\end{equation}

\begin{exe}
Show that $\langle\enspace,\enspace\rangle$ as defined in \eqref{eq:pairing} is a bilinear map $\Omega_p(M,\R)\times \Omega^p(M)\to \R$.
\end{exe}
By Stokes' theorem for chains we get 
\[
\langle\sigma,\dd\omega\rangle=\langle\de\sigma,\omega\rangle.
\]
In particular we have a sequence

\[
\begin{tikzcd}
0 \arrow[r] & {\Omega_n(M,\mathbb{R})} \arrow[r, "\partial"] & \dotsm \arrow[r, "\partial"] & {\Omega_{p}(M,\mathbb{R})} \arrow[r, "\partial"] & \dotsm \arrow[r, "\partial"] & {\Omega_0(M,\mathbb{R})} \arrow[r] & 0
\end{tikzcd}
\]

If we consider a map $F\colon M\to N$ between manifolds, it induces a graded linear map 
\begin{align}
\begin{split}
    F_*\colon \Omega_\bullet(M,\R)&\to \Omega_\bullet(N,\R),\\
    \sigma&\mapsto F\circ \sigma.
\end{split}
\end{align}
Hence, if $\dim M=\dim N=n,$ we get a chain complex
\[
\begin{tikzcd}[scale cd=0.85]
0 \arrow[r] & {\Omega_n(M,\mathbb{R})} \arrow[r, "\partial"] \arrow[d, "F_*"] & \dotsm \arrow[r, "\partial"] & {\Omega_{p}(M,\mathbb{R})} \arrow[d, "F_*"] \arrow[r, "\partial"] & {\Omega_{p-1}(M,\mathbb{R})} \arrow[r, "\partial"] \arrow[d, "F_*"] & \dotsm \arrow[r, "\partial"] & {\Omega_0(M,\mathbb{R})} \arrow[r] \arrow[d, "F_*"] & 0 \\
0 \arrow[r] & {\Omega_n(N,\mathbb{R})} \arrow[r, "\partial"]                  & \dotsm \arrow[r, "\partial"] & {\Omega_p(N,\mathbb{R})} \arrow[r, "\partial"]                    & {\Omega_{p-1}(N,\mathbb{R})} \arrow[r, "\partial"]                  & \dotsm \arrow[r, "\partial"] & {\Omega_0(N,\mathbb{R})} \arrow[r]                  & 0
\end{tikzcd}
\]

\begin{exe}
\label{exe:commutative}
Show that each square of the diagram commutes, i.e. $\de\circ F_*=F_*\circ\de$.
\end{exe}

\begin{defn}[Singular homology groups]
The \emph{$p$-th singular homology group} on a manifold $M$ with real coefficents is given by  
\[
H_p(M,\R):=\ker\de_{(p)}/\im \de_{(p+1)},
\]
where $\de_{(p)}\colon \Omega_{p}(M,\R)\to \Omega_{p-1}(M,\R)$.
\end{defn}

\begin{rem}
Elements of $\ker \de_{(p)}:=\{\sigma\in \Omega_{p}(M,\R)\mid \de_{(p)}\sigma=0\}$ are usually called \emph{$p$-cycles} and elements of $\im \de_{(p)}:=\{\sigma\in \Omega_p(M,\R)\mid \exists \tau\in \Omega_{p+1}(M,\R),\sigma=\de_{(p+1)}\tau\}$ are usually called \emph{$p$-boundaries}.
\end{rem}

\begin{exe}
Using Exercise \ref{exe:commutative}, show that $F_*$ descends to a graded linear map 
\[
F_*\colon H_\bullet(M,\R)\to H_\bullet(N,\R).
\]
\end{exe}
Thus we get a chain complex on the level of homology
\[
\begin{tikzcd}[scale cd=0.85]
0 \arrow[r] & {H_n(M,\mathbb{R})} \arrow[r] \arrow[d, "F_*"] & \dotsm \arrow[r] & {H_p(M,\mathbb{R})} \arrow[r] \arrow[d, "F_*"] & {H_{p-1}(M,\mathbb{R})} \arrow[r] \arrow[d, "F_*"] & \dotsm \arrow[r] & {H_0(M,\mathbb{R})} \arrow[r] \arrow[d, "F_*"] & 0 \\
0 \arrow[r] & {H_n(N,\mathbb{R})} \arrow[r]                  & \dotsm \arrow[r] & {H_p(N,\mathbb{R})} \arrow[r]                  & {H_{p-1}(N,\mathbb{R})} \arrow[r]                  & \dotsm \arrow[r] & {H_0(N,\mathbb{R})} \arrow[r]                  & 0
\end{tikzcd}
\]

\begin{exe}
Show that $\langle F_*\sigma,\omega\rangle=\langle\sigma,F^*\omega\rangle$ for all $\sigma\in \Omega_\bullet(M,\R)$ and all $\omega\in \Omega^\bullet(N)$.
\end{exe}

\subsection{de Rham cohomology and de Rham's theorem}
Let $M$ be an $n$-dimensional manifold and recall its de Rham complex
\begin{equation}
    0\to C^\infty(M)\xrightarrow{\dd}\Omega^1(M)\xrightarrow{\dd}\Omega^2(M)\xrightarrow{\dd}\dotsm \xrightarrow{\dd}\Omega^s(M)\xrightarrow{\dd}\Omega^{s+1}(M)\xrightarrow{\dd}\dotsm\xrightarrow{\dd}\Omega^n(M)\to 0
\end{equation}

\begin{defn}[de Rham cohomology groups]
We define the \emph{$s$-th de Rham cohomology group} by 
\[
H^s(M):=\ker\dd^{(s)}/\im \dd^{(s-1)}, 
\]
where $\dd^{(s)}\colon \Omega^s(M)\to \Omega^{s+1}(M)$. 
\end{defn}

\begin{rem}
we call element of $\ker\dd^{(s)}:=\{\omega\in \Omega^s(M)\mid \dd^{(s)}\omega=0\}$ \emph{closed forms} and elements of $\im\dd^{(s)}:=\{\omega\in \Omega^s(M)\mid \exists \eta\in \Omega^{s-1}(M), \omega=\dd^{(s-1)}\eta\}$ \emph{exact forms}.
\end{rem}

\begin{exe}
Using Stokes' theorem for chains, show that the bilinear map $\langle\enspace,\enspace\rangle$ defined in \eqref{eq:pairing} descends to a bilinear map 
\begin{equation}
\label{eq:pairing2}
H_p(M,\R)\times H^p(M)\to \R.
\end{equation}
\end{exe}

\begin{thm}[de Rham\cite{deRham1931}]
The bilinear map \eqref{eq:pairing2} is nondegenerate. In particular, 
\[
(H_p(M,\R))^*\cong H^p(M),\quad \forall p.
\]
\end{thm}

\section{Distributions and Frobenius' theorem}

\subsection{Plane distributions}

\begin{defn}[Plane distribution]
A \emph{$k$-plane distribution} $D$, or simply a \emph{$k$-distribution} or just a \emph{distribution}, on a smooth $n$-dimensional manifold $M$ is a collection $\{D_q\}_{q\in M}$ of linear $k$-dimensional subspaces $D_q\in T_qM$ for all $q\in M$. 
\end{defn}

\begin{rem}
The number $k$ is called the \emph{rank} of the distribution. Obviously, we want $k\leq n$.
\end{rem}

\begin{defn}[Distribution]
A \emph{$k$-distribution} $D$ on $M$ is called \emph{smooth} if every $q\in M$ has an open neighborhood $U$ and smooth vector fields $X_1,\ldots,X_k$ defined on $U$ such that 
\[
D_x=\mathrm{span}\{(X_1)_x,\ldots,(X_k)_x\},\quad \forall x\in U.
\]
The vector fields $X_1,\ldots,X_k$ are also called \emph{(local) generators} for $D$ on $U$.
\end{defn}

\begin{rem}
In the case when $U=M$, we speak of \emph{global generators}. The existence of local generators is required since there are certain interesting distributions which do not have global generators.
\end{rem}

We say that a vector field $X\in \mathfrak{X}(M)$ is a \emph{tangent} to a distribution $D$ if $X_q\in D_q$ for all $q\in M$. Note that any linear combination of tangent vector fields are again tangent. 

\begin{defn}[Involutive distribution]
A smooth distribution $D$ on $M$ is called \emph{involutive} if $\Gamma(D)$ is a Lie subalgebra of $\mathfrak{X}(M)$, i.e. when $[X,Y]\in\Gamma(D)$ for all $X,Y\in\Gamma(D)$.
\end{defn}

\begin{rem}
A distribution generated by vector fields $X_1,\ldots,X_k$ is involutive if and only if $[X_i,X_j]$ is a linear combination over $C^\infty(M)$ of the generators. The only-if side follows from the definition. The if-implication can be obtained by observing that a vector field tangent to the distribution is necessarily a linear combination of the generators. Moreover, we have
\[
\left[\sum_i f_i X_i,\sum_j g_jY_j\right]=\sum_{i,j}\Bigg(\Big(f_iX_i(g_j)-g_iX_i(f_j)\Big)X_j+f_ig_j[X_i,X_j]\Bigg).
\]
\end{rem}

\begin{rem}
Note that the previous discussion directly implies that any distribution of rank $1$ is involutive. Locally, it is generated by a single vector field $X$ and by the skew-symmetry of the Lie bracket, we have $[X,X]=0$.
\end{rem}

\begin{defn}[Push-forward of a distribution]
Let $D$ be a distribution on $M$ and let $F\colon M\to N$ be a diffeomorphism. We define the \emph{push-forward} $F_*D$ of $D$ by 
\[
(F_*D)_y:=\dd_{F^{-1}(y)}D_{F^{-1}(y)},\quad \forall y\in N.
\]
\end{defn}

\subsection{Frobenius' theorem}

\begin{defn}[Integral manifold]
An immersion $\psi\colon N\to M$ with $N$ connected is called an \emph{integral manifold} for a distribution $D$ on $M$ if 
\[
\dd_n\psi(T_nN)=D_{\psi(n)},\quad \forall n\in N.
\]
\end{defn}

\begin{rem}
An integral manifold that is not a proper restriction of an integral manifold is called \emph{maximal}.
\end{rem}

\begin{rem}
If $\psi$ is an embedding, restricting $\psi\colon N\to \psi(N)$ to its image allows us to rewrite the above condition as 
\[
\psi_*(TN)=D\big|_{\psi(N)},
\]
where the fact that $D$ can be restricted to $\psi(N)$, i.e. $D_q\in T_q\psi(N)$ for all $q\in \psi(N)$, is part of the condition.
\end{rem}

\begin{defn}[Integrable distribution]
A smooth distribution $D$ on $M$ is called \emph{integrable} if for all $q\in M$ there is an integral manifold $D$ passing through $q$.
\end{defn}

\begin{lem}
If a distribution $D$ is integrable, then $D$ is involutive.
\end{lem}

\begin{proof}
For all $q\in M$, we can find an integral manifold $\psi\colon N\to M$ with $q\in\psi(N)$. If $X$ and $Y$ are tangent to $D$, in a neighborhood of $q$ in $\psi(N)$ we can write them as push-forwards of vector fields $\widetilde{X}$ and $\widetilde{Y}$ on $N$. Since the push-forward preserves the Lie bracket and $TN$ is involutive, we see that in this neighborhood $Z:=[X,Y]$ is the push-forward of $[\widetilde{X},\widetilde{Y}]$ and hence tangent to $D$. Note also that this is indeed the Lie bracket of $X$ and $Y$, as they do not have components transverse to $\psi(N)$ by definition. We can compute $Z$ by this procedure at each point in $M$, which shows that $D$ is involutive.
\end{proof}

\begin{prop}
\label{prop:Frobenius}
Let $X$ be a vector field on a Hausdorff manifold $M$. Let $q\in M$ be a point such that $X_q\not=0$. Then there is a chart $(U,\phi_U)$ with $U\ni q$ such that $(\phi_U)_*X\big|_{U}$ is the constant vector field $(1,0,\ldots,0)$. As a consequence, if $\gamma$ is an integral curve of $X$ passing through $U$, then $\phi_U\circ\gamma$ is of the form $\{x\in \phi_U(U)\mid x^1(t)=x_0^1+t;\, x^j(t)=x_0^j,\, j>1\}$ where the $x_0^i$s are constants.
\end{prop}

\begin{thm}[Frobenius\cite{Frobenius1877}]
\label{thm:Frobenius}
Let $D$ be an involutive $k$-distribution on a smooth, Hausdorff $n$-dimensional manifold $M$. Then each point $q\in M$ has a chart neighborhood $(U,\phi)$ such that $\phi_*D=\mathrm{span}\left\{\frac{\de}{\de x^1},\ldots,\frac{\de}{\de x^k}\right\}$, where $x^1,\ldots,x^k$ are coordinates on $\phi(U)$.
\end{thm}

\begin{cor}
On a smooth, Hausdorff manifold a smooth distribution is involutive if and only if it is integrable.
\end{cor}

\begin{proof}[Proof of Theorem \ref{thm:Frobenius}]
We will use induction on the rank $k$ of the distribution. For $k=1$, This is the content of Proposition \ref{prop:Frobenius}. We assume that we have proved the theorem for rank $k-1$. Let $(X_1,\ldots,X_k)$ be generators of the distribution in a neighborhood of $q$. In particular, they are all not vanishing at $q$. By Proposition \ref{prop:Frobenius}, we can find a chart neighborhood $(V,\chi)$ of $q$ with $\chi(q)=0$ and $\chi_*X_1=\frac{\de}{\de y^1}$, where $y^1,\ldots,y^n$ are coordinates on $\chi(V)$. We can define new generators of $\chi_*D$ by 
\[
Y_1:=\chi_*X_1=\frac{\de}{\de y^1}
\]
and for $i>1$
\[
Y_i:=\chi_*X_i-(\chi_*X_i(y^1))\chi_*X_1.
\]
For $i>1$ we have $Y_i(y^1)=0$ and thus, for $i,j>1$, we have $[Y_i,Y_j](y^1)=0$. This means that the expansion of $[Y_i,Y_j]$ in the $Y_\ell$s does not contain $Y_1$. Hence, the distribution $D'$ defined on $S:=\{y\in \chi(V)\mid y^1=0\}$ as the span of $Y_2,\ldots,Y_k$ is involutive. By the induction assumption, we can find a neighborhood $U$ of $0$ in $S$ and a diffeomorphism $\tau$ such that $\tau_*Y_i=\frac{\de}{\de w^i}$, for $i=2,\ldots,w$, where $w^2,\ldots,w^n$ are coordinates on $\tau(U)$. Let $\widetilde{U}:=U\times(-\varepsilon,\varepsilon)$ for some $\varepsilon>0$ such that $\widetilde{U}\subset\chi(V)$. We then have the projection $\pi\colon \widetilde{U}\to U$. Finally, consider the diffeomorphism 
\begin{align*}
    \widetilde{\tau}\colon \widetilde{U}&\to\tau(U)\times(-\varepsilon,\varepsilon),\\
    (u,y^1)&\mapsto (\tau(u),y^1),
\end{align*}
and write $x^1=y^1$, $x^i=\tau^i(y^2,\ldots,y^n)=w^i$ for $i>1$. We write $Z_i:=\widetilde{\tau}_* Y_i$ for $i=1,\ldots,k$ being the generators of the distribution $\widetilde{D}:=\widetilde{\tau}_*\chi_*D$. Now since $\frac{\de x^i}{\de y^1}$ is equal to one if $i=1$ and zero otherwise, we get that $Z_1=\frac{\de}{\de x^1}$. For $i=2,\ldots,k$ and $j>1$ we have 
\[
\frac{\de}{\de x^1}(Z_i(x^j))=Z_1(Z_i(x^j))=[Z_1,Z_i](x^j)=\sum_{\ell=2}^kc^\ell_iZ_\ell(x^j),
\]
where the $c^\ell_i$ are functions that are guaranteed to exist by the involutivity of the distribution. For fixed $j$ and fixed $x^2,\ldots,x^n$ we regard these identities as ODEs in the variable $x^1$. Note that, for $i=2,\ldots,k$ and $j>k$, we have $Z_i(x^j)=0$ at $x^1=0$ (since at $x^1=0$ we have $Z_i=Y_i$). This means that $Z_i(x^j)=0$ for $i=2,\ldots,k$ and $j>k$, is the unique solution with this initial condition. These identities mean that $\widetilde{D}$ is the distribution spanned by the vector fields $\frac{\de}{\de x^1},\ldots,\frac{\de}{\de x^k}$.
\end{proof}

\chapter{Symplectic Geometry}

Symplectic geometry appears in the mathematical structure of the phase space $M=\R^{2n}$ for a classical mechanical system. The dynamical information is usually encoded in a function $H\in C^\infty(M)$ called the \emph{Hamiltonian}. In order to express the dynamics in terms of the flow lines, we want to extract a vector field $X_H\in \mathfrak{X}(M)$ out of $H$, i.e. consider a differential equation with respect to the change of $H$, similarly as for the canonical equations of motion in Hamiltonian mechanics. This means that we want to consider a map $\omega\colon TM\to T^*M$, or equivalently an element of $T^*M\otimes T^*M$, such that $\dd H=\iota_{X_H}\omega=\omega(X_H,\enspace)$. Additionally, we want that the choice of $X_H$ for each $H$ is unique in this way, which requires $\omega$ to be nondegenerate.
Moreover, we want that $H$ doesn't change along flow lines, i.e. $\dd H(X_H)=0$, meaning that $\omega(X_H,X_H)=0$ and thus we require $\omega$ to be alternating, which means that $\omega$ has to be a 2-form. Note that this also implies that the underlying space has to be even dimensional since every skew-symmetric linear map for odd dimensions is singular.
Finally, we want also that $\omega$ doesn't change under flow lines. Mathematically, this is expressed as the vanishing of the Lie derivative $L_{X_H}\omega=0$. Using Cartan's magic formula \eqref{eq:Cartan_magic_formula}, we get 
\[
L_{X_H}\omega=\dd\iota_{X_H}\omega+\iota_{X_H}\dd\omega=\dd(\dd H)+\iota_{X_H}\dd\omega=\dd\omega(X_H).
\]
Hence, if we require $\dd\omega(X_H)=0$ for the vector field induced by different Hamiltonians, we require $\omega$ to be closed, i.e. $\dd\omega=0$. Such a 2-form is then called \emph{symplectic}. In fact, it gives an area measure which induces a time-independent area by 
\[
A:=\int_\Sigma\omega
\]
for a suitable surface $\Sigma\subset M$. We will generalize this concept to the definition of a volume form on any $2n$-dimensional manifold. 
This chapter is based on \cite{daSilva01,Cattaneo_Manifolds,Arnold1978,Donaldson1996,DonaldsonKronheimer1990,McDuffSalamon1995,EliashbergGromov1998,DuistermaatHeckman1982,Weinstein1971,Weinstein1977,Weinstein1981,BerlineGetzlerVergne1992,AtiyahBott1984}.

\section{Symplectic manifolds}

\subsection{Symplectic form}

Let $V$ be an $m$-dimensional vector space over $\R$ and let $$\Omega\colon V\times V\to \R$$ be a bilinear map. 

\begin{defn}[Skew-symmetric]
The map $\Omega$ is called \emph{skew-symmetric} if $\Omega(u,v)=-\Omega(v,u)$ for all $u,v\in V$.
\end{defn}

\begin{thm}[Standard form]
\label{thm:standard_form}
Let $\Omega$ be a skew-symmetric bilinear map on $V$. Then there is a basis $u_1,\ldots,u_k,e_1,\ldots,e_n,f_1,\ldots,f_n$ of $V$ such that 
\begin{align}
    \Omega(u_i,v)&=0,\quad 1\leq i\leq k,\forall v\in V,\\
    \Omega(e_i,e_j)&=\Omega(f_i,f_j)=0,\quad 1\leq i,j\leq n\\
    \Omega(e_i,f_j)&=\delta_{ij},\quad  1\leq i,j\leq n.
\end{align}
\end{thm}

\begin{rem}
The basis in Theorem \ref{thm:standard_form} is not unique even if it is historically called \emph{canonical} basis. 
\end{rem}

\begin{proof}[Proof of Theorem \ref{thm:standard_form}]
Let $U:=\{u\in V\mid \Omega(u,v)=0,\,\, \forall v\in V\}$. Choose a basis $u_1,\ldots,u_k$ of $U$ and choose a complementary space $W$ to $U$ in $V$, i.e. such that 
\[
V=W\oplus U.
\]
Let $e_1\in W$ be a nonzero element. Then there is $f_1\in W$ such that $\Omega(e_1,f_1)\not=0$. Now assume that $\Omega(e_1,f_1)=1$. Let $W_1$ be the span of $e_1$ and $f_1$ and let 
\[
W^\Omega_1:=\{w\in W\mid \Omega(w,v)=0,\,\,\forall v\in W_1\}.
\]
We can then show that $W_1\cap W^\Omega_1=\{0\}$. Indeed, suppose that $v=ae_1+bf_1\in W_1\cap W^\Omega_1$. Then $0=\Omega(v,e_1)=-b$ and $0=\Omega(v,f_1)=a$ which together implies that $v=0$. Moreover, we have $W=W_1\oplus W^\Omega_1$. Indeed, suppose that $v\in W$ has $\Omega(v,e_1)=c$ and $\Omega(v,f_1)=d$. Then $v=(-cf_1+de_1)+(v+cf_1-de_1)$, where $-cf_1+de_1\in W_1$ and $v+cf_1-de_1\in W^\Omega_1$. Now let $e_2\in W^\Omega_1$ be a nonzero element. Then there is $f_2\in W^\Omega_1$ such that $\Omega(e_2,f_2)\not=0$. Now assume that $\Omega(e_2,f_2)=1$. Moreover, let $W_2$ be given by the span of $e_2$ and $f_2$. This construction can be continued until some point since $\dim V<\infty$ and thus we obtain 
\[
V=U\oplus W_1\oplus\dotsm \oplus W_n.
\]
where all the summands are orthogonal with respect to $\Omega$ and where $W_i$ has basis $e_i,f_i$ with $\Omega(e_i,f_i)=1$.
\end{proof}

\begin{rem}
Note that the dimension of the subspace $U\subset V$ does not depend on the choice of basis and hence is an invariant on $(V,\Omega)$. Since $\dim U+2n=\dim V$, we get that $n$ is an invariant of $(V,\Omega)$. We call the number $2n$ the \emph{rank} of $\Omega$.
\end{rem}

\subsection{Symplectic vector spaces}
Let $V$ be an $m$-dimensional real vector space and let $\Omega\colon V\times V\to \R$ be a bilinear form. Define the map $\tilde\Omega\colon V\to V^*$ to be the linear map defined by 
\[
\tilde\Omega(v)(u):=\Omega(v,u).
\]
Note that $\ker\tilde\Omega=U$.
\begin{defn}[Linear symplectic fom]
A skew-symmetric bilinear form $\Omega$ is called \emph{symplectic} (or \emph{nondegenerate}) if $\tilde\Omega$ is bijective, i.e. $U=\{0\}$. 
\end{defn}

\begin{rem}
We sometimes also call a symplectic form $\Omega$ a \emph{linear symplectic structure}.
\end{rem}

\begin{defn}[Symplectic vector space]
We call a vector space $V$ endowed with a linear symplectic structure $\Omega$ \emph{symplectic}.
\end{defn}

\begin{exe}
Let $\Omega$ be a symplectic structure.
Check that the map $\tilde\Omega$ is a bijection and that $\dim U=0$ so $\dim V$ is even. Moreover, check that a symplectic vector space $(V,\Omega)$ has a basis $e_1,\ldots, e_n,f_1,\ldots, f_n$ satisfying
\[
\Omega(e_i,f_j)=\delta_{ij},\qquad \Omega(e_i,e_j)=\Omega(f_i,f_j)=0.
\]
We will call such a basis \emph{symplectic}.
\end{exe}

\begin{defn}[Symplectic subspace]
A subspace $W\subset V$ is called \emph{symplectic} if $\Omega\big|_W$ is nondegenerate.
\end{defn}

\begin{exe}
Show that the subspace given by the span of $e_1$ and $f_1$ is symplectic.
\end{exe}

\begin{defn}[Isotropic subspace]
A subspace $W\subset V$ is called \emph{isotropic} if $\Omega\big|_W=0$.  
\end{defn}

\begin{exe}
Show that the subspace given by the span of $e_1$ and $e_2$ is isotropic.
\end{exe}

\begin{defn}[Symplectic orthogonal]
Let $W\subset V$ be a subspace of a symplectic vector space $(V,\Omega)$. The \emph{symplectic orthogonal} of $W$ is defined as 
\[
W^\Omega:=\{v\in V\mid \Omega(v,u)=0,\,\,\forall u\in W\}.
\]
\end{defn}

\begin{exe}
Show that $(W^\Omega)^\Omega=W$.
\end{exe}

\begin{exe}
Show that a subspace $W\subset V$ is isotropic if $W\subseteq W^\Omega$. Moreover, show that if $W$ is isotropic, then $\dim W\leq \frac{1}{2}\dim V$.
\end{exe}

\begin{defn}[Coisotropic subspace]
A subspace $W\subset V$ of a symplectic vector space $(V,\Omega)$ is called \emph{coisotropic} if 
\[
W^\Omega\subseteq W.
\]
\end{defn}

\begin{exe}
Show that every codimension 1 subspace $W\subset V$ is coisotropic.
\end{exe}

\begin{defn}[Lagrangian subspace]
An isotropic subspace $W\subset V$ of a symplectic vector space $(V,\Omega)$ is called \emph{Lagrangian} if it is maximal, i.e. $\dim W=\frac{1}{2}\dim V$.
\end{defn}

\begin{exe}
Show that a subspace $W\subset V$ of a symplectic vector space $(V,\Omega)$ is Lagrangian if and only if $W$ is isotropic and coisotropic if and only if $W=W^\Omega$.
\end{exe}

\begin{defn}[Standard symplectic vector space]
The \emph{standard symplectic vector space} is defined as the vector space $\R^{2n}$ endowed with the linear symplectic structure $\Omega_0$ defined such that the basis
\begin{align}
\begin{split}
    e_1&=(1,0,\ldots,0),\ldots,e_n=(0,\ldots,0,\underbrace{1}_{n},0,\ldots,0),\\
    f_1&=(0,\ldots,0,\underbrace{1}_{n+1},0,\ldots,0),\ldots,f_n=(0,\ldots,0,1)
\end{split}
\end{align}
is a symplectic basis. 
\end{defn}

\begin{rem}
We can extend $\Omega_0$ to other vectors by using its values on a basis and bilinearity.
\end{rem}

\begin{defn}[Linear symplectomorphism]
A \emph{linear symplectomorphism} $\varphi$ between symplectic vector spaces $(V,\Omega)$ and $(V',\Omega')$ is a linear isomorphism $\varphi\colon V\xrightarrow{\sim} V$ such that
\[
\varphi^*\Omega'=\Omega.
\]
\end{defn}

\begin{rem}
By definition we have 
\[
(\varphi^*\Omega')(u,v):=\Omega'(\varphi(u),\varphi(v)).
\]
\end{rem}

\begin{defn}[Symplectomorphic spaces]
We call two symplectic vector spaces $(V,\Omega)$ and $(V',\Omega')$ \emph{symplectomorphic} if there exists a symplectomorphism between them.
\end{defn}

\begin{exe}
Show that the relation of being symplectomorphic defines an equivalence relation in the set of all even-dimensional vector spaces. Moreover, show that every $2n$-dimensional symplectic vector space $(V,\Omega)$ is symplectomorphic to the standard symplectic vector space $(\R^{2n},\Omega_0)$.
\end{exe}

\subsection{Symplectic manifolds}

Let $\omega$ be a 2-form on a manifold $M$. Note that for each point $q\in M$, the map 
\[
\omega_q\colon T_qM\times T_qM\to \R
\]
is skew-symmetric bilinear on the tangent space $T_qM$. 

\begin{defn}[Symplectic form]
A 2-form $\omega$ is called \emph{symplectic} if $\omega$ is closed and $\omega_q$ is symplectic for all $q\in M$.
\end{defn}

\begin{defn}[Symplectic manifold]
A \emph{symplectic manifold} is a pair $(M,\omega)$ where $M$ is a manifold and $\omega$ is a symplectic form.
\end{defn}

\begin{ex}
\label{ex:standard_symplectic_manifold}
Let $M=\R^{2n}$ with coordinates $x_1,\ldots,x_n,y_1,\ldots,y_n$. Then one can check that the 2-form 
\[
\omega_0:=\sum_{i=1}^n\dd x_i\land \dd y_i
\]
is symplectic and that the set 
\[
\left\{\left(\frac{\de}{\de x_1}\right)_q,\ldots, \left(\frac{\de}{\de x_n}\right)_q,\left(\frac{\de}{\de y_1}\right)_q,\ldots,\left(\frac{\de}{\de y_n}\right)_q\right\}
\]
defines a symplectic basis of $T_qM$.
\end{ex}

\begin{ex}
Let $M=\C^n$ with coordinates $z_1,\ldots,z_n$. Then one can check that the 2-form 
\[
\omega_0:=\frac{\I}{2}\sum_{k=1}^n\dd z_k\land \dd \bar z_k
\]
is symplectic. This is similar to Example \ref{ex:standard_symplectic_manifold} by the identification $\C^n\cong \R^{2n}$ and $z_k=x_k+\I y_k$.
\end{ex}

\subsection{Symplectomorphisms}

\begin{defn}[Symplectomorphism]
Let $(M_1,\omega_1)$ and $(M_2,\omega_2)$ be $2n$-dimensional symplectic manifolds, and let $\varphi\colon M_1\to M_2$ be a a diffeomorphism. Then $\varphi$ is called a \emph{symplectomorphism} if 
\[
\varphi^*\omega_2=\omega_1.
\]
\end{defn}

\begin{rem}
Note that by definition we have 
\[
(\varphi^*\omega_2)_q(u,v)=(\omega_2)_{\varphi(q)}(\dd_q\varphi(u),\dd_q\varphi(v)),\quad \forall u,v\in T_qM.
\]
\end{rem}

The classification of symplectic manifolds up to symplectomorphisms is an interesting problem. The next theorem takes care of this locally. In fact, as any $n$-dimensional manifold looks locally like $\R^n$, one can show that any $2n$-dimensional symplectic manifold $(M,\omega)$ is locally symplectomorphic to $(\R^{2n},\omega_0)$. In particular, the dimension is the only local invariant of symplectic manifolds up to symplectomorphisms.

\begin{thm}[Darboux\cite{Darboux1882}]
\label{thm:Darboux}
Let $(M,\omega)$ be a $2n$-dimensional symplectic manifold and let $q\in M$. Then there is a coordinate chart $(\calU,x_1,\ldots,x_n,y_1,\ldots,y_n)$ centered at $q$ such that on $\calU$ we have \[
\omega=\sum_{i=1}^n\dd x_i\land \dd y_i.
\]
\end{thm}

\begin{defn}[Darboux chart]
A local coordinate chart $(\calU,x_1,\ldots,x_n,y_1,\ldots,y_n)$ is called a \emph{Darboux chart}.
\end{defn}

\begin{rem}
The proof of Theorem \ref{thm:Darboux} will be given later as an exercise (Exercise \ref{exe:proof_of_Darboux}) for a simple application of another important theorem (Moser's relative theorem). 
\end{rem}

\section{The cotangent bundle as a symplectic manifold}

Let $X$ be an $n$-dimensional manifold and let $M:=T^*X$ be its cotangent bundle. Consider coordinate charts $(\calU,x_1,\ldots,x_n)$ on $X$ with $x_i\colon \calU\to \R$. Then at any $x\in \calU$ we have a basis of $T^*_xX$ defined by the linear maps 
\[
\dd_x x_1,\ldots, \dd_x x_n.
\]
In particular, if $\xi\in T^*_xX$, then $\xi=\sum_{i=1}^n\xi_i \dd_x x_i$ for some $\xi_1,\ldots,\xi_n\in \R$. Note that this induces a map 
\begin{align}
\begin{split}
    T^*\calU&\to \R^{2n}\\
    (x,\xi)&\mapsto (x_1,\ldots,x_n,\xi_1,\ldots,\xi_n).
\end{split}
\end{align}
The chart $(T^*\calU,x_1,\ldots,x_n,\xi_1,\ldots,\xi_n)$ is a coordinate chart of $T^*X$. The transition functions on the overlaps are smooth; given two charts $(\calU,x_1,\ldots,x_n)$ and $(\calU',x_1'\ldots,x_n')$ and $x\in \calU\cap\calU'$ with $\xi\in T^*_xX$, then 
\[
\xi=\sum_{i=1}^n\xi_i\dd_xx_i=\sum_{i,j}\xi_i\left(\frac{\de x_i}{\de x_j}\right)\dd_xx_j'=\sum_{j=1}^n\xi_j'\dd_xx_j',
\]
where $\xi'_j=\sum_i\xi_i\left(\frac{\de x_i}{\de x_j}\right)$ is smooth. Hence, $T^*X$ is a $2n$-dimensional manifold.

\subsection{Tautological and canonical forms in coordinates}

Consider a coordinate chart $(\calU,x_1,\ldots,x_n)$ for $X$ with associated cotangent coordinates $(T^*\calU,x_1,\ldots,x_n,\xi_1,\ldots,\xi_n)$. Define a 2-form $\omega$ on $T^*\calU$ by 
\[
\omega:=\sum_{i=1}^n\dd x_i\land \dd \xi_i.
\]
we want to show that this expression is independent of the choice of coordinates. Indeed, consider the 1-form 
\[
\alpha:=\sum_{i=1}^n\xi_i\dd x_i,
\]
and note that $\omega=-\dd\alpha$. Let $(\calU,x_1,\ldots,x_n,\xi_1,\ldots,\xi_n)$ and $(\calU',x_1',\ldots,x_n',\xi_1',\ldots,\xi_n')$ be two coordinate charts on $T^*X$. As we have seen, on the intersection $\calU\cap \calU'$ they are related by $\xi_j'=\sum_{i}\xi_i\left(\frac{\de x_i}{\de x_j}\right)$. Since $\dd x_j'=\sum_{i}\left(\frac{\de x_j'}{\de x_i}\right)\dd x_i$, we get 
\[
\alpha=\sum_i\xi_i\dd x_i=\sum_j\xi_j'\dd x_j'=\alpha'.
\]
Hence, since $\alpha$ is intrinsically defined, so is $\omega$. This finishes the claim. 

\begin{defn}[Tautological form]
The 1-form $\alpha$ is called \emph{tautological form}.
\end{defn}

\begin{rem}
The tautological form is sometimes also called \emph{Liouville 1-form} and $\omega=-\dd\alpha$ is often called \emph{canonical symplectic form}.
\end{rem}

\subsection{Coordinate-free construction}

Let use denote by 
\begin{align}
\begin{split}
    \pi\colon M:=T^*X&\to X,\\
    q=(x,\xi)&\mapsto x
\end{split}
\end{align}
be the natural projection for $\xi\in T^*_xX$. We define the tautological 1-form $\alpha$ pointwise as
\[
\alpha_q=(\dd_q \pi)^*\xi\in T^*_qM.
\]
where $(\dd_q\pi)^*$ denotes the transpose of $\dd_q\pi$, i.e. $(\dd_q\pi)^*\xi=\xi\circ\dd_q\pi$. We have the three maps
\begin{align}
\pi\colon M:=T^*X&\to X,\\
\dd_q\pi\colon T_qM&\to T_xX,\\
(\dd_q\pi)^*\colon T^*_xX&\to T^*_qM.
\end{align}
In fact, we have 
\[
\alpha_q(v)=\xi\left((\dd_q\pi)v\right),\quad \forall v\in T_qM.
\]
The canonical symplectic form is then defined by
\[
\omega=-\dd\alpha.
\]

\begin{exe}
Show that the tautological form $\alpha$ is uniquely characterized by the property that, for every 1-form $\mu\colon X\to T^*X$ we have 
\[
\mu^*\alpha=\mu.
\]
\end{exe}

\subsection{Symplectic volume}
Let $V$ be a vector space of dimension $\dim V<\infty$.
Any skew-symmetric bilinear map $\Omega\in \bigwedge^2V^*$ is of the form 
\[
\Omega=e_1^*\land f^*_1+\dotsm +e^*_n\land f_n^*,
\]
where $u_1^*,\ldots,u_k^*,e_1^*,\ldots,e^*_n,f_1^*,\ldots,f_n^*$ is a basis of $V^*$ dual to the standard basis. Here we have set $\dim V=k+2n$. If $\Omega$ is also nondegenerate, i.e. a symplectic form on a vector space $V$ with $\dim V=2n$, then the $n$-th exterior power of
\[
\Omega^n:=\underbrace{\Omega\land\dotsm \land \Omega}_n
\]
does not vanish. 
\begin{exe}
\label{exe:symplectic_volume}
Show that this also holds for the $n$-th exterior power $\omega^n$ of a symplectic form $\omega$ on a $2n$-dimensional symplectic manifold $(M,\omega)$. Deduce that it defines a volume form on $M$.
\end{exe}

Exercise \ref{exe:symplectic_volume} shows that any symplectic manifold $(M,\omega)$ can be canonically oriented by the symplectic structure.

\begin{defn}[Symplectic volume]
Let $(M,\omega)$ be a symplectic manifold. Then the form 
\[
\frac{\omega^n}{n!}
\]
is called the \emph{symplectic volume} of $(M,\omega)$.
\end{defn}

\begin{rem}
The symplectic volume is sometimes also called \emph{Liouville volume}.
\end{rem}

\begin{exe}
Show that if, conversely, a given 2-form $\Omega\in \bigwedge^2V^*$ satisfies $\Omega^n\not=0$, then $\Omega$ is symplectic.
\end{exe}

\begin{exe}
Let $(M,\omega)$ be a $2n$-dimensional symplectic manifold. Show that, if $M$ is compact, the de Rham cohomology class $[\omega^n]\in H^{2n}(M)$ is non-zero. \emph{Hint: Use Stokes' theorem}. Conclude then that $[\omega]$ is not exact and show that for $n>1$, there are no symplectic structures on the sphere $S^{2n}$.
\end{exe}

\section{Lagrangian submanifolds}

\subsection{Lagrangian submanifolds of $T^*X$}

\begin{defn}[Lagrangian submanifold]
Let $(M,\omega)$ be a $2n$-dimensional symplectic manifold. A submanifold $L\subset M$ is called \emph{Lagrangian} if for any point $q\in L$ we get that $T_qL$ is a Lagrangian subspace of $T_qM$.
\end{defn}

\begin{rem}
Recall that this is equivalent to say that $L\subset M$ is Lagrangian if and only if $\omega_q\big|_{T_qL}=0$ and $\dim T_qL=\frac{1}{2}\dim T_qM$ for all $q\in L$. Equivalently, if $i\colon L\hookrightarrow M$ denotes the inclusion of $L$ into $M$, then $L$ is Lagrangian if and only if $i^*\omega=0$ and $\dim L=\frac{1}{2}\dim M$.
\end{rem}

Let $X$ be an $n$-dimensional manifold and let $M:=T^*X$ be its cotangent bundle. consider coordinates $x_1,\ldots,x_n$ on $U\subseteq X$ with cotangent coordinates $x_1,\ldots,x_n,\xi_1,\ldots,\xi_n$ on $T^*U$, then the tautological 1-form on $T^*X$ is given by 
\[
\alpha=\sum_i\xi_i\dd x_i
\]
and the canonical 2-form on $T^*X$ is 
\[
\omega=-\dd\alpha=\sum_i\dd x_i\land \dd\xi_i.
\]

\subsection{Conormal bundle}
Let $S$ be any $k$-dimensional submanifold of $X$. 

\begin{defn}[Conormal space]
The \emph{conormal space} at $x\in S$ is defined by 
\[
N_x^*S:=\{\xi\in T^*_xX\mid \xi(v)=0,\,\forall v\in T_xS\}.
\]
\end{defn}

\begin{defn}[Conormal bundle]
The \emph{conormal bundle} of $S$ is
\[
N^*S=\{(x,\xi)\in T^*X\mid x\in S,\, \xi\in N^*_xS\}.
\]
\end{defn}

\begin{exe}
Show that the conormal bundle $N^*S$ is an $n$-dimensional submanifold of $T^*X$.
\end{exe}

\begin{prop}
\label{prop:conormal}
Let $i\colon N^*S\hookrightarrow T^*X$ be the inclusion, and let $\alpha$ be the tautological 1-form on $T^*X$. Then 
\[
i^*\alpha=0.
\]
\end{prop}

\begin{defn}[Adapted coordinate chart]
A coordinate chart $(\calU,x_1,\ldots,x_n)$ on $X$ is said to be \emph{adapted} to a $k$-dimensional submanfiold $S\subset X$ if $S\cap\calU$ is described by $x_{k+1}=\dotsm =x_n=0$.
\end{defn}

\begin{proof}[Proof of Proposition \ref{prop:conormal}]
Let $(\calU,x_1,\ldots,x_n)$ be a coordinate system on $X$ centered at $x\in S$ and adapted to $S$, so that $\calU\cap S$ is described by $x_{k+1}=\dotsm=x_n=0$. Let $(T^*\calU,x_1,\ldots,x_n,\xi_1,\dots,\xi_n)$ be the associated cotangent coordinate system. The submanifold $N^*S\cap T^*\calU$ is then described by 
\begin{align*}
    x_{k+1}&=\dotsm =x_n=0,\\
    \xi_1&=\dotsm=\xi_k=0.
\end{align*}
Since $\alpha=\sum_i\xi_i\dd x_i$ on $T^*\calU$, we conclude that for $p\in N^*S$ we get 
\[
(i^*\alpha)_p=\alpha_p\big|_{T_p(N^*S)}=\sum_{i>k}\xi_i\dd x_i\bigg|_{\mathrm{span}\left(\frac{\de}{\de x_i}\right)_{i\leq k}}=0.
\]
\end{proof} 

\begin{defn}[Zero section]
The \emph{zero section} of $T^*X$ is defined by 
\[
X_0:=\{(x,\xi)\in T^*X\mid \xi=0\in T^*_xX\}.
\]
\end{defn}

\begin{exe}
Let $i_0\colon X_0\hookrightarrow T^*X$ be the inclusion of the zero section and let $\omega=-\dd\alpha$ be the canonical symplectic form on $T^*X$. Show that $i_0^*\omega$ and $X_0$ are Lagrangian submanifolds of $T^*X$.
\end{exe}

\begin{cor}
\label{cor:conormal}
For any submanifold $S\subset X$, the conormal bundle $N^*S$ is a Lagrangian submanifold of $T^*X$.
\end{cor}

\begin{exe}
prove Corollary \ref{cor:conormal}.
\end{exe}

\begin{rem}
If $S=\{x\}\subset X$, then the conormal bundle is given by a cotangent fiber $N^*S=T_x^*X$. If $S=X$, then the conormal bundle is the zero section of $T^*X$, i.e. $N^*S=X_0$. 
\end{rem}

\subsection{Applications to symplectomorphisms}
Let $(M_1,\omega_1)$ and $(M_2,\omega_2)$ be two $2n$-dimensional symplectic manifolds. Given a diffeomorphism $\varphi\colon M_1\xrightarrow{\sim}M_2$, when is it a symplectomorphism? Equivalently, when do we have 
\[
\varphi^*\omega_2=\omega_1?
\]
Consider two projection maps
\[
\begin{tikzcd}
                          & M_1\times M_2 \arrow[ld, "\mathrm{pr}_1"'] \arrow[rd, "\mathrm{pr}_2"] &     \\
M_1 \arrow[rr, "\varphi"] &                                                                        & M_2
\end{tikzcd}
\]
Then 
\[
\omega:=(\mathrm{pr}_1)^*\omega_1+(\mathrm{pr}_2)^*\omega_2
\]
is a 2-form on $M_1\times M_2$ which is closed:
\[
\dd\omega=(\mathrm{pr}_1)^*\underbrace{\dd\omega_1}_{=0}+(\mathrm{pr}_2)^*\underbrace{\dd\omega_2}_{=0}=0,
\]
and symplectic:
\[
\omega^{2n}=\binom{2n}{n}\left((\mathrm{pr}_1)^*\omega_1\right)^{n}\land\left((\mathrm{pr}_2)^*\omega_2\right)^n\not=0.
\]
Note that for all $\lambda_1,\lambda_2\in\R$ we get that
\[
\lambda_1(\mathrm{pr}_1)^*\omega_1+\lambda_2(\mathrm{pr})^*\omega_2
\]
is a symplectic form on $M_1\times M_2$. The \emph{twisted product form} on $M_1\times M_2$ is obtained by taking $\lambda_1=1$ and $\lambda_2=-1$. Namely, 
\[
\tilde{\omega}=(\mathrm{pr}_1)^*\omega_1-(\mathrm{pr}_2)^*\omega_2. 
\]
For a diffeomorphism $\varphi\colon M_1\xrightarrow{\sim} M_2$, we define its \emph{graph} by 
\begin{equation}
\label{eq:Graph}
\Gamma_\varphi:=\mathrm{Graph}\, \varphi=\{(q,\varphi(q))\mid q\in M_1\}\subset M_1\times M_2.
\end{equation}
\begin{rem}
Note that $\dim \Gamma_\varphi =2n$ and that $\Gamma_\varphi$ is the embedded image of $M_1$ in $M_1\times M_2$. The embedding is given by the map 
\begin{align*}
    \gamma\colon M_1&\to M_1\times M_2,\\
    q&\mapsto (q,\varphi(q)).
\end{align*}
\end{rem}

\begin{prop}
\label{prop:Graph_Lagrangian}
A diffeomorphism $\varphi$ is a symplectomorphism if and only if $\Gamma_\varphi$ is a Lagrangian submanifold of $(M_1\times M_2,\tilde{\omega})$.
\end{prop}

\begin{proof}
The graph $\Gamma_\varphi$ is Lagrangian if and only if $\gamma^*\tilde{\omega}=0$. But we have
\[
\gamma^*\tilde{\omega}=\gamma^*(\mathrm{pr}_1)^*\omega_1-\gamma^*(\mathrm{pr}_2)^*\omega_2=(\mathrm{pr}_1\circ \gamma)^*\omega_1-(\mathrm{pr}_2\circ\gamma)^*\omega_2,
\]
where $\mathrm{pr_1}\circ \gamma$ is the identity map on $M_1$ whereas $\mathrm{pr}_2\circ \gamma=\varphi$. Hence
\[
\gamma^*\tilde{\omega}=0\Longleftrightarrow \varphi^*\omega_2=\omega_1.
\]
\end{proof}

\section{Local theory}

\subsection{Isotopies and vector fields}

Let $M$ be a manifold and $\rho\colon M\times \R\to M$ a map with $\rho_t(q):=\rho(q,t)$.

\begin{defn}[Isotopy]
The map $\rho$ is an \emph{isotopy} if each $\rho_t\colon M\to M$ is a diffeomorphism and $\rho_0=\id_M$.
\end{defn}

Given an isotopy $\rho$ we can construct a \emph{time-dependent} vector field. This means that we get a family of vector fields $X_t$ for $t\in\R$, which at $q\in M$ satisfy
\[
X_t(q)=\frac{\dd}{\dd s}\rho_s(p)\bigg|_{s=t},\qquad \forall p=\rho_t^{-1}(q).
\]
Basically, this means
\begin{equation}
\label{eq:isotopy_differential_equation}
\frac{\dd\rho_t}{\dd t}=X_t\circ\rho_t.
\end{equation}
Conversely, let $X_t$ be a time-dependent vector field and assume either that $M$ is compact or $X_t$ is compactly supported for all $t$. Then there exists an isotopy $\rho$ satisfying \eqref{eq:isotopy_differential_equation}. 

Moreover, if $M$ is compact, we have a one-to-one correspondence
\begin{align*}
    \{\text{isotopies of $M$}\}&\longleftrightarrow\{\text{time-dependent vector fields on $M$}\},\\
    (\rho_t)_{t\in\R}&\longleftrightarrow (X_t)_{t\in\R}.
\end{align*}

\begin{defn}[Exponential map]
If a vector field $X_t=X$ is time-independent, we call the associated isotopy the \emph{exponential map} of $X$ and we denote it by $\exp tX$. 
\end{defn}

\begin{rem}
Note that the family $\{\exp tX\colon M\to M\mid t\in\R\}$ is the unique smooth family of diffeomorphisms satisfying the Cauchy problem
\begin{align*}
\exp tX\big|_{t=0}&=\id_M,\\
\frac{\dd}{\dd t}(\exp tX)(q)&=X(\exp tX(q)).
\end{align*}
\end{rem}

\begin{rem}
The exponential map is the same as the \emph{flow} of a vector field (see Section \ref{sec:vector_fields_and_differential_1-forms}, Definition \ref{defn:flow}). The flow of a time-dependent vector field is given by the corresponding isotopy. 
\end{rem}

\begin{exe}
Show that for a time-dependent vector field $X_t$ and $\omega\in\Omega^s(M)$ we have 
\begin{equation}
\label{eq:Lie_derivative_isotopy}
\frac{\dd}{\dd t}\rho_t^*\omega=\rho_t^*L_{X_t}\omega,
\end{equation}
where $\rho$ is the (local) isotopy generated by $X_t$.
\end{exe}

\begin{prop}
\label{prop:change_pullback_isotopy}
For a smooth family $(\omega_t)_{t\in \R}$ of $s$-forms, we have
\begin{equation}
\label{eq:change_pullback_isotopy}
\boxed{
\frac{\dd}{\dd t}\rho_t^*\omega_t=\rho_t^*\left(L_{X_t}\omega_t+\frac{\dd}{\dd t}\omega_t\right).}
\end{equation}
\end{prop}

\begin{rem}
Equation \eqref{eq:change_pullback_isotopy} will turn out to be very useful for the proof of Moser's theorem (see Section \ref{sec:Moser's_theorem}, Theorem \ref{thm:Moser1})
\end{rem}

\begin{proof}[Proof of Proposition \ref{prop:change_pullback_isotopy}]
If $f(x,y)$ is a real function of two variables, we can use the chain rule to get
\[
\frac{\dd}{\dd t}f(t,t)=\frac{\dd}{\dd x}f(x,t)\bigg|_{x=t}+\frac{\dd}{\dd y}f(t,y)\bigg|_{y=t}.
\]
Hence, we get 
\[
\frac{\dd}{\dd t}\rho^*_t\omega_t=\underbrace{\frac{\dd}{\dd x}\rho_x^*\omega_t\bigg|_{x=t}}_{\rho_x^* L_{X_t}\omega_t\big|_{x=t}\text{ by \eqref{eq:Lie_derivative_isotopy}}}+\underbrace{\frac{\dd}{\dd y}\rho_t^*\omega_y\bigg|_{y=t}}_{\rho_t^*\frac{\dd}{\dd y}\omega_y\big|_{y=t}}=\rho_t^*\left(L_{X_t}\omega_t+\frac{\dd}{\dd t}\omega_t\right).
\]
\end{proof}

\subsection{Tubular neighborhood theorem}

Let $M$ be an $n$-dimensional manifold and let $X\subset M$ be a $k$-dimensional submanifold and consider the inclusion map
\[
i\colon X\hookrightarrow M.
\]
By the differential of the inclusion $\dd i_x\colon T_xX\hookrightarrow T_xM$, we have an inclusion of the tangent space of $X$ at a point $x\in X$ into the tangent space of $M$ at the point $x$. 

\begin{defn}[Normal space]
The \emph{normal space} to $X$ at the point $x\in X$ is given by the $(n-k)$-dimensional vector space defined by the quotient
\[
N_xX:=T_xM/T_xX.
\]
\end{defn}

\begin{defn}[Normal bundle]
The \emph{normal bundle} is then given by 
\[
NX:=\{(x,v)\mid x\in X,\, v\in N_xX\}.
\]
\end{defn}

\begin{rem}
Using the natural projection, $NX$ is a vector bundle over $X$ of rank $n-k$ and hence as a manifold it is $n$-dimensional. The zero section of $NX$
\begin{align*}
    i_0\colon X&\hookrightarrow NX,\\
    x&\mapsto (x,0),
\end{align*}
embeds $X$ as a closed submanifold of $NX$.
\end{rem}

\begin{defn}[Convex neighborhood]
A neighborhood $\calU_0$ of the zero section $X$ in $NX$ is called \emph{convex} if the intersection $\calU_0\cap N_xX$ with each fiber is convex.
\end{defn}

\begin{thm}[Tubular neighborhood theorem]
There exists a convex neighborhood $\calU_0$ of $X$ in $NX$, a neighborhood $\calU$ of $X$ in $M$, and a diffeomorphism $\varphi\colon \calU_0\to \calU$ such that the following diagram commutes:
\[
\begin{tikzcd}
NX\supseteq\mathcal{U}_0 \arrow[rr, "\varphi"] &                                      & \mathcal{U}\subseteq M \\
                                              & X \arrow[lu,hook', "i_0"] \arrow[ru,hook, "i"'] &                       
\end{tikzcd}
\]
\end{thm}

\begin{rem}
Restricting to the subset $\calU_0\subseteq NX$, we obtain a submersion $\calU_0\xrightarrow{\pi_0}X$ with all fibers $\pi_0^{-1}(x)$ convex. We can extend this fibration to $\calU$ by setting $\pi:=\pi_0\circ \varphi^{-1}$ such that if $NX\supseteq\calU_0\xrightarrow{\pi_0}X$ is a fibration, then $M\supseteq\calU\xrightarrow{\pi}X$ is a fibration. This is called the \emph{tubular neighborhood fibration}.

\end{rem}

\subsection{Homotopy formula}
Let $\calU$ be a tubular neighborhood of a submanifold $X\subset M$. The restriction of de Rham cohomology groups
\[
i^*\colon H^s(\calU)\to H^s(X)
\]
by the inclusion map is surjective. By the tubular neighborhood fibration, $i^*$ is also injective since the de Rham cohomology is homotopy invariant. In fact, we have the following corollary:

\begin{cor}
\label{cor:cohomology}
For any degree $s$ we have 
\[
H^s(\calU)\cong H^s(X).
\]
\end{cor}

\begin{rem}
Corollary \ref{cor:cohomology} says that if $\omega$ is a closed $s$-form on $\calU$ and $i^*\omega$ is exact on $X$, then $\omega$ is exact.
\end{rem}

\begin{prop}
If a closed $s$-form $\omega$ on $\calU$ has restriction $i^*\omega=0$, then $\omega$ is exact, i.e. $\omega=\dd\mu$ for some $\mu\in \Omega^{s-1}(\calU)$. Moreover, we can choose $\mu$ such that $\mu_x=0$ for all $x\in X$.
\end{prop}

\begin{proof}
By using the map $\varphi\colon \calU_0\xrightarrow{\sim}\calU$, we can work over $\calU_0$. For $t\in [0,1]$, define a map 
\begin{align*}
    \rho_t\colon \calU_0&\to \calU_0,\\
    (x,v)&\mapsto (x,tv).
\end{align*}
This is well-defined since $\calU_0$ is convex. The map $\rho_1$ is the identity and $\rho_0=i_0\circ\pi_0$. Moreover, each $\rho_t$ fixes $X$, i.e. $\rho_t\circ i_0=i_0$. Hence, we say that the family $(\rho_t)_{t\in[0,1]}$ is a \emph{homotopy} from $i_0\circ \pi_0$ to the identity fixing $X$. The map $\pi_0$ is called \emph{retraction} because $\pi_0\circ i_0$ is the identity. The submanifold $X$ is then called a \emph{deformation retract} of $\calU$. 

A (de Rham) \emph{homotopy operator} between $\rho_0=i_0\circ \pi_0$ and $\rho_1=\id$ is a linear map \[
Q\colon \Omega^s(\calU_0)\to \Omega^{s-1}(\calU_0)
\]
satisfying the \emph{homotopy formula}
\begin{equation}
\label{eq:homotopy}
\boxed{
\id-(i_0\circ \pi_0)^*=\dd Q+Q\dd.}
\end{equation}
When $\dd\omega=0$ and $i_0^*\omega=0$, the operator $Q$ gives $\omega=\dd Q\omega$, so that we can take $\mu=Q\omega$. A concrete operator $Q$ is given by the formula
\begin{equation}
\label{eq:Q_homotopy}
Q\omega=\int_0^1\rho_t^*(\iota_{X_t}\omega)\dd t,
\end{equation}
where $X_t$, at the point $p=\rho_t(q)$, is the vector tangent to the curve $\rho_s(q)$ at $s=t$. We claim that the operator \eqref{eq:Q_homotopy} satisfies the homotopy formula. Indeed, we compute
\[
Q\dd\omega+\dd Q\omega=\int_0^1\rho_t^*(\iota_{X_t}\dd\omega)\dd t+\dd\int_0^1\rho_t^*(\iota_{X_t}\omega)\dd t=\int_0^1\rho_t^*(\underbrace{\iota_{X_t}\dd\omega+\dd\iota_{X_t}\omega}_{=L_{X_t}\omega})\dd t.
\]
Hence the result follows from \eqref{eq:Lie_derivative_isotopy} and the fundamental theorem of analysis:
\[
Q\dd\omega+\dd Q\omega=\int_0^1\frac{\dd}{\dd t}\rho_t^*\omega \dd t=\rho_1^*\omega-\rho_0^*\omega.
\]
This completes the proof since, for our case, we have that $\rho_t(x)=x$ is a constant curve for all $x\in X$ and for all $t$, so $X_t$ vanishes at all $x$ and all $t$. Hence, $\mu_x=0$.
\end{proof}

\section{Moser's theorem}
\label{sec:Moser's_theorem}
\subsection{Equivalences for symplectic structures}

Let $M$ be a $2n$-dimensional manifold with two symplectic forms $\omega_0$ and $\omega_1$, so that $(M,\omega_0)$ and $(M,\omega_1)$ are two symplectic manifolds.

\begin{defn}[Symplectomorphic]
$(M,\omega_0)$ and $(M,\omega_1)$ are \emph{symplectomorphic} if there is a diffeomorphism $\varphi\colon M\to M$ with $\varphi^*\omega_1=\omega_0$.
\end{defn}

\begin{defn}[Strongly isotopic]
$(M,\omega_0)$ and $(M,\omega_1)$ are \emph{strongly isotopic} if there is an isotopy $\rho_t\colon M\to M$ such that $\rho^*_1\omega_1=\omega_0$.
\end{defn}

\begin{defn}[Deformation-equivalent]
$(M,\omega_0)$ and $(M,\omega_1)$ are \emph{deformation-equivalent} if there is a smooth family $\omega_t$ of symplectic forms joining $\omega_0$ to $\omega_1$.
\end{defn}

\begin{defn}[Isotopic]
$(M,\omega_0)$ and $(M,\omega_1)$ are \emph{isotopic} if they are deformation-equivalent with the de Rham cohomology classes $[\omega_t]$ independent of $t$.
\end{defn}

\begin{rem}
We have \emph{strongly isotopic} $\Longrightarrow$ \emph{symplectomorphic}, and \emph{isotopic} $\Longrightarrow$ \emph{deformation-equivalent}. Moreover, we have \emph{strongly isotopic} $\Longrightarrow$ \emph{isotopic}, because if $\rho_t\colon M\to M$ is an isotopy such that $\rho_1^*\omega_1=\omega_0$, then the set $\omega_t:=\rho_t^*\omega_1$ is a smooth family of symplectic forms joining $\omega_1$ to $\omega_0$ and $[\omega_t]=[\omega_1]$ for all $t$ by the homotopy invariance of the de Rham cohomology.
\end{rem}

\subsection{Moser's trick}

Consider the following problem: Let $M$ be a $2n$-dimensional manifold and let $X\subset M$ be a $k$-dimensional submanifold. Moreover, consider neighborhoods $\calU_0$ and $\calU_1$ of $X$ and symplectic forms $\omega_0$ and $\omega_1$ on $\calU_0$ and $\calU_1$ respectively. Does there exist a symplectomorphism preserving $X$? More precisely, does there exist a diffeomorphism $\varphi\colon \calU_0\to \calU_1$ with $\varphi^*\omega_1=\omega_0$ and $\varphi(X)=X$?\\

We want to consider the extreme case when $X=M$ and we consider $M$ to be compact with symplectic forms $\omega_0$ and $\omega_1$. So the question will change to: are $(M,\omega_0)$ and $(M,\omega_1)$ symplectomorphic, i.e. does there exist a diffeomorphism $\varphi\colon M\to M$ such that $\varphi^*\omega_1=\omega_0$? 

Moser's question was whether we can find such a $\varphi$ which is homotopic to the identity on $M$. A necessary condition is 
\[
[\omega_0]=[\omega_1]\in H^2(M)
\]
because if $\varphi\sim \id_M$ then, by the homotopy formula \eqref{eq:homotopy}, there exists a homotopy operator $Q$ such that 
\[
(\id_M)^*\omega_1-\varphi^*\omega_1=\dd Q\omega_1+Q\underbrace{\dd\omega_1}_{=0}.
\]
Thus, we get 
\[
\omega_1=\varphi^*\omega_1+\dd Q\omega_1
\]
and hence 
\[
[\omega_1]=[\varphi^*\omega_1]=[\omega_0].
\]

So we ask ourselves whether, if $[\omega_0]=[\omega_1]$, there exist a diffeomorphism $\varphi$ homotopic to $\id_M$ such that $\varphi^*\omega_1=\omega_0$. In \cite{Moser1965}, \emph{Moser} proved that, with certain assumptions, this is true. Later, \emph{McDuff} showed that, in general, this is not true by constructing a counterexample \cite[Example 7.23]{McDuffSalamon1995}.

\begin{thm}[Moser (version I)\cite{Moser1965}]
\label{thm:Moser1}
Suppose that $M$ is compact, $[\omega_0]=[\omega_1]$ and that the 2-form $\omega_t=(1-t)\omega_0+t\omega_1$ is symplectic for all $t\in [0,1]$. Then there exists an isotopy $\rho\colon M\times \R\to M$ such that $\rho^*_t\omega_t=\omega_0$.
\end{thm}

\begin{rem}
In particular, $\varphi:=\rho_t\colon M\to M$ satisfies $\varphi^*\omega_1=\omega_0$.
\end{rem}

The argument for the proof is known as \emph{Moser's trick}.

\begin{proof}[Proof of Theorem \ref{thm:Moser1}]
Suppose that there exists an isotopy $\rho\colon M\times \R\to M$ such that $\rho_t^*\omega_t=\omega_0$ for $t\in[0,1]$. Let 
\[
X_t:=\frac{\dd\rho_t}{\dd t}\circ \rho_t^{-1},\quad \forall t\in\R.
\]
Then 
\[
0=\frac{\dd}{\dd t}(\rho_t^*\omega_t)=\rho^*_t\left(L_{X_t}\omega_t+\frac{\dd}{\dd t}\omega_t\right)
\]
which is equivalent to 
\begin{equation}
\label{eq:Moser_trick}
L_{X_t}\omega_t+\frac{\dd}{\dd t}\omega_t=0.
\end{equation}
Suppose conversely that we can find a smooth time-dependent vector field $X_t$ for $t\in\R$, such that \eqref{eq:Moser_trick} holds for $t\in[0,1]$. Since $M$ is compact, we can integrate $X_t$ to an isotopy $\rho\colon M\times \R\to M$ with 
\[
\frac{\dd}{\dd t}(\rho^*_t\omega_t)=0,
\]
which implies that 
\[
\rho^*_t\omega_t=\rho^*_0\omega_0=\omega_0.
\]
This means that everything reduces to solving \eqref{eq:Moser_trick} for $X_t$. First, note that, from $\omega_t=(1-t)\omega_0+t\omega_1$, we get 
\[
\frac{\dd}{\dd t}\omega_t=\omega_1-\omega_0.
\]
Second, since $[\omega_0]=[\omega_1]$, there is a 1-form $\mu$ such that 
\[
\omega_1-\omega_0=\dd\mu.
\]
Third, by Cartan's magic formula \eqref{eq:Cartan_magic_formula}, we have 
\[
L_{X_t}\omega_t=\dd\iota_{X_t}\omega_t+\iota_{X_t}\underbrace{\dd\omega_{t}}_{=0}.
\]
putting everything together, we need to find $X_t$ such that 
\[
\dd\iota_{X_t}\omega_t+\dd\mu=0.
\]
Clearly, it is sufficient to solve 
\[
\iota_{X_t}\omega_t+\mu=0.
\]
By nondegeneracy of $\omega_t$, we can solve this pointwise to obtain a unique (smooth) $X_t$.

\end{proof}

\begin{thm}[Moser (version II)\cite{Moser1965}]
Let $M$ be a compact manifold with symplectic forms $\omega_0$ and $\omega_1$. Suppose that $(\omega_t)_{t\in [0,1]}$ is a smooth family of closed 2-forms joining $\omega_0$ and $\omega_1$ and satisfying: 
\begin{enumerate}
    \item (cohomology assumption) $[\omega_t]$ is independent of $t$, i.e. 
    \[
    \frac{\dd}{\dd t}[\omega_t]=\left[\frac{\dd}{\dd t}\omega_t\right]=0,
    \]
    \item (nondegeneracy condition) $\omega_t$ is nondegenerate for all $t\in[0,1]$.
\end{enumerate}
Then there exists an isotopy $\rho\colon M\times \R\to M$ such that 
\[
\rho^*_t\omega_t=\omega_0,\quad \forall t\in[0,1].
\]
\end{thm}

\begin{proof}[Proof (Moser's trick)]

We have the following implications: Condition (1) implies that there exists a family of 1-forms $\mu_t$ such that 
\[
\frac{\dd}{\dd t}\omega_t=\dd\mu_t,\quad \forall t\in[0,1].
\]
Indeed, we can find a smooth family of 1-forms $\mu_t$ such that $\frac{\dd}{\dd t}\omega_t=\dd\mu_t$. The argument uses a combination of the \emph{Poincar\'e lemma} for compactly-supported forms and the \emph{Mayer--Vietoris sequence} in order to use induction on the number of charts in a good cover of $M$ (see \cite{BottTu} for a detailed discussion of the Poincar\'e lemma and the Mayer--Vietoris construction). Condition (2) implies that there exists a unique family of vector fields $X_t$ such that
\begin{equation}
\label{eq:Mosers_equation}
\boxed{
\iota_{X_t}\omega_t+\mu_t=0.}
\end{equation}
Equation \eqref{eq:Moser_trick} is called \emph{Moser's equation}. We can extend $X_t$ to all $t\in\R$. Let $\rho$ be the isotopy generated by $X_t$ ($\rho$ exists by compactness of $M$). Then, using Cartan's magic formula and Moser's equation, we indeed have 
\[
\frac{\dd}{\dd t}(\rho^*_t\omega_t)=\rho^*_t\left(L_{X_t}\omega_t+\frac{\dd}{\dd t}\omega_t\right)=\rho^*_t(\dd\iota_{X_t}\omega_t+\dd\mu_t)=0.
\]
\end{proof}

\begin{rem}
Note that we have used compactness of $M$ to be able to integrate $X_t$ for all $t\in\R$. If $M$ is not compact, we need to check the existence of a solution $\rho_t$ for the differential equation 
\[
\frac{\dd}{\dd t}\rho_t=X_t\circ \rho_t,\quad \forall t\in[0,1].
\]
\end{rem}

\begin{thm}[Moser (relative version)\cite{Moser1965}]
Let $M$ be a manifold, $X$ a compact submanifold of $M$, $i\colon X\hookrightarrow M$ the inclusion map, $\omega_0$ and $\omega_1$ two symplectic forms on $M$. Then if $\omega_0\vert_q=\omega_1\vert_q$ for all $q\in X$, we get that there exist neighborhoods $\calU_0,\calU_1$ of $X$ in $M$ and a diffeomorphism $\varphi\colon \calU_0\to \calU_1$ such that $\varphi^*\omega_1=\omega_0$ and the following diagram commutes:
\[
\begin{tikzcd}
\mathcal{U}_0 \arrow[rr, "\varphi"] &                                      & \mathcal{U}_1 \\
                                              & X \arrow[lu,hook', "i"] \arrow[ru,hook, "i"'] &                       
\end{tikzcd}
\]
\end{thm}

\begin{proof}
Choose a tubular neighborhood $\calU_0$ of $X$. The 2-form $\omega_1-\omega_0$ is closed on $\calU_0$ and $(\omega_1-\omega_0)_q=0$ for all $q\in X$. By the homotopy formula on the tubular neighborhood, there exists a 1-form $\mu$ on $\calU_0$ such that $\omega_1-\omega_0=\dd\mu$ and $\mu_q=0$ at all $q\in X$. Consider a family $\omega_t=(1-t)\omega_0+t\omega_1=\omega_0+t\dd\mu$ of closed 2-forms on $\calU_0$. Shrinking $\calU_0$ if necessary, we can assume that $\omega_t$ is symplectic for $t\in[0,1]$. Then we can solve Moser's equation $\iota_{X_t}\omega_t=-\mu$ and note that $X_t\vert_X=0$. Shrinking $\calU_0$ again if necessary, there is an isotopy $\rho\colon \calU_0\times [0,1]\to M$ with $\rho^*_t\omega_t=\omega_0$ for all $t\in[0,1]$. Since $X_t\vert_X=0$, we have $\rho_t\vert_X=\id_X$. Then we can set $\varphi:=\rho_1$ and $\calU_1:=\rho_1(\calU_0)$.
\end{proof}

\begin{exe}
\label{exe:proof_of_Darboux}
Prove Darboux's theorem (Theorem \ref{thm:Darboux}) by using the relative version of Moser's theorem for $X=\{q\}$.
\end{exe}

\section{Weinstein tubular neighborhood theorem}

\subsection{Weinstein Lagrangian neighborhood theorem}

\begin{thm}[Weinstein Lagrangian neighborhood theorem\cite{Weinstein1971}]
\label{thm:Weinstein_Lagrangian_neighborhood_theorem}
Let $M$ be a $2n$-dimensional manifold, $X$ a compact $n$-dimensional submanifold $i\colon X\hookrightarrow M$ the inclusion map, and $\omega_0$ and $\omega_1$ symplectic forms on $M$ such that $i^*\omega_0=i^*\omega_1=0$, i.e. $X$ is a Lagrangian submanifold of both $(M,\omega_0)$ and $(M,\omega_1)$. Then there exists neighborhoods $\calU_0$ and $\calU_1$ of $X$ in $M$ and a diffeomorphism $\varphi\colon\calU_0\to \calU_1$ such that $\varphi^*\omega_1=\omega_0$ and the following diagram commutes:
\[
\begin{tikzcd}
\mathcal{U}_0 \arrow[rr, "\varphi"] &                                      & \mathcal{U}_1 \\
                                              & X \arrow[lu,hook', "i"] \arrow[ru,hook, "i"'] &                       
\end{tikzcd}
\]
\end{thm}

\begin{thm}[Whitney extension theorem\cite{Whitney1934}]
Let $M$ be an $n$-dimensional manifold and $X$ a $k$-dimensional submanifold with $k<n$. Suppose that at each $q\in X$ we are given a linear isomorphism $L_q\colon T_qM\xrightarrow{\sim}T_qM$ such that $L_q\vert_{T_qX}=\id_{T_qX}$ and $L_q$ depends smoothly on $q$. Then there exists an embedding $h\colon \calN\to M$ of some neighborhood $\calN$ of $X$ in $M$ such that $h\vert_X=\id_X$ and $\dd_q h=L_q$ for all $q\in X$.
\end{thm}

\begin{proof}[Proof of Theorem \ref{thm:Weinstein_Lagrangian_neighborhood_theorem}]
Let $g$ be a \emph{Riemannian metric} on $M$, i.e. at each $q\in M$ we get that $g_q$ is a positive-definite inner product. Fix $q\in X$ and let $V:=T_qM$, $U:=T_qX$ and $W:=U^\perp$ be the orthogonal complement of $U$ in $V$ relative to $g_q$. Since $i^*\omega_0=i^*\omega_1=0$, the space $U$ is a Lagrangian subspace of both $(V,\omega_0\vert_{q})$ and $(V,\omega_1\vert_q)$. By the symplectic linear algebra, we canonically get from $U^\perp$ a linear isomorphism $L_q\colon T_qM\to T_qM$ such that $L_q\vert_{T_qX}=\id_{T_qX}$ and $L_q^*\omega_1\vert_p=\omega_0\vert_q$. Note that $L_q$ varies smoothly with respect to $q$ since the construction is canonical. By the Whitney extension theorem, there exists a neighborhood $\calN$ of $X$ and an embedding $h\colon \calN\hookrightarrow X$ with $h\vert_X=\id_X$ and $\dd_q h=L_q$ for $q\in X$. Hence, at any $q\in X$, we have
\[
(h^*\omega_1)_q=(\dd_qh)^*\omega_1\vert_q=L_q^*\omega_1\vert_q=\omega_0\vert_q.
\]
applying the relative version of Moser's theorem to $\omega_0$ and $h^*\omega_1$, we find a neighborhood $\calU_0$ of $X$ and an embedding $f\colon \calU_0\to \calN$ such that $f\vert_X=\id_X$ and $f^*(h^*\omega_1)=\omega_0$ on $\calU_0$. Then we can set $\varphi:=h\circ f$.
\end{proof}

\begin{thm}[Coisotropic embedding theorem]
\label{thm:Coisotropic_embedding}
Let $M$ be a $2n$-dimensional manifold, $X$ a $k$-dimensional submanifold with $k<n$, $i\colon X\hookrightarrow M$ the inclusion map and $\omega_0$ and $\omega_1$ two symplectic forms on $M$ such that $i^*\omega_0=i^*\omega_1$ with $X$ being coisotropic for both $(M,\omega_0)$ and $(M,\omega_1)$. Then there exist neighborhoods $\calU_0$ and $\calU_1$ of $X$ in $M$ and a diffeomorphism $\varphi\colon \calU_0\to \calU$ such that $\varphi^*\omega_1=\omega_0$ and the following diagram commutes:
\[
\begin{tikzcd}
\mathcal{U}_0 \arrow[rr, "\varphi"] &                                      & \mathcal{U}_1 \\
                                              & X \arrow[lu,hook', "i"] \arrow[ru,hook, "i"'] &                       
\end{tikzcd}
\]
\end{thm}

\begin{exe}
Prove Theorem \ref{thm:Coisotropic_embedding}. \emph{Hint: See} \cite{Weinstein1977} \emph{for some inspiration and} \cite{GuilleminSternberg1977,Gotay1982} \emph{for a proof.}
\end{exe}

\subsection{Weinstein tubular neighborhood theorem}

\begin{thm}[Weinstein tubular neighborhood theorem\cite{Weinstein1971}]
\label{thm:Weinstein_tubular_neighborhood_theorem}
Let $(M,\omega)$ be a symplectic manifold, $X$ a compact Lagrangian submanifold of $M$, $\omega_0$ the canonical symplectic form on $T^*X$, $i_0\colon X\hookrightarrow T^*X$ the Lagrangian embedding as the zero section, and $i\colon X\hookrightarrow M$ the Lagrangian embedding given by the inclusion. Then there exist neighborhoods $\calU_0$ of $X$ in $T^*X$, $\calU$ of $X$ in $M$ and a diffeomorphism $\varphi\colon \calU_0\to \calU$ such that $\varphi^*\omega=\omega_0$ and the following diagram commutes:
\[
\begin{tikzcd}
\mathcal{U}_0 \arrow[rr, "\varphi"] &                                      & \mathcal{U} \\
                                              & X \arrow[lu,hook', "i_0"] \arrow[ru,hook, "i"'] &                       
\end{tikzcd}
\]
\end{thm}

\begin{rem}
A similar statement for isotropic submanifolds was also proved by Weinstein in \cite{Weinstein1977,Weinstein1981}.
\end{rem}

\begin{proof}[Proof of Theorem \ref{thm:Weinstein_tubular_neighborhood_theorem}]
The proof uses the tubular neighborhood theorem and the Weinstein Lagrangian neighborhood theorem. We first use the tubular neighborhood theorem. Since $NX\cong T^*X$, we can find a neighborhood $\calN_0$ of $X$ in $T^*X$, a neighborhood $\calN$ of $X$ in $M$ and a diffeomorphism $\psi\colon \calN_0\to \calN$ such that the following diagram commutes:
\[
\begin{tikzcd}
\mathcal{N}_0 \arrow[rr, "\psi"] &                                      & \mathcal{N} \\
                                              & X \arrow[lu,hook', "i_0"] \arrow[ru,hook, "i"'] &                       
\end{tikzcd}
\]
Let $\omega_0$ be the canonical symplectic form on $T^*X$ and $\omega_1:=\psi^*\omega$. Note that $\omega_0$ and $\omega_1$ are symplectic forms on $\calN_0$. The submanifold $X$ is Lagrangian for both $\omega_0$ and $\omega_1$. 
Now we use the Weinstein Lagrangian neighborhood theorem. There exist neighborhoods $\calU_0$ and $\calU_1$ of $X$ in $\calN_0$ and a diffeomorphism $\theta\colon \calU_0\to \calU_1$ such that $\theta^*\omega_1=\omega_0$ and the following diagram commutes:
\[
\begin{tikzcd}
\mathcal{U}_0 \arrow[rr, "\theta"] &                                      & \mathcal{U}_1 \\
                                              & X \arrow[lu,hook', "i_0"] \arrow[ru,hook, "i_0"'] &                       
\end{tikzcd}
\]
Now we can take $\varphi=\psi\circ \theta$ and $\calU_1=\varphi(\calU_0)$. It is easy to check $\varphi^*\omega=\theta^*\underbrace{\psi^*\omega}_{\omega_1}=\omega_0$.
\end{proof}

\subsection{Application}

\begin{defn}[$C^1$-topology]
Let $X$ and $Y$ be two manifolds. A sequence of $C^1$ maps $f_i\colon X\to Y$ is said to \emph{converge in the $C^1$-topology} to a map $f\colon X\to Y$ if and only if the sequence $(f_i)$ itself and the sequence of the differentials $\dd f_i\colon TX\to TY$ converge uniformly on compact sets.
\end{defn}

\begin{rem}
Let $(M,\omega)$ be a symplectic manifold. Note that the graph of the identity map, denoted by $\Delta:=\Gamma_{\id_M}$ (recall Equation \eqref{eq:Graph} for the definition), is a Lagrangian submanifold of $(M\times M,(\mathrm{pr}_1)^*\omega-(\mathrm{pr_2})^*\omega)$ (see also Proposition \ref{prop:Graph_Lagrangian}). Moreover, we say that $f$ is \emph{$C^1$-close} to another map $g$, if $f$ is in some small neighborhood of $g$ in the $C^1$-topology.
\end{rem}

By the Weinstein tubular neighborhood theorem, there is a neighborhood $\calU$ of 
\[
\Delta\subset (M\times M,(\mathrm{pr}_1)^*\omega-(\mathrm{pr}_2)^*\omega)
\]
which is symplectomorphic to a neighborhood $\calU_0$ of $M$ in $(T^*M,\omega_0)$. Let $\varphi\colon \calU\to \calU_0$ be the symplectomorphism satisfying $\varphi(q,q)=(q,0)$ for all $q\in M$. 

Denote by 
\[
\mathrm{Sympl}(M,\omega):=\{f\colon M\xrightarrow{\sim}M\mid f^*\omega=\omega\}
\]
and suppose that $f\in \mathrm{Sympl}(M,\omega)$ is sufficiently \emph{$C^1$-close} to the identity, i.e. $f$ is in some sufficiently small neighborhood of the identity $\id_M$ in the $C^1$-topology. Then we can assume that the graph of $f$ lies inside of $\calU$. Let $j\colon M\hookrightarrow\calU$ be the embedding as $\Gamma_f$ and $i\colon M\hookrightarrow \calU$ the embedding as $\Gamma_{\id_M}=\Delta$. The map $j$ is sufficiently $C^1$-close to $i$. By the Weinstein theorem, $\calU\simeq\calU_0\subseteq T^*M$, so the above $j$ and $i$ induce two embeddings: $j_0\colon M\hookrightarrow \calU_0$ where $j_0=\varphi\circ j$ and $i_0\colon M\hookrightarrow \calU_0$ embedding as $0$-section. Hence, we have 
\[
\begin{tikzcd}
\mathcal{U} \arrow[rr, "\varphi"] &                                      & \mathcal{U}_0 \\
                                              & M \arrow[lu,hook', "i"] \arrow[ru,hook, "i_0"'] &                       
\end{tikzcd}
\]
\[
\begin{tikzcd}
\mathcal{U} \arrow[rr, "\varphi"] &                                      & \mathcal{U}_0 \\
                                              & M \arrow[lu,hook', "j"] \arrow[ru,hook, "j_0"'] &                       
\end{tikzcd}
\]
where $i(q)=(q,q)$, $i_0(q)=(q,0)$, $j(q)=(q,f(q))$ and $j_0(q)=\varphi(q,f(q))$ for $q\in M$. The map $j_0$ is sufficiently $C^1$-close to $i_0$. Thus, the image set $j_0(M)$ intersects each $T^*_qM$ at one point $\mu_q$ depending smoothly on $q$. The image of $j_0$ is the image of a smooth intersection $\mu\colon M\to T^*M$, that is, a 1-form $\mu=j_0\circ(\pi\circ j_0)^{-1}$. Therefore, 
\begin{equation}
\label{eq:graph}
\Gamma_f\simeq \{(q,\mu_q)\mid q\in M,\,\, \mu_q\in T^*_qM\}.
\end{equation}
Vice-versa, if $\mu$ is a 1-form sufficiently $C^1$-close to the zero 1-form, then we also get \eqref{eq:graph}. In fact, we get 
\[
\text{$\Gamma_f$ is Lagrangian}\Longleftrightarrow \text{$\mu$ is closed}.
\]

Hence, we can conclude that a small $C^1$-neighborhood of the identity in $\mathrm{Sympl}(M,\omega)$ is homeomorphic to a $C^1$-neighborhood of zero in the vector space of closed 1-forms on $M$. Thus, we have 
\[
T_{\id_M}(\mathrm{Sympl}(M,\omega))\simeq \{\mu\in \Omega^1(M)\mid \dd\mu=0\}.
\]
In particular, $T_{\id_M}(\mathrm{Sympl}(M,\omega))$ contains the space of exact 1-forms 
\[
\{\mu=\dd h\mid h\in C^\infty(M)\}\simeq C^\infty(M)/\{\text{locally constant functions}\}.
\]

\begin{thm}
Let $(M,\omega)$ be a compact symplectic manifold with $H^1(M)=0$. Then any symplectomorphism of $M$ which is sufficiently $C^1$-close to the identity has at least two fixed points.
\end{thm}

\begin{proof}
Suppose that any $f\in \mathrm{Sympl}(M,\omega)$ is sufficiently $C^1$-close to the identity $\id_M$. Then the graph $\Gamma_f$ is homotopic to a closed 1-form on $M$. Now the fact that $\dd\mu=0$ and $H^1(M)=0$ imply that $\mu=\dd h$ for some $h\in C^\infty(M)$. Since $M$ is compact, $h$ has at least two critical points. Note that fixed points of $f$ are equal to critical points of $h$. Moreover, fixed points of $f$ are equal to the intersection of $\Gamma_f$ with the diagonal $\Delta\subset M\times M$ which is again equal to $S:=\{q\mid \mu_q=\dd_qh=0\}$. Finally, note that the critical points of $h$ are exactly points in $S$.
\end{proof}

We say that a submanifold $Y$ of $M$ is \emph{$C^1$-close} to another submanifold $X$ of $M$, when there is a diffeomorphism $X\to Y$ which is, as a map into $M$, $C^1$-close to the inclusion $X\hookrightarrow M$.

\begin{thm}[Lagrangian intersection theorem]
\label{thm:Lagrangian_intersection}
Let $(M,\omega)$ be a symplectic manifold. Suppose that $X$ is a compact Lagrangian submanifold of $M$ with $H^1(X)=0$. Then every Lagrangian submanifold of $M$ which is $C^1$-close to $X$ intersects $X$ in at least two points. 
\end{thm}

\begin{exe}
Prove Theorem \ref{thm:Lagrangian_intersection}.
\end{exe}

\begin{defn}[Morse function]
\label{defn:Morse_function}
A \emph{Morse function} on a manifold $M$ is a function $h\colon M\to \R$ whose critical points are all nondegenerate, i.e. the Hessian at a critical point $q$ is not singular, $\det\left(\frac{\de^2h}{\de x_i\de x_j}\big|_q\right)\not=0$.
\end{defn}

Let $(M,\omega)$ be a symplectic manifold and suppose that $h_t\colon M\to \R$ is a smooth family of functions which is \emph{1-periodic}, i.e. $h_t=h_{t+1}$. Let $\rho\colon M\times \R\to M$ be the isotopy generated by the time-dependent vector field $X_t$ defined by $\omega(X_t,\enspace)=\dd h_t$. Then we say that $f$ is \emph{exactly homotopic to the identity} if $f=\rho_1$ for such $h_t$. Equivalently, using the notions which will be introduced in Section \ref{sec:Hamiltonian_mechanics}, we have the following definition:

\begin{defn}[Exactly homotopic to the identity]
A symplectomorphism $f\in \mathrm{Sympl}(M,\omega)$ is \emph{exactly homotopic to the identity} when $f$ is the time-1 map of an isotopy generated by some smooth time-dependent 1-periodic Hamiltonian function.
\end{defn}

\begin{defn}[Nondegenerate fixed point]
A fixed point $q$ of a function $f\colon M\to M$ is called \emph{nondegenerate} if $\dd_qf\colon T_qM\to T_qM$ is not singular.
\end{defn}

\begin{conj}[Arnold]
Let $(M,\omega)$ be a compact symplectic manifold, and $f\colon M\to M$ a symplectomorphism which is exactly homotopic to the identity. Then 
\[
\vert\{\text{fixed points of $f$}\}\vert\geq \min \vert \{\text{critical points of a smooth function on $M$}\}\vert.
\]
Using the notion of a \emph{Morse function} as in Definition \ref{defn:Morse_function}, we obtain
\begin{align*}
\vert\{\text{nondegenerate fixed points of $f$}\}\vert&\geq \min \vert\{\text{critical points of a Morse function on $M$}\}\vert\\&\geq \sum_{i=0}^{2n}\dim H^i(M,\R).
\end{align*}
\end{conj}

\begin{rem}
The Arnold conjecture has been proven by Conley--Zehnder, Floer, Hofer--Salamon, Ono, Fukaya--Ono, Liu--Tian by using \emph{Floer homology}. This is an infinite-dimensional version of \emph{Morse theory}. Sharper bound versions of the Arnold conjecture are still open.
\end{rem}

\begin{exe}
Compute the estimates for the number of fixed points on the compact symplectic manifolds $S^2$, $S^2\times S^2$ and $T^2:=S^1\times S^1$.
\end{exe}

\section{Hamiltonian mechanics}
\label{sec:Hamiltonian_mechanics}
\subsection{Hamiltonian and symplectic vector fields}
Let $(M,\omega)$ be a symplectic manifold and let $H\colon M\to \R$ be a smooth function. By nondegneracy of $\omega$, there is a unique vector field $X_H$ on $M$ such that 
\[
\iota_{X_H}\omega=\dd H.
\]
Suppose that $M$ is compact or that $X_H$ is complete. Let $\rho_t\colon M\to M$ for $t\in\R$ be the 1-parameter family of diffeomorphisms generated by $X_H$, i.e. satisfying the Cauchy problem
\begin{align*}
    \rho_0&=\id_M,\\
    \frac{\dd}{\dd t}\rho_t&=X_H\circ \rho_t.
\end{align*}

In fact, each diffeomorphism $\rho_t$ preserves $\omega$, i.e. $\rho^*_t\omega=\omega$ for all $t$. Indeed, note that
\[
\frac{\dd}{\dd t}\rho^*_t\omega=\rho_t^*L_{X_H}\omega=\rho^*_t(\dd\underbrace{\iota_{X_H}\omega}_{=\dd H}+\iota_{X_H}\underbrace{\dd\omega}_{=0})=0
\]

\begin{defn}[Hamiltonian vector field and Hamiltonian function]
A vector field $X_H$ as above is called the \emph{Hamiltonian vector field} with \emph{Hamiltonian function} $H$. 
\end{defn}

\begin{exe}
Let $X$ be a vector field on some manifold $M$. Then there is a unique vector field $X_\sharp$ on the cotangent bundle $T^*M$ whose flow is the lift of the flow of $X$. Let $\alpha$ be the tautological 1-form on $T^*M$ and let $\omega=-\dd \alpha$ be the canonical symplectic form on $T^*M$. Show that $X_\sharp$ is a Hamiltonian vector field with Hamiltonian function $H:=\iota_{X_\sharp}\alpha$.
\end{exe}

\begin{rem}
If $X_H$ is Hamiltonian, we get
\[
L_{X_H}H=\iota_{X_H}\dd H=\iota_{X_H}\iota_{X_H}\omega=0.
\]
Hence, Hamiltonian vector fields preserve their Hamiltonian functions and each integral curve $(\rho_t(x))_{t\in\R}$ of $X_H$ must be contained in a level set of $H$, i.e. 
\[
H(x)=(\rho^*_tH)(x)=H(\rho_t(x)),\quad \forall t.
\]
\end{rem}

\begin{defn}[Symplectic vector field]
A vector field $X$ on a symplectic manifold $(M,\omega)$ preserving $\omega$, i.e. $L_X\omega=0$, is called \emph{symplectic}.
\end{defn}

\begin{rem}
Note that $X$ is \emph{symplectic} if and only if $\iota_X\omega$ is closed and $X$ is \emph{Hamiltonian} if and only if $\iota_X\omega$ is exact.
\end{rem}

\begin{rem}
Locally, on every contractible set, every symplectic vector field is Hamiltonian. If $H^1(M)=0$, then globally every symplectic vector field is Hamiltonian. In general, $H^1(M)$ measures the obstruction for symplectic vector fields to be Hamiltonian.
\end{rem}

\subsection{Classical mechanics}
\label{subsec:CM}
Consider the Euclidean space $\R^{2n}$ with coordinates\\ $(q_1,\ldots,q_n,p_1,\ldots,p_n)$ and $\omega_0=\sum_j\dd q_j\land \dd p_j$. If the \emph{Hamilton equations}
\begin{align}
    \begin{split}
        \frac{\dd q_i}{\dd t}(t)&=\frac{\de H}{\de p_i},\\
        \frac{\dd p_i}{\dd t}(t)&=-\frac{\de H}{\de q_i}
    \end{split}
\end{align}
are satisfied, then the curve $\rho_t=(q(t),p(t))$ is an integral curve for $X_H$. Indeed, let 
\[
X_H=\sum_{i=1}^n\left(\frac{\de H}{\de p_i}\frac{\de}{\de q_i}-\frac{\de H}{\de q_i}\frac{\de}{\de p_i}\right).
\]
Then 
\begin{multline*}
    \iota_{X_H}\omega=\sum_{j=1}^n\iota_{X_H}(\dd q_j\land \dd p_j)=\sum_{j=1}^n\left[(\iota_{X_H}\dd q_j)\land \dd p_j-\dd q_j\land (\iota_{X_H}\dd p_j)\right]\\
    =\sum_{j=1}^n\left(\frac{\de H}{\de p_j}\dd p_j+\frac{\de H}{\de q_j}\dd q_j\right)=\dd H.
\end{multline*}

Consider the case when $n=3$. By Newton's second law, a particle with mass $m\in \R_{>0}$ moving in \emph{configuration space} $\R^3$ with coordinates $q:=(q_1,q_2,q_3)$ in a potential $V(q)$ moves along a curve $q(t)$ satisfying 
\[
m\frac{\dd^2q}{\dd t^2}=-\nabla V(q).
\]
One then introduces momentum coordinates $p_i:=m\frac{\dd q_i}{\dd t}$ for $i=1,2,3$ and a total energy function 
\begin{align*}
H(q,p)&=\text{kinetic energy} + \text{potential energy}\\
&=\frac{1}{2m}\| p\|^2+V(q)
\end{align*}

The \emph{phase space} is then given by $T^*\R^3=\R^6$ with coordinates $(q_1,q_2,q_3,p_1,p_2,p_3)$. Newton's second law in $\R^3$ is equivalent to the Hamilton equations in $\R^6$:
\begin{align*}
    \frac{\dd q_i}{\dd t}&=\frac{1}{m}p_i=\frac{\de H}{\de p_i},\\
    \frac{\dd p_i}{\dd t}&=m\frac{\dd^2 q_i}{\dd t^2}=-\frac{\de V}{\de q_i}=-\frac{\de H}{\de q_i}.
\end{align*}

\begin{rem}
The total energy $H$ is \emph{conserved} by the physical motion, i.e. $\frac{\dd}{\dd t}H=0$.
\end{rem}

\subsection{Brackets}

Recall that vector fields are differential operators on functions. In particular, if $X$ is a vector field and $f\in C^\infty(M)$, with $\dd f$ being the corresponding 1-form, then 
\[
X(f)=\dd f(X)=L_Xf.
\]
For two vector fields $X,Y$, there is a unique vector field $Z$ such that 
\[
L_Zf=L_X(L_Yf)-L_Y(L_Xf).
\]
The vector field $Z$ is called the \emph{Lie bracket} of the vector fields $X$ and $Y$ and is denoted by $Z=[X,Y]$, since $L_Z=[L_X,L_Y]$ is the commutator (see also Definition \ref{defn:commutator}).

\begin{prop}
If $X$ and $Y$ are symplectic vector fields on a symplectic manifold $(M,\omega)$, then $[X,Y]$ is a Hamiltonian vector field with Hamiltonian function $\omega(Y,X)$.
\end{prop}

\begin{proof}
We have 
\begin{align*}
    \iota_{[X,Y]}\omega&=L_X\iota_Y\omega-\iota_YL_X\omega\\
    &=\dd\iota_X\iota_Y\omega+\iota_X\underbrace{\dd\iota_Y\omega}_{=0}-\iota_Y\underbrace{\dd\iota_X\omega}_{=0}-\iota_Y\iota_X\underbrace{\dd\omega}_{=0}\\
    &=\dd(\omega(Y,X)).
\end{align*}
\end{proof}

\begin{defn}[Lie algebra]
A \emph{Lie algebra} is a vector space $\mathfrak{g}$ together with a \emph{Lie bracket} $[\enspace,\enspace]$, i.e. a bilinear map $[\enspace,\enspace]\colon\mathfrak{g}\times \mathfrak{g}\to \mathfrak{g}$ such that
\begin{enumerate}
    \item $[X,Y]=-[Y,X]$, $\forall X,Y\in\mathfrak{g}$ (antisymmetry)
    \item $[X,[Y,Z]]+[Y,[Z,X]]+[Z,[X,Y]]=0$, $\forall X,Y,Z\in \mathfrak{g}$ (Jacobi identity)
\end{enumerate}
\end{defn}

Denote by $\mathfrak{X}_H(M)$ the space of ]\emph{Hamiltonian vector fields} and by $\mathfrak{X}_S(M)$ the space of \emph{symplectic vector fields} on a manifold $M$.

\begin{cor}
The inclusions
\[
(\mathfrak{X}_H(M),[\enspace,\enspace])\subseteq (\mathfrak{X}_S(M),[\enspace,\enspace])\subseteq (\mathfrak{X}(M),[\enspace,\enspace])
\]
are inclusions of Lie algebras.
\end{cor}

\begin{defn}[Poisson bracket]
\label{defn:Poisson_bracket}
The \emph{Poisson bracket} of two functions $f,g\in C^\infty(M)$ on a symplectic manifold $(M,\omega)$ is defined by 
\begin{equation}
\label{eq:Poisson_bracket}
\{f,g\}:=\omega(X_f,X_g).
\end{equation}
\end{defn}

\begin{rem}
Note that we have $X_{\{f,g\}}=-[X_f,X_g]$ since $X_{\omega(X_f,X_g)}=[X_g,X_f]$.
\end{rem}

\begin{exe}
Show that the Poisson bracket $\{\enspace,\enspace\}$ satisfies the \emph{Jacobi identity}, i.e.
\[
\{f,\{g,h\}\}+\{g,\{h,f\}\}+\{h,\{f,g\}\}=0,\quad \forall f,g,h\in C^\infty(M).
\]
\end{exe}

\begin{defn}[Poisson algebra]
A \emph{Poisson algebra} is a commutative associative algebra $\calA$ with a Lie bracket $\{\enspace,\enspace\}\colon \calA\times\calA\to \calA$ satisfying the \emph{Leibniz rule}
\[
\{f,gh\}=\{f,g\}h+g\{f,h\},\quad \forall f,g,h\in \calA.
\]
\end{defn}

\begin{exe}
Show that the Poisson bracket defined as in \eqref{eq:Poisson_bracket} is satisfies the Leibniz rule and deduce that if $(M,\omega)$ is a symplectic manifold, then $(C^\infty(M),\{\enspace,\enspace\})$ is a Poisson algebra. 
\end{exe}

\begin{rem}
Note that we have a Lie algebra \emph{anti-homomorphism}
\begin{align*}
C^\infty(M)&\to\mathfrak{X}(M),\\
H&\mapsto X_H
\end{align*}
such that the Poisson bracket $\{\enspace,\enspace\}$ will correspond to $-[\enspace,\enspace]$.
\end{rem}

\subsection{Integrable systems}

\begin{defn}[Hamiltonian system]
A \emph{Hamiltonian system} is a triple $(M,\omega,H)$, where $(M,\omega)$ is a symplectic manifold and $H\in C^\infty(M)$ is a \emph{Hamiltonian function}.
\end{defn}

\begin{thm}
We have $\{f,H\}=0$ if and only if $f$ is constant along integral curves of $X_H$.
\end{thm}

\begin{proof}
Let $\rho_t$ be the flow of $X_H$. Then 
\begin{equation}
    \frac{\dd}{\dd t}(f\circ \rho_t)=\rho_t^*L_{X_H}f=\rho_t^*\iota_{X_H}\dd f=\rho^*_t\iota_{X_H}\iota_{X_f}\omega=\rho_t^*\omega(X_f,X_H)=\rho^*_t\{f,H\}.
\end{equation}
\end{proof}

\begin{rem}
Given a Hamiltonian system $(M,\omega,H)$, a function $f$ satisfying $\{f,H\}=0$ is called an \emph{integral of motion} or a  \emph{constant of motion}. In general, Hamiltonian system do not admit integrals of motion which are \emph{independent} of the Hamiltonian function. 
\end{rem}

\begin{defn}[Independent functions]
We say that functions $f_1,\ldots,f_n$ on a manifold $M$ are independent if their differentials $\dd_qf_1,\ldots, \dd_qf_n$ are linearly independent at all points $q\in M$ in some open dense subset of $M$.
\end{defn}

\begin{defn}[(Completely) integrable system]
A Hamiltonian system $(M,\omega,H)$ is \emph{(completely) integrable} if it has $n=\frac{1}{2}\dim M$ independent integrals of motion $f_1=H,f_2\ldots,f_n$, which are pairwise in \emph{involution} with respect to the Poisson bracket, i.e.
\[
\{f_i,f_j\}=0,\quad \forall i,j.
\]
\end{defn}

Let $(M,\omega,H)$ be an integrable system of dimension $2n$ with integrals of motion $f_1=H,f_2,\ldots,f_n$. Let $c\in \R^n$ be a regular value of $f:=(f_1,\ldots,f_n)$. Note that the corresponding \emph{level set}, $f^{-1}(c)$, is a Lagrangian submanifold, because it is $n$-dimensional and its tangent bundle is isotropic.

\begin{lem}
\label{lem:level_sets}
If the Hamiltonian vector fields $X_{f_1},\ldots, X_{f_n}$ are complete on the level sets $f^{-1}(c)$, then the connected components of $f^{-1}(c)$ are \emph{homogeneous spaces} for $\R^n$, i.e. are of the form $\R^{n-k}\times T^k$ for some $k$ with $0\leq k\leq n$, where $T^k$ denotes the $k$-torus.
\end{lem}

\begin{exe}
Prove Lemma \ref{lem:level_sets}. \emph{Hint: follow the flows to obtain coordinates}.
\end{exe}

Note that any compact component of $f^{-1}(c)$ must be a torus. These components, when they exist, are called \emph{Liouville tori}.

\begin{thm}[Arnold--Liouville\cite{Arnold1978,Liouville1855}]
Let $(M,\omega,H)$ be an integrable system of dimension $2n$ with integrals of motion $f_1=H,f_2,\ldots,f_n$. Let $c\in \R^n$ be a regular value of $f:=(f_1,\ldots,f_n)$. The corresponding level set $f^{-1}(c)$ is a Lagrangian submanifold of $M$.
\begin{enumerate}
    \item If the flows of $X_{f_1},\ldots, X_{f_n}$ starting at a point $q\in f^{-1}(c)$ are complete, then the connected component of $f^{-1}(c)$ containing $q$ is a homogeneous space for $\R^n$. With respect to this affine structure, that component has coordinates $\phi_1,\ldots,\phi_n$, known as \emph{angle coordinates}, in which the flows of the vector fields $X_{f_1},\ldots,X_{f_n}$ are linear.
    \item There are coordinates $\psi_1,\ldots,\psi_n$, known as \emph{action coordinates}, complementary to the angle coordinates such that the $\psi_i$s are integrals of motion and 
    \[
    \phi_1,\ldots,\phi_n,\psi_1,\ldots,\psi_n 
    \]
    form a Darboux chart.
\end{enumerate}
\end{thm}

\begin{exe}[Pendulum]
The \emph{pendulum} is a mechanical system consisting of a massless rod of length $\ell$, where one end is fixed and the other end has some mass $m$ attached to it which can oscillate in a vertical plane. Assume that gravity is constant pointing vertically downwards, and that it is the only external force acting on this system.
\begin{enumerate}
    \item Let $\theta$ be the oriented angle between the rod and the vertical direction. Let $\xi$ be the coordinate along the fibers of $T^*S^1$ induced by the standard angle coordinate on $S^1$. Show that the function 
    \begin{align*}
        H\colon T^*S^1&\to \R,\\
        (\theta,\xi)&\mapsto \underbrace{\frac{\xi^2}{2m\ell^2}}_{=:K}+\underbrace{m\ell(1-\cos\theta)}_{=:V}
    \end{align*}
    is an appropriate Hamiltonian function to describe the pendulum. More precisely, check that gravity corresponds to the potential energy $V(\theta)=m\ell(1-\cos\theta)$ (universal constants are omitted), and that the kinetic energy is given by $K(\theta,\xi)=\frac{\xi^2}{2m\ell^2}$.
    \item For simplicity, assume that $m=\ell=1$. draw the level curves of $H$ in the $(\theta,\xi)$-plane. Show that there is a real number $c\in\R$ such that for some real number $h\in\R$ with $0<h<c$ the level curve $H=h$ is a disjoint union of closed curves. Show that the projection of each of these curves onto the $\theta$-axis is an interval of length less than $\pi$. Show that neither of these assertions is true if $h>c$. What types of motion are described by these two types of curves? What about the case $H=c$?
    \item compute the critical points of $H$. Show that, modulo $2\pi$ in $\theta$, there are exactly two critical points: a critical point $s$ where $H$ vanishes, and a critical point $u$ where $H$ equals $c$. These points are called \emph{stable} and \emph{unstable} points of $H$, respectively. Justify this terminology, i.e., show that a trajectory of the Hamiltonian vector field of $H$ whose initial point is close to $s$ stays close to $s$ forever, and show that this is not the case for $u$. What is happening physically?
\end{enumerate}
\end{exe}

\newpage
\section{Moment maps}

\subsection{Smooth actions}

\begin{defn}[Lie group]
A \emph{Lie group} is a manifold $G$ equipped with a group structure where the group operations \emph{multiplication} and \emph{taking the inverse} are smooth maps.
\end{defn}

\begin{defn}[Representation]
A \emph{representation} of a Lie group $G$ on a vector space $V$ is a group homomorphism 
\[
G\to \GL(V).
\]
\end{defn}

\begin{defn}[Action]
Let $M$ be a manifold and denote by 
\[
\mathrm{Diff}(M):=\{\varphi\colon M\xrightarrow{\sim}M\mid \varphi \text{ diffeomorphism}\}
\]
the \emph{diffeomorphism group} of $M$. An \emph{action} of a Lie group $G$ on $M$ is a group homomorphism
\begin{align*}
    \Psi\colon G&\to \mathrm{Diff}(M),\\
    g&\mapsto \Psi_g.
\end{align*}
\end{defn}

\begin{rem}
We will only consider \emph{left actions} where $\Psi$ is a Homomorphism. A \emph{right action} is defined with $\Psi$ being an anti-homomorphism.
\end{rem}

\begin{defn}[Evaluation map]
The \emph{evaluation map} associated with an action 
\[
\Psi\colon G\to \mathrm{Diff}(M)
\]
is 
\begin{align*}
    \ev_\Psi\colon M\times G&\to M,\\
    (q,g)&\mapsto \Psi_g(q).
\end{align*}
\end{defn}

\begin{rem}
The action $\Psi$ is \emph{smooth} if $\ev_\Psi$ is smooth.
\end{rem}

\subsection{Symplectic and Hamiltonian actions}

Let $(M,\omega)$ be a symplectic manifold, and $G$ a Lie group. Let $\Psi\colon G\to \mathrm{Diff}(M)$ be a (smooth) action. 

\begin{defn}[Symplectic action]
The action $\Psi$ is a \emph{symplectic action} if 
\[
\Psi\colon G\to \mathrm{Sympl}(M,\omega)\subset \mathrm{Diff}(M),
\]
i.e., $G$ \emph{acts by symplectomorphisms}.
\end{defn}

\begin{rem}
It is easy to show that there is a one-to-one correspondence between \emph{complete vector fields} on a manifold $M$ and \emph{smooth actions} of $\R$ on $M$. One can show this by associating to a complete vector field $X$ its exponential map $\exp tX$ and, vice-versa, to a smooth action $\Psi$ its derivative $\frac{\dd\Psi_t(q)}{\dd t}\big|_{t=0}=:X_q$. Thus, we also get a one-to-one correspondence between \emph{complete symplectic vector fields} on $M$ and \emph{symplectic actions} of $\R$ on $M$.
\end{rem}

\begin{defn}[Hamiltonian action]
A symplectic action $\Psi$ of $S^1$ or $\R$ on $(M,\omega)$ is \emph{Hamiltonian} if the vector field generated by $\Psi$ is Hamiltonian. Equivalently, an action $\Psi$ of $S^1$ or $\R$ on $(M,\omega)$ is \emph{Hamiltonian} if there is a function $H\colon M\to \R$ with $\iota_X\omega=\dd H$, where $X$ is the vector field generated by $\Psi$.
\end{defn}

\begin{rem}
When $G$ is not a product of copies of $S^1$ or $\R$, the solution is to use an upgraded Hamiltonian function, which is called a \emph{moment map}. 
\end{rem}

\subsection{Adjoint and coadjoint representations}

Let $G$ be a Lie group. For $g\in G$, let 
\begin{align*}
    L_g\colon G&\to G,\\
    a&\mapsto g\cdot a
\end{align*}
be \emph{left multiplication} by $g$.

\begin{defn}[left-invariant vector field]
A \emph{vector field} $X$ on $G$ is called \emph{left-invariant} if 
\[
(L_g)_*X=X,\quad \forall g\in G.
\]
\end{defn}

\begin{rem}
There are similar \emph{right} notions.
\end{rem}

Let $\mathfrak{g}$ be the vector space of all left-invariant vector fields on $G$. Together with the Lie bracket $[\enspace,\enspace]$ of vector fields, $\mathfrak{g}$ forms a Lie algebra. It is called the \emph{Lie algebra of the Lie group $G$}.

\begin{exe}
Show that the map 
\begin{align*}
    \mathfrak{g}&\to T_eG,\\
    X&\mapsto X_e,
\end{align*}
where $e$ is the identity element in $G$, is an isomorphism of vector spaces.
\end{exe}

Note that any Lie group $G$ acts on itself by \emph{conjugation}:
\begin{align*}
    G&\to \mathrm{Diff}(M),\\
    g&\mapsto \Psi_g
\end{align*}
where $\Psi_g(a)=g\cdot a\cdot g^{-1}$. The derivative at the identity of $\Psi_g$ is an invertible linear map 
\[
\Ad_g\colon \mathfrak{g}\to \mathfrak{g}.
\]
Note that we have identified the Lie algebra $\mathfrak{g}$ with the tangent space $T_eG$. 

\begin{defn}[Adjoint representation]
The \emph{adjoint representation} (or \emph{adjoint action}) of $G$ on $\mathfrak{g}$ is
\begin{align*}
    \mathrm{Ad}\colon G&\to \GL(\mathfrak{g}),\\
    g&\mapsto \Ad_g.
\end{align*}
\end{defn}

\begin{exe}
Check that for matrix Lie groups
\[
\frac{\dd}{\dd t}\Ad_{\exp tX}Y\bigg|_{t=0}=[X,Y],\quad \forall X,Y\in \mathfrak{g}.
\]
\emph{Hint: use that for a matrix group $G$ (i.e. a subgroup of $\GL_n(\R)$ for some $n$) we have}
\[
\Ad_g(Y)=gYg^{-1},\quad\forall g\in G,\forall Y\in \mathfrak{g},
\]
\emph{and}
\[
[X,Y]=XY-YX,\quad \forall X,Y\in\mathfrak{g}.
\]
\end{exe}

Let $\langle\enspace,\enspace\rangle$ be the natural pairing between $\mathfrak{g}^*$ and $\mathfrak{g}$ defined as 
\begin{align*}
    \langle\enspace,\enspace\rangle\colon \mathfrak{g}^*\times \mathfrak{g}&\to \R,\\
    (\xi,X)&\mapsto \langle \xi,X\rangle:=\xi(X).
\end{align*}

For $\xi\in\mathfrak{g}^*$, we define $\Ad^*_g\xi$ by the property
\[
\langle \Ad^*_g\xi,X\rangle=\langle \xi,\Ad_{g^{-1}}X\rangle,\quad \forall X\in\mathfrak{g}.
\]

\begin{defn}[Coadjoint representation]
The collection of maps $\Ad^*_g$ forms the \emph{coadjoint representation} (or \emph{coadjoint action}) of $G$ on $\mathfrak{g}^*$:
\begin{align*}
    \Ad^*\colon G&\to \GL(\mathfrak{g}^*),\\
    g&\mapsto \Ad^*_g.
\end{align*}
\end{defn}

\begin{exe}
Show that for all $g,h\in G$, we have
\[
\Ad_g\circ \Ad_h=\Ad_{gh},\qquad \Ad^*_g\circ \Ad^*_h=\Ad^*_{gh}.
\]
\end{exe}

\subsection{Hamiltonian actions}

Let $(M,\omega)$ be a symplectic manifold, $G$ a Lie group, and $\Psi\colon G\to \mathrm{Sympl}(M,\omega)$ a (smooth) symplectic action, i.e. a group homomorphism such that the evaluation map $\ev_\Psi(g,q):=\Psi_g(q)$ is smooth. Moreover, let $\mathfrak{g}$ be the Lie algebra of $G$ and $\mathfrak{g}^*$ its dual space. 

\begin{defn}[Hamiltonian action]
\label{defn:Hamiltonian_action}
The action $\Psi$ is called \emph{Hamiltonian} if there exists a map 
\[
\mu\colon M\to \mathfrak{g}^*
\]
such that
\begin{enumerate}
    \item for each $X\in \mathfrak{g}$, let 
    \begin{itemize}
        \item the map 
        \begin{align*}
            \mu^X\colon M&\to \R,\\
            q&\mapsto \mu^X(q):=\langle \mu(q),X\rangle,
        \end{align*}
        be the component of $\mu$ along $X$,
        \item $X^\#$ be the vector field on $M$ generated by the 1-parameter subgroup $(\exp tX)_{t\in\R}\subseteq G$.
    \end{itemize}
    Then 
    \[
    \dd\mu^X=\iota_{X^\#}\omega,
    \]
    i.e., $\mu^X$ is a Hamiltonian function for the vector field $X^\#$.
    \item $\mu$ is \emph{equivariant} with respect to the given action $\Psi$ of $G$ on $M$ and the coadjoint action $\Ad^*$ of $G$ on $\mathfrak{g}^*$:
    \[
    \mu\circ\Psi_g=\Ad^*_g\circ \mu,\quad \forall g\in G.
    \]
\end{enumerate}
\end{defn}

\begin{defn}[Moment map]
A map $\mu$ as in Definition \ref{defn:Hamiltonian_action} is called a \emph{moment map}.
\end{defn}

\begin{defn}[Hamiltonian $G$-space]
The quadruple $(M,\omega,G,\mu)$ is called a \emph{Hamiltonian $G$-space}.
\end{defn}

\begin{defn}[Comoment map]
Let $G$ be a connected Lie group. A \emph{comoment map} is a map 
\[
\mu^*\colon \mathfrak{g}\to C^\infty(M),
\]
such that 
\begin{enumerate}
    \item $\mu^*(X):=\mu^X$ is a Hamiltonian function for the vector field $X^\#$, 
    \item $\mu^*$ is a Lie algebra homomorphism:
    \[
    \mu^*[X,Y]=\{\mu^*(X),\mu^*(Y)\},
    \]
    where $\{\enspace,\enspace\}$ denotes the Poisson bracket on $C^\infty(M)$.
\end{enumerate}
\end{defn}

\begin{rem}
The condition for an action to be Hamiltonian can be equivalently rephrased by using the notion of a comoment map instead of a moment map.
\end{rem}

\section{Symplectic reduction}

\subsection{Orbit spaces}

Let $\Psi\colon G\to \mathrm{Diff}(M)$ be any action. 

\begin{defn}[Orbit]
The \emph{orbit} of $G$ through $q\in M$ is 
\[
\calO_q:=\{\Psi_g(q)\mid g\in G\}.
\]
\end{defn}

\begin{defn}[Stabilizer]
The \emph{stabilizer} (or \emph{isotropy}) of $q\in M$ is the subgroup 
\[
G_q:=\{g\in G\mid \Psi_g(q)=q\}.
\]
\end{defn}

\begin{exe}
Show that if $p$ is in the orbit of $q$, then $G_p$ and $G_q$ are conjugate subgroups.
\end{exe}

\begin{defn}[Transitive/Free/Locally free]
The action $\Psi$ is called 
\begin{itemize}
    \item \emph{transitive} if there is just one orbit,
    \item \emph{free} if all stabilizers are trivial $\{e\}$,
    \item \emph{locally free} if all stabilizers are discrete.
\end{itemize}
\end{defn}

Let $\sim$ be the orbit equivalence relation, i.e. for $q,p\in M$, we define $q\sim p$ if and only if $q$ and $p$ are on the same orbit. 

\begin{defn}[Orbit space]
The space $M/_\sim=:M/G$ is called the \emph{orbit space}.  
\end{defn}

\begin{rem}
We can endow $M/G$ with the \emph{quotient topology} with respect to the projection 
\begin{align*}
    \pi\colon M&\to M/G,\\
    q&\mapsto \calO_q.
\end{align*}
\end{rem}

\subsection{Principal bundles}

Let $G$ be a Lie group and let $B$ be a manifold.

\begin{defn}[Principal $G$-bundle]
A \emph{principal $G$-bundle over $B$} is a manifold $P$ with a smooth map $\pi\colon P\to B$ satisfying:
\begin{enumerate}
    \item $G$ acts freely on $P$ (from the left),
    \item $B$ is the orbit space for this action and $\pi$ is the point-orbit projection, and
    \item there is an open covering of $B$ such that to each set $\calU$ in that covering corresponds a map $\varphi_\calU\colon \pi^{-1}(\calU)\to \calU\times G$ with 
    \[
    \varphi_\calU(q)=(\pi(q),s_\calU(q)),\quad s_\calU(g\cdot q)=g\cdot s_\calU(q),\quad \forall q\in\pi^{-1}(\calU).
    \]
\end{enumerate}
The $G$-valued maps $s_\calU$ are determined by the corresponding $\varphi_\calU$. Condition (3) is called the property of being \emph{locally trivial}.
\end{defn}

\begin{rem}
The manifold $B$ is usually called the \emph{base}, the manifold $P$ is called the \emph{total space}, the Lie group $G$ is called the \emph{structure group}, and the map $\pi$ is called the \emph{projection}. We can represent a principal $G$-bundle by the following diagram:
\[
\begin{tikzcd}
G \arrow[r, hook] & P \arrow[d, "\pi"] \\
                  & B                 
\end{tikzcd}
\]
\end{rem}

\begin{thm}
\label{thm:quotient_manifold}
If a compact Lie group $G$ acts freely on a manifold $M$, then $M/G$ is a manifold and the map $\pi\colon M\to M/G$ is a principal $G$-bundle.
\end{thm}

\begin{proof}
First, we will show that, for any $q\in M$, the $G$-orbit through $q$ is a compact embedded submanifold of $M$ diffeomorphic to $G$. Note that the evaluation map 
\begin{align*}
\ev\colon G\times M&\to M,\\
(g,q)&\mapsto\ev(g,q):=g\cdot q 
\end{align*}
is smooth since the action is smooth. We claim that, for $q\in M$, the map $\ev_q$ provides the desired embedding. Note that the image of $\ev_q$ is the $G$-orbit through $q$. Since the action of $G$ is free, we get that $\ev_q$ is injective. Clearly, the map $\ev_q$ is proper, since a compact (and hence closed) subset $N$ of $M$ has inverse image $(\ev_q)^{-1}(N)$ being a closed subset of a compact Lie group $G$, hence compact. We still have to show that $\ev_q$ is an immersion. Note that for $X\in\mathfrak{g}\cong T_eG$ we have 
\[
\dd_e\ev_q(X)=0\Longleftrightarrow X^\#_q=0\Longleftrightarrow X=0,
\]
since the action is free. Hence, we can conclude that $\dd_e\ev_q$ is injective. Thus, at any point $g\in G$, for $X\in T_gG$, we have 
\[
\dd_g\ev_q(X)=0\Longleftrightarrow \dd_e(\ev_q\circ R_g)\circ \dd_gR_{g^{-1}}(X)=0,
\]
where $R_g\colon G\to G$ denotes \emph{right multiplication} by $g$. On the other hand, $\ev_q\circ R_g=\ev_{g\cdot q}$ has an injective differential at the identity $e$, and $\dd_g R_{g^{-1}}$ is an isomorphism. This implies that $\dd_g\ev_q$ is always injective.

In fact, one can show that even if the action is not free, the $G$-orbit through $q$ is a compact embedded submanifold of $M$. In that case, the orbit is diffeomorphic to the quotient of $G$ by the isotropy of $q$, i.e. 
\[
\calO_q\cong G/G_q.
\]
Let $S$ be a transverse section to $\calO_q$ at $q$. We call $S$ a \emph{slice}. Choose coordinates $x_1,\ldots, x_n$ centered at $q$ such that 
\begin{align*}
    \calO_q\cong G\colon x_1&=\dotsm =x_k=0,\\
    S\colon x_{k+1}&=\dotsm =x_n=0.
\end{align*}
Let $S_\epsilon:=S\cap B_\epsilon(0,\R^n)$, where $B_\epsilon(0,\R^n)$ denotes the ball of radius $\epsilon$ centered at $0$ in $\R^n$. Let $\eta\colon G\times S\to M$, $\eta(g,s)=g\cdot s$. Then we can apply the following equivariant version of the tubular neighborhood theorem:
\begin{thm}[Slice theorem]
Let $G$ be a compact Lie group $G$ acting on a manifold $M$ such that $G$ acts freely at $q\in M$. For sufficiently small $\epsilon$, $\eta\colon G\times S_\epsilon\to M$ maps $G\times S_\epsilon$ diffeomorphically onto a $G$-invariant neighborhood $\calU$ of the $G$-orbit through $q$.
\end{thm}
We will use the following corollaries of the Slice theorem:
\begin{cor}
If the action of $G$ is free at $q$, then the action is free on $\calU$.
\end{cor}

\begin{cor}
The set of points where $G$ acts freely is open. 
\end{cor}

\begin{cor}
The set $G\times S_\epsilon\cong \calU$ is $G$-invariant. Hence, the quotient 
\[
\calU/G\cong S_\epsilon
\]
is smooth. 
\end{cor}
Now we can conclude the proof that $M/G$ is a manifold and $\pi\colon M\to M/G$ is a smooth fiber map. For $q\in M$, let $p=\pi(q)\in M/G$. Choose a $G$-invariant neighborhood $\calU$ of $q$ as in the slice theorem: $\calU\cong G\times S_\epsilon$. Then $\pi(\calU)=\calU/G=:\calV$ is an open neighborhood of $p$ in $M/G$. By the slice theorem, we get that $S_\epsilon \xrightarrow{\sim}\calV$ is a homeomorphism. We will use such neighborhoods $\calV$ as charts on $M/G$. We want to show that the transition functions associated with these charts are smooth. For this, consider two $G$-invariant open sets $\calU_1,\calU_2$ in $M$ and corresponding slices $S_1,S_2$ of the $G$-action. Then $S_{12}:=S_1\cap\calU_2$ and $S_{21}:=S_2\cap\calU_1$ are both slices for the $G$-action on $\calU_1\cap\calU_2$. To compute the transition map $S_{12}\to S_{21}$, consider the diagram
\[
\begin{tikzcd}
S_{12} \arrow[r, "\sim"] & \mathrm{id}\times S_{12} \arrow[r, hook] & G\times S_{12} \arrow[rd, "\sim"]  &                                \\
                         &                                          &                                    & \mathcal{U}_1\cap\mathcal{U}_2 \\
S_{21} \arrow[r, "\sim"] & \mathrm{id}\times S_{21} \arrow[r, hook] & G\times S_{21} \arrow[ru, "\sim"'] &                               
\end{tikzcd}
\]
Then the composition
\[
\begin{tikzcd}
S_{12} \arrow[r, hook] & \mathcal{U}_1\cap\mathcal{U}_2 \arrow[r, "\sim"] & G\times S_{21} \arrow[r, "\mathrm{pr}_2"] & S_{21}
\end{tikzcd}
\]
is smooth. Finally, we need to show that $\pi\colon M\to M/G$ is a smooth fiber map. For $q\in M$ with $p:=\pi(q)\in M/G$, choose a $G$-invariant neighborhood $\calU$ of the $G$-orbit through $q$ of the form $\eta\colon G\times S_\epsilon\xrightarrow{\sim}\calU$. Then $\calV=\calU/G\simeq S_\epsilon$ is the corresponding neighborhood of $p\in M/G$:
\[
\begin{tikzcd}
M\supseteq\mathcal{U} \arrow[d, "\pi"'] \arrow[r, "\eta^{-1}"] & G\times S_\epsilon \arrow[r, "\sim"] & G\times \mathcal{V} \arrow[d, "\mathrm{pr}_2"] \\
M/G\supseteq\mathcal{V} \arrow[rr, "\mathrm{id}"]              &                                      & \mathcal{V}                                   
\end{tikzcd}
\]
since the projection $\mathrm{pr}_2$ is smooth. it is then easy to check that the transition maps for the bundle defined by $\pi$ are smooth. We leave this as an exercise.
\end{proof}

\subsection{The Marsden--Weinstein theorem}

\begin{thm}[Marsden--Weinstein\cite{MarsdenWeinstein1974}]
\label{thm:Marsden-Weinstein}
Let $(M,\omega,G,\mu)$ be a Hamiltonian $G$-space for a compact Lie group $G$. Let $i\colon \mu^{-1}(0)\hookrightarrow M$ be the inclusion map. Assume that $G$ acts freely on $\mu^{-1}(0)$. Then 
\begin{itemize}
    \item the orbit space $M_\mathrm{red}:=\mu^{-1}(0)/G$ is a manifold,
    \item $\pi\colon \mu^{-1}(0)\to M_\mathrm{red}$ is a principal $G$-bundle, and 
    \item there is a symplectic form $\omega_\mathrm{red}$ on $M_\mathrm{red}$ satisfying $i^*\omega=\pi^*\omega_\mathrm{red}$.
\end{itemize}
\end{thm}

\begin{defn}[Reduction]
The pair $(M_\mathrm{red},\omega_\mathrm{red})$ is called the \emph{reduction} of $(M,\omega)$ with respect to $G$ and $\mu$.
\end{defn}

\begin{rem}
The reduction is sometimes also called \emph{reduced space}, \emph{symplectic quotient} or the \emph{Marsden--Weinstein quotient}.
\end{rem}

\begin{lem}
\label{lem:Marsden-Weinstein}
Let $(V,\omega)$ be a symplectic vector space. Suppose that $I$ is an isotropic subspace of $V$. Then $\omega$ induces a canonical symplectic form $\Omega$ on $I^\omega/I$, where $I^\omega$ denotes the symplectic orthocomplement of $I$.
\end{lem}

\begin{exe}
Prove Lemma \ref{lem:Marsden-Weinstein}.
\end{exe}

\begin{proof}[Proof of Theorem \ref{thm:Marsden-Weinstein}]
Note that since $G$ acts freely on $\mu^{-1}(0)$, we get that $\dd_q\mu$ is surjective for all $q\in \mu^{-1}(0)$ since 
\[
\im \dd_q\mu=\mathrm{Ann}(\mathfrak{g}_q):=\{\xi\in\mathfrak{g}^*\mid \langle \xi,X\rangle=0,\,\, \forall X\in\mathfrak{g}_q\},
\]
where $\mathrm{Ann}$ denotes the \emph{annihilator} and $\mathfrak{g}_q$ the Lie algebra of the stabilizer $G_q$ which in this case is trivial and hence $\im \dd_q\mu=\mathfrak{g}^*$. Thus, $0$ is a \emph{regular value} and therefore $\mu^{-1}(0)$ is a closed submanifold of codimension $\dim G$. The first part of Theorem \ref{thm:Marsden-Weinstein} is just an application of Theorem \ref{thm:quotient_manifold} to the free action of $G$ on $\mu^{-1}(0)$. 

Lemma \ref{lem:Marsden-Weinstein} gives a canonical symplectic structure on the quotient $T_q\mu^{-1}(0)/T_q\calO_q$ since, by the fact that $G$ acts freely on $\mu^{-1}(0)$, we have that $\calO_q\cong G$ and thus 
\[
T_q\mu^{-1}(0)=\ker \dd_q\mu=(T_q\calO_q)^{\omega_q}.
\]
One can indeed check that $T_q\calO_q$ is an isotropic subspace of $T_qM$.
The point $[q]\in M_\mathrm{red}=\mu^{-1}(0)/G$ has tangent space $T_{[q]}M_\mathrm{red}\cong T_q\mu^{-1}(0)/T_q\calO_q$. Thus Lemma \ref{lem:Marsden-Weinstein} defines a nondegenerate 2-form $\omega_\mathrm{red}$ on $M_\mathrm{red}$. This is well-defined because $\omega$ is $G$-invariant.

By construction, we have $i^*\omega=\pi^*\omega_\mathrm{red}$, where 
\[
\begin{tikzcd}
\mu^{-1}(0) \arrow[r, "i", hook] \arrow[d, "\pi"'] & M \\
M_\mathrm{red}                                            &  
\end{tikzcd}
\]
Hence, $\pi^*\dd\omega_\mathrm{red}=\dd\pi^*\omega_\mathrm{red}=\dd i^*\omega=i^*\dd\omega=0$. The closedness of $\omega_\mathrm{red}$ follows from the injectivity of $\pi^*$.
\end{proof}

\subsection{Noether's theorem}
Let $(M,\omega,G,\mu)$ be a Hamiltonian $G$-space.

\begin{thm}[Noether\cite{Noether1918}]
A function $f\colon M\to \R$ is $G$-invariant if and only if $\mu$ is constant along the trajectories of the Hamiltonian vector field of $f$.
\end{thm}

\begin{proof}
Let $X_f$ be the Hamiltonian vector field of $f$. Moreover, let $X\in\mathfrak{g}$ and 
\[
\mu^X=\langle\mu, X\rangle\colon M\to \R.
\]
Then we have 
\begin{equation}
    L_{X_f}\mu^X=\iota_{X_f}\dd\mu^X=\iota_{X_f}\iota_{X^\#}\omega=-\iota_{X^\#}\iota_{X_f}\omega=-\iota_{X^\#}\dd f=-L_{X^\#}f=0.
\end{equation}
because $f$ is $G$-invariant.
\end{proof}

\begin{defn}[Integral of motion/Symmetry]
A $G$-invariant function $f\colon M\to \R$ is called an \emph{integral of motion} of the Hamiltonian $G$-space $(M,\omega,G,\mu)$. If $\mu$ is constant along the trajectories of a Hamiltonian vector field $X_f$, then the corresponding 1-parameter group of diffeomorphisms $(\exp tX_f)_{t\in\R}$ is called a \emph{symmetry} of $(M,\omega,G,\mu)$.
\end{defn}

\begin{rem}
The \emph{Noether principle} asserts that there is a one-to-one correspondence between symmetries and integrals of motion. 
\end{rem}

\section{The Duistermaat--Heckman theorems}

\subsection{Duistermaat--Heckman polynomial}

Let $(M,\omega)$ be a $2n$-dimensional symplectic manifold and consider its Liouville volume $\frac{\omega^n}{n!}$.

\begin{defn}[Liouville measure]
The \emph{Liouville measure} (or \emph{symplectic measure}) of a Borel subset $\calU\subset M$ is given by 
\[
m_\omega(\calU)=\int_\calU\frac{\omega^n}{n!}.
\]
\end{defn}

Let $G$ be a torus and suppose that $(M,\omega,G,\mu)$ is a Hamiltonian $G$-space such that the moment map $\mu$ is proper. Moreover denote by $\mathfrak{g}$ the Lie algebra of $G$.

\begin{defn}[Duistermaat--Heckman measure]
The \emph{Duistermaat--Heckman measure}, denoted by $m_\mathrm{DH}$, on $\mathfrak{g}^*$ is the pushforward of $m_\omega$ by $\mu\colon M\to \mathfrak{g}^*$. That is, 
\[
m_\mathrm{DH}(U):=(\mu_*m_\omega)(U)=\int_{\mu^{-1}(U)}\frac{\omega^n}{n!}
\]
for any Borel subset $U\subset\mathfrak{g}^*$.
\end{defn}

\begin{rem}
For a compactly supported function $h\in C^\infty(\mathfrak{g}^*)$, we can define its integral with respect to the Duistermaat--Heckman measure by
\[
\int_{\mathfrak{g}^*}h\dd m_\mathrm{DH}:=\int_M(h\circ \mu)\frac{\omega^n}{n!}.
\]
On $\mathfrak{g}^*$, regarded as a vector space $\R^n$, there is also the Lebesgue measure $m_0$. The relation between $m_\mathrm{DH}$ and $m_0$ is given through the \emph{Radon--Nikodym derivative}, denoted by $\frac{\dd m_\mathrm{DH}}{\dd m_0}$, which is a generalized function such that 
\[
\int_{\mathfrak{g}^*}h\dd m_\mathrm{DH}=\int_{\mathfrak{g}^*}\frac{\dd m_\mathrm{DH}}{\dd m_0}\dd m_0.
\]
\end{rem}

\begin{thm}[Duistermaat--Heckman \cite{DuistermaatHeckman1982}]
The Duistermaat--Heckman measure is a piecewise polynomial multiple of the Lebesgue measure $m_0$ on $\mathfrak{g}^*\cong\R^n$. That is, the Radon--Nikodym derivative 
\[
f:=\frac{\dd m_\mathrm{DH}}{\dd m_0}
\]
is piecewise polynomial. More precisely, for any Borel subset $U\subset\mathfrak{g}^*$, we have 
\[
m_\mathrm{DH}(U)=\int_Uf(x)\dd x,
\]
where $\dd x:=\dd m_0(x)$ denotes the Lebesgue volume form on $U$ and $f\colon \mathfrak{g}^*\cong\R^n\to \R$ is polynomial on any region consisting of regular values of $\mu$.
\end{thm}

\begin{rem}
The Radon--Nikodym derivative $f$ is also called the \emph{Duistermaat--Heckman polynomial}. 
\end{rem}

\begin{ex}
Consider the Hamiltonian $S^1$-space $(S^2,\omega=\dd\theta\land \dd h, S^1,\mu=h)$. The image of $\mu$ is given by the interval $[-1,1]$. The Lebesgue measure of $[a,b]\subseteq [-1,1]$ is given by 
\[
m_0([a,b])=b-a.
\]
The Duistermaat--Heckman measure of $[a,b]$ is given by 
\[
m_\mathrm{DH}([a,b])=\int_{\{(\theta,h)\in S^2\,\mid\, a\leq h\leq b\}}\dd\theta\land \dd h=2\pi(b-a).
\]
\end{ex}

\begin{rem}
As a consequence, one can show that the spherical area between two horizontal circles depends only on the vertical distance between them. This result was actually already known by Archimedes. 
\end{rem}

\begin{cor}
For the standard Hamiltonian action of $S^1$ on $(S^2,\omega)$, we have 
\[
m_\mathrm{DH}=2\pi m_0.
\]
\end{cor}

\subsection{Local form for reduced spaces}

Let $(M,\omega,G,\mu)$ be a Hamiltonian $G$-space, where $G$ is an $n$-torus. Assume that $\mu$ is proper. If $G$ acts freely on $\mu^{-1}(0)$, then is also acts freely on nearby levels $\mu^{-1}(t)$ for $t\in \mathfrak{g}^*$ and $t$ near to $0$ (which we will denote by $t\approx 0$). Consider the reduced spaces 
\[
M_\mathrm{red}:=\mu^{-1}(0)/G,\qquad M_t:=\mu^{-1}(t)/G
\]
with reduced symplectic forms $\omega_\mathrm{red}$ and $\omega_t$. We want to see what the relation between these reduced spaces when regarded as symplectic manifolds. For simplicity, we want to assume that $G$ is the circle $S^1$. Let $Z:=\mu^{-1}(0)$ and let $i\colon Z\hookrightarrow M$ be the inclusion. Fix a form $\alpha\in \Omega^1(Z)$ for the principal bundle 
\[
\begin{tikzcd}
S^1 \arrow[r, hook] & Z \arrow[d, "\pi"] \\
                  & M_\mathrm{red}                 
\end{tikzcd}
\]
which means that $L_{X^\#}\alpha=0$ and $\iota_{X^\#}\alpha=1$, where $X^\#$ is the infinitesimal generator for the $S^1$-action. Using $\alpha$, we construct a 2-form on the product manifold $Z\times(-\varepsilon,\varepsilon)$ by 
\[
\sigma:=\pi^*\omega_\mathrm{red}-\dd(x\alpha),
\]
where $x$ is a linear coordinate on the interval $(-\varepsilon,\varepsilon)\subset \R\cong\mathfrak{g}^*$. We want to abuse notation to shorten the symbols for forms on $Z\times(-\varepsilon,\varepsilon)$ which are given by pullback via the projection onto each factor.

\begin{lem}
The 2-form $\sigma$ is symplectic for $\varepsilon$ small enough.
\end{lem}

\begin{proof}
Clearly, $\sigma$ is closed. Note that at $x=0$, we have 
\[
\sigma\big|_{x=0}=\pi^*\omega_\mathrm{red}+\alpha\land \dd x,
\]
which satisfies
\[
\sigma\big|_{x=0}\left(X^\#,\frac{\de}{\de x}\right)=1,
\]
and thus $\sigma$ is nondegenerate along $Z\times\{0\}$. Since nondegeneracy is an open condition, we can conclude that $\sigma$ is nondegenerate for $x$ in a sufficiently small neighborhood of $0$.
\end{proof}

\begin{rem}
Note that $\sigma$ is invariant with respect to the $S^1$-action on the first factor of $Z\times(-\varepsilon,\varepsilon)$. Actually, this $S^1$-action is Hamiltonian with moment map given by the projection onto the second factor
\[
x\colon Z\times(-\varepsilon,\varepsilon)\to (-\varepsilon,\varepsilon),
\]
which can be easily shown by
\[
\iota_{X^\#}\sigma=-\iota_{X^\#}\dd(x\alpha)=-L_{X^\#}(x\alpha)+\dd\iota_{X^\#}(x\alpha)=\dd x.
\]
\end{rem}

\begin{lem}
\label{lem:DH1}
There is an equivariant symplectomorphism between a neighborhood of $Z$ in $M$ and a neighborhood of $Z\times\{0\}$ in $Z\times(-\varepsilon,\varepsilon)$, intertwining the two moment maps, for $\varepsilon$ small enough.
\end{lem}

\begin{proof}
The inclusion $i\colon Z\hookrightarrow Z\times(-\varepsilon,\varepsilon)$ as $Z\times \{0\}$ and the natural inclusion $i\colon Z\hookrightarrow M$ are $S^1$-equivariant coisotropic embeddings. They actually satisfy $i_0^*\sigma=i^*\omega$ since both sides are equal to $\pi^*\omega_\mathrm{red}$, and the moment maps coincide on $Z$ since $i^*_0x=0=i^*\mu$. If we replace $\varepsilon$ by a smaller positive number whenever necessary, the result follows from an equivariant version of the coisotropic embedding theorem (Theorem \ref{thm:Coisotropic_embedding}).
\end{proof}

Hence, in order to compare the reduced spaces 
\[
M_t=\mu^{-1}(t)/S^1,\quad t\approx 0,
\]
we work with $Z\times(-\varepsilon,\varepsilon)$ and compare instead the reduced spaces
\[
x^{-1}(t)/S^1,\quad t\approx 0.
\]

\begin{prop}
\label{prop:DH1}
The reduced space $(M_t\omega_t)$ is symplectomorphic to 
\[
(M_\mathrm{red},\omega_\mathrm{red}-t\beta),
\]
where $\beta$ is the curvature of the connection 1-form $\alpha$.
\end{prop}

\begin{proof}
Using Lemma \ref{lem:DH1}, we can see that $(M_t,\omega_t)$ is symplectomorphic to the reduced space at level $t$ for the Hamiltonian $S^1$-space $(Z\times(-\varepsilon,\varepsilon),\sigma,S^1,x)$. Since $x^{-1}(t)=Z\times\{t\}$, where $S^1$ acts on the first factor, all manifolds $x^{-1}(t)/S^1$ are diffeomorphic to $Z/S^1=M_\mathrm{red}$. For the symplectic forms, let $\iota_t\colon Z\times\{0\}\hookrightarrow Z\times(-\varepsilon,\varepsilon)$ be the inclusion map. The restriction of $\sigma$ to $Z\times\{t\}$ is 
\[
\iota_t^*\sigma=\pi^*\omega_\mathrm{red}-t\dd\alpha.
\]
By definition of the curvature, we have $\dd\alpha=\pi^*\beta$. Hence, the reduced symplectic form on $x^{-1}(t)/S^1$ is given by $\omega_\mathrm{red}-t\beta$.
\end{proof}

\begin{rem}
Informally, Proposition \ref{prop:DH1} says that the reduced forms $\omega_t$ vary linearly in $t$, for $t\approx0$. On the other hand, the identification of $M_t$ with $M_\mathrm{red}$ as abstract manifolds is not natural. Nevertheless, any two such identifications are isotopic. Using the homotopy invariance of the de Rham cohomology classes, we can obtain the following theorem.
\end{rem}

\begin{thm}[Duistermaat--Heckman\cite{DuistermaatHeckman1982}]
\label{thm:DH2}
The cohomology class of the reduced symplectic form $[\omega_t]$ varies linearly in $t$. More specifically, we have 
\[
[\omega_t]=[\omega_\mathrm{red}]+tc,
\]
where $c=[-\beta]\in H^2(M_\mathrm{red})$ is the \emph{first Chern class} of the $S^1$-bundle $Z\to M_\mathrm{red}$.
\end{thm}

\begin{rem}
A connection on a principal bundle is given by a Lie algebra-valued 1-form. We can identify the Lie algebra $S^1$ with $2\pi \I\R$ by using the exponential map $\exp\colon \mathfrak{g}\cong 2\pi\I\R\to S^1$, $\xi\mapsto \exp(\xi)$. If we consider a principal $S^1$-bundle, this identification leads to the fact that the infinitesimal action maps the generator $2\pi\I$ of $2\pi\I\R$ to the generating vector field $X^\#$. If we have a connection 1-form $A$, it is then an imaginary-valued 1-form on the total space satisfying $L_{X^\#}A=0$ and $\iota_{X^\#}A=2\pi\I$. Its curvature 2-form $B$ is then an imaginary-valued 2-form on the base satisfying $\pi^*B=\dd A$. By the \emph{Chern--Weil isomorphism} \cite{Chern1952,Weil1949,BottTu}, the \emph{first Chern class} of the principal $S^1$-bundle is $c=\left[\frac{\I}{2\pi}B\right]$. However, we want to identify the Lie algebra $S^1$ with $\R$ and implicitly use the exponential map $\exp\colon \mathfrak{g}\cong\R\to S^1$, $t\mapsto \exp(2\pi\I t)$. Thus, if we consider a principal $S^1$-bundle, the infinitesimal action maps the generator $1$ of $\R$ to $X^\#$, where now a connection 1-form $\alpha$ is an ordinary 1-form on the total space satisfying $L_{X^\#}\alpha=0$ and $\iota_{X^\#}\alpha=1$. The curvature $\beta$ is now an ordinary 2-form on the base satisfying $\pi^*\beta=\dd\alpha$. This implies that we have $A=2\pi\I\alpha$, $B=2\pi\I\beta$ and the first Chern class $c=[-\beta]$. 
\end{rem}

\subsection{Variation of the symplectic volume}

Consider a Hamiltonian $S^1$-space 
\[
(M,\omega,S^1,\mu)
\]
of dimension $2n$ and let $(M_x,\omega_x)$ be its reduced space at level $x$. Proposition \ref{prop:DH1} and Theorem \ref{thm:DH2} tell us that for $x$ in a sufficiently narrow neighborhood of 0, the symplectic volume of $M_x$ given by 
\[
\mathrm{vol}(M_x):=\int_{M_x}\frac{\omega^{n-1}_x}{(n-1)!}=\int_{M_\mathrm{red}}\frac{(\omega_\mathrm{red}-x\beta)^{n-1}}{(n-1)!},
\]
is a polynomial in $x$ of degree $n-1$. This volume can be also expressed as 
\[
\mathrm{vol}(M_x)=\int_Z\frac{\pi^*(\omega_\mathrm{red}-x\beta)^{n-1}}{(n-1)!}\land \alpha.
\]
Here, $\alpha$ is a chosen connection 1-form on the $S^1$-bundle $Z\to M_\mathrm{red}$ and $\beta$ is its curvature 2-form. The Duistermaat--Heckman measure for a Borel subset $U\subset (-\varepsilon,\varepsilon)$ is
\[
m_\mathrm{DH}(U)=\int_{\mu^{-1}(U)}\frac{\omega^n}{n!}.
\]
Note that $(\mu^{-1}((-\varepsilon,\varepsilon)),\omega)$ is symplectomorphic to $(Z\times(-\varepsilon,\varepsilon),\sigma)$ and moreover they are isomorphic as Hamiltonian $S^1$-spaces. Thus, we can obtain 
\[
m_\mathrm{DH}(U)=\int_{Z\times U}\frac{\sigma^n}{n!}.
\]
Since $\sigma=\pi^*\omega_\mathrm{red}-\dd(x\alpha)$, we can express its $n$-th power by 
\[
\sigma^n=n(\pi^*\omega_\mathrm{red}-x\dd\alpha)^{n-1}\land\alpha\land\dd x.
\]
Using Fubini's theorem, we get
\[
m_\mathrm{DH}(U)=\int_U\left(\int_Z\frac{\pi^*(\omega_\mathrm{red}-x\beta)^{n-1}}{(n-1)!}\land\alpha\right)\land \dd x.
\]
Hence, the Radon--Nikodym derivative of $m_\mathrm{DH}$ with respect to the Lebesgue measure $\dd x$ is given by 
\[
f(x)=\int_Z\frac{\pi^*(\omega_\mathrm{red}-x\beta)^{n-1}}{(n-1)!}\land\alpha=\mathrm{vol}(M_x).
\]

\begin{rem}
Note that the previous discussion, for $x\approx0$, $f(x)$ is indeed polynomial in $x$. This actually also holds for a neighborhood of any other regular value of $\mu$, since we can change the moment map $\mu$ by an arbitrary additive constant.
\end{rem}

\begin{exe}[Equivariant cohomology]
Let $M$ be a manifold with an $S^1$-action and $X^\#$ the vector field on $M$ generated by $S^1$. The algebra of \emph{$S^1$-equivariant forms on $M$} is the algebra of $S^1$-invariant forms on $M$ tensored with complex polynomials in $x$, i.e. 
\[
\Omega^\bullet_{S^1}(M):=(\Omega^\bullet(M))^{S^1}\otimes_\R\C[x].
\]
Note that the product $\land$ on $\Omega^\bullet_{S^1}(M)$ combines the wedge product on $\Omega^\bullet(M)$ with the product of polynomials on $\C[x]$.
\begin{enumerate}
    \item We grade $\Omega^\bullet_{S^1}(M)$ by adding the usual grading on $\Omega^\bullet(M)$ to a grading on $\C[x]$ where the monomial $x$ has degree $2$. Show that $(\Omega^\bullet_{S^1}(M),\land)$ is a \emph{supercommutative graded algebra}, i.e. 
    \[
    \underline{\alpha}\land\underline{\beta}=(-1)^{\deg\underline{\alpha}\deg\underline{\beta}}\underline{\beta}\land\underline{\alpha},\quad \forall \underline{\alpha},\underline{\beta}\in\Omega^\bullet_{S^1}(M).
    \]
    \item On $\Omega^\bullet_{S^1}(M)$ we define an operator 
    \[
    \dd_{S^1}:=\dd\otimes 1-\iota_{X^\#}\otimes x.
    \]
    This means that for an element $\underline{\alpha}=\alpha\otimes p(x)$ we have 
    \[
    \dd_{S^1}\underline{\alpha}=\dd\alpha\otimes p(x)-\iota_{X^\#}\alpha\otimes xp(x).
    \]
    The operator $\dd_{S^1}$ is called \emph{Cartan differentiation}. Show that $\dd_{S^1}$ is a subderivation of degree 1, i.e. show that it increases the degree by 1 and that it satisfies the \emph{super Leibniz rule}
    \[
    \dd_{S^1}(\underline{\alpha}\land\underline{\beta})=(\dd_{S^1}\underline{\alpha})\land \underline{\beta}+(-1)^{\deg\underline{\alpha}}\underline{\alpha}\land\dd_{S^1}\underline{\beta}.
    \]
    \item Show that $(\dd_{S^2})^2=0$. \emph{Hint: Use Cartan's magic formula}.
\end{enumerate}
We can hence deduce that the sequence 
\[
0\to \Omega_{S^1}^0(M)\xrightarrow{\dd_{S^1}}\Omega^1_{S^1}(M)\xrightarrow{\dd_{S^1}}\Omega^2_{S^1}(M)\xrightarrow{\dd_{S^1}}\dotsm
\]
is a graded complex whose cohomology is usually called \emph{equivariant cohomology}. The equivariant cohomology of a topological space $M$ endowed with a continuous action of a topological group $G$ is the cohomology of the diagonal quotient $(M\times EG)/G$, where $EG$ denotes the \emph{universal bundle} of $G$, i.e. $EG$ is a contractible space where $G$ acts freely. Cartan showed that, for the action of a compact Lie group $G$ on some manifold $M$, the de Rham complex $(\Omega^\bullet_G(M),\dd_G)$ computes the equivariant cohomology, where $\Omega^\bullet_G(M)$ denotes the $G$-equivariant forms on $M$. For equivariant cohomology in the symplectic context see also \cite{AtiyahBott1984}. The $k$-th equivariant cohomology group is given by 
\[
H^k_{S^1}(M):=\frac{\ker \dd_{S^1}\colon \Omega^{k}_{S^1}(M)\to \Omega^{k+1}_{S^1}(M)}{\mathrm{im}\, \dd_{S^1}\colon \Omega^{k-1}_{S^1}(M)\to \Omega^k_{S^1}(M)}
\]
\begin{enumerate}
\setcounter{enumi}{3}
    \item What is the equivariant cohomology of a point?
    \item What is the equivariant cohomology of $S^1$ with its multiplication action on itself?
    \item Show that the equivariant cohomology of a manifold $M$ with free $S^1$-action is isomorphic to the ordinary cohomology of the quotient space $M/S^1$.\\ 
    \emph{Hint: Consider the projection $\pi\colon M\to M/S^1$.\\ 
    Show that $\pi^*\colon H^\bullet(M/S^1)\to H^\bullet_{S^1}(M)$, $[\alpha]\mapsto [\pi^*\alpha\otimes 1]$ is a well-defined isomorphism. Choose a connection on the principal $S^1$-bundle $M\to M/S^1$, i.e. a 1-form $\theta$ on $M$ such that $L_{X^\#}\theta=0$ and $\iota_{X^\#}\theta=1$. Recall also that a form $\beta$ on $M$ is of type $\pi^*\alpha$ for some $\alpha$ if and only if it is basic, i.e. $L_{X^\#}\beta=0$ and $\iota_{X^\#}\beta=0$.
    }
    \item Let $(M,\omega)$ be a symplectic manifold with an $S^1$-action. Moreover, let $\mu\in C^\infty(M)$ be a real function. Consider the equivariant form 
    \[
    \underline{\omega}:=\omega\otimes 1+\mu\otimes x.
    \]
    Show that $\underline{\omega}$ is \emph{equivariantly closed}, i.e. $\dd_{S^1}\underline{\omega}=0$, if and only if $\mu$ is a moment map. The equivariant form $\underline{\omega}$ is called the \emph{equivariant symplectic form}.
    \item Let $M$ be a $2n$-dimensional compact oriented manifold with an $S^1$-action. Suppose that the set $M^{S^1}$ consisting of fixed points for the $S^1$-action is finite. Consider an $S^1$-invariant form $\alpha^{(2n)}$ which is the top degree part of an equivariantly closed form of even degree, i.e. $\alpha^{(2n)}\in \Omega^{2n}(M)^{S^1}$ is such that there is some $\underline{\alpha}\in\Omega^\bullet_{S^1}(M)$ with 
    \[
    \underline{\alpha}=\alpha^{(2n)}+\alpha^{(2n-2)}+\dotsm+\alpha^{(0)},\quad \underline{\alpha}^{(2k)}\in(\Omega^{2k}(M))^{S^1}\otimes_\R\C[x],\quad \dd_{S^1}\underline{\alpha}=0.
    \]
    \begin{enumerate}
        \item Show that the restriction of $\alpha^{(2n)}$ to $M\setminus M^{S^1}$ is exact.\\
        \emph{Hint: The generator $X^\#$ of the $S^1$-action does not vanish on $M\setminus M^{S^1}$. Thus, we can define a connection on $M\setminus M^{S^1}$ by $\theta(Y)=\frac{\langle Y,X^\#\rangle}{\langle X^\#,X^\#\rangle}$, where $\langle\enspace,\enspace\rangle$ is some $S^1$-invariant metric on $M$. Use $\theta\in \Omega^1(M\setminus M^{S^1})$ to get the primitive of $\alpha^{(2n)}$ by starting at $\alpha^{(0)}$.
        }
        \item Compute the integral of $\alpha^{(2n)}$ over $M$.\\
        \emph{Hint: Use Stokes' theorem to localize the answer near the fixed points.}
    \end{enumerate}
    This exercise is a very particular case of the \emph{Atiyah--Bott localization theorem} for equivariant cohomology \cite{AtiyahBott1984}.
    \item What is the integral of the symplectic form $\omega$ on a surface with a Hamiltonian $S^1$-action which is free outside a finite set of fixed points?\\
    \emph{Hint: Use (7) and (8).}
\end{enumerate}
\end{exe}

\chapter{Poisson Geometry}
Poisson geometry combines methods of differential geometry and noncommutative geometry. Motivated by the dynamical structure induced from the setting of classical mechanics, it connects to the notion of symplectic geometry and provides the actual mathematical structure to talk about deformation quantization. In fact, as we will see, the phase space can be regarded as a \emph{Poisson manifold} with \emph{Poisson structure} coming from the canonical symplectic form since every symplectic manifold naturally induces a Poisson manifold. In this chapter we will introduce Poisson manifolds, Lie algebroids, Courant algebroids, Dirac manifolds, the local behaviour of Poisson manifolds and the induced symplectic foliation. Finally, we will introduce Poisson maps in the regular sense. We will not cover the notion of \emph{Morita equivalence} \cite{Morita1958}, which gives another construction for morphisms between Poisson manifolds, but we refer to \cite{Xu1991,Xu2004,GinzburgLu1992,JurcoSchuppWess2002,Schwarz1998} for the interested reader.
This chapter is mainly based on \cite{Weinstein1983,GinzburgWeinstein1992,BursztynCrainic2005,Conn1995,Cattaneo2004,GinzburgGolubev2001,DufourZung2005,Courant1990a,Courant1990b,BursztynWeinstein2004,GuttRawnsleySternheimer2005}.
Let $P$ denote a manifold for the entire chapter.

\section{Poisson manifolds}

\subsection{Poisson structures and the Schouten--Nijenhuis bracket}

\begin{defn}[Poisson structure]
A \emph{Poisson structure} on $P$ is an $\R$-bilinear Lie bracket $\{\enspace,\enspace\}$ on $C^\infty(P)$ satisfying the Leibniz rule
\[
\{f,gh\}=\{f,g\}h+g\{f,h\},\quad \forall f,g,h\in C^\infty(P).
\]
\end{defn}

\begin{defn}[Casimir function]
A function $f\in C^\infty(P)$ is called \emph{Casimir} if the Hamiltonian vector field $X_f=\{f,\enspace\}$ of $f$ vanishes. 
\end{defn}

\begin{rem}
By the Leibniz rule of the Poisson bracket $\{\enspace,\enspace\}$, there exists a bivector field $\pi\in \mathfrak{X}^2(P):=\Gamma(\bigwedge^2TP)$ such that 
\[
\{f,g\}=\pi(\dd f,\dd g).
\]
\end{rem}

\begin{defn}[Schouten--Nijenhuis bracket]
The \emph{Schouten--Nijenhuis bracket} is a unique extension of the Lie bracket of vector fields to a graded bracket of multivector fields. It is defined by 
\begin{multline}
\label{eq:Schouten-bracket}
[X_1\land \dotsm \land X_m,Y_1\land\dotsm \land Y_n]:=\sum_{1\leq i\leq m\atop 1\leq j\leq m}\,(-1)^{i+j}[X_i,Y_j]X_1\land\dotsm \land X_{i-1}\land X_{i+1}\land\dotsm\land X_m\land\\
\land Y_1\land\dotsm \land Y_{j-1}\land Y_{j+1}\land \dotsm \land Y_{n}
\end{multline}
for vector fields $X_1,\ldots,X_m,Y_1,\ldots,Y_n$ and by 
\[
[f,X_1\land\dotsm \land X_m]:=-\iota_{X_1\land\dotsm \land X_m}\dd f
\]
for a function $f$.
\end{defn}

\begin{exe}
Show that the Schouten--Nijenhuis bracket satisfies the graded Jacobi identity:
\[
(-1)^{(\vert X\vert-1)(\vert Z\vert-1)}[X,[Y,Z]]+(-1)^{(\vert Y\vert-1)(\vert X\vert-1)}[Y,[Z,X]]+(-1)^{(\vert Z\vert-1)(\vert Y\vert-1)}[Z,[X,Y]]=0,
\]
where $\vert\enspace\vert$ denotes the degree operation, i.e. $\vert X\vert=k$ if $X\in\mathfrak{X}^k(P)$.
\end{exe}

\begin{exe}
\label{ex:vanishing_SN_bracket}
Show that the Jacobi identity for $\{\enspace,\enspace\}$ is equivalent to the condition 
\[
[\pi,\pi]=0.
\]
See Section \ref{subsubsec:DGLA_Poisson} for the computation. 
\end{exe}

In local coordinates $(x_1,\ldots,x_n)$ we can determine the components of the bivector field $\pi$ as 
\[
\pi^{ij}(x)=\{x_i,x_j\}.
\]

\begin{defn}[Poisson manifold]
A pair $(P,\pi)$, where $\pi$ is a Poisson bivector field on $P$, is called a \emph{Poisson manifold}. 
\end{defn}

\begin{defn}[Symplectic Poisson structure]
If the bivector field $\pi$ is invertible at each point $x$, it is called \emph{nondegenerate} or \emph{symplectic}.
\end{defn}

\begin{rem}
Note that, if $\pi$ is symplectic, then we can define a symplectic form on $P$. Indeed, we can locally define the matrices 
\[
(\omega_{ij})=(-\pi^{ij})^{-1},
\]
which defines globally a 2-form $\omega\in \Omega^2(P)$. Moreover, the condition $[\pi,\pi]=0$ implies that $\dd\omega=0$. 
\end{rem}

\subsection{Examples of Poisson structures}
\begin{ex}{Constant Poisson structure}
Let $P=\R^n$ and suppose that $\pi^{ij}(x)$ is constant. Then we can find coordinates on $P$ 
\[
(q_1,\ldots, q_k,p_1,\ldots,p_k,e_1,\ldots,e_\ell),\quad 2k+\ell=n,
\]
such that 
\[
\pi=\sum_{1\leq i\leq k}\frac{\de}{\de q_i}\land \frac{\de}{\de p_i}.
\]
Expressing it in terms of a bracket, we get 
\[
\{f,g\}=\sum_{1\leq i\leq k}\left(\frac{\de f}{\de q_i}\frac{\de g}{\de p_i}-\frac{\de f}{\de p_i}\frac{\de g}{\de q_i}\right).
\]
Note that this agrees with the usual Poisson bracket on $C^\infty(T^*\R^n)$ in Hamiltonian mechanics (see Section \ref{subsec:CM}). Note that all coordinates $e_j$ for $0\leq j\leq \ell$ are Casimirs with respect to $\{\enspace,\enspace\}$. 
\end{ex}

\begin{ex}[Poisson structures on $\R^2$]
Let $P=\R^2$ and consider a smooth function $f\in C^\infty(P)$. Such a function $f$ induces a Poisson structure on $P$ by 
\[
\{x_1,x_2\}:=f(x_1,x_2).
\]
Moreover, any Poisson structure on $\R^2$ is of this form. 
\end{ex}

\begin{ex}[Lie-Poisson structure]
\label{ex:Lie-Poisson}
Let $P$ be a finite-dimensional vector space $V$ with coordinates $(x_1,\ldots,x_n)$. We can define a \emph{linear} Poisson structure by 
\[
\{x_i,x_j\}:=\sum_{1\leq k\leq n}c_{ij}^kx_k,
\]
where $c_{ij}^k$ are determined constants with $c_{ij}^k=-c_{ji}^k$. Such a Poisson structure is usually called \emph{Lie-Poisson structure}, since the Jacobi identity of $\{\enspace,\enspace\}$ implies that the $c_{ij}^k$ are the structure constants of a Lie algebra $\mathfrak{g}$, which might be identified naturally with $V^*$. Thus, we might also identify $V\cong \mathfrak{g}^*$. Conversely, any Lie algebra $\mathfrak{g}$ with structure constants $c_{ij}^k$ defines through $\{\enspace,\enspace\}$ a linear Poisson structure on $\mathfrak{g}^*$. 
\end{ex}

\begin{rem}
Deformation quantization of a Lie-Poisson structure on $\mathfrak{g}^*$ leads to the definition of the universal enveloping algebra $U(\mathfrak{g})$. The elements in the center of $U(\mathfrak{g})$ are usually called \emph{Casimir elements}. These exactly correspond to the center of the Poisson algebra of functions on $\mathfrak{g}^*$, and hence, by extension, the Casimir functions for the center of any Poisson algebra. 
\end{rem}

\section{Dirac manifolds}
\subsection{Courant algebroids}
\label{subsec:Courant_algebroids}
We want to formulate a generalization of Poisson structures and closed 2-forms. 

\begin{defn}[Presymplectic form]
A closed 2-form is called \emph{presymplectic}.
\end{defn}

Note that each 2-form $\omega$ on $P$ corresponds to a bundle map 
\begin{align}
\label{eq:bundle_map_1}
\begin{split}
    \tilde\omega\colon TP&\to T^*P,\\
    v&\mapsto\tilde\omega(v):=\omega(v,\enspace).
\end{split}
\end{align}

We can define a similar bundle map for a bivector field $\pi\in\mathfrak{X}^2(P)$ by 
\begin{align}
\label{eq:bundle_map_2}
\begin{split}
    \tilde\pi\colon T^*P&\to TP,\\
    \alpha&\mapsto \tilde \pi(\alpha):=\pi(\enspace,\alpha).
\end{split}
\end{align}

such that $\iota_{\tilde\pi(\alpha)}\beta=\beta(\tilde\pi(\alpha))=\pi(\beta,\alpha)$. The matrix representing $\tilde \pi$ in the basis $(\dd x_i)$ and $(\frac{\de}{\de x_i})$ for local coordinates $(x_1,\ldots,x_n)$ of $P$, is (up to a sign) given by 
\[
\pi^{ij}(x)=\{x_1,x_j\}.
\]
Thus, bivector fields (or 2-forms) are nondegenerate if and only if the associated bundle maps are invertible. Using these bundle maps, we can describe closed 2-forms and Poisson bivector fields as a subbundles of $TP\oplus T^*P$. Indeed, consider the graphs
\begin{align*}
    L_\omega&:=\Gamma_{\tilde\omega}=\mathrm{graph}\,\tilde\omega,\\
    L_\pi&:=\Gamma_{\tilde\pi}=\mathrm{graph}\,\tilde\pi.
\end{align*}

Let us introduce the following canonical structure on $TP\oplus T^*P$:
\begin{enumerate}
    \item The symmetric bilinear form 
    \begin{align*}
        \langle\enspace,\enspace\rangle_+\colon TP\oplus T^*P\times TP\oplus T^*P&\to \R,\\
        ((X,\alpha),(Y,\beta))&\mapsto \langle(X,\alpha),(Y,\beta)\rangle_+:=\alpha(Y)+\beta(X).
    \end{align*}
    \item The bracket
    \begin{align*}
        [\![\enspace,\enspace]\!]\colon \Gamma(TP\oplus T^*P)\times\Gamma(TP\oplus T^*P)&\to \Gamma(TP\oplus T^*P),\\
        ((X,\alpha),(Y,\beta))&\mapsto[\![(X,\alpha),(Y,\beta)]\!]:=([X,Y],L_X\beta-\iota_Y\dd\alpha).
    \end{align*}
\end{enumerate}

\begin{defn}[Lie algebroid]
A \emph{Lie algebroid} is a vector bundle $E$ endowed with a Lie bracket $[\enspace,\enspace]$ on the space of sections $\Gamma(E)$ together with an \emph{anchor map} $\rho\colon E\to TP$ such that for all $X,Y\in \Gamma(E)$ and $f\in C^\infty(P)$
\begin{enumerate}
    \item $[X,fY]=\rho(X)(f)\cdot Y+f[X,Y]$,
    \item $\rho([X,Y])=[\rho(X),\rho(Y)]$.
\end{enumerate}
Here $\rho(X)(f)=L_{\rho(X)}f$ denotes the derivative of $f$ along the vector field $\rho(X)$.
\end{defn}

\begin{defn}[Courant algebroid]
A \emph{Courant algebroid} is a vector bundle $E$ equipped with a nondegenerate symmetric bilinear form $\langle\enspace,\enspace\rangle\colon E\times E\to \R$, a bilinear bracket $[\enspace,\enspace]\colon \Gamma(E)\times \Gamma(E)\to \Gamma(E)$ and a bundle map $\rho\colon E\to TP$ satisfying the following properties:
\begin{enumerate}
    \item $[X,[Y,Z]]=[[X,Y],Z]+[Y,[X,Z]]$ (left Jacobi identity),
    \item $\rho([X,Y])=[\rho(X),\rho(Y)]$,
    \item $[X,fY]=f[X,Y]+\rho(X)(f)\cdot Y$, (Leibniz rule)
    \item $[X,X]=\frac{1}{2}\mathsf{D}\langle X,X\rangle$, where 
    \[
    \mathsf{D}:=\rho^*\dd\colon C^\infty(P)\xrightarrow{\dd}\Omega^1(P)\xrightarrow{\rho^*}E^*\cong E.
    \]
    \item $\rho(X)\langle Y,Z\rangle=\langle[X,Y],Z\rangle+\langle Y,[X,Z]\rangle$, (self-adjoint)
\end{enumerate}
\end{defn}

\begin{defn}[Split Courant algebroid]
The bundle $E:=TP\oplus T^*P$ together with the bracket $\langle\enspace,\enspace\rangle_+$ and $[\![\enspace,\enspace]\!]$ is an example of a \emph{Courant algebroid} which we call the \emph{split Courant algebroid}. 
\end{defn}

\subsection{Dirac structures}

\begin{prop}
\label{prop:Dirac}
A subbundle $L\subset TP\oplus T^*P$ is of the form $L_\pi$ (resp. $L_\omega$) for a bivector field $\pi$ (resp. 2-form $\omega$) if and only if 
\begin{enumerate}
    \item $TP\cap L=\{0\}$ (resp. $L\cap T^*P=\{0\}$) at all points of $P$,
    \item{$L$ is maximal isotropic with respect to $\langle\enspace,\enspace\rangle_+$,}
\end{enumerate}
Moreover, $[\pi,\pi]=0$ (resp. $\dd\omega=0$) if and only if
\begin{enumerate}
    \item[(3)]{$\Gamma(L)$ is closed under the Courant bracket $[\![\enspace,\enspace]\!]$, i.e.
    \[
    [\![\Gamma(L),\Gamma(L)]\!]\subset\Gamma(L).
    \]
    }
\end{enumerate}
\end{prop}

\begin{defn}[Dirac structure]
A \emph{Dirac structure} on $P$ is a subbundle $L\subset TP\oplus T^*P$ which is maximal isotropic with respect to $\langle\enspace,\enspace\rangle_+$ and whose sections are closed under the Courant bracket $[\![\enspace,\enspace]\!]$.
\end{defn}

\begin{defn}[Dirac manifold]
A tuple $(P,L)$, where $L$ is a Dirac structure on $P$ is called a \emph{Dirac manifold}.
\end{defn}

\begin{rem}
Dirac structures are exactly those which satisfy conditions (2) and (3) of Proposition \ref{prop:Dirac}, but do not necessarily appear as the graph of some bivector field or 2-form.
\end{rem}

\begin{defn}[Almost Dirac structure]
If $L$ only satisfies condition (2) of Proposition \ref{prop:Dirac}, it is called an \emph{almost Dirac structure}.
\end{defn}

\begin{rem}
Usually, condition (3) of Proposition \ref{prop:Dirac} is referred to as the \emph{integrability condition} of a Dirac structure.
\end{rem}

\begin{ex}[Regular foliations]
Let $F\subseteq TP$ be a subbundle, and let $\mathrm{Ann}(F)\subset T^*P$ be its annihilator. Then $L=F\oplus \mathrm{Ann}(F)$ is an almost Dirac structure. It is a Dirac structure if and only if $F$ satisfies the \emph{Frobenius condition} (see Theorem \ref{thm:Frobenius})
\[
[\Gamma(F),\Gamma(F)]\subset \Gamma(F).
\]
Hence, \emph{regular foliations} are examples of Dirac structures.
\end{ex}

\begin{ex}[Linear Dirac structure]
\label{ex:linear_Dirac_structures}
If $V$ is a finite-dimensional real vector space, then a \emph{linear Dirac structure} on $V$ is a subspace $L\subset V\oplus V^*$ which is maximal isotropic with respect to the symmetric pairing $\langle\enspace,\enspace\rangle_+$.

Let $L$ be a linear Dirac structure on $V$. Let $\mathrm{pr}_1\colon V\oplus V^*\to V$ and $\mathrm{pr}_2\colon V\oplus V^*\to V^*$ be the canonical projections, and consider the subspace
\[
W:=\mathrm{pr}_1(L)\subset V.
\]
Then $L$ induces a skew-symmetric bilinear form $\theta$ on $W$ defined by 
\begin{equation}
\label{eq:bilinear_form}
\theta(X,Y):=\alpha(Y),
\end{equation}
where $X,Y\in W$ and $\alpha\in V^*$ such that $(X,\alpha)\in L$.

\begin{exe}
Show that $\theta$ is well-defined, i.e. \eqref{eq:bilinear_form} is independent of the choice of $\alpha$.
\end{exe}

Conversely, any pair $(W,\theta)$, where $W\subseteq V$ is a subspace and $\theta$ is a skew-symmetric bilinear form $W$, defines a linear Dirac structure by 
\[
L:=\{(X,\alpha)\mid X\in W,\, \alpha\in V^*\,\, \text{such that}\,\, \alpha\vert_W=\iota_X\theta\}.
\]
\begin{exe}
Check that this $L$ is a linear Dirac structure on $V$ with associated subspace $W$ and bilinear form $\theta$.
\end{exe}
\end{ex}

Note that Example \ref{ex:linear_Dirac_structures} induces a simple way in which linear Dirac structures can be restricted to subspaces.

\subsection{Dirac structures for constrained manifolds}
\begin{ex}[Restriction of Dirac structures to subspaces]
\label{ex:restriction_Dirac_structures}
Let $L$ be a linear Dirac structure on $V$, let $U\subseteq V$ be a subspace and consider the pair $(W,\theta)$ associated to $L$. Then $U$ inherits the linear Dirac structure $L_U$ from $L$ defined by 
\[
W_U:=W\cap U,\qquad \theta_U:=i^*\theta,
\]
where $i\colon U\hookrightarrow V$ denotes the inclusion.
\end{ex}

\begin{exe}
Show that there is a a canonical isomorphism 
\[
L_U\cong \frac{L\cap(U\oplus V^*)}{L\cap \mathrm{Ann}(U)}.
\]
\end{exe}

\begin{rem}
Let $(P,L)$ be a Dirac manifold, and let $i\colon N\hookrightarrow P$ be a submanifold. The constructions of Example \ref{ex:restriction_Dirac_structures}, when applied to $T_xN\subseteq T_xP$ for all $x\in P$, defines a maximal isotropic \emph{subbundle} $L_N\subset TN\oplus T^*N$. However, $L_N$ might not be a continuous family of subspaces. When $L_N$ \emph{is} a continuous family, it is a smooth bundle which then directly satisfies the integrability condition (condition (3) of Proposition \ref{prop:Dirac}). Hence, $L_N$ defines a Dirac structure.
\end{rem}

\begin{ex}[Moment level sets]
Let $\mu\colon P\to \mathfrak{g}^*$ be the moment map for a Hamiltonian action of a Lie group $G$ on a Poisson manifold $(P,\pi)$. Let $\xi\in\mathfrak{g}^*$ be a regular value for $\mu$ and let $G_\xi$ be the isotropy group at $\xi$ with respect to the coadjoint action, and consider
\[
Q=\mu^{-1}(\xi)\hookrightarrow P.
\]
At each point $x\in Q$, we have a linear Dirac structure on $T_xQ$ given by 
\[
(L_Q)_x:=\frac{L_x\cap(T_xQ\oplus T_x^*P)}{L_x\cap\mathrm{Ann}(T_xQ)}.
\]
One can verify that $L_Q$ defines a smooth bundle by verifying that $L_x\cap \mathrm{Ann}(T_xQ)$ has constant dimension. In fact, one can show that $L_x\cap\mathrm{Ann}(T_xQ)$ has constant dimension if and only if the stabilizer groups of the $G_\xi$-action on $Q$ have constant dimension, which happens if the $G_\xi$-orbits on $Q$ have constant dimension. In this case, $L_Q$ is a Dirac structure on $Q$.
\end{ex}

\section{Symplectic leaves and local structure of Poisson manifolds}
\subsection{Local and regular Poisson structures}
Let $\pi$ be a symplectic Poisson structure on $P$. Then Darboux's theorem asserts that, around each point of $P$, one can find coordinates $(q_1,\ldots,q_k,p_1,\ldots,p_k)$ such that 
\[
\pi=\sum_{1\leq i\leq k}\frac{\de}{\de q_i}\land \frac{\de}{\de p_i}.
\]
The symplectic form is then given by 
\[
\omega=\sum_{1\leq i\leq k}\dd q_i\land \dd p_i.
\]
In general, the image of $\tilde\pi\colon T^*P\to TP$ defined as in \eqref{eq:bundle_map_2} induces an integrable singular distribution on $P$. In fact, $P$ is a disjoint union of \emph{leaves} $\calO$ satisfying 
\[
T_x\calO=\tilde\pi(T^*_xP),\quad \forall x\in P.
\]

\begin{defn}[Regular Poisson structure]
A Poisson structure $\pi$ is called \emph{regular} if $\tilde\pi$ has locally constant rank.
\end{defn}

\begin{rem}
If the considered Poisson structure is regular, then it defines a \emph{foliation} in the ordinary sense.  Note that one can always find an open dense subset of $P$ where this is in fact the case. We call this the \emph{regular part} of $P$.
\end{rem}

Locally, by using the Darboux theorem, if $\pi$ has constant rank $k$ around a given point, there exist coordinates $(q_1,\ldots,q_k,p_1,\ldots,p_k,e_1\ldots,e_\ell)$ such that 
\[
\{q_i,p_j\}=\delta_{ij},\quad \{q_i,q_j\}=\{p_i,p_j\}=\{q_i,e_j\}=\{p_i,e_j\}=0.
\]
\subsection{Local splitting and symplectic foliation}
\begin{thm}[Local splitting theorem]
\label{thm:splitting}
Around any point $x_0$ of a Poisson manifold $(P,\pi)$ there exist local coordinates 
\[
(q_1,\ldots,q_k,p_1,\ldots,p_k,e_1,\ldots,e_\ell),\quad (q,p,e)(x_0)=(0,0,0)
\]
such that 
\begin{equation}
\label{eq:splitting}
\pi=\sum_{1\leq i\leq k}\frac{\de}{\de q_i}\land \frac{\de}{\de p_i}+\frac{1}{2}\sum_{1\leq i,j\leq \ell}\eta^{ij}(e)\frac{\de}{\de e_i}\land \frac{\de}{\de e_j}
\end{equation}
with $\eta^{ij}(0)=0$.
\end{thm}

\begin{rem}
We have a \emph{symplectic factor} in the splitting \eqref{eq:splitting} associated to the coordinates $(q_i,p_i)$ and a factor where all Poisson brackets vanish at $e=0$ associated to the coordinates $(e_j)$. The latter factor is often called the \emph{totally degenerate factor}.
We can identify the symplectic factor with an open subset of the leaf $\calO_{x_0}$ through $x_0$.  
If we consider the foliation given by the collection of all leaves
\[
P=\bigsqcup_{x\in P}\calO_x
\]
we see that $\pi$ canonically defines a singular foliation of $P$ by \emph{symplectic leaves}, since $\pi$ induces a symplectic structure on each leaf. One usually refers to the totally degenerate factor, which is locally well-defined up to isomorphism, as the \emph{transverse structure} to $\pi$ along a given leaf. 
\end{rem}

\begin{ex}[Symplectic leaves of Poisson structures on $\R^2$]
Let $f\colon \R^2\to \R$ be a smooth function, and consider the Poisson structure on $\R^2$ defined by 
\[
\{x_1,x_2\}:=f(x_1,x_2).
\]
The connected components of the set $S_f:=\{(x_1,x_2)\in\R^2\mid f(x_1,x_2)\not=0\}$ are the 2-dimensional symplectic leaves. Note that in the set $S_f$, each point is a symplectic leaf.
\end{ex}

\begin{ex}[Symplectic leaves of Lie-Poisson structures]
Consider a Lie algebra $\mathfrak{g}$ with dual $\mathfrak{g}^*$ equipped with its Lie-Poisson structure as in Example \ref{ex:Lie-Poisson}. The symplectic leaves are then the coadjoint orbits for any connected Lie group with Lie algebra $\mathfrak{g}$. Since $\{0\}$ is always an orbit, a Lie-Poisson structure is not regular unless $\mathfrak{g}$ is commutative.
\end{ex}

\begin{exe}
Describe the symplectic leaves in the duals of the Lie algebras $\mathfrak{su}(2)$ and $\mathfrak{sl}(2,\R)$. 
\end{exe}

\begin{rem}[Linearization problem]
Linearizing the functions $\eta^{ij}$ as in Theorem \ref{thm:splitting} at $x_0$, we can write 
\[
\{e_i,e_j\}=\sum_{1\leq k \leq \ell}c^k_{ij}e_k+O(e^2).
\]
Thus, it turns out that $c^k_{ij}$ defines a Lie-Poisson structure on the normal space to the symplectic leaf at $x_0$. The \emph{linearization problem} consists of determining whether one can choose suitable \emph{transverse} coordinates $(e_1,\ldots,e_\ell)$ with respect to which $O(e^2)$ vanishes. If the Lie algebra structure on the conormal bundle to a symplectic leaf determined by linearization of $\pi$ at a point $x_0$ is semi-simple and of \emph{compact type}, then $\pi$ is linearizable around $x_0$ through a smooth change of coordinates.
\end{rem}

\section{Poisson maps}
\subsection{Two definitions}
\begin{defn}[Poisson map (version I)]
\label{defn:Poisson_map}
Let $(P_1,\pi_1)$ and $(P_2,\pi_2)$ be Poisson manifolds. A smooth map $\psi\colon P_1\to P_2$ is a \emph{Poisson map} if $\psi^*\colon C^\infty(P_2)\to C^\infty(P_1)$ is a homomorphism of Poisson algebras, i.e. 
\[
\psi^*\{f,g\}_2=\{\psi^*f,\psi^*g\}_1,\quad \forall f,g\in C^\infty(P_2).
\]
\end{defn}

Equivalently, we can reformulate Definition \ref{defn:Poisson_map} in terms of Poisson bivector fields and Hamiltonian vector fields. 

\begin{defn}[Poisson map (version II)]
\label{defn:Poisson_map_II}
Let $(P_1,\pi_1)$ and $(P_2,\pi_2)$ be Poisson manifolds. A smooth map $\psi\colon P_1\to P_2$ is a \emph{Poisson map} if and only if either of the following two equivalent conditions hold:
\begin{enumerate}
    \item $\psi_*\pi_1=\pi_2$, i.e. $\pi_1$ and $\pi_2$ are $\psi$-related.
    \item $X_f=\psi_*(X_{\psi^*f})$ for all $f\in C^\infty(P_2)$.
\end{enumerate}
\end{defn}

\begin{rem}
Condition (2) of Definition \ref{defn:Poisson_map_II} shows that trajectories of $X_{\psi^*f}$ project to those of $X_f$ if $\psi$ is a Poisson map. However, $X_f$ being complete does not imply that $X_{\psi^*f}$ is complete. Hence, one define a Poisson map $\psi\colon P_1\to P_2$ to be \emph{complete} if for all $f\in C^\infty(P_2)$ such that $X_f$ is complete, then $X_{\psi^*f}$ is complete. 
\end{rem}

\subsection{Examples of Poisson maps}
\begin{ex}[Complete functions]
Consider $\R$ as a Poisson manifold endowed with the zero Poisson structure. Then any map $f\colon P\to \R$ is a Poisson map, which is \emph{complete} if and only if $X_f$ is a complete vector field.
\end{ex}

\begin{exe}
Find the Poisson manifolds for which the set of complete functions is closed under addition.
\end{exe}

\begin{ex}[Open subsets of symplectic manifolds]
Let $(P,\pi)$ be a symplectic manifold, and let $U\subseteq P$ be an open subset. Then the inclusion $U\hookrightarrow P$ is complete if and only if $U$ is closed. More, generally, the image of a complete Poisson map is a union of symplectic leaves.
\end{ex}

\begin{exe}
Show that the inclusion of every symplectic leaf in a Poisson manifold is a complete Poisson map.
\end{exe}

\begin{exe}
\label{exe:fibration}
let $P_1$ be a Poisson manifold and let $P_2$ be symplectic. Show that then any Poisson map $\psi\colon P_1\to P_2$ is a submersion. Furthermore, if $P_2$ is connected and $\psi$ is complete, then $\psi$ is surjective (assuming that $P_1$ is nonempty).
\end{exe}

\begin{rem}
Exercise \ref{exe:fibration} gives a first hint to the fact that complete Poisson maps with symplectic target must be \emph{fibrations}. In fact, if $P_1$ is symplectic and $\dim P_1=\dim P_2$, then a complete Poisson map $\psi\colon P_1\to P_2$ is a covering map. In general, a complete Poisson map $\psi\colon P_1\to P_2$, where $P_2$ is symplectic, is a locally trivial symplectic fibration with a flat \emph{Ehresmann connection}: the horizontal lift in $T_xP_1$ of a vector $X\in T_{\psi(x)}P_2$ is defined as
\[
\tilde\pi_1((\dd_x\psi)^*\tilde\pi_2^{-1}(X)).
\]
The horizontal subspaces define a foliation whose leaves are coverings of $P_2$, and $P_1$ and $\psi$ are completely determined, up to isomorphism, by the \emph{holonomy}
\begin{equation}
\label{eq:holonomy}
\pi_1(P_2,x)\to \mathrm{Aut}(\psi^{-1}(x)),
\end{equation}
where $\pi_1(P_2,x)$ in \eqref{eq:holonomy} denotes the \emph{fundamental group} of $P_2$ with base point $x$.
\end{rem}

\chapter{Deformation Quantization}
In this chapter we will consider the mathematical framework needed to understand Kontsevich's construction of deformation quantization and more generally the formality theorem. We will start by introducing the general concept of deformation quantization (star products) and answer the existence problem first for the local symplectic case (including a global approach). Then we will move to the more general Poisson case where we build everything up for the formality theorem which implies the existence of a deformation quantization as an application for a special case. Then we will discuss the construction of the explicit star product provided by Kontsevich, using the notion of graphs. At the end, we will also mention the formality result using the notion of operads.
This chapter is mainly based on \cite{Moy,Fedosov1994,Kontsevich1999,K,Tamarkin2003,CI,GuttRawnsleySternheimer2005,CKTB}.

\section{Star products}
\subsection{Formal deformations}
Given a commutative ring $k$, the set of \emph{formal power series} $k[\![X]\!]$ can be thought of polynomials with coefficients in $k$ in some formal parameter $X$ with infinite terms, not concerning any convergence issue. One can show that $k[\![X]\!]$ carries in fact a ring structure.
More precisely, one can define $k[\![X]\!]$ as the completion of $k[X]$ by using a particular metric, which endows $k[\![X]\!]$ with the structure of a topological ring.

Let $\calA$ be a commutative associative algebra with unit over some commutative base ring $k$ and let $\hbar$ denote a formal parameter\footnote{In the introduction we have denoted the formal parameter by $\epsilon$. We want to use $\hbar$ instead to emphasize the reduced Planck constant from the physics literature, which takes the role of a (small) deformation parameter in this theory. More precisely, the deformation parameter is usually given by $\frac{\I\hbar}{2}$.}. 

\begin{defn}[Formal deformation]
A \emph{formal deformation} of $\calA$ is the algebra $\calA[\![\hbar]\!]$ of formal power series over the ring $k[\![\hbar]\!]$ of formal power series.
\end{defn}

\begin{rem}
Elements of the deformed algebra $\calA[\![\hbar]\!]$ are of the form
\[
C=\sum_{i\geq 0}c_i\hbar^i,\quad c_i\in\calA
\]
and the product between such elements is given by the Cauchy formula
\[
\left(\sum_{i\geq 0}a_i\hbar^i\right)\bullet_\hbar\left(\sum_{j\geq 0}b_j\hbar^j\right)=\sum_{r\geq 0}\left(\sum_{\ell\geq 0}a_{r-\ell}b_\ell\right)\hbar^r.
\]
\end{rem}

\begin{defn}[Star product (general)]
A \emph{star product} is a $k[\![\hbar]\!]$-linear associative product $\star$ on $\calA[\![\hbar]\!]$ which deforms the trivial extension $\bullet_\hbar\colon \calA[\![\hbar]\!]\otimes_{k[\![\hbar]\!]}\calA[\![\hbar]\!]\to \calA[\![\hbar]\!]$ in the sense that for any two $v,w\in \calA[\![\hbar]\!]$ we have 
\[
v\star w=v\bullet_\hbar w,\quad \mathrm{mod}\,\,\hbar.
\]
\end{defn}

\begin{rem}
In fact, the star product is a formal noncommutative deformation of the usual pointwise product for the case where $\calA=C^\infty(P)$ for some Poisson manifold $(P,\pi)$ and $k=\R$. 
\end{rem}

\begin{defn}[Star product]
A \emph{star product} on a Poisson manifold $(P,\pi)$ is an $\R[\![\hbar]\!]$-bilinear map 
\begin{align}
    \begin{split}
        \star\colon C^\infty(P)[\![\hbar]\!]\times C^\infty(P)[\![\hbar]\!]&\to C^\infty(P)[\![\hbar]\!]\\
        (f,g)&\mapsto f\star g
    \end{split}
\end{align}
such that 
\begin{enumerate}
    \item $f\star g=fg+\sum_{i\geq1}B_i(f,g)\hbar^i$,
    \item $(f\star g)\star h=f\star (g\star h)$, \hspace{0.3cm} $\forall f,g,h\in C^\infty(P)$,
    \item $1\star f=f\star 1=f$, \hspace{0.3cm} $\forall f\in C^\infty(P)$,
\end{enumerate}
\end{defn}

\begin{rem}
The $B_i$ are bidifferential operators on $C^\infty(P)$ of \emph{globally bounded order}. We can write
\[
B_i(f,g)=\sum_{K,L}\beta^{KL}_i\de_Kf\de_Lg,
\]
where we sum over all multi-indices $K=(k_1,\ldots,k_m)$ and $L=(\ell_1,\ldots,\ell_n)$ for some lengths $m,n\in \mathbb{N}$. The $\beta^{KL}_i$ are smooth functions which are non-zero only for finitely many choices of the multi-indices $K$ and $L$.
\end{rem}

\begin{rem}
A star product on $\calA[\![\hbar]\!]$ is sometimes also called a \emph{formal deformation} of $\calA\subset \calA[\![\hbar]\!]$.
\end{rem}

\subsection{Moyal product}

Consider the standard symplectic manifold $(\R^{2n},\omega_0)$ endowed with the canonical symplectic form $\omega_0$ regarded as a Poisson manifold with Poisson structure induced by $\omega_0$. Moreover, choose Darboux coordinates $(q,p)=(q_1,\ldots,q_n,p_1,\ldots,p_n)$ on $\R^{2n}$.

\begin{defn}[Moyal product\cite{Moy}]
The \emph{Moyal product} is the star product on $C^\infty(\R^{2n})$ defined by 
\begin{equation}
\label{eq:Moyal}
f\star g:=f(q,p)\exp\left(\frac{\I\hbar}{2}\left(\overleftarrow{\de}_q\overrightarrow{\de}_p-\overleftarrow{\de}_p\overrightarrow{\de}_q\right)\right)g(q,p),\qquad f,g\in C^\infty(\R^{2n}),
\end{equation}
where $\overleftarrow{\de}$ denotes the derivative on $f$ and $\overrightarrow{\de}$ the derivative on $g$. 
\end{defn}

\begin{rem}
Note that in \eqref{eq:Moyal} the formal parameter is actually replaced by $\hbar\to\frac{\I\hbar}{2}$, which is typical in the physics literature, especially in quantum field theory. 
\end{rem}

If we consider $\R^d$ regarded as a Poisson manifold endowed with a constant Poisson structure $\pi$, we can define a star product similarly to \eqref{eq:Moyal} by 
\begin{equation}
    \label{eq:Moyal_general}
    f\star g(x)=\exp\left(\frac{\I\hbar}{2}\pi^{ij}\frac{\de}{\de x^i}\frac{\de}{\de y^j}\right)f(x)g(y)\bigg|_{y=x},\qquad f,g\in C^{\infty}(\R^d).
\end{equation}

\begin{prop}
The star product defined in \eqref{eq:Moyal_general} is associative for any choice of $\pi^{ij}$.
\end{prop}

\begin{proof}
\begin{align*}
    ((f\star g)\star h)(x)&=\exp\left(\frac{\I\hbar}{2}\pi^{ij}\frac{\de}{\de x^i}\frac{\de}{\de z^j}\right)(f\star g)(x)h(z)\bigg|_{x=z}\\
    &=\exp\left(\frac{\I\hbar}{2}\pi^{ij}\left(\frac{\de}{\de x^i}+\frac{\de}{\de y^i}\right)\frac{\de}{\de z^j}\right)\exp\left(\frac{\I\hbar}{2}\pi^{k\ell}\frac{\de}{\de x^k}\frac{\de}{\de y^\ell}\right)f(x)g(y)h(z)\bigg|_{x=y=z}\\
    &=\exp\left(\frac{\I\hbar}{2}\left(\pi^{ij}\frac{\de}{\de x^i}\frac{\de}{\de z^j}+\pi^{k\ell}\frac{\de}{\de y^k}\frac{\de}{\de z^\ell}+\pi^{mn}\frac{\de}{\de x^m}\frac{\de}{\de y^n}\right)\right)f(x)g(y)h(z)\bigg|_{x=y=z}\\
    &=\exp\left(\frac{\I\hbar}{2}\pi^{ij}\frac{\de}{\de x^i}\left(\frac{\de}{\de y^j}+\frac{\de}{\de z^j}\right)\right)\exp\left(\frac{\I\hbar}{2}\pi^{ij}\frac{\de}{\de y^k}\frac{\de}{\de z^\ell}\right)f(x)g(y)h(z)\bigg|_{x=y=z}\\
    &=(f\star (g\star h))(x).
\end{align*}
\end{proof}

\subsection{Fedosov's globalization approach}
\label{subsec:Fedosov}
The Moyal product as in \eqref{eq:Moyal} is only defined locally on $\R^{2n}$. In \cite{Fedosov1994} Fedosov showed how to obtain a star product on any symplectic manifold $(M,\omega)$ by using a symplectic connection. 

\begin{defn}[Symplectic connection]
A \emph{symplectic connection} on a symplectic manifold $(M,\omega)$ is a torsion-free connection $\nabla$ preserving the tensor $(\omega_{ij})$, i.e. $\nabla_i\omega_{j\ell}=0$ where $\nabla_i$ denotes the covariant derivative with respect to $\frac{\de}{\de x^i}$ for local coordinates $(x^i)$ on $M$.
\end{defn}

\begin{thm}[Fedosov\cite{Fedosov1994}]
Let $(M,\omega)$ be a $2n$-dimensional symplectic manifold. Consider the bundle\footnote{We will consider the \emph{completed} symmetric algebra, denoted by $\widehat{\Sym}$. The completion of the underlying tensor product is necessary for certain topological reasons. See also Section \ref{subsec:dual_view}.} $\calW:=\widehat{\Sym}(T^*M)$ and let $\nabla$ be a symplectic connection of the \emph{Weyl bundle} $\calW[\![\hbar]\!]$ and denote by $\star$ the Moyal product on $(\R^{2n},\omega_0)$. Then one can define a global star product $\star_M$ on $C^\infty(M)$ by 
\[
f\star_Mg=\sigma(\sigma^{-1}(f)\star\sigma^{-1}(g)),
\]
where\footnote{We have denoted by $H^0_D(\calW)$ the space of $D$-flat sections on $\calW$.} $\sigma\colon H^0_D(\calW)\xrightarrow{\sim} \mathcal{Z}$ with $D$ a flat connection ($D^2=0$) on $\calW$ induced by the symplectic connection $\nabla$ and $\mathcal{Z}$ the center of $\calW$.
\end{thm}

\begin{rem}
Note that taking the center $\mathcal{Z}$ of $\calW$ coincides with the center $\mathcal{Z}_\hbar$ of $\calW[\![\hbar]\!]$ which is given by $C^\infty(M)$.
\end{rem}

\begin{rem}
Another approach to globalize the Moyal product was given by DeWilde and Lecomte \cite{DeWildeLecomte1983} using Darboux's theorem. However, it turns out that Fedosov's approach is more important in the sense that it provides a natural extension to the general Poisson case (see Section \ref{sec:globalization}).
\end{rem}

\subsection{Equivalent star products}
\begin{exe}
\label{exe:Poisson_bracket}
Check that the $B_i$ of a star product $\star$ are \emph{strict} bidifferential operators, i.e. there is no term in order zero and
\[
B_i(1,f)=B_i(f,1)=0,\quad \forall i\in\N.
\]
Moreover, check that the skew-symmetric part $B^-_1$ of the first bidifferential operator, defined as
\[
B_1^-(f,g):=\frac{1}{2}\left(B_1(f,g)-B_1(g,f))\right),
\]
satisfies:
\begin{itemize}
    \item $B_1^-(f,g)=-B_1^-(g,f)$,
    \item $B_1^-(f,gh)=gB_1^-(f,h)+B_1^-(f,g)h$,
    \item $B_1^-(B_1^-(f,g),h)+B_1^-(B_1^-(g,h),f)+B_1^-(B_1^-(h,f),g)=0$.
\end{itemize}
Deduce that $\{f,g\}:=B_1^-(f,g)$ is a Poisson structure.
\end{exe}

\begin{rem}
Exercise \ref{exe:Poisson_bracket} shows that given a star product $\star$, we can always deduce a Poisson structure by the formula
\begin{equation}
\label{eq:limi_star_product}
\{f,g\}:=\lim_{\hbar\to 0}\frac{f\star g-g\star f}{\hbar}.
\end{equation}
On the other hand we want to consider the following problem: Given a Poisson manifold $(P,\pi)$, can we define an associative, but possibly noncommutative, product $\star$ on the algebra of smooth functions, which is the pointwise product and such that \eqref{eq:limi_star_product} is satisfied.
\end{rem}

\begin{defn}[Equivalent star products]
\label{defn:equivalent_star_products}
Two star products $\star$ and $\star'$ on $C^\infty(P)$ are said to be \emph{equivalent} if and only if there is a linear operator $D\colon C^\infty(P)[\![\hbar]\!]\to C^\infty(P)[\![\hbar]\!]$ of the form 
\[
Df:=f+\sum_{i\geq 1}D_if\hbar^i,
\]
such that 
\begin{equation}
\label{eq:equivalent}
f\star' g=D^{-1}(Df\star Dg),
\end{equation}
where $D^{-1}$ denotes the inverse in the sense of formal power series.
\end{defn}

\begin{rem}
Note that by the definition of a star product, the $D_i$ have to be differential operators vanishing on constants.
\end{rem}

\begin{lem}
For any equivalence class of star products, there is a representative whose first term $B_1$ in the $\hbar$-expansion is skew-symmetric.
\end{lem}

\begin{proof}
Given any star product
\[
f\star g=fg+B_1(f,g)\hbar+B_2(f,g)\hbar^2+\dotsm
\]
one can define an equivalent star product $\star'$ as in \eqref{eq:equivalent} by using a formal differential operator 
\[
D=\id+D_1\hbar+D_2\hbar^2+\dotsm
\]
The skew-symmetric condition for $B_1'$ implies that 
\begin{equation}
\label{eq:differential}
D_1(fg)=D_1f g+fD_1g+\underbrace{\frac{1}{2}\left(B_1(f,g)+B_1(g,f)\right)}_{=:B_1^+(f,g)},
\end{equation}
which can be checked for polynomials and further be completed to hold to any smooth functions on $P$. If we choose $D_1$ to vanish on linear functions, Equation \eqref{eq:differential} describes uniquely how $D_1$ acts on quadratic terms, given by the symmetric part $B_1^+$ of $B_1$. Namely, we get 
\[
D_1(x^ix^j)=B_1^+(x^i,x^j):=\frac{1}{2}\left(B_1(x^i,x^j)+B_1(x^j,x^i)\right),
\]
where $(x^i)$ are local coordinates on $P$. By the associativity of $\star$, we get a well-defined operator since it does not depend on the way how we order the factors. 
\end{proof}

\begin{exe}
Check that $D_1((x^ix^j)x^\ell)=D_1(x^i(x^jx^\ell))$.
\end{exe}

\begin{rem}
A natural problem appearing in this setting concerns \emph{classification} of star products.
In fact, it was shown that equivalence classes of star products on a symplectic manifold $(M,\omega)$ are in one-to-one correspondence with elements in $H^{2}(M)[\![\hbar]\!]$ \cite{Deligne1995,GuttRawnsley1999}. The general construction was formulated by Kontsevich \cite{K}. He constructed an explicit formula for a star product on $\R^d$ endowed with any Poisson structure. However, we want to start with a much more general statement, formulated and also proved by Kontsevich \cite{K}, which provides a solution to the existence of deformation quantization for a special case as we will see. This general result is called \emph{formality}.
\end{rem}

\section{Formality}

\subsection{Some formal setup}
For a Poisson manifold $(P,\pi)$, one can define a bracket
\begin{equation}
    \label{eq:bracket}
    \{f,g\}_\hbar:=\sum_{m\geq 0}\hbar^m\sum_{0\leq i,j,\ell\leq m\atop i+j+\ell=m}\pi_i(\dd f_j,\dd g_\ell),
\end{equation}
where 
\[
f=\sum_{j\geq0}f_j\hbar^j,\qquad g=\sum_{\ell\geq 0}g_\ell\hbar^\ell.
\]

\begin{defn}[Formal Poisson structure]
    We call
    \[
    \pi_\hbar:=\pi_0+\pi_1\hbar+\pi_2\hbar^2+\dotsm
    \]
    a formal \emph{Poisson structure} if $\{\enspace,\enspace\}_\hbar$ defined as in \eqref{eq:bracket} is a Lie bracket on $C^\infty(P)[\![\hbar]\!]$. 
\end{defn}

The gauge group in this case is given by \emph{formal diffeomorphisms}, i.e. formal power series of the form 
\begin{equation}
\label{eq:formal_diffeomorphisms}
\Psi_\hbar:=\exp(\hbar X),
\end{equation}
where $X:=\sum_{i\geq 0}X_i\hbar^i\in\mathfrak{X}(P)[\![\hbar]\!]$ is a \emph{formal vector field}, i.e. a formal power series where each coefficients $X_i\in \mathfrak{X}(P)$ are vector fields on $P$. 

\begin{defn}[Baker--Campbell--Hausdorff formula]
For two vector fields $X$ and $Y$ on a manifold we have the \emph{Baker--Campbell--Hausdorff} (BCH) formula for solving the equation $\exp(X)\exp(Y)=\exp(Z)$ for a vector field $Z$:
\begin{align}
\begin{split}
\label{eq:BCH}
\exp(X)\cdot \exp(Y)&=\exp\left(X+\left(\int_0^1\psi\left(\exp(\ad_{X})\exp(t\ad_{Y})\right)\dd t\right)Y\right)\\
&=\exp\left(X+Y+\frac{1}{2}[X,Y]+\dotsm\right),
\end{split}
\end{align}
where $\psi(x):=\frac{x\log x}{x-1}=1-\sum_{n\geq 0}\frac{(1-x)^n}{n(n+1)}$ and $\ad_X=[X,\enspace]$. 
\end{defn}

\begin{exe}
Check that on the set of formal vector fields $\mathfrak{X}(P)[\![\hbar]\!]$ on $P$ there is a group structure regarding the exponential induced by the BCH formula, i.e. we have a group structure on $\{\exp(\hbar X)\mid X\in\mathfrak{X}(P)[\![\hbar]\!]\}$ by
\[
\exp(\hbar X)\cdot \exp(\hbar Y):=\exp\left(\hbar X+\hbar Y+\frac{1}{2}\hbar^2 [X,Y]+\dotsm\right).
\]
\end{exe}

The group action is given by 
\begin{equation}
    \label{eq:formal_group_action}
    \exp(\hbar X)_*\pi_\hbar:=\sum_{m\geq 0}\hbar^m\sum_{0\leq i,j,\ell\leq m\atop i+j+\ell=m}(L_{X_i})^j\pi_\ell
\end{equation}

\begin{rem}
Kontsevich's result was to find an identification between the set of star products modulo the action of the differential operators $D$ as in Definition \ref{defn:equivalent_star_products} and the set of formal Poisson structures modulo the gauge group consisting of formal diffeomorphisms as \eqref{eq:formal_diffeomorphisms}.
\end{rem}

\subsection{Differential graded Lie algebras}

\begin{defn}[Graded Lie algebra]
A \emph{graded Lie algebra} (GLA) is a $\Z$-graded vector space $\mathfrak{g}=\bigoplus_{i\in\Z}\mathfrak{g}^i$ endowed with a bilinear operation
\[
[\enspace,\enspace]\colon\mathfrak{g}\otimes \mathfrak{g}\to \mathfrak{g}
\]
satisfying 
\begin{enumerate}
    \item $[a,b]\in\mathfrak{g}^{\alpha+\beta}$ (homogeneity).
    \item $[a,b]=-(-1)^{\alpha\beta}[b,a]$ (graded skew-symmetry),
    \item $[a,[b,c]]=[[a,b],c]+(-1)^{\alpha\beta}[b,[a,c]]$ (graded Jacobi identity)
\end{enumerate}
for all $a\in\mathfrak{g}^\alpha$, $b\in\mathfrak{g}^{\beta}$ and $c\in\mathfrak{g}^\gamma$.
\end{defn}

\begin{defn}[Shift]
\label{defn:shift}
Given any graded vector space $\mathfrak{g}=\bigoplus_{i\in\mathbb{Z}}\mathfrak{g}^i$, we can obtain a new graded vector space $\mathfrak{g}[n]$ by shifting each component by $n$, i.e.
\[
\mathfrak{g}[n]:=\bigoplus_{i\in\Z}\mathfrak{g}[n]^i, \qquad \mathfrak{g}[n]^i:=\mathfrak{g}^{i+n}.
\]
\end{defn}

\begin{defn}[Differential graded Lie algebra]
A \emph{differential graded Lie algebra} (DGLA) is a GLA $\mathfrak{g}$ together with a differential $\dd\colon\mathfrak{g}\to \mathfrak{g}$, i.e. a linear operator of degree $+1$ ($\dd\colon \mathfrak{g}^i\to \mathfrak{g}^{i+1}$) which satisfies the Leibniz rule
\[
\dd[a,b]=[\dd a,b]+(-1)^\alpha[a,\dd b],\quad a\in \mathfrak{g}^\alpha,b\in\mathfrak{g}^\beta
\]
and squares to zero ($\dd\circ \dd=0$).
\end{defn}

\begin{ex}
Note that one can turn any Lie algebra into a DGLA concentrated in degree zero by using the trivial zero differential $\dd=0$.
\end{ex}

Given a DGLA $\mathfrak{g}$, we can consider its cohomology 
\[
H^i(\mathfrak{g}):=\ker(\dd\colon \mathfrak{g}^i\to \mathfrak{g}^{i+1})/\im(\dd\colon \mathfrak{g}^{i-1}\to\mathfrak{g}^i).
\]

\begin{rem}
Note that $H(\mathfrak{g}):=\bigoplus_iH^i(\mathfrak{g})$ has a natural structure of a graded vector space and further has a GLA structure defined on equivalence classes $[a],[b]\in H(\mathfrak{g})$ by 
\[
[[a],[b]]_{H(\mathfrak{g})}:=[[a,b]_{\mathfrak{g}}],
\]
where $[\enspace,\enspace]_\mathfrak{g}$ denotes the Lie bracket on $\mathfrak{g}$. Moreover, extending the GLA $H(\mathfrak{g})$ by the zero differential will turn it into a DGLA.
\end{rem}

\begin{rem}
Note that each morphism $\varphi\colon \mathfrak{g}_1\to\mathfrak{g}_2$ between two DGLAs induces a morphism $H(\varphi)\colon H(\mathfrak{g}_1)\to H(\mathfrak{g}_2)$ between the cohomologies. In particular, we are interested in \emph{quasi-isomorphisms} between DGLAs.
\end{rem}

\begin{defn}[Quasi-isomorphism]
A \emph{quasi-isomorphism} is a morphism between DGLAs which induces an isomorphism on the level of cohomology.
\end{defn}

\begin{rem}
The notion of quasi-isomorphism naturally holds more generally for any (co)chain complex and the induced (co)homology theory.
\end{rem}

\begin{defn}[Quasi-isomorphic]
Two DGLAs $\mathfrak{g}_1$ and $\mathfrak{g}_2$ are called \emph{quasi-isomorphic} if there is a quasi-isomorphism $\varphi\colon \mathfrak{g}_1\to \mathfrak{g}_2$.
\end{defn}

\begin{rem}
The existence of a quasi-isomorphism $\varphi\colon \mathfrak{g}_1\to \mathfrak{g}_2$ does not imply that there is also a quasi-isomorphism \emph{inverse} $\varphi^{-1}\colon \mathfrak{g}_2\to \mathfrak{g}_1$ and therefore they do not directly define an equivalence relation. We will deal with this issue by considering the category of \emph{$L_\infty$-algebras}.
\end{rem}

\subsection{$L_\infty$-algebras} 
The notion of an $L_\infty$-algebra, first introduced in \cite{Stasheff1992}, gives a generalization of a DGLA. It has a lot of applications in modern constructions of quantum field theory and homotopy theory. In this section, we will define what an $L_\infty$-algebra is and show how one can extract an $L_\infty$-structure out of any DGLA. Moreover, we briefly explain its representation in terms of higher brackets. 
We refer to \cite{Stasheff1992,LadaStasheff1993,LadaMarkl1995} for more details.

\begin{defn}[Formal]
A DGLA $\mathfrak{g}$ is called \emph{formal} if it is quasi-isomorphic to its cohomology, regarded as a DGLA with zero differential and the induced bracket.
\end{defn}

\begin{rem}
Kontsevich's \emph{formality theorem} \cite{K} consists of the result that the DGLA of multidifferential operators (see later) is formal.
\end{rem}

\begin{defn}[Graded coalgebra]
A \emph{graded coalgebra} (GCA) on a base ring $k$ is a $\Z$-graded vector space $\mathfrak{h}=\bigoplus_{i\in\Z}\mathfrak{h}^i$ endowed with a \emph{comultiplication}, i.e. a graded linear map 
\[
\Delta\colon \mathfrak{h}\to \mathfrak{h}\otimes \mathfrak{h}
\]
such that 
\[
\Delta(\mathfrak{h}^i)\subset \bigoplus_{j+\ell}\mathfrak{h}^j\otimes \mathfrak{h}^\ell
\]
and which also satisfies the \emph{coassociativity} condition
\[
(\Delta\otimes \id)\Delta(a)=(\id\otimes \Delta)\Delta(a),\quad \forall a\in \mathfrak{h}.
\]
\end{defn}

\begin{defn}[GCA with counit]
We say that a GCA $\mathfrak{h}$ has a \emph{counit} if there is a morphism
\[
\varepsilon\colon \mathfrak{h}\to k
\]
such that $\varepsilon(\mathfrak{h}^i)=0$ for all $i>0$ and
\[
(\varepsilon\otimes\id)\Delta(a)=(\id\otimes \varepsilon)\Delta(a)=a,\quad\forall a\in \mathfrak{h}.
\]
\end{defn}

\begin{defn}[Cocommutative GCA]
We say that a GCA $\mathfrak{h}$ is \emph{cocommutative} if there is a \emph{twisting map} defined on homogeneous elements $x,y\in\mathfrak{h}$ by
\begin{align*}
    \mathsf{T}\colon \mathfrak{h}\otimes \mathfrak{h}&\to \mathfrak{h}\otimes \mathfrak{h},\\
    x\otimes y&\mapsto \mathsf{T}(x\otimes y):=(-1)^{\vert x\vert\vert y\vert}y\otimes x,
\end{align*}
where $\vert x\vert\in \Z$ denotes the degree of $x\in\mathfrak{h}$, and extended linearly, such that
\[
\mathsf{T}\circ \Delta=\Delta.
\]
\end{defn}

\begin{ex}[Tensor algebra]
Let $V$ be a (graded) vector space over $k$ and consider its tensor algebra $T(V)=\bigoplus_{n=0}^\infty V^{\otimes n}$. Denote by $1$ the unit of $k$. Then, $T(V)$ can be endowed with a coalgebra structure. We define a comultiplication $\Delta_{T(V)}$ on homogeneous elements by 
\[
\Delta_{T(V)}(v_1\otimes\dotsm \otimes v_n):=1\otimes(v_1\otimes\dotsm\otimes v_n)+\sum_{j=1}^{n-1}(v_1\otimes\dotsm\otimes v_j)\otimes(v_{j+1}\otimes\dotsm \otimes v_n)+(v_1\otimes\dotsm\otimes v_n)\otimes 1.
\]
In particular, we have a counit $\varepsilon_{T(V)}$ as the canonical projection $\varepsilon_{T(V)}\colon T(V)\to V^{\otimes 0}:=k$. Note that if we consider the \emph{reduced tensor algebra} $\bar T(V):=\bigoplus_{n=1}^\infty V^{\otimes n}$, then the projection $\bar\pi\colon T(V)\to \bar T(V)$ and the inclusion $\bar i\colon T(V)\hookrightarrow \bar T(V)$ induce a comultiplication on $\bar T(V)$, thus a coalgebra structure but without counit. 
\end{ex}

\begin{ex}[Symmetric and exterior algebra]
\label{ex:reduced_tensor_algebra}
Let $V$ be a (graded) vector space over $k$. Then we can also consider the \emph{Symmetric algebra} $\Sym(V)$ and \emph{exterior algebra} $\bigwedge V$. Note that for the graded case we have to take the quotients with respect to the two-sided ideals generated by homogeneous elements of the form $v\otimes w-\mathsf{T}(v\otimes w)$ and $v\otimes w+\mathsf{T}(v\otimes w)$.
They can be endowed with a coalgebra structure (without counit if constructed for the reduced tensor algebra). Indeed, e.g. we can define a comultiplication $\Delta_{\Sym(V)}$ on homogeneous elements $v\in V$ by 
\[
\Delta_{\Sym(V)}(v):=1\otimes v+v\otimes 1,
\]
\end{ex}

\begin{defn}[Coderivation]
A \emph{coderivation} of degree $\ell$ on some GCA $\mathfrak{h}$ is a graded linear map $\delta\colon \mathfrak{h}^i\to \mathfrak{h}^{i+\ell}$ which satisfies the \emph{co-Leibniz identity}
\[
\Delta\delta(v)=(\delta\otimes \id)\Delta(v)+((-1)^{\ell\vert v\vert}\id\otimes \delta)\Delta(v),\quad \forall v\in \mathfrak{h}^{\vert v\vert}.
\]
\end{defn}

\begin{defn}[Codifferential]
A \emph{codifferential} $Q$ on a coalgebra is a coderivation of degree $+1$ which squares to zero ($Q^2=0$).
\end{defn}

\begin{defn}[$L_\infty$-algebra]
An \emph{$L_\infty$-algebra} is a graded vector space $\mathfrak{g}$ over $k$ endowed with a degree $+1$ coalgebra differential $Q$ on the reduced symmetric space $\overline{\Sym}(\mathfrak{g}[1])$. 
\end{defn}

\begin{defn}[$L_\infty$-morphism]
An \emph{$L_\infty$-morphism} $F\colon(\mathfrak{g},Q)\to (\widetilde{\mathfrak{g}},\widetilde{Q})$ is a morphism
\[
F\colon \overline{\Sym}(\mathfrak{g}[1])\to \overline{\Sym}(\widetilde{\mathfrak{g}}[1]).
\]
of graded coalgebras, which also commutes with the differentials, i.e. 
\[
F Q=\widetilde{Q} F.
\]
\end{defn}

\begin{rem}
As in the dual case an algebra morphism $f\colon \Sym(\calA)\to \Sym(\calA)$ (resp. a derivation $\delta\colon \Sym(\calA)\to \Sym(\calA)$) is uniquely determined by its restriction to the algebra $\calA=\Sym^1(\calA)$ because of the homomorphism condition $f(a\otimes b)=f(a)\otimes f(b)$ (resp. Leibniz rule), an $L_\infty$-morphism $F$ (resp. a coderivation $Q$) is uniquely determined by its projection onto the first component $F^1$ (resp. $Q^1$). 
\end{rem}

Let us introduce the generalized notation $F_j^i$ (resp. $Q_j^i$) for the projection to the $i$-th component of the target vector space restricted to the $j$-th component of the domain. We can then rephrase the condition for $F$ (resp. $Q$) to be an $L_\infty$-morphism (resp. a codifferential). Indeed, using this notation, we get that $QQ$, $FQ$, and $\widetilde{Q}F$ are coderivations. It is sufficient to show this on the projection to the first factor for each term.

\begin{defn}
A coderivation $Q$ is a codifferential if and only if
\begin{equation}
\label{eq:L_infty_differential}
\sum_{1\leq i\leq n}Q_i^1Q_n^i=0,\quad \forall n\in \N.
\end{equation}
\end{defn}

\begin{defn}
A morphism $F$ of graded coalgebras is an $L_\infty$-morphism if and only if 
\begin{equation}
\label{eq:L_infty_morphism}
\sum_{1\leq i\leq n}F_i^1Q_n^i=\sum_{1\leq i\leq n}\widetilde{Q}_i^1F_{n}^i,\quad \forall n\in \N.
\end{equation}
\end{defn}

\begin{ex}
For the case when $n=1$, we get 
\[
Q_1^1Q_1^1=0,\qquad F_1^1Q_1^1=\widetilde{Q}_1^1F_1^1.
\]
Therefore, every coderivation $Q$ induces the structure of a cochain complex of vector spaces on $\mathfrak{g}$ and every $L_\infty$-morphism restricts to a cochain map $F_1^1$.
\end{ex}

\begin{rem}
We can generalize the definitions for a DGLA to this case. In particular, we can define a quasi-isomorphism of $L_\infty$-algebras to be an $L_\infty$-morphism $F$ such that $F_1^1$ is a quasi-isomorphism of cochain complexes. Similarly, one can extend the notion of \emph{formality}.
\end{rem}

\begin{lem}
\label{lem:L_infty_relation}
Let $F\colon (\mathfrak{g},Q)\to (\widetilde{\mathfrak{g}},\widetilde{Q})$ be an $L_\infty$-morphism. If $F$ is a quasi-isomorphism it admits a \emph{quasi-inverse}, i.e. there is an $L_\infty$-morphism $G\colon (\widetilde{\mathfrak{g}},\widetilde{Q})\to (\mathfrak{g},Q)$ which induces the inverse isomorphism in the corresponding cohomologies.
\end{lem}

\begin{rem}
Lemma \ref{lem:L_infty_relation} implies an equivalence relation defined by $L_\infty$-quasi-isomorphisms, i.e. two $L_\infty$-algebras are $L_\infty$-quasi-isomorphic if and only if there is an $L_\infty$-quasi-isomorphism between them. This solves the problem that one has to face regarding DGLAs, where the equivalence relation only holds on the level of quasi-isomorphisms. 
\end{rem}

Let us see how we can extract an $L_\infty$-structure from any DGLA $\mathfrak{g}$ which indeed implies that an $L_\infty$-algebra is a generalization of a DGLA. Define $Q_1^1$ to be a multiple of the differential. For $n=2$, we can write condition \eqref{eq:L_infty_differential} as 
\begin{equation}
\label{eq:example_condition_n=2}
Q_1^1Q_2^1+Q_2^1Q_2^2=0.
\end{equation}
We identify $Q_1^1$ with the differential $\dd$ (up to some sign) and note that \eqref{eq:example_condition_n=2} suggests that $Q_2^1$ should be expressed in terms of the bracket $[\enspace,\enspace]$ since \eqref{eq:example_condition_n=2} has the same form as the compatibility condition between the bracket and the differential. 
In particular, we get
\begin{align}
    Q_1^1(a):=(-1)^\alpha \dd a,\qquad &a\in \mathfrak{g}^\alpha,\\
    \label{eq:bracket_Q}
    Q_2^1(b\otimes c):=(-1)^{\beta(\gamma-1)}[b,c],\qquad &b\in\mathfrak{g}^\beta,c\in\mathfrak{g}^\gamma,\\
    Q_n^1=0,\qquad &\forall n\geq 3. 
\end{align}

For $n=3$, we get 
\begin{equation}
\label{eq:example_condition_n=3}
Q_1^1Q_3^1+Q_2^1Q_3^2+Q_3^1Q_3^3=0.
\end{equation}

If we insert the definition of the bracket \eqref{eq:bracket_Q} and expand $Q_3^2$ in terms of $Q_2^1$ in \eqref{eq:example_condition_n=3}, we get

\begin{multline}
\label{eq:L_infty_Jacobi}
    (-1)^{(\alpha+\beta)(\gamma-1)}\left[(-1)^{\alpha(\beta-1)}[a,b],c\right]\\+(-1)^{(\alpha+\gamma)(\beta-1)}(-1)^{(\gamma-1)(\beta-1)}\left[(-1)^{\alpha(\gamma-1)}[a,c],b\right]\\+(-1)^{(\beta+\gamma)(\alpha-1)}(-1)^{(\beta+\gamma)(\alpha-1)}\left[(-1)^{\beta(\gamma-1)}[b,c],a\right]=0.
\end{multline}

\begin{rem}
Note that, by rearranging signs, \eqref{eq:L_infty_Jacobi} is the same as the \emph{graded Jacobi identity}.
\end{rem}

Similar constructions hold for a DGLA morphism $F\colon \mathfrak{g}\to \widetilde{\mathfrak{g}}$. It induces an $L_\infty$-morphism $\bar F$ which is completely determined by its first component $\bar F_1^1:=F$. The nontrivial conditions for $n=1$ and $n=2$ on $\bar F$, induced by \eqref{eq:L_infty_morphism}, are given by 
\begin{align*}
    \bar F_1^1Q_1^1(f)=\widetilde{Q}_1^1\bar F_1^1(f)&\Leftrightarrow F(\dd f)=\widetilde{\dd} F(f),\\
    \bar F_1^1Q_2^1(f\otimes g)+\bar F_1^2Q_2^2(f\otimes g)=\widetilde{Q}_1^1\bar F_1^2(f\otimes g)+\widetilde{Q}_2^1\bar F_2^2(f\otimes g)&\Leftrightarrow F\left([f,g]\right)=[F(f),F(g)].
\end{align*}

\begin{rem}
\label{rem:homotopy_Jacobi}
For $n=3$, where the $Q_3^1$ do not vanish, one can show that \eqref{eq:L_infty_Jacobi} would have been satisfied up to \emph{homotopy}, i.e. up to a term of the form
\[
\dd\rho(g,h,k)\pm \rho(\dd g,h,k)\pm \rho(g,\dd h,k)\pm \rho(g,h,\dd k),
\]
where $\rho\colon \bigwedge^3\mathfrak{g}\to \mathfrak{g}[-1]$. In this case, one says that $\mathfrak{g}$ has the structure of a \emph{homotopy Lie algebra}.
\end{rem}

\begin{rem}
The concept of Remark \ref{rem:homotopy_Jacobi} can be generalized by introducing the canonical isomorphism between the symmetric and exterior algebra (usually called \emph{d\'eclage isomorphism}) given by 
\begin{align*}
\mathrm{dec}_n\colon \Sym^n(\mathfrak{g}[1])&\xrightarrow{\sim} \bigwedge^n\mathfrak{g}[n],\\
x_1\otimes\dotsm \otimes x_n&\mapsto (-1)^{\sum_{1\leq i\leq n}(n-i)(\vert x_i\vert-1)}x_1\land\dotsm\land x_n,
\end{align*}
to define for each $n$ a \emph{multibracket} of degree $2-n$
\[
[\enspace,\ldots,\enspace]_n\colon \bigwedge^n\mathfrak{g}\to \mathfrak{g}[2-n]
\]
by starting from the corresponding $Q_n^1$. In fact, \eqref{eq:L_infty_differential} induces an infinite family of conditions on these multibrackets. A graded vector space $\mathfrak{g}$ together with such a family of operators is called a \emph{strong homotopy Lie algebra} (SHLA) (see also \cite{LadaMarkl1995}).
\end{rem}

\begin{defn}[Maurer--Cartan equation]
The \emph{Maurer--Cartan equation} of a DGLA $\mathfrak{g}$ is given by 
\begin{equation}
    \label{eq:Maurer-Cartan}
    \dd a+\frac{1}{2}[a,a]=0,\quad a\in \mathfrak{g}^1.
\end{equation}
\end{defn}

\begin{rem}
It is easy to show that the set of solutions to the Maurer--Cartan equation \eqref{eq:Maurer-Cartan} is preserved under the action of a morphism between DGLAs. Moreover, the group action by the gauge group that is canonically defined through the degree zero part of any formal DGLA also preserves the set of solutions to \eqref{eq:Maurer-Cartan}. This can be extended to the formal setting\footnote{Note that we can set $\mathfrak{g}[\![\hbar]\!]:=\mathfrak{g}\otimes k[\![\hbar]\!]$ and show that this is again a DGLA. As we have seen before, the gauge group is then formally defined as $G:=\exp(\hbar\mathfrak{g}^0[\![\hbar]\!])$, where $\mathfrak{g}^0[\![\hbar]\!]$ denotes the Lie algebra given as the degree zero part of $\mathfrak{g}[\![\hbar]\!]$. Note that the action of $\hbar\mathfrak{g}^1[\![\hbar]\!]$ is defined by generalizing the adjoint action. Namely, we have
\[
\exp(\hbar X)a:=\sum_{n\geq 0}\frac{(\mathrm{ad}_X)^n}{n!}(a)-\sum_{n\geq 0}\frac{(\mathrm{ad}_X)^n}{(n+1)!}(\dd X)=a+\hbar[X,a]-\hbar\dd X+O(\hbar^2),\quad \forall X\in\mathfrak{g}^0[\![\hbar]\!],a\in \mathfrak{g}^1[\![\hbar]\!].
\]}. For a DGLA $\mathfrak{g}$, we will denote by $\mathrm{MC}(\mathfrak{g})$ the set of solutions to \eqref{eq:Maurer-Cartan} in $\mathfrak{g}$. 
\end{rem}

\subsection{The DGLA of multivector fields $\calV$}

Recall from Section \ref{subsec:Multivectorfields_differentialforms} that a multivector field of degree $j\geq 1$ on some manifold $M$ is an element
\[
X\in\mathfrak{X}^j(M):=\Gamma\left(\bigwedge^j TM\right).
\]
In local coordinates, we can write 
\[
X=\sum_{1\leq i_1<\dotsm <i_j\leq \dim M}X^{i_1\dotsm i_j}\de_{i_1}\land\dotsm \land \de_{i_j}.
\]
Hence, if we consider the collection of these vector spaces in all degrees, we naturally get a graded vector space structure
\[
\widetilde{\calV}(M):=\bigoplus_{j\geq 0}\widetilde{\calV}^j(M),\qquad \widetilde{\calV}^j(M):=\begin{cases}C^\infty(M),&j=0\\ \mathfrak{X}^j(M),&j\geq 1\end{cases}
\]
Using the \emph{Schouten--Nijenhuis bracket} as defined in \eqref{eq:Schouten-bracket}, we can endow $\widetilde{\calV}(M)$ with a GLA structure. We will write $[\enspace,\enspace]_{\mathrm{SN}}$ to indicate that we mean the Schouten--Nijenhuis bracket. We note/recall the following identities for the Schouten Nijenhuis bracket:
\begin{enumerate}
    \item $[X,Y]_{\mathrm{SN}}=-(-1)^{(\vert X\vert+1)(\vert Y\vert+1)}[Y,X]_\mathrm{SN}$,
    \item $[X,Y\land Z]_\mathrm{SN}=[X,Y]_\mathrm{SN}\land Z+(-1)^{(\vert Y\vert+1)\vert Z\vert}Y\land[X,Z]_\mathrm{SN}$,
    \item $[X,[Y,Z]_\mathrm{SN}]_\mathrm{SN}=[[X,Y]_\mathrm{SN},Z]_\mathrm{SN}+(-1)^{(\vert X\vert+1)(\vert Y\vert+1)}[Y,[X,Z]_\mathrm{SN}]_\mathrm{SN}$.
\end{enumerate}

In order, to obtain the correct signs, we have to shift everything by $+1$. Thus, we define the GLA of multivector fields to be $\calV(M):=\widetilde{\calV}(M)[1]$. Namely, we have
\[
\calV(M):=\bigoplus_{j\geq -1}\calV^j(M),\qquad \calV^j(M):=\widetilde{\calV}^{j+1}(M).
\]
Finally, using the zero differential, we can turn it into a DGLA $(\calV(M),[\enspace,\enspace]_{\mathrm{SN}},\dd=0)$.

\subsubsection{The case of Poisson bivector fields}
\label{subsubsec:DGLA_Poisson}
Consider a manifold $M$ together with a Poisson bivector field $\pi\in \calV^1(M)$. Then we can consider its induced Poisson bracket $\{\enspace,\enspace\}$ on $C^\infty(M)$. Recall that, by Exercise \ref{ex:vanishing_SN_bracket}, we can translate the Jacobi identity of the Poisson bracket to the condition $[\pi,\pi]_{\mathrm{SN}}=0$. Indeed, we have 
\begin{multline*}
    \{\{f,g\},h\}+\{\{g,h\},f\}+\{\{h,f\},g\}=0\Leftrightarrow\\
    \pi^{ij}\de_j\pi^{k\ell}\de_i f\de_kg\de_\ell h+\pi^{ij}\de_j\pi^{k\ell}\de_ig\de_kh\de_\ell f+\pi^{ij}\de_j\pi^{k\ell}\de_ih\de_k\de_\ell g=0\Leftrightarrow\\
    \pi^{ij}\de_j\pi^{k\ell}\de_i\land \de_k\land \de_\ell=0\Leftrightarrow [\pi,\pi]_{\mathrm{SN}}=0.
\end{multline*}

Note that, since we have endowed $\calV(M)$ with the zero differential, we get that Poisson bivector fields are exactly solutions to the Maurer--Cartan equation on the DGLA $\calV(M)$:
\[
\dd\pi+\frac{1}{2}[\pi,\pi]_{\mathrm{SN}}=0,\quad \pi\in\calV^1(M).
\]
In particular, \emph{formal} Poisson structures $\{\enspace,\enspace\}_\hbar$ are associated to formal bivectors $\pi_\hbar\in \hbar\calV^1(M)[\![\hbar]\!]$.

\subsection{The DGLA of multidifferential operators $\calD$}

Note that to any associative algebra $\calA$ over a field $k$, we can assign the complex of multilinear maps given by 
\[
\calC:=\sum_{j\geq-1}\calC^j,\qquad \calC^j:=\Hom_k(\calA^{\otimes(j+1)},\calA).
\]

Define a family of composition $\{\circ_i\}$ such that for an $(m+1)$-linear operator $\phi\in \calC^{m}$ and an $(n+1)$-linear operator $\psi\in\calC^{n}$ we have 
\[
(\phi\circ_i\psi)(f_0\otimes\dotsm\otimes f_{m+n}):=\phi(f_0\otimes \dotsm\otimes f_{i-1}\otimes\psi(f_i\otimes \dotsm\otimes f_{i+n})\otimes f_{i+n+1}\otimes\dotsm\otimes f_{m+n}).
\]
If we sum over all the different ways of composition, we get 
\begin{equation}
\label{eq:composition}
\phi\circ\psi:=\sum_{0\leq i\leq m}(-1)^{in}\phi\circ_i\psi,
\end{equation}
which we can use to endow $\calC$ with a GLA structure. 

\begin{defn}[Gerstenhaber bracket]
The \emph{Gerstenhaber bracket} on the graded vector space $\calC$ is defined by
\begin{align}
\begin{split}
\label{eq:Gerstenhaber_bracket}
[\enspace,\enspace]_\mathrm{G}\colon \calC^m\otimes\calC^n&\to \calC^{m+n},\\
(\phi,\psi)&\mapsto[\phi,\psi]_\mathrm{G}:=\phi\circ\psi-(-1)^{mn}\psi\circ \phi.
\end{split}
\end{align}
\end{defn}

\begin{rem}
There is a more general notion of a \emph{Gerstenhaber algebra} \cite{Ger1963}, which describes the connection between \emph{supercommutative rings} \cite{Varadarajan2004} and \emph{graded Lie superalgebras} \cite{Kac1977}. In particular, it consists of a graded commutative algebra endowed with a Poisson bracket of degree $-1$. This structure plays also a fundamental role in other constructions of theoretical physics, such as in \emph{DeDonder--Weyl theory} \cite{DeDonder1930,Weyl1935} (a generalization of the Hamiltonian formalism to field theory) and in the \emph{Batalin--Vilkovisky formalism} \cite{BV1,BV2,BV3} (a way of treating perturbative quantization of gauge theories).
\end{rem}

\begin{prop}
The graded vector space $\calC$ together with the Gerstenhaber bracket \eqref{eq:Gerstenhaber_bracket} is a GLA, which we call \emph{Hochschild GLA}.
\end{prop}

\begin{proof}
Clearly, the Gerstenhaber bracket is linear. Note that the sign $(-1)^{mn}$ ensures that it is also (graded) skew-symmetric, since 
\[
[\phi,\psi]_\mathrm{G}=-(-1)^{mn}\bigg(\psi\circ\phi-(-1)^{mn}\phi\circ\psi\bigg)=-(-1)^{mn}[\psi,\phi]_\mathrm{G},\quad \forall \phi\in \calC^m,\psi\in\calC^n.
\]
Next, we need to make sure that $[\enspace,\enspace]_\mathrm{G}$ satisfies the (graded) Jacobi identity
\begin{equation}
\label{eq:Gerstenhaber_Jacobi}
[\phi,[\psi,\chi]_\mathrm{G}]_\mathrm{G}=[[\phi,\psi]_\mathrm{G}]_\mathrm{G}+(-1)^{mn}[\psi,[\phi,\chi]_\mathrm{G}]_\mathrm{G},\quad\forall \phi\in\calC^m,\psi\in\calC^m,\chi\in\calC^p.
\end{equation}
Note that the left hand side of \eqref{eq:Gerstenhaber_Jacobi} gives
\begin{multline}
\label{eq:expansion}
    \bigg(\phi\circ\psi-(-1)^{mn}\psi\circ\phi\bigg)\circ\chi-(-1)^{(m+n)p}\chi\circ\bigg(\phi\circ\psi-(-1)^{mn}\psi\circ\phi\bigg)\\
    =\sum_{0\leq i\leq m\atop 0\leq k\leq m+n}(-1)^{in+kp}(\phi\circ_i\psi)\circ_k\chi-\sum_{0\leq j\leq n\atop 0\leq k\leq m+n}(-1)^{m(j+n)+kp}(\psi\circ_j\phi)\circ_k\chi\\
    -\sum_{0\leq i\leq m\atop 0\leq k\leq p}(-1)^{(m+n)(k+p)+in}\chi\circ_k(\phi\circ_i\psi)+\sum_{0\leq j\leq n\atop 0\leq k\leq p}(-1)^{(m+n)(k+p)+m(j+n)}\chi\circ_k(\psi\circ_j\phi).
\end{multline}
We can decompose the first sum according to the following rule 
\[
(\phi\circ_i\psi)\circ_k\chi=\begin{cases}(\phi\circ_k\chi)\circ_i\psi,&k<i\\ \phi\circ_i(\psi\circ_{k-i}\chi),&i\leq k\leq i+n\\ (\phi\circ_{k-n}\chi)\circ_i\psi,&i+n<k\end{cases}
\]
into a term of the form 
\begin{equation}
\label{eq:first_sum}
\sum_{0\leq i\leq m\atop i\leq k\leq i+n}(-1)^{in+kp}\phi\circ_i(\psi\circ_{k-i}\chi)=\sum_{0\leq i\leq m\atop 0\leq k\leq n}(-1)^{(n+p)i+kp}\phi\circ_i(\psi\circ_k\chi).
\end{equation}
Note that the sign of \eqref{eq:first_sum} is the same as the one in the corresponding term coming from $(\phi\circ\psi)\circ\chi$ on the left hand side of \eqref{eq:expansion} plus the terms in which the $i$-th and $k$-th composition commute. These then cancel the corresponding terms coming from the expansion of the second term of the right hand side of \eqref{eq:Gerstenhaber_Jacobi}.
Using the same approach for the remaining terms, we get the proof.
\end{proof}

\begin{rem}
Note that associative multiplications on $\calA$ are elements in $\calC^1$. In particular, if we denote by $\mathsf{m}\in\calC^1$ a multiplication map, the associativity condition is given by 
\[
\mathsf{m}(\mathsf{m}(f\otimes g)\otimes h)-\mathsf{m}(f\otimes \mathsf{m}(g\otimes h))=0.
\]
This is in fact equivalent to the condition that the Gerstenhaber bracket of $\mathsf{m}$ with itself vanishes. Indeed, we have 
\begin{multline*}
[\mathsf{m},\mathsf{m}]_\mathrm{G}(f\otimes g\otimes h)=\sum_{0\leq i\leq 1}(-1)^i(\mathsf{m}\circ_i\mathsf{m})(f\otimes g\otimes h)-(-1)^1\sum_{0\leq i\leq 1}(-1)^i(\mathsf{m}\circ_i\mathsf{m})(f\otimes g\otimes h)\\
=2\bigg(\mathsf{m}(\mathsf{m}(f\otimes g)\otimes h)-\mathsf{m}(f\otimes\mathsf{m}(g\otimes h))\bigg).
\end{multline*}
\end{rem}

Now, for an element $\phi$ of a (DG) Lie algebra $\mathfrak{g}$ (of degree $\ell$), we get that 
\[
\mathrm{ad}_\phi:=[\phi,\enspace]
\]
is a derivation (of degree $\ell$). Indeed, by the Jacobi identity, we have 
\[
\mathrm{ad}_\phi[\psi,\xi]=[\mathrm{ad}_\phi\psi,\xi]+(-1)^{\ell m}[\psi,\mathrm{ad}_\phi\xi],\quad \forall \psi\in\mathfrak{g}^m,\xi\in \mathfrak{g}^n.
\]

\begin{defn}[Hochschild differential]
The \emph{Hochschild differential} is given by 
\begin{align}
    \begin{split}
        \label{eq:Hochschild_differential}
        \dd_\mathsf{m}\colon \calC^i&\to \calC^{i+1}\\
        \psi&\mapsto \dd_\mathsf{m}\psi:=[\mathsf{m},\psi]_\mathrm{G}.
    \end{split}
\end{align}
\end{defn}

\begin{prop}
The GLA $\calC$ endowed with the Hochschild differential $\dd_\mathsf{m}$ is a DGLA.
\end{prop}

\begin{proof}
The only thing we need to check is that the Hochschild differential indeed squares to zero (i.e. $\dd_\mathsf{m}\circ \dd_\mathsf{m}=0$). This follows immediately from the Jacobi identity and the associativity condition of $\mathsf{m}$ expressed in terms of the Gerstenhaber bracket:
\begin{multline*}
(\dd_\mathsf{m}\circ\dd_\mathsf{m})\psi=[\mathsf{m},[\mathsf{m},\psi]_\mathrm{G}]_\mathrm{G}=[[\mathsf{m},\mathsf{m}]_\mathrm{G},\psi]_\mathrm{G}-[\mathsf{m},[\mathsf{m},\psi]_\mathrm{G}]_\mathrm{G}\\
=-[\mathsf{m},[\mathsf{m},\psi]_\mathrm{G}]_\mathrm{G}\Leftrightarrow \dd_\mathsf{m}\circ \dd_\mathsf{m}=0.
\end{multline*}
\end{proof}

\begin{rem}
In fact, we explicitly have
\begin{multline*}
(\dd_\mathsf{m}\psi)(f_0\otimes\dotsm \otimes f_{n+1})=\sum_{0\leq i\leq n}(-1)^{i+1}\psi(f_0\otimes \dotsm\otimes f_{i-1}\otimes\mathsf{m}(f_i\otimes f_{i+1})\otimes\dotsm \otimes f_{n+1})\\
+\mathsf{m}(f_0\otimes \psi(f_1\otimes\dotsm\otimes f_{n+1}))+(-1)^{n+1}\mathsf{m}(\psi(f_0\otimes\dotsm\otimes f_n)\otimes f_{n+1}).
\end{multline*}
\end{rem}

Let now $\calA=C^\infty(M)$ and consider the and consider the subalgebra $\widetilde{\calD}(M)\subset \calC$, which is the (graded) vector space given as the collection 
\[
\widetilde{\calD}^i:= \bigoplus_i\widetilde{\calD}^i
\]
of subspaces $\widetilde{\calD}^i(M)\subset \calC^i$ which consist of differential operators on $C^\infty(M)$. One can check that $\widetilde{\calD}(M)$ is closed with respect to the Gerstenhaber bracket $[\enspace,\enspace]_\mathrm{G}$ and the differential $\dd_\mathsf{m}$ and hence is a DGL subalgebra\footnote{Note that we will also allow differential operators of degree zero, i.e. operators which do not ``derive'' anything. Thus, the associative product $\mathsf{m}$ is also an element of $\widetilde{\calD}^1(M)$.} of $\calC$.\\

Moreover, we want to restrict everything to differential operators which vanish on constant functions. They will in fact form another DGL subalgebra $\calD(M)\subset \widetilde{\calD}(M)$. Note that we want to have this restriction because of the fact that the coefficients in the expansion of a star product vanish on constant functions, i.e. $B_i(1,f)=0$ for all $i\in\N$ and $f\in C^\infty(M)$. Note however that the Hochschild differential $\dd_\mathsf{m}$ is not inner derivation anymore for $\widetilde{\calD}(M)$ since the multiplication does not vanish on constants.\\

Consider an element $B\in \calD^1(M)$. Then we want to regard $\mathsf{m}+B$ as a deformation of the original product. The associativity condition for $\mathsf{m}+B$ implies
\[
[\mathsf{m}+B,\mathsf{m}+B]_\mathrm{G}=0,
\]
and by associativity of $\mathsf{m}$ we get 
\[
[\mathsf{m},B]_\mathrm{G}=[B,\mathsf{m}]_\mathrm{G}=\dd_\mathsf{m}B.
\]
This does exactly lead to the Maurer--Cartan equation:
\begin{equation}
    \label{eq:MC_operator}
    \dd_\mathsf{m}B+\frac{1}{2}[B,B]_\mathrm{G}=0.
\end{equation}

\begin{rem}
If we consider the \emph{formal} version of $\calD(M)$, we can see that the deformed product does exactly describe a star product since $B\in\hbar \calD^1(M)[\![\hbar]\!]$. Moreover, the gauge group is given by formal differential operators and the action on the star product is given by \eqref{eq:equivalent} since the adjoint action, by the definition of the Gerstenhaber bracket, is exactly the composition of $D_i$ with $B_j$.
\end{rem} 

\subsection{The Hochschild--Kostant--Rosenberg map}

In \cite{HochschildKostantRosenberg1962}, Hochschild--Kostant--Rosenberg have shown that the cohomology of the algebra of multidifferential operators and the algebra of multivector fields, which is equal to its cohomology, are isomorphic. They established an isomorphism 
\begin{equation}
\label{eq:HKR}
\mathrm{HKR}\colon H(\widetilde{\calD}(M))\xrightarrow{\sim}\widetilde{\calV}(M)\cong H(\widetilde{\calV}(M)).
\end{equation}
In particular, this isomorphism is induced by the map 
\begin{align}
    \begin{split}
        \label{eq:first_term}
        U^{(0)}_1\colon \widetilde{\calV}(M)&\to \widetilde{\calD}(M),\\
        \xi_0\land\dotsm \land \xi_n&\mapsto\bigg(f_0\otimes\dotsm \otimes f_n\mapsto \frac{1}{(n+1)!}\sum_{\sigma\in S_{n+1}}\mathrm{sign}(\sigma)\xi_{\sigma(0)}(f_0)\dotsm \xi_{\sigma(n)}(f_n)\bigg)
    \end{split}
\end{align}

\begin{rem}
This is extended to vector fields of order zero by the identity map. However, the map \eqref{eq:first_term} fails to preserve the Lie structure. This can be easily checked for order two. Indeed, given homogeneous bivector fields $\chi_1\land \chi_2$ and $\xi_1\land \xi_2$, we get
\[
U^{(0)}_1\bigg([\chi_1\land\chi_2,\xi_1\land\xi_2]_{\mathrm{SN}}\bigg)\not=\left[U^{(0)}_1(\chi_1\land\chi_2),U^{(0)}_1(\xi_1\land\xi_2)\right]_\mathrm{G}.
\]
The left hand side gives 
\begin{multline*}
    U^{(0)}_1\bigg([\chi_1,\xi_1]\land \chi_2\land\xi_2-[\chi_1,\xi_2]\land\chi_2\land \xi_1-[\chi_2,\xi_1]\land \chi_1\land\xi_2+[\chi_2,\xi_2]\land\chi_1\land \xi_1\bigg)(f\otimes g\otimes h)\\
    =\frac{1}{6}\bigg(\chi_1\big(\xi_1(f)\chi_2(g)\xi_2(h)\big)-\xi_1\big(\chi_1(f)\chi_2(g)\xi_2(h)\big)\\
    -\chi_1\big(\xi_2(f)\chi_2(g)\xi_1(h)\big)+\xi_2\big(\chi_1(f)\chi_2(g)\xi_1(h)\big)-\chi_2\big(\xi_1(f)\chi_1(g)\xi_2(h)\big)+\xi_1\big(\chi_2(f)\chi_1(g)\xi_2(h)\big)\\
    +\chi_2\big(\xi_2(f)\chi_1(g)\xi_1(h)\big)+\xi_2\big(\chi_2(f)\chi_1(g)\xi_1(h)\big)\bigg)+\mathrm{permutations}.
\end{multline*}
The right hand side is given by 
\[
\frac{1}{4}\bigg(\chi_1\big(\xi_1(f)\xi_2(g)\big)\chi_2(h)+\dotsm\bigg).
\]
One can, however, still show that the difference of the two terms is given by the image of a closed term in the cohomology of $\calD(M)$. Thus, there is some hope to control everything and still extend it to a Lie algebra morphism whose first order approximation is given by \eqref{eq:first_term}. Again, to resolve this problem, we will consider an $L_\infty$-morphism $U$. 
\end{rem}

\subsection{The dual point of view}
\label{subsec:dual_view}
Let $V$ be a vector space and note that we can naturally identify polynomials on $V$ as the \emph{symmetric functions} on the dual space $V^*$
\[
f(v):=\sum_j\frac{1}{j!}f_j(v\otimes\dotsm \otimes v),\quad \forall v\in V,
\]
where $f_j\in \Sym^j(V^*)$. We want to extend this to the case when $V$ is a graded vector space. In this case, we need to consider the exterior algebra instead. Note that, if we consider the \emph{injective limit} of the $\Sym^j(V^*)$ (resp. $\bigwedge^jV^*$) endowed with the induced topology, we can define the corresponding completion $\widehat{\Sym}(V^*)$ (resp. $\widehat{\bigwedge V^*}$). Now we can define a function in a formal neighborhood of $0$ to be given by the formal Taylor expansion in the $\hbar$
\[
f(\hbar v):=\sum_j \frac{\hbar^j}{j!}f_j(v\otimes\dotsm \otimes v),\quad \forall v\in V.
\]
Thus, a vector field $X$ on $V$ can be regarded as a derivation on $\widehat{\bigwedge V^*}$ and the Leibniz rule makes sure that $X$ is completely determined by its restriction on $V^*$. Similarly, if we have an algebra morphism
\[
\varphi\colon \widehat{\bigwedge W^*}\to\widehat{\bigwedge V^*}, 
\]
it induces a map 
\[
f:=\varphi^*\colon \widehat{\bigwedge V}\to \widehat{\bigwedge W}
\]
whose components $f_j$ are completely defined by their projection on $W$ as the $\varphi_j$ are defined by their restriction on $W^*$.

\begin{defn}[(Formal) pointed map]
A \emph{(formal) pointed map} is an algebra morphism between the reduced exterior algebras
\[
\varphi\colon \widehat{\overline{\bigwedge W^*}}\to \widehat{\overline{\bigwedge V^*}},
\]
where the overline means that we are considering the \emph{reduced tensor algebra} $\bar T(W^*)$ (resp. $\bar T(V^*)$) as in Example \ref{ex:reduced_tensor_algebra}.
\end{defn}

\begin{defn}[Pointed vector field]
A \emph{pointed vector field} $X$ on a manifold $M$ is a vector field in $\mathfrak{X}(M)$ which fixes the point zero, i.e.
\[
X(f)(0)=0,\quad \forall f\in C^\infty(M).
\]
\end{defn}

\begin{defn}[Cohomological pointed vector field]
A pointed vector field $X$ is called \emph{cohomological} (or \emph{$Q$-field}) if and only if it commutes with itself, i.e. $X^2=\frac{1}{2}[X,X]=0$.
\end{defn}

\begin{defn}[Pointed $Q$-manifold]
A \emph{pointed $Q$-manifold} is a (formal) pointed manifold together with a cohomological vector field.
\end{defn}

Consider now a Lie algebra $\mathfrak{g}$. Note that the bracket $[\enspace,\enspace]\colon \bigwedge^2\mathfrak{g}\to \mathfrak{g}$ induces a linear map 
\[
[\enspace,\enspace]^*\colon\mathfrak{g}^*\to \bigwedge^2\mathfrak{g}^*.
\]
This can be extended to the whole exterior algebra by a map
\[
\delta\colon \bigwedge\mathfrak{g}^*\to \bigwedge\mathfrak{g}^*[1]
\]
such that it satisfies the Leibniz rule and $\delta\big|_{\mathfrak{g}^*}=[\enspace,\enspace]^*$. 

\begin{rem}
We can regard the exterior algebra as an \emph{odd} analogue of a manifold on which $\delta$ takes the role of a (pointed) vector field. Note also that, since $[\enspace,\enspace]$ satisfies the Jacobi identity, we have $\delta^2=0$ and thus $\delta$ defines a cohomological (pointed) vector field.
\end{rem}

Consider two Lie algebras $\mathfrak{g}$ and $\mathfrak{h}$ together with the corresponding cohomological vector fields $\delta_\mathfrak{g}$ and $\delta_\mathfrak{h}$ on the respective exterior algebras. Then a Lie algebra morphism $\varphi\colon \mathfrak{g}\to \mathfrak{h}$ will correspond to a chain map $\varphi^*\colon \mathfrak{h}^*\to \mathfrak{g}^*$ since 
\begin{equation}
\label{eq:correspondence_Q-manifold}
\varphi\bigg([\enspace,\enspace]_\mathfrak{g}\bigg)=[\varphi(\enspace),\varphi(\enspace)]_\mathfrak{h}\Longleftrightarrow \delta_\mathfrak{g}\circ \varphi^*=\varphi^*\circ\delta_\mathfrak{h}.
\end{equation}

\begin{rem}
Equation \eqref{eq:correspondence_Q-manifold} gives a first hint to the correspondence between $L_\infty$-algebras and pointed $Q$-manifolds. Since a Lie algebra is a particular case of a DGLA, which again can be endowed with an $L_\infty$-structure, we can see that the map $\varphi$ satisfies exactly the same condition for the first component of an $L_\infty$-morphism as in \eqref{eq:L_infty_morphism} for $n=1$.
\end{rem}

Let us extend this picture to the case of a graded vector space $Z$ over a field $k$. Let us decompose the graded space $Z$ into its \emph{even} and \emph{odd} part, i.e. 
\[
Z=Z_{[0]}\oplus Z_{[1]}
\]
where $Z_{[0]}$ denotes the even part, i.e. $Z_{[0]}:=\bigoplus_{i\in 2\mathbb{Z}}Z^i$ and $Z_{[1]}:=\bigoplus_{i\in 2\mathbb{Z}+1}Z^i$. Here we have denoted by $[0],[1]$ the equivalence classes in $\mathbb{Z}/2$. For a vector space $W$, denote by $\Pi W$ the (odd) space defined by a \emph{parity reversal} on $W$, which we can also denote by $W[1]$ using the notation of Definition \ref{defn:shift}.
Let us define $V:=Z_{[0]}$ and $\Pi W:=Z_{[1]}$. Then $Z=V\oplus \Pi W$ and functions can be identified by elements of 
\[
\Sym(Z^*):=\Sym(V^*)\otimes \bigwedge W^*.
\]
We can express the condition for a vector field $\delta\colon \Sym(Z^*)\to \Sym(Z^*)[1]$ to be \emph{cohomological} in terms of its coefficients
\[
\delta_j\colon \Sym^j(Z^*)\to \Sym^{j+1}(Z^*)
\]
by expanding the equation $\delta^2=0$. Hence, we get an infinite family of equations

\begin{align}
\label{eq:equation_1}
\delta_0\delta_0&=0,\\
\label{eq:equation_2}
\delta_1\delta_0+\delta_0\delta_1&=0,\\
\label{eq:equation_3}
\delta_2\delta_0+\delta_1\delta_1+\delta_0\delta_2&=0,\\
&\vdots\nonumber
\end{align}

Let now $m_j:=(\delta_j\vert_{Z^*})^*$ be the dual coefficients and let $\langle\enspace,\enspace\rangle\colon Z^*\otimes Z\to k$. Then we can rewrite these conditions in terms of $m_j\colon \Sym^j(Z)\to Z$. Note that the first Equation \eqref{eq:equation_1}, which is now given by $m_0m_0=0$, implies that $m_0$ is a differential on $Z$ and hence induces a cohomology $H_{m_0}(Z)$. For $j=1$, i.e. for the second Equation \eqref{eq:equation_2}, we get 
\begin{equation}
\label{eq:first_term_equation_2}
\langle\delta_1\delta_0 f,x\otimes y\rangle=\langle\delta_0f,m_1(x\otimes y)\rangle=\langle f,m_0(m_1(x\otimes y))\rangle 
\end{equation}
and 
\begin{multline}
\label{eq:second_term_equation_2}
   \langle\delta_0\delta_1 f,x\otimes y\rangle=\langle\delta_1f,m_0(x)\otimes y\rangle+(-1)^{\vert x\vert}\langle \delta_1f,x\otimes m_0(y)\rangle\\
   =\langle f,m_1(m_0(x)\otimes y)\rangle+(-1)^{\vert x\vert}\langle f,m_1(x\otimes m_0(y))\rangle.
\end{multline}
If we put \eqref{eq:first_term_equation_2} and \eqref{eq:second_term_equation_2} together, we get that $m_0$ is a derivation with respect to the multiplication defined by $m_1$.

\begin{rem}
If we consider $Z:=\mathfrak{g}[1]$ and consider the exterior algebra by the d\'eclage isomorphism ($\Sym^n(\mathfrak{g}[1])\xrightarrow{\sim}\bigwedge^n\mathfrak{g}[n]$), we can interpret $m_1$ as a bilinear skew-symmetric operator on $\mathfrak{g}$.
\end{rem}

Note that Equation \eqref{eq:equation_3} implies that $m_1$ is indeed a Lie bracket for which the Jacobi identity holds up to terms containing $m_0$. Since $m_0$ is a differential, this means that $m_1$ is a Lie bracket up to \emph{homotopy}. Hence, considering all equations, we will get a strong homotopy Lie algebra structure on $\mathfrak{g}$.

\begin{rem}
Note that this will lead to a one-to-one correspondence between pointed $Q$-manifolds and SHLAs. Since SHLAs are equivalent to $L_\infty$-algebras, we get a one-to-one correspondence between pointed $Q$-manifolds and $L_\infty$-algebras.
\end{rem}

\begin{defn}[$Q$-map]
A \emph{$Q$-map} is a (formal) pointed map between two $Q$-manifolds $Z$ and $\widetilde{Z}$ which commutes with the $Q$-fields, i.e. a map
\begin{align}
\label{eq:Q-map}
    \begin{split}
        \varphi\colon\overline{\Sym}(\widetilde{Z}^*)&\to \overline{\Sym}(Z^*)\\
       \text{such that}\qquad \varphi\circ \widetilde{\delta}&=\delta\circ \varphi.
    \end{split}
\end{align}
\end{defn}

We want to express the condition \eqref{eq:L_infty_morphism} more explicitly in this setting by using $Q$-maps. Consider the vector field $\delta$ and its restriction to $\widetilde{Z}$. We define the coefficients of the dual map as
\[
U_j:=\left(\varphi_j\vert_{Z^*}\right)^*\colon \Sym^j(Z)\to \widetilde{Z}.
\]
Then we can express the condition \eqref{eq:Q-map} of a $Q$-map on the dual coefficients. The first equation gives
\begin{multline*}
\langle\varphi(\widetilde{\delta} f),x\rangle=\langle\delta \varphi(f),x\rangle\Rightarrow \langle\widetilde{\delta}_0f, U_1(x)\rangle=\langle \varphi(f),m_0(x)\rangle\\
\Rightarrow \langle f,\widetilde{m}_0(U_1(x))\rangle=\langle f,U_1(m_0(x))\rangle.
\end{multline*}

This tells us that the first coefficient $U_1$ is a chain map with respect to the differential defined by the first coefficient of the $Q$-structure:
\[
H(U_1)\colon H_{m_0}(Z)\to H_{\widetilde{m}_0}(\widetilde{Z}).
\]

\begin{rem}
Similarly, we can consider the equation for the next coefficient:
\begin{multline}
    \widetilde{m}_1(U_1(x)\otimes U_1(y))+\widetilde{m}_1(U_2(x\otimes y))\\
    =U_2(m_0(x)\otimes y)+(-1)^{\vert x\vert}U_2(x\otimes m_0(y))+U_1(m_1(x\otimes y)).
\end{multline}
This shows that $U_1$ preserves the Lie structure induced by $m_1$ and $\widetilde{m}_1$ up to terms containing $m_0$ and $\widetilde{m}_0$, i.e. up to \emph{homotopy}. Recall that this solves the problem that the map $U^{(0)}_1$, defined in \eqref{eq:first_term}, is a chain map but not a DGLA morphism. Namely, a $Q$-map $U$ (or equivalently an $L_\infty$-morphism) induces a map $U_1$ which has the same property as $U^{(0)}_1$.
\end{rem}

\begin{defn}[Koszul sign]
Let $V$ be a vector space endowed with a graded commutative product $\otimes$. 
The \emph{Koszul sign} $\varepsilon(\sigma)$ of a permutation $\sigma$ is the sign defined by 
\[
x_1\otimes\dotsm\otimes x_n=\varepsilon(\sigma)x_{\sigma(1)}\otimes\dotsm\otimes x_{\sigma(n)},\quad x_i\in V.
\]
\end{defn}

\begin{defn}[Shuffle permutation]
An \emph{$(\ell,n-\ell)$-shuffle permutation} is a permutation $\sigma$ of $(1,\ldots,n)$ such that $\sigma(1)<\dotsm<\sigma(\ell)$ and $\sigma(\ell+1)<\dotsm<\sigma(n)$. 
\end{defn}

\begin{rem}
The shuffle permutation associated to a partition $I_1\sqcup\dotsm\sqcup I_j=\{1,\ldots,n\}$ is the permutation that takes at first all the elements index by the subset $I_1$ in the given order, then those of $I_2$ and so on.
\end{rem}

The $n$-th coefficient of $U$ satisfies the following condition:
\begin{multline}
\label{eq:n-th_coefficients}
    \widetilde{m}_0(U_n(x_1\otimes\dotsm\otimes x_n))+\frac{1}{2}\sum_{I\sqcup J=\{1,\ldots,n\}\atop I,J\not=\varnothing}\,\, \varepsilon_x(I,J)\widetilde{m}_1(U_{\vert I\vert}(x_I)\otimes U_{\vert J\vert}(x_J))\\
    =\sum_{1\leq j\leq n}\varepsilon^j_x U_n(m_0(x_j)\otimes x_1\otimes\dotsm\otimes \widehat{x}_j\otimes\dotsm\otimes x_n)+\\+\frac{1}{2}\sum_{j\not=\ell}\varepsilon^{j\ell}_xU_{n-1}(m_1(x_j\otimes x_\ell)\otimes x_1\otimes\dotsm\otimes \widehat{x}_j\otimes\dotsm\otimes \widehat{x}_\ell\otimes\dotsm\otimes x_n),
\end{multline}
where $\varepsilon_x(I,J)$ is the \emph{Koszul sign} associated to the $(\vert I\vert,\vert J\vert)$-shuffle permutation for the partition $I\sqcup J=\{1,\ldots,n\}$ and $\varepsilon^j_x$ (resp. $\varepsilon^{j\ell}_x$) for the particular case $I=\{j\}$ (resp. $I=\{j,\ell\}$). Moreover, we have set $x_I:=\bigotimes_{i\in I}x_i$.

\begin{rem}
We will specialize \eqref{eq:n-th_coefficients} to an $L_\infty$-morphism where $Z$ is chosen to be the DGLA $\calV$ of multivector fields and $\widetilde{Z}$ the DGLA $\calD$ of multidifferential operators. 
\end{rem}

\subsection{Formality of $\calD$ and classification of star products on $\R^d$}

We want to analyze the relation between the formality of $\calD$ and the classification problem of all star products on $\R^d$.
Recall that the \emph{associativity} of the star product and the \emph{Jacobi identity} for a bivector field are equivalent to certain Maurer--Cartan equations.

\begin{defn}[Maurer--Cartan equation on (formal) $L_\infty$-algebras]
The Maurer--Cartan equation on a (formal) $L_\infty$-algebra $(\mathfrak{g}[\![\hbar]\!],Q)$ is given by 
\begin{equation}
\label{eq:MC_L_infty}
Q(\exp(\hbar x))=0,\qquad x\in\mathfrak{g}^1[\![\hbar]\!],
\end{equation}
where the exponential $\exp$ maps an element of degree $+1$ to a formal power series in $\hbar \mathfrak{g}[\![\hbar]\!]$.
\end{defn}

\begin{rem}
Note that, from the dual point of view, condition \eqref{eq:MC_L_infty} tells us that $x$ is a fixed point of the cohomological vector field $\delta$. This means that for each $f\in \Sym(\mathfrak{g}^*[\![\hbar]\!][1])$ we have 
\begin{equation}
\label{eq:MC_2}
\delta f(\hbar x)=0.
\end{equation}
Moreover, using that $(\delta f)_j=\delta_{j-1}f$, we can expand Equation \eqref{eq:MC_2} into a formal Taylor series. In fact, using the pairing
\[
\langle \delta_{j-1}f,x\otimes\dotsm\otimes x\rangle=\langle f,m_{j-1}(x\otimes\dotsm \otimes x)\rangle,
\]
we can write the Maurer--Cartan equation \eqref{eq:MC_L_infty} as
\begin{equation}
    \label{eq:dual_MC}
    \sum_{j\geq 1}\frac{\hbar^j}{j!}m_{j-1}(x\otimes\dotsm\otimes x)=\hbar m_0(x)+\frac{\hbar^2}{2}m_1(x\otimes x)+O(\hbar^3)=0.
\end{equation}
\end{rem}

For two DGLAs $\mathfrak{g}$ and $\mathfrak{h}$, an $L_\infty$-morphism $\varphi\colon\Sym(\mathfrak{h}^*[\![\hbar]\!][1])\to \Sym(\mathfrak{g}^*[\![\hbar]\!][1])$ preserves the Maurer--Cartan equation \eqref{eq:dual_MC}, similarly as a morphism of DGLAs preserves the Maurer--Cartan equation \eqref{eq:Maurer-Cartan}. In particular, if $x$ is a solution of \eqref{eq:dual_MC} on $\mathfrak{g}[\![\hbar]\!]$, we get that 
\[
U(\hbar x)=\sum_{j\geq 1}\frac{\hbar^j}{j!}U_j(x\otimes\dotsm\otimes x)
\]
is also a solution of \eqref{eq:dual_MC} on $\mathfrak{h}[\![\hbar]\!]$.

\begin{rem}
The action of the gauge group on solutions to the Maurer-Cartan equation \eqref{eq:Maurer-Cartan} can be generalized to the case of $L_\infty$-algebras. Namely, if $x$ and $x'$ are equivalent modulo this generalized action, then their images under $U$ are still equivalent solutions.
\end{rem}

\begin{rem}
Again, we are especially interested in the case where $\mathfrak{g}=\calV$ and $\mathfrak{h}=\calD$. Thus, we get an $L_\infty$-morphism $U$ which gives as a formula of how to construct an associative star product out of any (formal) Poisson bivector $\pi$, given by 
\begin{equation}
\label{eq:star_product_L_infty}
U(\pi)=\sum_{j\geq 0}\frac{\hbar^j}{j!}U_j(\pi\land\dotsm \land \pi),
\end{equation}
where we insert the coefficient of degree zero to be the original non-deformed product. Note that if $U$ is a quasi-isomorphism, the correspondence between (formal) Poisson structures on a manifold $M$ and formal deformations of the pointwise product on $C^\infty(M)$ are one-to-one. In particular, as soon as we have a formality map, we have solved the problem of \emph{existence} and \emph{classification} of star products on $M$.
\end{rem}

\section{Kontsevich's star product}

We want to construct an explicit expression of the formality map $U$. The idea is to introduce graphs and to rewrite things in those terms. An easy example is given by the Moyal product $f\star g$ for two smooth functions $f$ and $g$ on $(\R^{2n},\omega_0)$. We can write $f\star g$ as the sum of graphs in the following way:

\begin{figure}[ht]
\centering
\begin{tikzpicture}[scale=0.5]
\tikzset{Bullet/.style={fill=black,draw,color=#1,circle,minimum size=0.5pt,scale=0.5}}
\node[Bullet=black, label=below: $f$] (f1) at (1,-2) {};
\node[Bullet=black, label=below: $g$] (g1) at (3,-2) {};
\node[Bullet=black, label=below: $f$] (f2) at (7,-2) {};
\node[Bullet=black, label=below: $g$] (g2) at (9,-2) {};
\node[Bullet=black, label=below: $f$] (f3) at (13,-2) {};
\node[Bullet=black, label=below: $g$] (g3) at (15,-2) {};
\node[Bullet=black, label=below: $f$] (f4) at (19,-2) {};
\node[Bullet=black, label=below: $g$] (g4) at (23,-2) {};
\node[Bullet=gray] (pi2) at (8,0.5) {};
\node[Bullet=gray] (pi31) at (12,0.5) {};
\node[Bullet=gray] (pi32) at (16,0.5) {};
\node[Bullet=gray] (pi41) at (18,0.5) {};
\node[Bullet=gray] (pi42) at (21,0.5) {};
\node[Bullet=gray] (pi43) at (24,0.5) {};
\draw[fermion] (pi2) -- (f2);
\draw[fermion] (pi2) -- (g2);
\draw[fermion] (pi31) -- (f3);
\draw[fermion] (pi31) -- (g3);
\draw[fermion] (pi32) -- (f3);
\draw[fermion] (pi32) -- (g3);
\draw[fermion] (pi41) -- (f4);
\draw[fermion] (pi41) -- (g4);
\draw[fermion] (pi42) -- (f4);
\draw[fermion] (pi42) -- (g4);
\draw[fermion] (pi43) -- (f4);
\draw[fermion] (pi43) -- (g4);
\node[] (plus1) at (5,-0.5) {+};
\node[] (plus2) at (11,-0.5) {+};
\node[] (plus3) at (17,-0.5) {+};
\node[] (plus4) at (26,-0.5) {$+\quad\ldots$};
\draw (0,-2) -- (4,-2);
\draw (6,-2) -- (10,-2);
\draw (12,-2) -- (16,-2);
\draw (18,-2) -- (24,-2);
\end{tikzpicture}
\caption{The Moyal product $f\star g$ represented in terms of graphs. The gray vertices represent the Poisson tensor $\pi$ induced by $\omega_0$. The term with $n$ Poisson vertices in the sum represents the $n$-th term in the formal power series of the Moyal product. For each gray vertex, the two outgoing arrows represent the derivatives $\de_i$ and $\de_j$ acting on $f$ and $g$.} 
\label{fig:Moyal_graphs}
\end{figure}
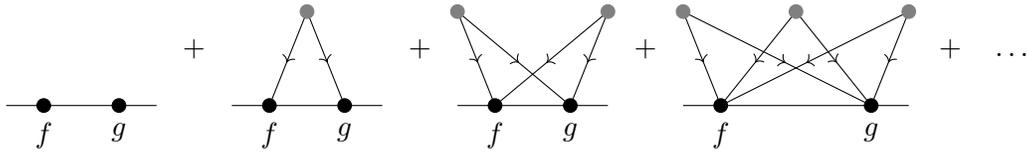

\begin{rem}
This construction can be generalized by allowing more than two outgoing arrows for gray vertices when considering general multivector fields and allowing incoming arrows for gray vertices (i.e. derivations of the tensor coefficients for the assigned multivector fields). Note that there are no incoming arrows for the Moyal product since the Poisson structure is constant. In Kontsevich's star product, we assign to each graph $\Gamma$ a multidifferential operator $B_\Gamma$ and a \emph{weight} $w_\Gamma$ such that the map $U$ that sends an $n$-tuple of multivector fields to the corresponding weighted sum over all possible graphs in this set of multidifferential operators is an $L_\infty$-morphism.
\end{rem}

\subsection{Data for the construction}

\subsubsection{Admissible graphs}

\begin{defn}[Admissible graphs]
The set $\calG_{n,\bar n}$ of \emph{admissible graphs} consists of all connected graphs $\Gamma$ which satisfy the following properties:
\begin{itemize}
    \item The set of vertices $V(\Gamma)$ is decomposed in two ordered subsets $V_1(\Gamma)$ and $V_2(\Gamma)$ isomorphic to $\{1,\ldots,n\}$ respectively $\{\bar 1,\ldots,\bar n\}$ whose elements are called vertices of \emph{first type} respectively \emph{second type};
    \item The following inequalities involving the number of vertices of the two types are satisfied: $n\geq 0,\bar n\geq 0$ and $2n+\bar n-2\geq 0$;
    \item The set of edges $E(\Gamma)$ is finite and does not contain \emph{short loops}, i.e. edges starting and ending at the same vertex;
    \item All edges $E(\Gamma)$ are oriented and start from a vertex of the first type;
    \item The set of edges starting at a given vertex $v\in V_1(\Gamma)$, which will be denoted by $\mathrm{Star}(v)$, is ordered.
\end{itemize}
\end{defn}

See Figure \ref{fig:admissible_graphs} and \ref{fig:non-admissible_graphs} for examples of admissible and non-admissible graphs respectively.

\begin{figure}[ht]
\centering
\begin{tikzpicture}[scale=0.7]
\tikzset{Bullet/.style={fill=black,draw,color=#1,circle,minimum size=0.5pt,scale=0.5}}
\node[Bullet=black, label=below: $\bar 1$] (bar11) at (1,-2) {};
\node[Bullet=black, label=below: $\bar 2$] (bar21) at (5,-2) {};
\node[Bullet=black, label=below: $\bar 1$] (bar12) at (10,-2) {};
\node[Bullet=black, label=below: $\bar 2$] (bar22) at (12,-2) {};
\node[Bullet=black, label=below: $\bar 3$] (bar32) at (14,-2) {};
\node[Bullet=gray, label=left: $1$] (v11) at (0,0) {};
\node[Bullet=gray, label=above: $2$] (v21) at (3,1) {};
\node[Bullet=gray, label=right: $3$] (v31) at (6,0) {};
\node[Bullet=gray, label=left: $1$] (v12) at (12,0) {};
\node[Bullet=gray, label=right: $2$] (v22) at (15,1) {};
\draw[fermion] (v11) -- (bar11);
\draw[fermion] (v11) -- (bar21);
\draw[fermion] (v21) -- (v11);
\draw[fermion] (v21) -- (v31);
\draw[fermion] (v31) -- (bar11);
\draw[fermion] (v31) -- (bar21);
\draw[fermion] (v12) -- (bar12);
\draw[fermion] (v12) -- (bar22);
\draw[fermion] (v12) -- (bar32);
\draw[fermion] (v22) -- (v12);
\draw[fermion] (v22) -- (bar32);
\draw (0,-2) -- (6,-2);
\draw (9,-2) -- (15,-2);
\end{tikzpicture}
\caption{Examples of admissible graphs.} 
\label{fig:admissible_graphs}
\end{figure}

\begin{figure}[ht]
\centering
\begin{tikzpicture}[scale=0.7]
\tikzset{Bullet/.style={fill=black,draw,color=#1,circle,minimum size=0.5pt,scale=0.5}}
\node[Bullet=black, label=below: $\bar 1$] (bar11) at (1,-2) {};
\node[Bullet=black, label=below: $\bar 2$] (bar21) at (5,-2) {};
\node[Bullet=black, label=below: $\bar 1$] (bar12) at (10,-2) {};
\node[Bullet=black, label=below: $\bar 2$] (bar22) at (14,-2) {};
\node[Bullet=gray, label=left: $1$] (v11) at (0,0) {};
\node[Bullet=gray, label=above: $2$] (v21) at (3,1) {};
\node[Bullet=gray, label=below: $1$] (v12) at (9,0) {};
\node[Bullet=gray, label=right: $2$] (v22) at (15,0) {};
\draw[fermion] (bar11) -- (v11);
\draw[fermion] (v11) -- (bar21);
\draw[fermion] (v11) -- (v21);
\draw[fermion] (v21) -- (bar21);
\draw[fermion] (v12) -- (bar12);
\draw[fermion] (v12) -- (bar22);
\draw[fermion] (v22) -- (bar12);
\draw[fermion] (v22) -- (bar22);
\path[->,every loop/.style={looseness=40}] (v12)
         edge  [in=160,out=60,loop] (v12); 
\draw (0,-2) -- (6,-2);
\draw (9,-2) -- (15,-2);
\end{tikzpicture}
\caption{Examples of non-admissible graphs.} 
\label{fig:non-admissible_graphs}
\end{figure}

\subsubsection{The multidifferential operators $B_\Gamma$}
Let us consider pairs $(\Gamma,\xi_1\otimes\dotsm\otimes \xi_n)$ consisting of a graph $\Gamma\in \calG_{n,\bar n}$ with $2n+\bar n-2$ edges and of a tensor product of $n$ multivector fields on $\R^d$. We want to understand how we can associate to such a pair a multidifferential operator $B_\Gamma\in \calD^{\bar n-1}$. This is done in the following way:
\begin{itemize}
    \item Associate to each vertex $v$ of first type with $k$ outgoing arrows the skew-symmetric tensor $\xi^{j_1\dotsm j_k}_i$ corresponding to a given $\xi_i$ via the natural identification.
    \item Place a function at each vertex of second type.
    \item Associate to the $\ell$-th arrow in $\mathrm{Star}(v)$ a partial derivative with respect to the coordinate labeled by the $\ell$-th index of $\xi_i$ acting on the function or the tensor appearing at its endpoint.
    \item Multiply such elements in the order prescribed by the labeling of the graph.
\end{itemize}

\begin{ex}
Denote by $\Gamma_1$ the first graph in Figure \ref{fig:admissible_graphs}. Then, to a triple of bivector fields $(\xi_1,\xi_2,\xi_3)$ with $\xi_1=\sum_{i_1<i_2}\xi_1^{i_1i_2}\de_{i_1}\land \de_{i_2}$, $\xi_2=\sum_{j_1<j_2}\xi_2^{j_1j_2}\de_{j_1}\land \de_{j_2}$ and $\xi_3=\sum_{\ell_1<\ell_2}\xi_3^{\ell_1\ell_2}\de_{\ell_1}\land \de_{\ell_2}$, we associate the bidifferential operator
\[
U_{\Gamma_1}(\xi_1\land\xi_2\land\xi_3)(f\otimes g):=\xi_2^{j_1j_2}\de_{j_2}\xi_1^{i_1i_2}\de_{j_2}\xi_3^{\ell_1\ell_2}\de_{i_1}\de_{\ell_1}f\de_{i_2}\de_{\ell_2}g.
\]
\end{ex}

\begin{ex}
Denote by $\Gamma_2$ the second graph in Figure \ref{fig:admissible_graphs}. Then, to a tuple of multivector fields $(\chi_1,\chi_2)$ with $\chi_1=\sum_{i_1<i_2}\chi_1^{i_1i_2}\de_{i_1}\land \de_{i_2}$ and $\chi_2=\sum_{j_1<j_2<j_3}\chi_2^{j_1j_2j_3}\de_{j_1}\land \de_{j_2}\land \de_{j_3}$, we associate the tridifferential operator
\[
U_{\Gamma_2}(\chi_1\land\chi_2)(f\otimes g\otimes h):=\chi_1^{i_1i_2}\de_{i_1}\chi_2^{j_1j_2j_3}\de_{j_1}f\de_{j_2}g\de_{j_3}\de_{i_2}h.
\]
\end{ex}

\begin{rem}
In particular, for each admissible graph $\Gamma\in\calG_{n,\bar n}$, we get a linear map $U_\Gamma\colon \bigwedge^n\calV\to \calD$ which is equivariant with respect to the action of the symmetric group. Kontsevich's main construction was to choose \emph{weights} $w_\Gamma\in\R$ such that the linear combination
\[
U:=\sum_{n,\bar n \geq 0}\,\,\sum_{\Gamma\in\calG_{n,\bar n}} w_\Gamma B_\Gamma
\]
becomes an $L_\infty$-morphism.
\end{rem}

\subsubsection{Weights of graphs}

The weight $w_\Gamma$ associated to an admissible graph $\Gamma\in\calG_{n,\bar n}$ is defined by an integral of a differential form $\omega_\Gamma$ over (a suitable compactification of) some \emph{configuration space} $C^+_{n,\bar n}$. Namely, it is given by \begin{equation}
    \label{eq:Kontsevich_weight}
    w_\Gamma:=\prod_{1\leq k\leq n}\frac{1}{\vert \mathrm{Star}(k)\vert!}\frac{1}{(2\pi)^{2n+\bar n-2}}\int_{\bar C^+_{n,\bar n}}\omega_\Gamma.
\end{equation}
This holds for the case when $\Gamma$ has exactly $2n+\bar n-2$ edges otherwise the weight is just zero.

Let us look at the integral \eqref{eq:Kontsevich_weight} more carefully. 

\subsubsection{Configuration spaces}
Consider an embedding of the graph $\Gamma$ into the $2$-dimensional upper half-space $\mathbb{H}^2$ (also called \emph{upper half-plane}) as defined in \eqref{eq:upper_half-space} where we want to identify $\R^2\cong\mathbb{C}$. Moreover, denote by $\mathbb{H}^2_{+}:=\{z\in \mathbb{H}\mid \Im(z)>0\}$ where $\Im(z)$ denotes the \emph{imaginary part} of $z$. 

\begin{defn}[Open configuration space]
Define the \emph{open configuration space} of the distinct $n+\bar n$ vertices of an admissible graph $\Gamma\in\calG_{n,\bar n}$ as the smooth manifold
\begin{equation}
    \label{eq:open_conf}
    \mathrm{Conf}_{n,\bar n}:=\big\{(z_1,\ldots,z_n,z_{\bar 1},\ldots,z_{\bar n})\in \mathbb{C}^{n+\bar n}\mid z_i\in \mathbb{H}^2_+, \, z_{\bar i}\in\R,\, z_i\not=z_j\,\, \text{for $i\not=j$},\, z_{\bar i}\not=z_{\bar j}\,\,\text{for $\bar i\not=\bar j$}\big\}
\end{equation}
\end{defn}

\begin{rem}
\label{rem:conf}
More precisely, we need to consider everything up to \emph{scaling} and \emph{translation}. Hence, we consider the Lie group $G$ consisting of translations in the horizontal direction and rescaling such that the action on $z\in\mathbb{H}^2$ is given by 
\[
z\mapsto az+b,\qquad a\in \R_+,b\in\R.
\]
The action of $G$ is free whenever the number of vertices is $2n+\bar n-2\geq 0$. Thus the quotient space of $\mathrm{Conf}_{n,\bar n}$ with respect to the $G$-action, which we will denote by $C_{n,\bar n}$, is again a smooth manifold of (real) dimension $2n+\bar n-2$. If there are no vertices of second type (i.e. $\bar n=0$), one can define the open configuration space as a subset of $\mathbb{C}^n$ instead of $\mathbb{H}^n$ and one can consider the Lie group which consists of rescaling and translation in any direction. The corresponding quotient space $C_n$ for $n\geq 2$ still a smooth manifold but now of dimension $2n-3$.
\end{rem} 

\begin{defn}[Connected configuration space]
Define the \emph{connected configuration space} to be the subset of $C_{n,\bar n}$ given by 
\begin{equation}
    \label{eq:connected_conf}
    C^+_{n,\bar n}:=\big\{(z_1,\ldots,z_n,z_{\bar 1},\ldots,z_{\bar n})\in C_{n,\bar n}\mid z_{\bar i}<z_{\bar j}\,\,\text{for $\bar i<\bar j$}\big\}
\end{equation}
\end{defn}

\begin{rem}
One can show that $C^+_{n,\bar n}$ is again a smooth manifold which is now connected.
\end{rem}

Let us define an \emph{angle map} 
\[
\phi\colon C_{2,0}\to S^1, 
\]
which associates to each pair of distinct points $z_1,z_2\in \mathbb{H}^2$ the angle between the geodesics with respect to the \emph{Poincar\'e metric} connecting $z_1$ to $\I\infty$ and to $z_2$, in the counterclockwise direction (see Figure \ref{fig:angle_map}). Formally, we have 
\begin{equation}
    \label{eq:angle_map}
    \phi(z_1,z_2):=\mathrm{arg}\left(\frac{z_2-z_1}{z_2-\bar z_1}\right)=\frac{1}{2\I}\log\left(\frac{(z_2-z_1)(\bar z_2-z_1)}{(z_2-\bar z_1)(\bar z_2-\bar z_1)}\right).
\end{equation}

\begin{figure}[ht] 
\centering
\begin{tikzpicture}[scale=0.8]
\node[vertex, label=right: $z_1$] (u) at (2,2.23){};
\node[vertex, label=above: $z_2$] (v) at (0,3){};
\node[label=above: $\mathbb{H}^2$] (H) at (-2.5,3){};
\node[label=right: $\I\infty$] (infty) at (2,4.5){};
\draw[very thick] (-4,0) -- (4,0);
\draw[dashed] (2,4) -- (2,5);
\draw[thick] (2,0) -- (2,4);
\draw[dashed] (-3,0) -- (3,0) arc(0:180:3);
\draw[thick, fermion] (u) arc (48.18:90:3);
\draw[fermion] (2,3) arc (90:170:0.5);
\node[] (phi) at (1.7,3.3){$\phi$};
\end{tikzpicture}
\caption{Illustration of the angle map $\phi$.}
\label{fig:angle_map}
\end{figure}
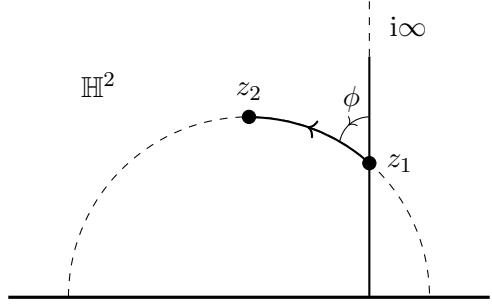

The differential of $\phi$ is then a 1-form on $C_{2,0}$ which we can pull-back to the configuration space including points in $\R$ by the projection $p_e$ associated to each edge $e=(z_i,z_j)$ of $\Gamma$ defined by 
\begin{align*}
    p_e\colon C_{n,\bar n}&\to C_{2,0},\\
    (z_1,\ldots,z_n,z_{\bar 1},\ldots,z_{\bar n})&\mapsto (z_i,z_j).
\end{align*}
Hence we get a 1-form $\dd\phi_e:=p_e^*\dd\phi\in\Omega^1(C_{n,\bar n})$. The differential form $\omega_\Gamma$ in \eqref{eq:Kontsevich_weight} is then defined as 
\[
\omega_\Gamma:=\bigwedge_{e\in E(\Gamma)}\dd\phi_e.
\]

\begin{rem}
Note that the ordering one these 1-forms is induced by the ordering on the set of edges by the first-type vertices and the ordering on $\mathrm{Star}(v)$. Moreover, note that we indeed obtain a top-form  on $C_{n,\bar n}$ (resp. $C_{n,\bar n}^+$) as long as we consider graphs with exactly $2n+\bar n-2$ edges since $\dim C_{n,\bar n}=\dim C^+_{n,\bar n}=2n+\bar n-2$.  
\end{rem}

\begin{rem}
In order to make sense of the integral \eqref{eq:Kontsevich_weight}, we need to make sure that it converges. Note that, by construction of $\phi$, we get that $\dd\phi$ is not defined as soon two points approach each other. This is the reason why we in fact need a suitable \emph{compactification} $\bar C^+_{n,\bar n}$ of the connected configuration space $C^+_{n,\bar n}$.
\end{rem}

\begin{thm}[\cite{FulMacPh,AS2}]
For any configuration space $C_{n,\bar n}$ (resp. $C_n$) there exists a compactification $\bar C_{n,\bar n}$ (resp. $\bar C_n$) whose interior is the open configuration space and such that the projections $p_e$, the angle map $\phi$ (and thus the differential form $\omega_\Gamma$) extend smoothly to the corresponding compactifications.
\end{thm}

\begin{rem}
This compactification is usually called \emph{FMAS compactification} for \emph{Fulton--MacPherson} who proved a first result of this in the algebro-geometrical setting for configuration spaces of points in non-singular algebraic varieties \cite{FulMacPh} and \emph{Axelrod--Singer} who proved this in the smooth-geometrical setting for configuration spaces of points in smooth manifolds \cite{AS2}. 
\end{rem}

\begin{rem}
The compactified configuration space is a compact smooth manifold with \emph{corners}. 
\end{rem}

\begin{defn}[Manifold with corners]
A manifold with corners of dimension $m$ is a topological Hausdorff space $M$ which is locally homeomorphic to $\R^{m-n}\times \R^n_+$ with $n=0,\ldots, m$. 
\end{defn}

\begin{rem}
The points $x\in M$ of a manifold with corners of dimension $m$ whose local expression in some (and hence in any) chart has the form 
\[
(x_1,\ldots,x_{m-n},0,\ldots,0)
\]
are said to be of \emph{type $n$} and form submanifolds of $M$ which are called \emph{strata of codimension $n$}.
\end{rem}

\subsection{Proof of Kontsevich's formula}

To show that $U$ is indeed an $L_\infty$-morphism, we need to show the following points:
\begin{enumerate}
    \item The first component of the restriction of $U$ to $\calV$ is up to a shift in degrees given by the natural map $U_1^{(0)}$ as defined in \eqref{eq:first_term}.
    \item $U$ is a graded linear map of degree zero.
    \item $U$ satisfies the equations \eqref{eq:L_infty_morphism} for an $L_\infty$-morphism. 
\end{enumerate}

\subsubsection{$U_1$ coincides with $U_1^{(0)}$}
\begin{lem}
The map 
\[
U_1\colon \calV\to \calD
\]
is the natural map $U^{(0)}_1$ that identifies each multivector field with the corresponding multidifferential operator. 
\end{lem}

\begin{proof}
Consider the set $\calG_{1,\bar n}$ of admissible graphs with one vertex of first type and $\bar n$ vertices of second type. It is easy to see that this set only contains one element $\Gamma_{\bar n}$ which has $2\cdot 1+\bar n-2=\bar n$ arrows outgoing from the single vertex of first type and each arrow is incoming to a vertex of second type (see Figure \ref{fig:single_graph}).

\begin{figure}[ht]
\centering
\begin{tikzpicture}[scale=0.7]
\tikzset{Bullet/.style={fill=black,draw,color=#1,circle,minimum size=0.5pt,scale=0.5}}
\node[Bullet=black, label=below: $\bar 1$] (bar1) at (1,-2) {};
\node[Bullet=black, label=below: $\bar 2$] (bar2) at (3,-2) {};
\node[Bullet=black, label=below: $\bar 3$] (bar3) at (5,-2) {};
\node[Bullet=black, label=below: $\bar n$] (barn) at (10,-2) {};
\node[label=below: $\dotsm$] (dots) at (6.5,0) {};
\node[Bullet=gray, label=above: $1$] (v1) at (5.5,1) {};
\draw[fermion] (v1) -- (bar1);
\draw[fermion] (v1) -- (bar2);
\draw[fermion] (v1) -- (bar3);
\draw[fermion] (v1) -- (barn);
\draw (0,-2) -- (6,-2);
\draw[dashed] (6,-2) -- (9,-2);
\draw (9,-2) -- (11,-2);
\end{tikzpicture}
\caption{The single element $\Gamma_{\bar n}\in \calG_{1,\bar n}$.} 
\label{fig:single_graph}
\end{figure}
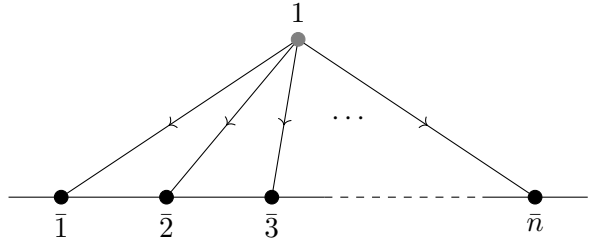

Hence, for a multivector field $\xi$ of degree $k$ we associate the multidifferential operator 
\[
U_{\Gamma_{\bar n}}(\xi)\colon f_{\bar 1}\otimes \dotsm \otimes f_{\bar n}\mapsto w_{\Gamma_{\bar n}}\xi^{i_{\bar 1}\dotsm i_{\bar n}}\de_{i_{\bar 1}}f_{\bar 1}\dotsm \de_{i_{\bar n}}f_{\bar n}.
\]
One can easily check that 
\[
\int_{\bar C_{1,\bar n}}\omega_{\Gamma_{\bar n}}=\int_{\bar C_{1,\bar n}}\,\,\bigwedge_{e\in E(\Gamma_{\bar n})}\dd\phi_e =(2\pi)^{\bar n}.
\]
Using \eqref{eq:Kontsevich_weight}, this will give us
\[
w_{\Gamma_{\bar n}}=\frac{1}{\bar n!}.
\]
This shows that $U_1$ is indeed the analogue for $U^{(0)}_1$ in \eqref{eq:first_term} which induces the HKR isomorphism \eqref{eq:HKR}.
\end{proof}

\subsubsection{Checking the degrees}
\begin{lem}
The $n$-th component 
\[
U_n:=\sum_{\bar n\geq 1}\,\sum_{\Gamma\in\calG_{n,\bar n}}w_\Gamma B_\Gamma
\]
has the correct degree for $U$ to be an $L_\infty$-morphism.
\end{lem}

\begin{proof}
Note that to a vertex $v_i$ with $\vert\mathrm{Star}(v_i)\vert$ outgoing arrows, we associate an element in $\calV^{\vert\mathrm{Star}(v_i)\vert}=\widetilde{\calV}^{\vert\mathrm{Star}(v_i)\vert+1}$. On the other hand, each graph $\Gamma$ with $\bar n$ vertices of second type and $n$ multivector fields $\xi_1,\ldots,\xi_n$ induce a multidifferential operator $U_n(\xi_1\land\dotsm \land \xi_n)$ of degree $s=\bar n-1$. Since we are only considering graphs with exactly $2n+\bar n-2$ edges and since 
\[
\vert E(\Gamma)\vert=\sum_{1\leq i\leq n}\vert\mathrm{Star}(v_i)\vert,
\]
we can write the degree of $U_n(\xi_1\land\dotsm \land \xi_n)$ as
\[
s=(2n+\bar n-2)+1-n=\vert E(\Gamma)\vert+1-n=\sum_{1\leq i\leq n}\vert\mathrm{Star}(v_i)\vert+1-n.
\]
This is exactly the degree for the $n$-th component of an $L_\infty$-morphism.
\end{proof}

\subsubsection{Reformulation of the $L_\infty$-condition in terms of graphs}
Next we discuss the geometric proof of the formality statement. We have to extend the morphism $U$ with a degree zero component which represents the usual multiplication between smooth maps. Therefore, we can consider the special case of the $L_\infty$-condition \eqref{eq:n-th_coefficients} where $m_0,\widetilde{m}_0$ are given by the Taylor coefficients $U_n$:
\begin{multline}
\label{eq:formality_decoding}
    \sum_{0\leq \ell\leq n}\,\,\sum_{-1\leq k\leq m}\,\,\sum_{0\leq i\leq m-k}\varepsilon_{kim}\sum_{\sigma\in S_{\ell,n-\ell}}\varepsilon_\xi(\sigma)U_\ell\bigg(\xi_{\sigma(1)}\land\dotsm \land \xi_{\sigma(\ell)}\bigg)\\
    \bigg(f_0\otimes\dotsm \otimes f_{i-1}\otimes U_{n-\ell}\bigg(\xi_{\sigma(\ell+1)}\land\dotsm\land \xi_{\sigma(n)}\bigg)\big(f_i\otimes\dotsm\otimes f_{i+k}\big)\otimes f_{i+k+1}\otimes\dotsm\otimes f_m\bigg)\\
    =\sum_{i\not=j=1}^n \varepsilon^{ij}_\xi U_{n-1}\bigg(\xi_i\circ\xi_j\land\xi_1\land\dotsm\land\widehat{\xi}_i\land\dotsm\land \widehat{\xi}_j\land\dotsm\land\xi_n\bigg)\big(f_0\otimes\dotsm \otimes f_n\big),
\end{multline}
where $(\xi_j)_{j=1,\ldots,n}$ are multivector fields, $f_0,\ldots,f_m$ are smooth maps on which the multidifferential operator is acting, $S_{\ell,n-\ell}$ is the subset of $S_n$ consisting of $(\ell,n-\ell)$-shuffle permutations, the product $\xi_i\circ\xi_j$ is defined in a way such that the Schouten--Nijenhuis bracket can be expressed in terms of this composition by a formula similar to the one relating the Gerstenhaber bracket to the analogous composition $\circ$ on $\calD$ in \eqref{eq:composition} and the signs involved are defined as follows: $\varepsilon_{kim}:=(-1)^{k(m+i)}$, $\varepsilon_\xi(\sigma)$ is the Koszul sign associated to the permutation $\sigma$ and $\varepsilon^{ij}_\xi$ is defined as in \eqref{eq:n-th_coefficients}.

\begin{rem}
Note that \eqref{eq:formality_decoding} carries the formality condition since the left hand side corresponds to the Gerstenhaber bracket of multidifferential operators and the right hand side to (a part of) the Schouten--Nijenhuis bracket (since the differential on $\calV$ is identical to zero). Recall that the differential on $\calD$ is is given by the Hochschild differential $[\mathsf{m},\enspace]_\mathrm{G}$ which in \eqref{eq:formality_decoding} is replaced by $U_0$.
\end{rem}

We can now rewrite \eqref{eq:formality_decoding} in terms of admissible graphs and weights to prove that it holds. Note that the difference between the left hand side and right hand side of \eqref{eq:formality_decoding} can be formulated as a linear combination 
\begin{equation}
\label{eq:difference_formality}
\sum_{\Gamma\in\calG_{n,\bar n}} c_\Gamma U_\Gamma(\xi_1\land\dotsm\land \xi_n)(f_0\otimes\dotsm \otimes f_n).
\end{equation}
Recall that we want to consider admissible graphs $\Gamma$ with exactly $2n+\bar n-2$ edges. In fact, \eqref{eq:formality_decoding} is satisfied if $c_\Gamma=0$ for all such $\Gamma$. We will show that this is indeed the case.\\ 

\subsubsection{The key is Stokes' theorem}
In order to prove that all these coefficients $c_\Gamma$ vanish, we will use Stokes' theorem for manifolds with corners. Similarly as for the usual Stokes' theorem, we can replace an integral of an exact form over some manifold $M$ as the integral of its primitive form over the boundary $\de M$. In particular, we have 
\begin{equation}
\label{eq:Stokes_formality}
    \int_{\de\bar C^+_{n,\bar n}}\omega_\Gamma=\int_{\bar C^+_{n,\bar n}}\dd\omega_\Gamma=0, \quad\forall \Gamma\in\calG_{n,\bar n}
\end{equation}
since $\dd\phi_e$ is closed and $\bar C^+_{n,\bar n}$ is compact. Let us expand the left hand side of \eqref{eq:Stokes_formality} in order to show that it is exactly given by the coefficients $c_\Gamma$. Therefore, we want to take a closer look on the manifold $\de \bar C^+_{n,\bar n}$. Recall that we have restricted the weights in \eqref{eq:formality_decoding} to be equal to zero if the form degree does not match the dimension of the underlying manifold over which we integrate. Hence, we only have to consider the codimension $1$ strata of $\de \bar C^+_{n,\bar n}$ which have dimension $2n+\bar n-3$. Note that the dimension of $\de\bar C^+_{n,\bar n}$ is equal to the number of edges and thus of the 1-forms $\dd\phi_e$. 

\begin{rem}
\label{rem:collision}
Intuitively, one can think of the boundary of $\bar C^+_{n,\bar n}$ to be represented by the degenerate configuration in which some of the $n+\bar n$ points \emph{collide with each other}. 
\end{rem}

\subsubsection{Classification of boundary strata}
Using Remark \ref{rem:collision}, we can classify the codimension $1$ strata of $\de\bar C^+_{n,\bar n}$ as follows:
\begin{itemize}
    \item (Strata of type S1) These are strata in which $i\geq 2$ points in $\mathbb{H}_+^2$ collide together to a point which still lies in $\mathbb{H}^2_+$. Points in such a stratum can be locally described by 
    \begin{equation}
        \label{eq:S1}
        C_i\times C_{n-i+1,\bar n}.
    \end{equation}
    The first term represents the relative position of the colliding points when we look at them under a \emph{magnifying glass}. The second term is the space of all remaining points plus the point which occurs as the point on which the first $i$ points have collapsed to.
    \item (Strata of type S2) These are strata in which $i>0$ points in $\mathbb{H}^2_+$ and $j>0$ points in $\R$ with $2i+j-2\geq 0$ collide to a single point on $\R$. Points in such a stratum can be locally represented by 
    \begin{equation}
        \label{eq:S2}
        C_{i,j}\times C_{n-i,\bar n-j+1}.
    \end{equation}
    The two terms are similarly given as in \eqref{eq:S1}. 
\end{itemize}

See Figure \ref{fig:S1} and \ref{fig:S2} for an illustration of a stratum of type S1 and type S2 respectively.

\begin{figure}[ht]
\centering
\begin{tikzpicture}[scale=0.7]
\tikzset{Bullet/.style={fill=black,draw,color=#1,circle,minimum size=0.5pt,scale=0.5}}
\node[Bullet=black, label=below: $\bar 1$] (bar1) at (1,-2) {};
\node[Bullet=black, label=below: $\bar 2$] (bar2) at (3,-2) {};
\node[Bullet=black, label=below: $\bar 3$] (bar3) at (5,-2) {};
\node[Bullet=black, label=below: $\bar n$] (barn) at (10,-2) {};
\node[Bullet=gray, label=left: $\ell$] (vi) at (5.5,1) {};
\node[Bullet=gray, label=above: $1$] (v1) at (1,0) {};
\node[Bullet=gray, label=above: $n$] (vn) at (8,0) {};
\node[Bullet=gray, label=above: $2$] (v2) at (3,3) {};
\node[Bullet=gray] (vc1) at (7,4) {};
\node[Bullet=gray] (vc2) at (9,3) {};
\node[Bullet=gray] (vc3) at (8.5,4.5) {};
\node[label=left: $\mathbb{H}^2$] (H) at (0,2.5) {};
\draw (0,-2) -- (6,-2);
\draw[dashed] (6,-2) -- (9,-2);
\draw (9,-2) -- (11,-2);
\draw[dashed] (8,4) circle (2cm);
\draw[dashed] (6,4) -- (vi);
\draw[dashed] (8,2) -- (vi);
\end{tikzpicture}
\caption{Example of a stratum of type S1. Note that here $i=3$ points have collapsed together to a point $\ell\in\mathbb{H}^2_+$.} 
\label{fig:S1}
\end{figure}
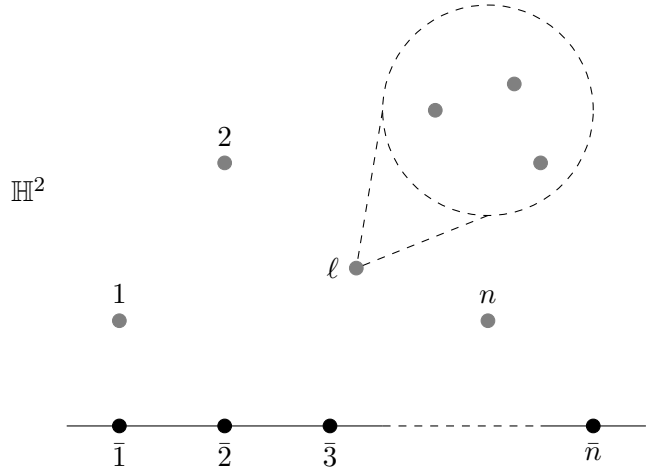

\begin{figure}[ht]
\centering
\begin{tikzpicture}[scale=0.7]
\tikzset{Bullet/.style={fill=black,draw,color=#1,circle,minimum size=0.5pt,scale=0.5}}
\node[Bullet=black, label=below: $\bar 1$] (bar1) at (1,-2) {};
\node[Bullet=black, label=below: $\bar 2$] (bar2) at (3,-2) {};
\node[Bullet=black, label=below: $\bar 3$] (bar3) at (5,-2) {};
\node[Bullet=black, label=below: $\bar n$] (barn) at (10,-2) {};
\node[Bullet=black, label=below: $\bar \ell$] (barell) at (7.5,-2) {};
\node[Bullet=black] (barc1) at (7,0) {};
\node[Bullet=black] (barc2) at (8,0) {};
\node[Bullet=gray, label=above: $1$] (v1) at (1,0) {};
\node[Bullet=gray, label=above: $n$] (vn) at (12,0) {};
\node[Bullet=gray, label=above: $2$] (v2) at (3,3) {};
\node[Bullet=gray] (vc1) at (6.5,1.5) {};
\node[Bullet=gray] (vc2) at (8,1) {};
\node[Bullet=gray] (vc3) at (7.5,2) {};
\node[label=left: $\mathbb{H}^2$] (H) at (0,2.5) {};
\draw (0,-2) -- (6,-2);
\draw (5.8,0) -- (9.2,0);
\draw[dashed] (6,-2) -- (9,-2);
\draw (9,-2) -- (11,-2);
\draw[dashed] (7.5,1) circle (2cm);
\draw[dashed] (5.8,0) -- (barell);
\draw[dashed] (9.2,0) -- (barell);
\end{tikzpicture}
\caption{Example of a stratum of type S2. Note that here $i=3$ and $j=2$ points have collapsed together to a point $\bar \ell\in\R$.} 
\label{fig:S2}
\end{figure}
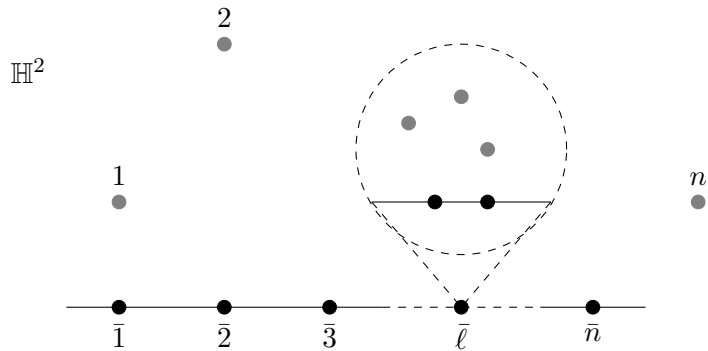

Now we can split the integral on the left hand side of \eqref{eq:Stokes_formality} into the sum of strata of type S1 and type S2. For the strata of type S1, we will distinguish two subcases for the $i$ collapsing vertices. Since the integral vanishes if the form degree is not equal to the dimension of the underlying manifold, one can show that the only contributions come from graphs $\Gamma$ whose subgraphs $\Gamma_1$, spanned by the collapsing vertices, contain exactly $2i-3$ edges.

When $i=2$, there is only one edge $e$, an hence in the first integral of the decomposition of \eqref{eq:S1} the differential $\dd\phi_e$ is integrated over $C_2\cong S^1$ and thus we get a factor of $2\pi$ which cancels the coefficient in \eqref{eq:Kontsevich_weight}. 
The remaining integral represents the weight of the corresponding quotient graph $\Gamma_2$ which is obtained from $\Gamma$ after the contraction of the edge $e$. In particular, to the vertex $j$ of first type, which results from this contraction, we associate the $j$-composition of the two multivector fields that were associated to the endpoints of $e$. Hence, summing over all graphs and all strata of this subtype, we get the right hand side of \eqref{eq:formality_decoding}.

\subsubsection{A trick using logarithms}
When $i\geq 3$, the integral corresponding to this stratum involves the product of $2i-3$ angle forms over $C_i$. According to a Lemma of Kontsevich, this integral vanishes. 

\begin{lem}[Kontsevich\cite{K}]
\label{lem:Kontsevich}
The integral over the configuration space $C_n$ of $n\geq 3$ points in the upper half-plane of any $\dim C_n:=2n-3$ angle forms $\dd\phi_{e_i}$ with $i=1,\ldots, n$ vanishes for $n\geq 3$.
\end{lem}

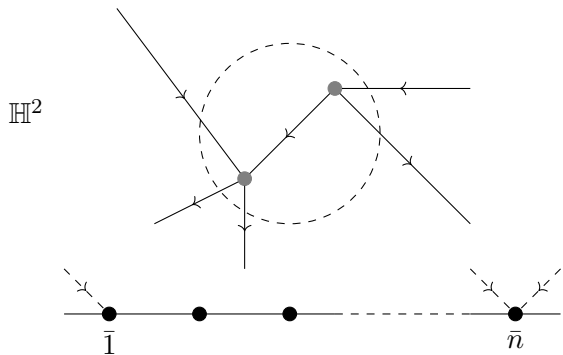
\begin{figure}[ht]
\centering
\begin{tikzpicture}[scale=0.6]
\tikzset{Bullet/.style={fill=black,draw,color=#1,circle,minimum size=0.5pt,scale=0.5}}
\node[Bullet=black, label=below: $\bar 1$] (bar1) at (1,-2) {};
\node[Bullet=black] (bar2) at (3,-2) {};
\node[Bullet=black] (bar3) at (5,-2) {};
\node[Bullet=black, label=below: $\bar n$] (barn) at (10,-2) {};
\node[Bullet=gray] (v1) at (4,1) {};
\node[Bullet=gray] (v2) at (6,3) {};
\node[] (n1) at (1,5) {};
\node[label=left: $\mathbb{H}^2$] (H) at (0,2.5) {};
\draw (0,-2) -- (6,-2);
\draw[dashed] (6,-2) -- (9,-2);
\draw (9,-2) -- (11,-2);
\draw[dashed] (5,2) circle (2cm);
\draw[fermion] (v2) -- (v1);
\draw[fermion] (n1) -- (v1);
\draw[fermion] (v1) -- (2,0);
\draw[fermion] (v1) -- (4,-1);
\draw[fermion] (9,3) -- (v2);
\draw[fermion] (v2) -- (9,0);
\draw[fermion,dashed] (0,-1) -- (bar1);
\draw[fermion,dashed] (11,-1) -- (barn);
\draw[fermion,dashed] (9,-1) -- (barn);
\end{tikzpicture}
\caption{Example of a non-vanishing term.} 
\label{fig:non-vanishing}
\end{figure}

\begin{figure}[ht]
\centering
\begin{tikzpicture}[scale=0.6]
\tikzset{Bullet/.style={fill=black,draw,color=#1,circle,minimum size=0.5pt,scale=0.5}}
\node[Bullet=black, label=below: $\bar 1$] (bar1) at (1,-2) {};
\node[Bullet=black] (bar2) at (3,-2) {};
\node[Bullet=black] (bar3) at (5,-2) {};
\node[Bullet=black, label=below: $\bar n$] (barn) at (10,-2) {};
\node[Bullet=gray] (v1) at (4,1) {};
\node[Bullet=gray] (v2) at (6,3) {};
\node[Bullet=gray] (v3) at (6,1) {};
\node[] (n1) at (1,5) {};
\node[label=left: $\mathbb{H}^2$] (H) at (0,2.5) {};
\draw (0,-2) -- (6,-2);
\draw[dashed] (6,-2) -- (9,-2);
\draw (9,-2) -- (11,-2);
\draw[dashed] (5,2) circle (2cm);
\draw[fermion] (v2) -- (v1);
\draw[fermion] (n1) -- (v1);
\draw[fermion] (v1) -- (2,0);
\draw[fermion] (v1) -- (4,-1);
\draw[fermion] (9,3) -- (v2);
\draw[fermion] (v2) -- (9,0);
\draw[fermion] (v1) -- (v3);
\draw[fermion] (v3) -- (v2);
\draw[fermion] (v3) -- (7.5,0);
\draw[fermion,dashed] (0,-1) -- (bar1);
\draw[fermion,dashed] (11,-1) -- (barn);
\draw[fermion,dashed] (9,-1) -- (barn);
\end{tikzpicture}
\caption{Example of a vanishing term.} 
\label{fig:non-vanishing}
\end{figure}
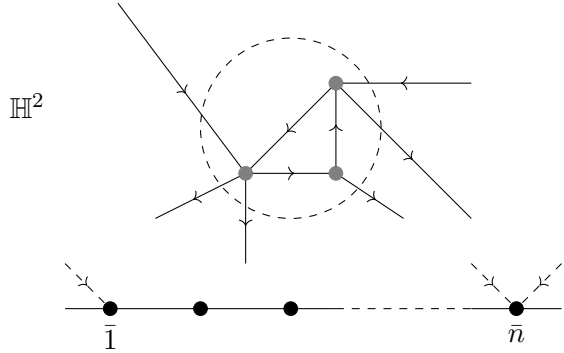

\begin{proof}[Proof of Lemma \ref{lem:Kontsevich}]
First we need to restrict the integral to an even number of angle forms. This can be done by identifying $C_n$ with the subset of $\mathbb{H}^n$, where one of the endpoints of $e_1$ is set to be the origin and the second is forced to lie on the unit circle (note that this can always be done by considering the action of the Lie group in the definition of $C_n$). Then the integral decomposes into a product of integrals of $\dd\phi_e$ over $S^1$ and the remaining $2n-4=:2N$ forms integrated over the resulting complex manifold $U$ given by the isomorphism
\[
C_n\cong S^1\times U.
\]
Then the claim follows by the following computation:
\begin{align}
    \begin{split}
        \int_U\bigwedge^{2N}_{j=1}\dd\,\mathrm{arg}(f_j)&=\int_U\bigwedge^{2N}_{j=1}\dd\log(\vert f_j\vert)=\int_{\bar U}\calI\bigg(\dd\log(\vert f_1\vert)\bigwedge^{2N}_{j=2}\dd\log(\vert f_j\vert)\bigg)\\
        &=\int_{\bar U}\dd\bigg(\calI\bigg(\log(\vert f_1\vert)\bigwedge^{2N}_{j=2}\dd\log(\vert f_j\vert)\bigg)\bigg)=0.
    \end{split}
\end{align}
Let us see what is happening here. First, note that the angle function $\phi_{e_j}$ was expressed in terms of the argument of the (holomorphic) function $f_j$. In particular, $f_j$ is just the difference of the coordinates of the endpoints of $e_j$. The first equality follows by the decompositions
\begin{align*}
\dd\,\mathrm{arg}(f_j)&=\frac{1}{2\I}\big(\dd\log(f_j)-\dd\log(\bar f_j)\big),\\
\dd\log(\vert f_j\vert)&=\frac{1}{2}\big(\dd\log(f_j)+\dd\log(\bar f_j)\big).
\end{align*}

Thus, the product of $2N$ of these expressions is a linear combination of products of $k$ holomorphic and $2N-k$ anti-holomorphic forms. By a result of complex analysis, one can show that in the integration over the complex manifold $U$, the only terms which do not vanish are the ones where $k=N$. Then one can observe that the terms coming from the first decomposition are equal to those coming from the second decomposition.

In the second equality the integral of the differential form is replaced by an integral over a suitable 1-form with values in the space of \emph{distributions} over the compactification $\bar U$ of $U$. One can show (this is another Lemma) that such a map $\calI$, which sends usual 1-forms to distributional ones, commutes with the differential and thus by Stokes' theorem we get the claim.
\end{proof}

\subsubsection{Last step: vanishing terms for type S2 strata}
Finally, let us consider the strata of type S2. There we will have a similar dimensional argument, i.e. it analogously restricts the possible non-vanishing terms to the condition that the subgraph $\Gamma_1$, spanned by $i+j$ collapsing vertices of first and second type respectively, contains exactly $2i+j-2$ edges. Similarly as before, for the quotient graph $\Gamma_2$ obtained by contracting $\Gamma_1$, the claim is that the only non-vanishing contributions come from the graphs for which both graphs obtained from a given $\Gamma$ are admissible. In this case we get a decomposition of the weight $w_\Gamma$ into the product $w_{\Gamma_1}\cdot w_{\Gamma_2}$ which in general, by the conditions on the number of edges of $\Gamma$ and $\Gamma_1$, does not vanish.

The only remaining thing to check is that we do not have \emph{bad edges} by contraction (see Figure \ref{fig:non-admissible_graphs}). The only such possibility appears when $\Gamma_2$ contains an edge which starts from a vertex of second type. In this case the corresponding integral vanishes because it contains the differential of an angle map evaluated on the pair $(z_1,z_2)$ where $z_1$ is constrained to lie in $\R$ and such maps vanish for every $z_2$ because the angle is measured with respect to the Poincar\'e metric (recall Figure \ref{fig:angle_map} for an intuitive picture).

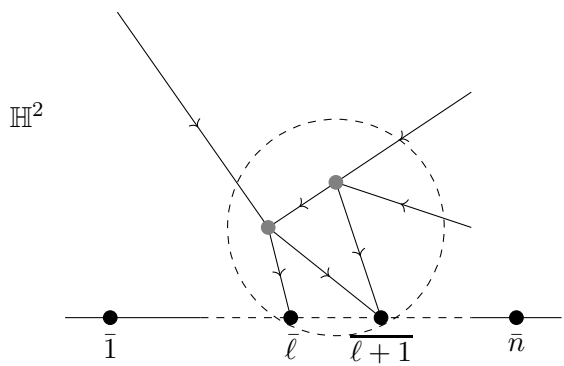
\begin{figure}[ht]
\centering
\begin{tikzpicture}[scale=0.6]
\tikzset{Bullet/.style={fill=black,draw,color=#1,circle,minimum size=0.5pt,scale=0.5}}
\node[Bullet=black, label=below: $\bar 1$] (bar1) at (1,-2) {};
\node[Bullet=black, label=below: $\bar \ell$] (bar2) at (5,-2) {};
\node[Bullet=black, label=below: $\overline{\ell+1}$] (bar3) at (7,-2) {};
\node[Bullet=black, label=below: $\bar n$] (barn) at (10,-2) {};
\node[Bullet=gray] (v1) at (4.5,0) {};
\node[Bullet=gray] (v2) at (6,1) {};
\node[] (n1) at (1,5) {};
\node[label=left: $\mathbb{H}^2$] (H) at (0,2.5) {};
\draw (0,-2) -- (3,-2);
\draw[dashed] (3,-2) -- (9,-2);
\draw (9,-2) -- (11,-2);
\draw[dashed] (6,0) circle (2.4cm);
\draw[fermion] (v2) -- (v1);
\draw[fermion] (n1) -- (v1);
\draw[fermion] (v1) -- (bar2);
\draw[fermion] (v1) -- (bar3);
\draw[fermion] (9,3) -- (v2);
\draw[fermion] (9,0) -- (v2);
\draw[fermion] (v2) -- (bar3);
\end{tikzpicture}
\caption{Example of an admissible quotient. } 
\label{fig:admissible_quotient}
\end{figure}

\begin{figure}[ht]
\centering
\begin{tikzpicture}[scale=0.6]
\tikzset{Bullet/.style={fill=black,draw,color=#1,circle,minimum size=0.5pt,scale=0.5}}
\node[Bullet=black, label=below: $\bar 1$] (bar1) at (1,-2) {};
\node[Bullet=black, label=below: $\bar \ell$] (bar2) at (5,-2) {};
\node[Bullet=black, label=below: $\overline{\ell+1}$] (bar3) at (7,-2) {};
\node[Bullet=black, label=below: $\bar n$] (barn) at (10,-2) {};
\node[label=below: bad edge] (be) at (2,0) {}; 
\node[Bullet=gray] (v1) at (4.5,0) {};
\node[Bullet=gray] (v2) at (6,1) {};
\node[] (n1) at (1,5) {};
\node[label=left: $\mathbb{H}^2$] (H) at (0,2.5) {};
\draw (0,-2) -- (3,-2);
\draw[dashed] (3,-2) -- (9,-2);
\draw (9,-2) -- (11,-2);
\draw[dashed] (6,0) circle (2.4cm);
\draw[fermion] (v2) -- (v1);
\draw[fermion] (n1) -- (v1);
\draw[fermion] (v1) -- (bar2);
\draw[fermion] (v1) -- (1,0);
\draw[fermion] (9,3) -- (v2);
\draw[fermion] (9,0) -- (v2);
\draw[fermion] (v2) -- (bar3);
\end{tikzpicture}
\caption{Example of a non-admissible quotient. Such a term vanishes.} 
\label{fig:non-admissible_quotient}
\end{figure}
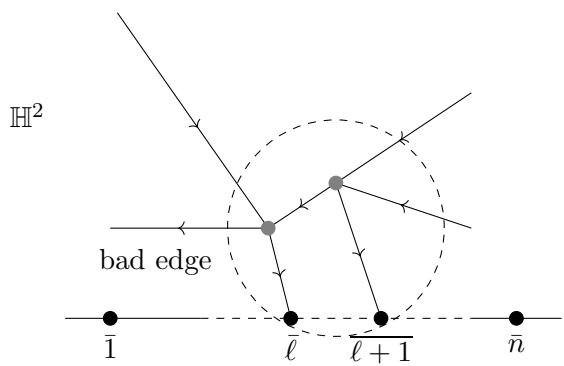

This means that the only non-vanishing terms correspond to the case when we plug the differential operator associated to $\Gamma_1$ as the $k$-th argument of the one associated to $\Gamma_2$, where $k$ is the vertex of the second type emerging from the collapse. Summing over all these possibilities, the strata of type S2 exactly corresponds to the left hand side of \eqref{eq:formality_decoding}.

We have thus proven that $U$ is indeed an $L_\infty$-morphism and since the first coefficients $U_1$ are given by $U^{(0)}_1$ it is also a quasi-isomorphism and hence determines uniquely a star product for any bivector field $\pi$ on $\R^d$ by the formula \eqref{eq:star_product_L_infty}.

\section{Globalization of Kontsevich's star product}
\label{sec:globalization}
Kontsevich's formula for the star product only gives a quantization for the case when $M=\R^d$ for a general Poisson structure $\pi$, and thus only describes a local description for the general case. Already in \cite{K}, Kontsevich gave a globalization method which was very briefly mentioned, but it was explicitly constructed by Cattaneo, Felder and Tomassini in \cite{CFT,CattaneoFelderTomassini2002} in a similar way as Fedosov did for the symplectic case of the Moyal product \cite{Fedosov1994} (see also Section \ref{subsec:Fedosov}). Their approach uses a flat connection $\bar D$ on a vector bundle over $M$ such that the algebra of the horizontal sections with respect to $\bar D$ is a quantization of $C^\infty(M)$. Consider the vector bundle $E_0\to M$ of infinite jets of functions endowed with a flat connection $D_0$. The fiber $(E_0)_x$ over $x\in M$ is naturally a commutative algebra and carries the Poisson structure induced fiberwise by the Poisson structure on $C^\infty(M)$. The map which associates to any globally defined function its infinite jet at each point $x\in M$ is a Poisson isomorphism onto the Poisson algebra of horizontal sections of $E_0$ with respect to $D_0$. Since the star product gives a deformation of the pointwise product on $C^\infty(M)$, we want to have a \emph{quantum version} of the vector bundle and the flat connection in order to get an similar isomorphism. The vector bundle $E\to M$ is defined in terms of a section $\phi^\mathrm{aff}$ of the fiber bundle $M^\mathrm{aff}\to M$, where $M^\mathrm{aff}$ denotes the quotient of the manifold $M^\mathrm{coor}$ of jets of coordinate systems on $M$ by the action of the group $GL(d,\R)$ of linear diffeomorphims given by $E:=(\phi^\mathrm{aff})^*\widetilde{E}$, where $\widetilde{E}$ is the bundle of $\R[\![\hbar]\!]$-modules 
$M^\mathrm{coor}\times_{GL(d,\R)}\R[\![y^1,\ldots,y^d]\!][\![\hbar]\!]\to M^\mathrm{aff}$.
Note that the section $\phi^\mathrm{aff}$ can be realized explicitly by a collection of jets at $0$ of maps $\phi_x\colon \R^d\to M$ such that $\phi_x(0)=x$ for all $x\in M$ (modulo the action of $GL(d,\R)$). Thus, we can assume for simplicity that we have fixed a representative $\phi_x$ of the equivalence class for each open set of a given covering, hence realizing a trivialization of the bundle $E$. So from now on we will identify $E$ with the trivial bundle with fiber given by $\R[\![y^1,\ldots,y^d]\!][\![\hbar]\!]$. Therefore, $E$ realizes the desired quantization, since it is isomorphic (as a bundle of $\R[\![\hbar]\!]$-modules) to the bundle $E_0[\![\hbar]\!]$ whose elements are formal power series with infinite jets of functions as coefficients.

\subsection{The multiplication, connection and curvature maps}
To define the star product and the connection on $E$, one has to introduce new objects whose existence and properties will be byproducts of the formality theorem. For a Poisson structure $\pi$ and two vector fields $\xi$ and $\eta$ on $\R^d$, we define 
\begin{align}
    \label{eq:multiplication}
    P(\pi)&:=\sum_{j\geq 0}\frac{\hbar^j}{j!}U_j(\pi\land\dotsm \land \pi),\\
    \label{eq:connection}
    A(\xi,\pi)&:=\sum_{j\geq 0}\frac{\hbar^j}{j!}U_{j+1}(\xi\land\pi\land\dotsm\land\pi),\\
    \label{eq:curvature}
    F(\xi,\eta,\pi)&:=\sum_{j\geq 0}\frac{\hbar^j}{j!}U_{j+2}(\xi\land\eta\land\pi\land\dotsm\land\pi).
\end{align}

It is easy to show that the degree of the multidifferential operators on the right hand sides of \eqref{eq:multiplication}, \eqref{eq:connection} and \eqref{eq:curvature} induce that $P(\pi)$ is a (formal) bidifferential operator, $A(\xi,\pi)$ is a differential operator and $F(\xi,\eta,\pi)$ is a function. Indeed, $P(\pi)$ is actually just the star product associated to $\pi$. In particular, the maps $P,A$ and $F$ are elements of degree $0,1$ and $2$ respectively of the Lie algebra cohomology complex of (formal) vector fields with values in the space of local polynomial maps, i.e. multidifferential operators which depend polynomially on $\pi$. An element of degree $j$ of this complex is a map that sends $\xi_1\land\dotsm\land\xi_j$ to a multidifferential operator $S(\xi_1\land\dotsm\land \xi_j\land\pi)$. The differential $\delta$ of this complex can then be defined as 
\begin{multline*}
    \delta S(\xi_1\land\dotsm\land\xi_{j+1}\land\pi):=\sum_{1\leq i\leq j+1}(-1)^i\frac{\de}{\de t}\Big|_{t=0}S\Big(\xi_1\land\dotsm\land\widehat{\xi}_i\land\dotsm\land\xi_{j+1}\land(\Phi^{\xi_i}_t)_*\pi\Big)\\
    +\sum_{i<\ell}(-1)^\ell S\Big([\xi_i,\xi_\ell]\land\xi_1\land\dotsm\land\widehat{\xi}_i\land\dotsm\land\widehat{\xi}_\ell\land\dotsm\land\xi_{j+1}\land\pi\Big),
\end{multline*}
where $\Phi^\xi_t$ denotes the flow of the vector field $\xi$. The associativity of the star product can now be expressed by 
\[
P\circ (P\otimes \id-\id\otimes P)=0.
\]
This follows from the formality theorem and hence the following equations do hold in the a similar way:
\begin{equation}
\label{eq:Eq1}
P(\pi)\circ (A(\xi,\pi)\otimes \id+\id\otimes A(\xi,\pi))=A(\xi,\pi)\circ P(\pi)+\delta P(\xi,\pi),
\end{equation}
\begin{equation}
\label{eq:Eq2}
P(\pi)\circ (F(\xi,\eta,\pi)\otimes \id-\id\otimes F(\xi,\eta,\pi))=-A(\xi,\pi)\circ A(\eta,\pi)+A(\eta,\pi)\circ A(\xi,\pi)+\delta A(\xi,\eta,\pi),
\end{equation}
\begin{equation}
\label{eq:Eq3}
-A(\xi,\pi)\circ F(\eta,\zeta,\pi)-A(\eta,\pi)\circ F(\zeta,\xi,\pi)-A(\zeta,\pi)\circ F(\xi,\eta,\pi)=\delta F(\xi,\eta,\zeta,\pi)
\end{equation}

Equation \eqref{eq:Eq1} describes the fact that under coordinate transformation induced by the vector field $\xi$, the star product $P(\pi)$ is changed to an equivalent one up to higher order terms. Equations \eqref{eq:Eq2} and \eqref{eq:Eq3} will be used for the construction of the relations between a connection 1-form $A$ and its curvature $F_A$.
For the explicit computation of the configuration space integrals, which arise for the weight computation in the Taylor coefficients of $U_j$, we can also describe the lowest order terms in the expansion of $P, A$ and $F$ and their action on functions:
\begin{enumerate}
    \item $P(\pi)(f\otimes g)=fg+\hbar\pi(\dd f,\dd g)+O(\hbar^2)$,
    \item $A(\xi,\pi)=\xi+O(\hbar)$, where we identify $\xi$ with a first order differential operator on the right hand side,
    \item $A(\xi,\pi)=\xi$, if $\xi$ is a linear vector field,
    \item $F(\xi,\eta,\pi)=O(\hbar)$,
    \item $P(\pi)(1\otimes f)=P(\pi)(f\otimes 1)=f$,
    \item $A(\xi,\pi)1=0$.
\end{enumerate}
    
Equations $(1)$ and $(5)$ have already been introduced before as the defining conditions for a star product. The equations involving the connection $A$ are used to construct a connection $D$ on sections on $E$. A section $f\in\Gamma(E)$ is locally given by a map $x\mapsto f_x$, where for every $y$, $f_x(y)$ is a formal power series with coefficients given by infinite jets. On this space we can introduce a deformed product $\star$ which will be the desired star product on $C^\infty(M)$ after we have identified horizontal section with ordinary functions. Let us denote analogously by $\pi_x$ the push-forward by $\phi_x^{-1}$ of the Poisson bivector field $\pi$ on $\R^d$. Then we can define the deformed product by the formal bidifferential operator $P(\pi_x)$ similarly as we did for the usual product by $P(\pi)$:
\[
(f\star g)_x(y):=f_x(y)g_x(y)+\hbar\pi^{ij}_x(y)\frac{\de f_x(y)}{\de y^i}\frac{\de g_x(y)}{\de y^j}+O(\hbar^2).
\]
One can define the connection $D$ on $\Gamma(E)$ by 
\[
(Df)_x:=\dd_xf+A^M_xf,
\]
where $\dd_xf$ denotes the de Rham differential of $f$ regarded as a function with values in $\R[\![y^1,\ldots,y^d]\!][\![\hbar]\!]$ and the formal connection $1$-form is given through its action on a tangent vector $\xi$ by 
\[
A^M_x(\xi):=A(\widehat{\xi}_x,\pi_x),
\]
where $A$ is the operator defined as in \eqref{eq:connection} evaluated on the multivector fields $\xi$ and $\pi$ given through the local coordinate system defined by $\phi_x$. Note that since the coefficients $U_j$ of the formality map that appear in the definition of $P$ and $A$ are polynomial in the derivatives of the coordinate of the arguments $\xi$ and $\pi$, all results which hold for $P(\pi)$ and $A(\xi,\pi)$ are inherited by their counterparts.
In fact, Equations $(1)$ and $(5)$ above ensure that $\star$ is an associative deformation of the pointwise product on sections and Equations $(2)$ and $(3)$ can be used to prove that $D$ is independent of the choice of $\phi$ and hence induces a global connection on $E$.

\subsection{Construction of solutions for a Fedosov-type equation}
We can extend $D$ and $\star$ by the (graded) Leibniz rule to the whole complex of formal differential forms $\Omega^\bullet(E):=\Omega^\bullet(M)\otimes_{C^\infty(M)}\Gamma(E)$ and use \eqref{eq:Eq2} to get the following lemma.

\begin{lem}
Let $F^M$ be the $E$-valued 2-form given by $x\mapsto F^M_x$ where $F_x^M(\xi,\eta):=F(\widehat{\xi}_x,\widehat{\eta}_x,\pi_x)$ for any pair of vector fields $\xi,\eta$. Then $F^M$ represents the curvature of $D$ and the two are related to each other and to the star product by the usual identities:
\begin{enumerate}
    \item $D(f\star g)=D(f)\star g+f\star D(g)$,
    \item $D^2=[F^M,\enspace]_\star$,
    \item $DF^M=0$.
\end{enumerate}
\end{lem}

\begin{proof}
We can deduce these identities from the relations \eqref{eq:Eq1}, \eqref{eq:Eq2} and \eqref{eq:Eq3}, in which the star commutator $[f,g]_\star:=f\star g-g\star f$ is already implicitly defined, after identifying the complex of formal multivector fields endowed with the differential $\delta$ with the complex of formal multidifferential operators with the de Rham differential. Such an isomorphism was explicitly given in \cite{CFT}.
\end{proof}

\begin{rem}
A connection $D$ which satisfies the above relation on a bundle $E$ of associative algebras is called \emph{Fedosov connection} with \emph{Weyl curvature} $F$. It is the kind of connection Fedosov introduced in order to give a global construction for the symplectic case (Section \ref{subsec:Fedosov}) \cite{Fedosov1994}.  
\end{rem}

The final step towards a globalization is to deform the connection $D$ into a new connection $\bar D$ which has the same properties and moreover has vanishing Weyl curvature, i.e. $\bar D^2=0$. Hence, we can define the complex $H^j_{\bar D}(E)$ and thus the (sub)algebra of horizontal sections $H^0_{\bar D}(E)$. 

\begin{lem}
\label{lem:Fedosov_connection}
Let $D$ be a Fedosov connection on $E$ with Weyl curvature $F$ and $\gamma$ an $E$-valued 1-form. Then 
\[
\bar D:=D+[\gamma,\enspace]_\star
\]
is also a Fedosov connection with Weyl curvature $\bar F=F+D\gamma+\gamma\star \gamma$.
\end{lem}

\begin{proof}
Let $f$ be a section in $\Gamma(E)$. Then 
\begin{align*}
    \bar D^2 f&=[\bar F,f]_\star+D[\gamma,f]_\star+[\gamma,Df]_\star+[\gamma,[\gamma,f]_\star]_\star\\
    &=[F,f]_\star+[D\gamma,f]_\star+[\gamma,[\gamma,f]_\star]_\star\\
    &=\left[F+D\gamma+\frac{1}{2}[\gamma,\gamma]_\star, f\right]_\star,
\end{align*}
where the last equality follows from the Jacobi identity of the star commutator $[\enspace,\enspace]_\star$, since each associative product induces a Lie bracket given by the commutator. if we apply $\bar D$ to the obtained curvature, we get 
\begin{align*}
    \bar D\left(F+D\gamma+\frac{1}{2}[\gamma,\gamma]_\star\right)&=D^2\gamma+\frac{1}{2}[D\gamma,\gamma]_\star-\frac{1}{2}[\gamma,D\gamma]_\star+[\gamma,F+D\gamma]_\star\\
    &=[F,\gamma]_\star+[\gamma,F]_\star\\
    &=0,
\end{align*}
where we again used the (graded) Jacobi identity and the (graded) skew-symmetry of $[\enspace,\enspace]_\star$.
\end{proof}

\begin{lem}
\label{lem:Fedosov_connection2}
Let $D$ be a Fedosov connection on a bundle $E=E_0[\![\hbar]\!]$ and $F$ its Weyl curvature and let 
\[
D=D_0+\hbar D_1+\dotsm,\qquad F=F_0+\hbar F_1+\dotsm
\]
be their expansion as formal power series. If $F_0=0$ and the second cohomology of $E_0$ with respect to $D_0$ is trivial, then there exists a 1-form $\gamma$ such that $\bar D$ has vanishing Weyl curvature.
\end{lem}

\begin{proof}
By Lemma \ref{lem:Fedosov_connection}, we can equivalently say that there exists a solution to the equation 
\[
\bar F=F+D\gamma+\frac{1}{2}[\gamma,\gamma]_\star=0.
\]
We can explicitly construct a solution by induction on the order of $\hbar$. We start from $\gamma_0=0$ and assume that $\gamma^{(j)}$ is a solution modulo $\hbar^{(j+1)}$. We can add to $\bar F^{(j)}=F+D\gamma^{(j)}+\frac{1}{2}\left[\gamma^{(j)},\gamma^{(j)}\right]_\star$ the next term $\hbar^{j+1}D_0\gamma_{j+1}$ to get $\bar F^{j+1}$ modulo higher terms. From the equation $D\bar F^{(j)}+\left[\gamma^{(j)},\bar F^{(j)}\right]_\star=0$ and the induction hypothesis $\bar F^{{j}}=0$ modulo $\hbar^{(j+1)}$, we get $D_0\bar F^{(j)}=0$. Now since $H^2_{D_0}(E_0)=0$, we can invert $D_0$ in order to define $\gamma_{j+1}$ in terms of the lower order terms $\bar F^{(j)}$ in a way such that $\bar F^{(j+1)}=0$ modulo $\hbar^{(j+2)}$. This completes the induction step.
\end{proof}

\begin{rem}
Note that in our case $D$ is a deformation of the natural flat connection $D_0$ on sections of the bundle of infinite jets. Thus, the hypothesis of Lemma \ref{lem:Fedosov_connection2} is indeed satisfied and we can find a flat connection $\bar D$ which is still a good deformation of $D_0$. In the setting of formal geometry, a flat connection as $D_0$ is called a \emph{Grothendieck connection} \cite{Grothendieck1968,Katz1970} which is a certain generalization of the \emph{Gauss--Manin connection} \cite{Manin1958}.  
Another more technical lemma, which uses the notion of \emph{homological perturbation theory}, actually gives an isomorphism between the algebra of horizontal sections $H^0_{\bar D}(E)$ and its non-deformed counterpart $H^0_{D_0}(E_0)$ which is isomorphic to $C^\infty(M)$. This implies the globalization procedure. 
\end{rem}

\begin{rem}
In \cite{D}, Dolgushev gave another proof of Kontsevich's formality theorem for general manifolds by using covariant tensors instead of $\infty$-jets of multidifferential operators and multivector fields which is intrinsically local. In particular, he also formulated the deformation quantization construction on \emph{Poisson orbifolds}.
\end{rem}

\section{Operadic approach to formality and Deligne's conjecture}

\subsection{Operads and algebras}

\begin{defn}[Algebraic operad]
\label{defn:alg_op}
An (algebraic) \emph{operad} (of vector spaces) consists of the following;
\begin{enumerate}
    \item A collection of vector spaces $P(n)$, $n\geq 0$,
    \item An action of the symmetric group $S_n$ on $P(n)$ for all $n$,
    \item An identity element $\id_P\in P(1)$,
    \item Compositions $m_{n_1,\ldots,n_k}$:
    \[
    P(k)\otimes (P(n_1)\otimes\dotsm \otimes P(n_k))\to P(n_1+\dotsm +n_k),
    \]
    for all $k\geq 0$ and $n_1,\ldots,n_k\geq 0$ satisfying a list of axioms. 
\end{enumerate}
\end{defn}

\begin{rem}
We want to obtain the list of axioms for an operad by looking at some examples.
\end{rem}

\begin{ex}[Endomorphism operad]
\label{ex:endomorphism_operad}
Consider the very simple operad $P(n):=\mathsf{Hom}(V^{\otimes n},V)$ for some vector space $V$. The action of the symmetric group and the identity element are obvious. The compositions are defined by 
\begin{multline*}
    (m_{n_1\ldots,n_k}(\phi\otimes(\psi_1\otimes\dotsm \otimes\psi_k)))(v_1\otimes\dotsm \otimes v_{n_1+\dotsm+n_k})\\
    :=\phi(\psi_1(v_1\otimes\dotsm\otimes v_{n_1})\otimes\dotsm\otimes\psi_k(v_{n_1}+\dotsm+v_{n_{k-1}+1}\otimes\dotsm \otimes v_{n_1+\dotsm +n_k})),
\end{multline*}
where $\phi\in P(k)=\mathsf{Hom}(V^{\otimes k},V)$, $\psi_i\in P(n_i)=\mathsf{Hom}(V^{\otimes n_i},V)$ for $i=1,\ldots,k$. This operad is called the \emph{endomorphism operad} of a vector space $V$. 
\end{ex}

\begin{rem}
We can see from Example \ref{ex:endomorphism_operad} that there should be an assoicativity axiom for multiple compositions, various compatibilities for actions of symmetric groups, and evident relations for compositions including the identity element.
\end{rem}

\begin{ex}[Associative operad]
Consider the operad $P=\mathsf{Assoc}_1$. The $n$-th component $P(n)=\mathsf{Assoc}_1(n)$ for $n\geq 0$ is defined as the collection of all universal (functorial) $n$-linear operations $\calA^{\otimes n}\to \calA$ of associative algebras $\calA$ with unit. The space $\mathsf{Assoc}_1(n)$ has dimension $n!$, and is spanned by the operations
\[
a_1\otimes\dotsm \otimes a_n\mapsto a_{\sigma(1)}\dotsm a_{\sigma(n)},\quad \sigma\in S_n.
\]
We can identify $\mathsf{Assoc}_1(n)$ with the subspace of free associative unital algebras in $n$ generators consisting of expressions which are multilinear in each generator.
\end{ex}

\begin{defn}[Algebra over an operad]
An \emph{algebra over an operad} $P$ consists of a vector space $\calA$ and a collection of multilinear maps $f_n\colon P(n)\otimes \calA^{\otimes n}\to \calA$ for all $n\geq 0$ satisfying the following axioms:
\begin{enumerate}
    \item The map $f_n$ is $S_n$-equivariant for any $n\geq 0$,
    \item We have $f_1(\id_P\otimes a)=a$ for all $a\in \calA$,
    \item All compositions in $P$ map to compositions of multilinear operations on $\calA$.
\end{enumerate}
\end{defn}

\begin{rem}
We often also call an algebra $\calA$ over an operad $P$ a \emph{$P$-algebra}.
\end{rem}

\begin{ex}
The algebra over the operad $\mathsf{Assoc}_1$ is an associative unital algebra. If we replace the 1-dimensional space $\mathsf{Assoc}_1(0)$ by the zero space $0$, we obtain an operad $\mathsf{Assoc}$ describing associative algebras possibly without unit.
\end{ex}

\begin{ex}[Lie operad]
There is an operad $\mathsf{Lie}$ such that $\mathsf{Lie}$-algebras are Lie algebras. The dimension of the $n$-th component $\mathsf{Lie}(n)$ is $(n-1)!$ for $n\geq 0$ and zero for $n=0$.
\end{ex}

\begin{thm}
Let $P$ be an operad and $V$ a vector space. Then the free $P$-algebra $\mathsf{Free}_P(V)$ generated by $V$ is naturally isomorphic as a vector space to 
\[
\bigoplus_{n\geq 0}(P(n)\otimes V^{\otimes n})_{S_n},
\]
where the subscript $S_n$ denotes the quotient space of coinvariants for the diagonal action of the symmetric group.
\end{thm}

\begin{rem}
The free algebra $\mathsf{Free}_P(V)$ is defined by the usual categorical \emph{adjunction property}, i.e. the set $\mathsf{Hom}_{P\text{-algebras}}(\mathsf{Free}_P(V),\calA)$ (i.e. homomorphisms in the category of $P$-algebras) is naturally isomorphic to the set $\mathsf{Hom}_{\text{vector spaces}}(V,\calA)$ for any $P$-algebra $\calA$.
\end{rem}

\subsection{Topological operads}

We can replace vector spaces by \emph{topological spaces} in the definition of an operad. The tensor product is then replaced by the Cartesian product.

\begin{defn}[Topological operad]
\label{defn:top_op}
A \emph{topological operad} consists of the following:
\begin{enumerate}
    \item A collection of topological spaces $P(n)$ for $n\geq 0$,
    \item A continuous action of the symmetric group $S_n$ on $P(n)$ for all $n$,
    \item An identity element $\id_P\in P(1)$,
    \item Compositions $m_{n_1,\ldots,n_k}$:
    \[
    P(k)\times(P(n_1)\times\dotsm\times P(n_k))\to P(n_1+\dotsm+n_k),
    \]
    which are continuous maps for all $k\geq 0$ and $n_1,\ldots,n_k\geq 0$ satisfying a list of axioms similarly to the ones in Definition \ref{defn:alg_op}.
\end{enumerate}
\end{defn}

\begin{ex}[Topological version of endomorphism operad]
The analog version to the endopmorphism operad is the operad $P$ such that for any $n\geq 0$, we get a topological space $P(n)$ which is the space of continuous maps from $X^n\to X$, where $X$ is some compact topological space.
\end{ex}

\begin{rem}
In general, one can define an operad and an algebra over an operad in any arbitrary symmetric monoidal category $\mathscr{C}^\otimes$, i.e. a category endowed with the functor 
\[
\otimes\colon\mathscr{C}\times \mathscr{C}\to \mathscr{C},
\]
the identity element $1_\mathscr{C}\in \mathrm{Obj}(\mathscr{C})$ and various coherence conditions for associativity, commutativity of $\otimes$ and so on. 

In particular, we want to consider the symmetric monoidal category $\mathsf{Complexes}$ of \emph{$\mathbb{Z}$-graded cochain complexes} of abelian groups (or vector spaces over some field). Such operads are called \emph{dg-operads}. Note that each component $P(n)$ is a cochain complex, i.e. a vector space decomposed into a direct sum $P(n)=\bigoplus_{i\in\mathbb{Z}}P(n)^i$ and endowed with a differential $\dd\colon P(n)^i\to P(n)^{i+1}$ which is of degree $+1$ satisfying $\dd^2=0$.
\end{rem}

\begin{rem}
We can construct an operad of cochain complexes out of a topological operad by considering a version of the \emph{singular chain complex}. For a topological space $X$ we denote by $\mathsf{Chains}(X)$ the complex concentrated in negative degrees, whose $(-k)$-th component for $k\geq 0$ consists of the formal finite additive combinations
\[
\sum_{1\leq i\leq N}n_if_i,\quad n_i\in\mathbb{Z},\, N\in\mathbb{Z}_{\geq 0}.
\]
of continuous maps $f_i\colon [0,1]^k\to X$ (\emph{singular cubes} in $X$) modulo the following relations:
\begin{enumerate}
    \item $f\circ \sigma=\mathrm{sign}(\sigma)f$ for any $\sigma\in S_k$ acting on the standard curbe $[0,1]^k$ by permutations of coordinates,
    \item $f'\circ\mathrm{pr}_{k\to k-1}=0$, where $\mathrm{pr}_{k\to k-1}\colon[0,1]^k\to[0,1]^{k-1}$ is the projection onto the first $(k-1)$ coordinates and $f'\colon[0,1]^{k-1}\to X$ is a continuous map.
\end{enumerate}
The boundary operator is defined similarly as for singular chains. Note that, in contrast to singular chains, for cubical chains we have an external product map 
\[
\bigotimes_{i\in I}\mathsf{Chains}(X_i)\to \mathsf{Chains}\left(\prod_{i\in I}X_i\right). 
\]
If $P$ is a topological operad, then $\mathsf{Chains}(P(n))$ for $n\geq 0$ has a natural operad structure in the category of complexes of abelian groups. The compositions are then given in terms of the external tensor product of cubical chains. On the level of cohomology, we obtain an operad $H(P)$ of $\mathbb{Z}$-graded abelian groups which are complexes endowed with the zero differential. This is called the \emph{homology} of the operad $P$.
\end{rem}

\subsection{The little disks operad}

Let $d\geq 1$ and denote by $G_d$ the $(d+1)$-dimensional Lie group acting on $\R^d$ by affine transformations $u\mapsto \lambda u+v$ where $\lambda>0$ is a real number and $v\in\R^d$. This group acts simply and transitively on the space of closed disks in $\R^d$ endowed with the usual Euclidean metric. The disk with center $v$ and radius $\lambda$ is given by a transformation of $G_d$ with parameters $(\lambda,v)$ of the standard disk
\[
D_0:=\{(x_1,\ldots,x_d)\in\R^d\mid x_1^2+\dotsm+x_d^2\leq 1\}.
\]

\begin{defn}[Little disks operad]
The \emph{little disks operad} $C_d$ is a topological operad with the following structure:
\begin{enumerate}
    \item $C_d(0)=\varnothing$,
    \item $C_d(1)=\{\id_{C_d}\}$,
    \item The space $C_d(n)$ is the space of configurations of $n$ disjoint disks $(D_i)_{1\leq i\leq n}$ inside the standard disk $D_0$ for $n\geq 2$.
\end{enumerate}
The composition 
\[
C_d(k)\times (C_d(n_1)\times\dotsm\times C_d(n_k))\to C_d(n_1+\dotsm +n_k)
\]
is obtained by applying elements from $G_d$ associated with disks $(D_i)_{1\leq i\leq n}$ in the configuration in $C_d(k)$ to configurations in all $C_d(n_i)$, $i=1,\ldots,k$ and putting the resulting configurations together. The action of the symmetric group $S_n$ on $C_d(n)$ is given by renumerations of indices of disks $(D_i)_{1\leq i\leq n}$ (see Figure \ref{fig:little_disks_composition}).
\end{defn}

\begin{figure}[ht]
\centering
\begin{tikzpicture}[scale=0.6]
\node[] (11) at (-1,-1) {$1$};
\node[] (21) at (0.5,1.8) {$2$};
\node[] (31) at (2.5,0) {$3$};
\node[] (22) at (12.5,1.8) {$3$};
\node[] (32) at (14.5,0) {$4$};
\node[] (32) at (10.5,-1.5) {$1$};
\node[] (32) at (11.5,0) {$2$};
\draw[] (0,0) circle (4cm);
\draw[] (-1,-1) circle (2cm);
\draw[] (0.5,1.8) circle (1cm);
\draw[] (2.5,0) circle (1.3cm);
\draw[] (12,0) circle (4cm);
\draw[dashed] (11,-1) circle (2cm);
\draw[] (10.5,-1.5) circle (0.5cm);
\draw[] (11.5,0) circle (0.8cm);
\draw[] (12.5,1.8) circle (1cm);
\draw[] (14.5,0) circle (1.3cm);
\draw[->] (5,0) -- (7,0);
\end{tikzpicture}
\caption{Example for the compositions of little disks.} 
\label{fig:little_disks_composition}
\end{figure}
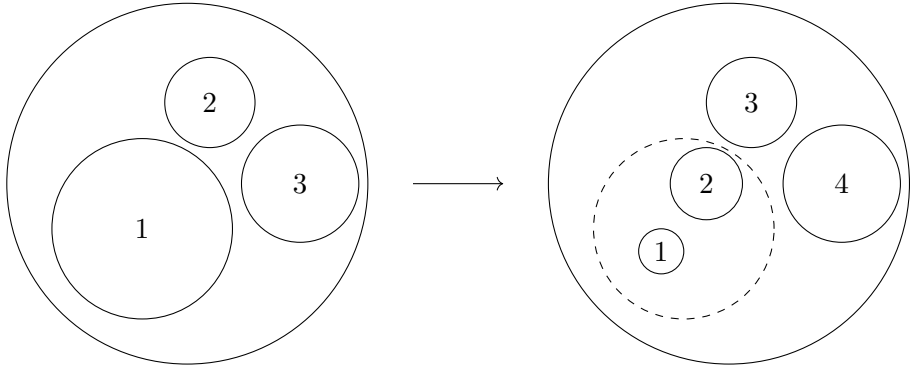

\begin{rem}
Note that $C_d(n)$ is homotopy equivalent to the configuration space $\mathrm{Conf}_n(\R^d)$ of $n$ pairwise distinct points in $\R^d$. Obviously, we can consider the map $C_d(n)\to \mathrm{Conf}_n(\R^d)$ which takes a collection of disjoint disks to the collection of the centers. Note that in particular, we have that $\mathrm{Conf}_2(\R^d)$ (and hence $C_2(n)$) is homotopy equivalent to the $(d-1)$-dimensional sphere $S^{d-1}$. The homotopy equivalence is given by the map
\[
(v_1,v_2)\mapsto \frac{v_1-v_2}{\| v_1-v_2\|}\in S^{d-1}\subset \R^d.
\]
\end{rem}

\begin{rem}
Since we are abusing notation, one should not confuse the operad $C_d$ with the configuration space $C_n$ as in Remark \ref{rem:conf}.
\end{rem}

\subsection{Deligne's conjecture}
Denote the Hochschild complex of an associative algebra $\calA$ concentrated in positive degree by
\[
\calC^\bullet(A,A):=\mathsf{Hom}_{\text{vector space}}(\calA^{\otimes n},\calA),\quad n\geq 0
\]
and the Hochschild differential by $\dd_\mathrm{H}$, given by the formula
\begin{multline*}
    (\dd_\mathrm{H}\phi)(a_1\otimes\dotsm\otimes a_{n+1}):=a_1\phi(a_2\otimes\dotsm\otimes a_n)+\sum_{1\leq i\leq n}(-1)^i\phi(a_1\otimes\dotsm\otimes a_ia_{i+1}\otimes\dotsm\otimes a_{n+1})+\dotsm\\\dotsm+(-1)^{n+1}\phi(a_1\otimes\dotsm\otimes a_n)a_{n+1},\quad \forall \phi\in \calC^n(\calA,\calA).
\end{multline*}

\begin{conj}[Deligne]
There exists a natural action of the operad $\mathsf{Chains}(C_2)$ on the Hochschild complex $\calC^\bullet(\calA,\calA)$ for any associative algebra $\calA$.
\end{conj}

\begin{rem}
The proof of this conjecture was given by a combination of results of \cite{Tamarkin1998,Tamarkin2003} and (a higher version) also in \cite{Kontsevich1999}. See also \cite{McClureSmith1999,KontsevichSoibelman2000}.
\end{rem}

\subsection{Formality of chain operads}
\begin{thm}[Kontsevich--Tamarkin\cite{Tamarkin2003,Kontsevich1999}]
The operad $\mathsf{Chains}(C_d)\otimes \R$ of complexes of real vector spaces is quasi-isomorphic to its cohomology operad endowed with the zero differential.
\end{thm}

\chapter{Quantum Field Theoretic Approach to Deformation Quantization}

Kontsevich's formula is in its nature a pure algebraic construction. However, the concept of deformation quantization should give a physical quantization procedure which opens the question whether it is possible to naturally extract the star product out of a quantum field theory, i.e. a perturbative \emph{Feynman path integral quantization} \cite{Feynman1942,Feynman1949,Feynman1950,FeynmanHibbs1965,P}. The appearance of graphs in Kontsevich's construction might already inspire to regard them as Feynman diagrams of a certain field theory. Indeed, using the \emph{Poisson sigma model} \cite{I,SS1,CF2}, one can show that its perturbative quantization produces the desired outcome. It is however necessary to formulate everything in a suitable formalism since it appears to be a \emph{gauge theory} (i.e. the theory has symmetries). It turns out that the Batalin--Vilkovisky (BV) formalism \cite{BV1,BV2,BV3} (see \cite{S} for a more mathematical formulation and \cite{Mnev2019} for a good introduction) is needed in order to deal with this particular sigma model. In this chapter we will describe the mathematical framework for functional integrals, construct the Moyal product out of a path integral quantization, describe the Faddeev--Popov \cite{FP} and BRST \cite{BRS1,BRS2,BRS3,Tyutin1976} method to deal with gauge theories and introduce the Poisson sigma model. Moreover, we will explicitly show how the perturbative quantization of the Poisson sigma model for linear Poisson structures produces Kontsevich's star product (this only requires the BRST formalism). Finally, we state the theorem for general Poisson structures on $\R^d$ (this construction needs the BV formalism) \cite{CF1,CF2}.
This chapter is based on \cite{Zinn-Justin1994,P,CKTB,CF1,CF2,Kontsevich1994,Mnev2019,Fedosov1996}. 

\section{Functional integrals}

\subsection{Functional integrals and expectation values}

For a field theory construction we want to consider a \emph{space of fields} $\calM$ which in most cases is given by the space of sections for some vector bundle and an \emph{action function} $S\colon \calM\to \R$. The action is usually given \emph{locally} (see also Definition \ref{defn:local}), i.e., roughly, as an integral over some Lagrangian density $S(\phi)=\int \mathscr{L}(\phi,\de\phi,\ldots,\de^N\phi)$ for $N\in\mathbb{Z}_{>0}$. 

\begin{defn}[Expectation value]
\label{defn:expectation}
For a function $\calO$ on $\calM$ we define the \emph{expectation value}
\[
\langle\calO\rangle:=\frac{\displaystyle{\int}_\calM \exp\left(\I S/\hbar\right)\calO}{\displaystyle{\int}_\calM \exp\left(\I S/\hbar\right)}.
\]
\end{defn}

\begin{defn}[Observables]
An \emph{observable} is a function $\calO$ on $\calM$ whose expectation value is well-defined.
\end{defn}

\begin{rem}
Integrals as in Definition \ref{defn:expectation} are called \emph{functional integrals} or \emph{path integrals}. One can consider a perturbative evaluation of such integrals by expanding $S$ around a nondegenerate critical point and defining the integral as a formal power series in $\hbar$ with coefficients given by Gaussian expectation values. If $S$ carries certain \emph{symmetries}, the critical points will always be degenerate. In this case we speak of a \emph{gauge theory}. There are different methods to deal with such theories such a the \emph{Faddeev--Popov ghost method} \cite{FP}, the \emph{BRST method} \cite{BRS1,BRS2,Tyutin1976} and the \emph{Batalin--Vilkovisky method} \cite{BV1,BV2,BV3}. Interestingly, it was shown that the field theoretic construction of Kontsevich's star product for general Poisson structures requires the Batalin--Vilkovisky formalism to deal with the gauge theory given by the Poisson sigma model \cite{CF1}.
\end{rem}

\subsection{Gaussian integrals}
Let $A$ be a positive-definite symmetric matrix on $\R^n$ which we assume to be endowed with the Lebesgue measure $\dd^nx$ and the Euclidean inner product $\langle\enspace,\enspace\rangle$ (note that $n$ has to be even). Then 
\[
I(\lambda):=\int_{\R^n}\exp\left(-\frac{\lambda}{2}\langle x,Ax\rangle\right)\dd^nx=\frac{(2\pi)^{\frac{n}{2}}}{\lambda^{\frac{n}{2}}}\frac{1}{\sqrt{\det A}},\quad \lambda>0.
\]
If we continue $I$ to the whole complex plane without the negative real axis, we get 
\[
I(-\I)=(2\pi)^{\frac{n}{2}}\exp\left(\frac{\I\pi n}{4}\right)\frac{1}{\sqrt{\det A}},\qquad I(\I)=(2\pi)^{\frac{n}{2}}\exp\left(\frac{-\I\pi n}{4}\right)\frac{1}{\sqrt{\det A}}.
\]
Thus, when $A$ is negative-definite, we can define the integral by 
\[
\int_{\R^n}\exp\left(\frac{\I}{2}\langle x,Ax\rangle\right)\dd^nx=\int_{\R^n}\exp\left(-\frac{\I}{2}(-\langle x,Ax\rangle)\right)\dd^nx=(2\pi)^{\frac{n}{2}}\exp\left(\frac{-\I\pi n}{4}\right)\frac{1}{\sqrt{\vert\det A\vert}}.
\]

For the case when $A$ is nondegenerate (not necessarily positive- or negative-definite), we get 
\[
\int_{\R^n}\exp\left(\frac{\I}{2}\langle x,Ax\rangle\right) \dd^nx=(2\pi)^{\frac{n}{2}}\exp\left(\frac{\I\pi\,\mathrm{sign}\,A}{4}\right)\frac{1}{\sqrt{\vert\det A\vert}},
\]
where $\mathrm{sign}\,A$ denotes the \emph{signature} of $A$.

We want to compute expectation values with respect to a Gaussian distribution. Let us denote such an expectation value by $\langle\enspace,\enspace\rangle_0$. Define first the generating function 
\begin{multline*}
Z(J):=\int_{\R^n}\exp\left(\frac{\I}{2}\langle x,Ax\rangle+\frac{\I}{2}\langle J,x\rangle\right) \dd^nx\\
=(2\pi)^{\frac{n}{2}}\exp\left(\frac{\I\pi\,\mathrm{sign}\,A}{4}\right)\frac{1}{\sqrt{\vert\det A\vert}}\exp\left(\frac{\I}{2}\langle J,A^{-1}J\rangle\right).
\end{multline*}
Then we get 
\begin{multline*}
\langle x^{i_1}\dotsm x^{i_k}\rangle_0=\frac{\displaystyle{\int}_{\R^n}\exp\left(\frac{\I}{2}\langle x,Ax\rangle\right) x^{i_1}\dotsm x^{i_k}\dd^nx}{\displaystyle{\int}_{\R^n}\exp\left(\frac{\I}{2}\langle x,Ax\rangle\right)\dd^nx}\\
=\frac{\frac{\de}{\de J^{i_1}}\dotsm \frac{\de}{\de J^{i_k}}Z(J)\vert_{J=0}}{Z(0)}=\frac{\de}{\de J^{i_1}}\dotsm \frac{\de}{\de J^{i_k}}\exp\left(\frac{\I}{2}\langle J,A^{-1}J\rangle\right)\bigg|_{J=0}.
\end{multline*}

\begin{rem}
Note that $\langle x^{i_1}\dotsm x^{i_k}\rangle_0$ vanishes if $k$ is odd and is a sum of products of matrix elements of the inverse of $A$ if $k$ is even. For example, if $k=2$ we have $\langle x^{i}x^j\rangle_0=\I(A^{-1})^{ij}$ and if $k=2s$, we get 
\[
\langle x^{i_1}\dotsm x^{i_{2s}}\rangle_0=\I^s\sum_{\sigma\in S_{2s}}\frac{1}{2^ss!}(A^{-1})^{i_{\sigma(1)}i_{\sigma(2)}}\dotsm (A^{-1})^{i_{\sigma(2s-1)}i_{\sigma(2s)}}.
\]
\end{rem}

\begin{thm}[Wick]
\label{thm:Wick}
Denote by $P(s)$ the set of \emph{pairings}, i.e. permutations $\sigma\in S_{2s}$ with the property that $\sigma(2i-1)<\sigma(2i)$ for $i=1,\ldots,s$ and $\sigma(1)<\sigma(3)<\dotsm <\sigma(2s-3)<\sigma(2s-1)$. Then we have 
\begin{equation}
\label{eq:Wick}
\langle x^{i_1}\dotsm x^{i_{2s}}\rangle_0=\I^s\sum_{\sigma\in P(s)}(A^{-1})^{i_{\sigma(1)}i_{\sigma(2)}}\dotsm (A^{-1})^{i_{\sigma(2s-1)}i_{\sigma(2s)}}.
\end{equation}
\end{thm}

\subsubsection{Infinite-dimensional case}

By \eqref{eq:Wick}, the infinite-dimensional case only makes sense if $A$ is invertible. However, usually $A$ will be a differential operator and thus $G:=A^{-1}$ will denote the distributional kernel of its inverse, i.e. its \emph{Green function}. So if e.g. $A$ is a differential operator on functions on some manifold $\Sigma$, we get 
\begin{multline}
\begin{split}
\label{eq:infinite-dimensional_Wick}
\langle\phi(x_1)\dotsm\phi(x_{2s})\rangle_0=\frac{\displaystyle{\int} \exp\left(\frac{\I}{2\hbar}\int_\Sigma \phi A\phi\right)\phi(x_1)\dotsm \phi(x_{2s})\mathscr{D}[\phi]}{\displaystyle{\int} \exp\left(\frac{\I}{2\hbar}\int_\Sigma \phi A\phi\right)\mathscr{D}[\phi]}\\
:=(\I\hbar)^s\sum_{\sigma \in P(s)}G\big(x_{\sigma(1)},x_{\sigma(2)}\big)\dotsm G\big(x_{\sigma(2s-1)},x_{\sigma(2s)}\big),
\end{split}
\end{multline}
where $\phi$ denotes a function on $\Sigma$, $\mathscr{D}[\phi]$ denotes the \emph{formal Lebesgue measure} on the space of functions and $x_1,\ldots,x_{2s}$ are distinct points in $\Sigma$.

\begin{rem}[Normal ordering]
Equation \eqref{eq:infinite-dimensional_Wick} is usually extended to the singular case when points coincide by restricting the sum to pairings with the property that $x_{\sigma(2i-1)}\not= x_{\sigma(2i)}$ for all $i$. This is called \emph{normal ordering} since it corresponds to the usual normal ordering in the operator formulation of Gaussian field theories.
\end{rem}

\begin{rem}[Propagator]
In Gaussian integrals all expectation values are given in terms of the expectation value of the quadratic monomials. These are usually called \emph{2-point functions} or \emph{propagators}. For the case of a Gaussian field theory defined by a differential operator, propagator is just another name for the Green function.
\end{rem}

\begin{rem}[Derivatives]
One can extend the definition of expectation values to derivatives of fields by linearity. In particular, for multi-indices $I_1,\ldots,I_{2s}$ we set
\begin{equation}
\label{eq:derivative_Green_functions}
    \langle\de_{I_1}\phi(x_1)\dotsm \de_{I_{2s}}\phi(x_{2s})\rangle_0=(\I\hbar)^s\frac{\de^{\vert I_1\vert}}{\de x_1^{I_1}}\dotsm \frac{\de^{\vert I_{2s}\vert}}{\de x_{2s}^{I_{2s}}}\sum_{\sigma\in P(s)}G\big(x_{\sigma(1)},x_{\sigma(2)}\big)\dotsm G\big(x_{\sigma(2s-1)},x_{\sigma(2s)}\big),
\end{equation}
where the derivatives on the right hand side are meant in the distributional sense.
\end{rem}

\subsection{Integration of Grassmann variables}
\label{subsec:Grassmann_integration}
Let $V$ be a vector space. Recall that the exterior algebra $\bigwedge V^*$ is regarded as the algebra of functions on the odd vector space $\Pi V$. Choosing a basis and orientation, we can identify $\bigwedge^{\mathrm{top}}V^*$ with $\R$. The composition of this isomorphism with the projection $\bigwedge V^*\to \bigwedge^\mathrm{top}V^*$ gives a map $\bigwedge V^*\to \R$ which will be denoted by $\int_{\Pi V}$ and we call it the \emph{integral on $\Pi V$}.

Let $B\in \End(V)$ and regard it as an element of $V^*\otimes V$ and thus as a function on 
\[
\Pi V^*\times \Pi V:= \Pi(V^*\oplus V).
\]
There is a natural identification of $\bigwedge^\mathrm{top}(V^*\oplus V)$ with $\R$. Thus, we have 
\[
\int_{\Pi V^*\times \Pi V}\exp(B)=\det B.
\]
If $B$ is nondegenerate, we define the expectation value of a function $f$ on $\Pi V^*\times \Pi V$ by 
\[
\langle f\rangle_0:=\frac{\displaystyle{\int}_{\Pi V^*\times \Pi V}\exp(B)f}{\displaystyle{\int}_{\Pi V^*\times \Pi V}\exp(B)}.
\]
Let $\{e_i\}$ be a basis of $V$ and denote by $\{\bar e^i\}$ the corresponding dual basis. Then $\bigwedge (V^*\oplus V)$ can be identified with the Grassmann algebra generated by the anticommuting \emph{coordinate functions} $\bar e^i$ and $e_j$. Functions on $\Pi V^*\times \Pi V$ are then linear combinations of monomials $e_{j_1}\dotsm e_{j_r}\bar e^{i_1}\dotsm \bar e^{i_s}$. To $B$ we associate the function 
\[
\langle \bar e,Be\rangle=\bar e^j B^i_j e_i,
\]
where $\langle\enspace,\enspace\rangle$ denotes the canonical pairing between $V^*$ and $V$. Hence, we can write 
\[
\int \exp\left(\langle \bar e,Be\rangle\right)=\det B,
\]
and 
\[
\langle e_{j_1}\dotsm e_{j_r}\bar e^{i_1}\dotsm \bar e^{i_s}\rangle_0=\frac{\displaystyle{\int} \exp\left(\langle \bar e,Be\rangle\right) e_{j_1}\dotsm e_{j_r}\bar e^{i_1}\dotsm \bar e^{i_s}}{\displaystyle{\int} \exp\left(\bar e^j B^i_j e_i\right)}.
\]
It is easy to see that $\langle e_{j_1}\dotsm e_{j_r}\bar e^{i_1}\dotsm \bar e^{i_s}\rangle_0$ vanishes if $r\not=s$ and that 
\[
\langle e_{j_1}\dotsm e_{j_s}\bar e^{i_1}\dotsm \bar e^{i_s}\rangle_0=\sum_{\sigma\in S_{s}}\mathrm{sign}(\sigma)(B^{-1})^{i_{\sigma(1)}}_{j_1}\dotsm (B^{-1})^{i_{\sigma(s)}}_{j_s}.
\]

\begin{rem}[Odd vector fields]
A vector field on $\Pi V$ is by definition a graded derivation of $\bigwedge V^*$. In particular, an endomorphism $X$ of $\bigwedge V^*$ is a vector field of degree $\vert X\vert$ if 
\[
X(fg)=X(f)g+(-1)^{\vert X\vert r}f X(g),\quad \forall f\in \bigwedge^r V^*,\forall g\in \bigwedge V^*,\forall r.
\]
A \emph{right} vector field $X$ of degree $\vert X\vert$ is an endomorphism $f\mapsto (f)X$ of $\bigwedge V^*$ that satisfies
\[
(fg)X=f(g)X+(-1)^{\vert X\vert s}(f)Xg,\quad \forall f\in \bigwedge V^*,\forall g\in \bigwedge^sV^*,\forall s.
\]
We can identify the vector space of all vector fields on $\Pi V$ with $\bigwedge V^*\otimes V$. Note that the elements of $V$ can be regarded as constant vector fields. Integration gives us 
\[
\int_{\Pi V}X(f)=0,\quad \forall f,
\]
if $X$ is a constant vector field. In general, one can define the \emph{divergence} $\Div X$ of $X$ by 
\[
\int_{\Pi V}X(f)=\int_{\Pi V}\Div X f,\quad \forall f.
\]
If $X=g\otimes v$ for $g\in\bigwedge^s V^*$ and $v\in V$, we get 
\[
\Div X=(-1)^{s+1}\iota_vg.
\]
\end{rem}

\section{The Moyal product as a path integral quantization}
\label{sec:Moyal_product_path_integral}
Consider the symplectic manifold $T^*\R^n$ endowed with the canonical symplectic form $\omega_0=\sum_{1\leq i\leq n}\dd p_i\land \dd q^i$, where $(q^1,\ldots,q^n,p_1,\ldots,p_n)\in T^*\R^n\cong \R^{2n}$. Denote by $\alpha=\sum_{1\leq i\leq n}p_i\dd q^i$ the Liouville 1-form such that $\omega_0=\dd\alpha$. Consider a path $\gamma\colon I\to T^*\R^n$, where $I$ is a 1-dimensional manifold, and define an action $S(\gamma)=\int_I\gamma^*\alpha$. Let us write $\gamma(t)=(Q(t),P(t))$ for $t\in I$. Then we can rewrite the action as
\begin{equation}
\label{eq:QM_action}
S(Q,P)=\int_{t\in I}P_i\frac{\dd}{\dd t}Q^i\dd t.
\end{equation}
Given a Hamiltonian function $H$, one can deform the action to 
\begin{equation}
\label{eq:deformed_action}
S_H(Q,P)=\int_{t\in I}\left(P_i\frac{\dd}{\dd t}Q^i+H(Q(t),P(t),t)\right)\dd t.
\end{equation}
Let now $I=S^1$ and, in order to make the quadratic form nondegenerate, we choose a basepoint $\infty\in S^1$. Let 
\[
\calM:=\{(Q,P)\in C^\infty(S^1,T^*\R^n)\}
\]
and 
\[
\calM(q,p):=\{(Q,P)\in C^\infty(S^1,T^*\R^n)\mid Q(\infty)=q,\, P(\infty)=p\}.
\]
The path integral can then be defined by imposing Fubini's theorem:
\begin{equation}
\label{eq:Fubini}
\int_\calM(\dotsm)=\int_{(q,p)\in T^*\R^n}\mu(q,p)\int_{\calM(q,p)}(\dotsm)
\end{equation}
where $\mu$ denotes a measure on $T^*\R^n$. One can check that the quadratic form in $S$ is nondegenerate when restricted to $\calM(q,p)$. Hence, we can compute
\[
\langle\calO\rangle_0(q,p):=\frac{\displaystyle{\int}_{\calM(q,p)}\exp\left(\I S/\hbar\right)\calO}{\displaystyle{\int}_{\calM(q,p)}\exp\left(\I S/\hbar\right)},
\]
where $\calO$ is some function on $\calM$ that is either polynomial or a formal power series in $Q$ and $P$. Using the fact that the denominator is constant (infinite though), we can use \eqref{eq:Fubini} to write 
\[
\langle \calO\rangle_0:=\frac{\displaystyle{\int}_{(q,p)\in T^*\R^n}\mu(q,p)\langle \calO\rangle_0(q,p)}{\displaystyle{\int}_{(q,p)\in T^*\R^n}\mu(q,p)}
\]
whenever the functions $\langle\calO\rangle_0(q,p)$ and $1$ are integrable. The latter condition prevents us from choosing $\mu$ to be the Liouville volume $\frac{\omega_0^n}{n!}$. However, at this point one can also forget about the denominator and define 
\[
\langle\calO\rangle_0':=\int_{(q,p)\in T^*\R^n}\mu(q,p)\langle\calO\rangle_0(q,p)
\]
such that the Liouville volume is allowed. In fact, we could also choose $\mu$ to be the delta measure peaked at a point $(q,p)\in T^*\R^n$. In this case we have
\[
\langle\calO\rangle_0=\langle\calO\rangle_0'=\langle\calO\rangle_0(q,p).
\]
If we have fixed $(q,p)\in T^*\R^n$, we can can use the \emph{change of variables} $Q=q+\tilde Q$ and $P=p+\tilde P$, where $(\tilde Q,\tilde P)$ is a map from $S^1$ to $T^*\R^n$ which vanishes at $\infty$. This is the same as a map $\R\to T^*\R^n$ which vanishes at $\infty$. The action is then given by 
\[
S(Q,P)=S(q+\tilde Q,p+\tilde P)=\int_\R \tilde P_i\frac{\dd}{\dd t}\tilde Q^i\dd t.
\]
Expectation values are given by 
\[
\langle \calO(Q,P)\rangle_0(q,p)=\left\langle\calO(q+\tilde Q,p+\tilde P)\right\rangle^\sim_0:=\frac{\displaystyle{\int} \exp\left(\I S/\hbar\right)\calO(q+\tilde Q,p+\tilde P)\mathscr{D}[\tilde P]\mathscr{D}[\tilde Q]}{\displaystyle{\int}\exp\left(\I S/\hbar\right)\mathscr{D}[\tilde P]\mathscr{D}[\tilde Q]}.
\]

\subsection{The propagator}

We need to find the Green function for the differential operator $\frac{\dd}{\dd t}$. It is given in terms of the \emph{sign function}:
\[
\left(\frac{\dd}{\dd t}\right)^{-1}(u,v)=\theta(u-v)=\frac{1}{2}\mathrm{sig}(u-v):=\begin{cases}\frac{1}{2},& u>v\\ -\frac{1}{2},& u<v\end{cases}
\]
As a consequence, we get 
\[
\left\langle \tilde P_i(u)\tilde Q^j(v)\right\rangle^\sim_0=\I\hbar \theta(v-u)\delta^j_i
\]
and more generally
\[
\left\langle \tilde P_{i_1}(u_1)\dotsm\tilde P_{i_s}(u_s)\tilde Q^{j_1}(v_1)\dotsm \tilde Q^{j_s}(v_s)\right\rangle^\sim_0=(\I\hbar)^s\sum_{\sigma\in S_s}\theta(v_{\sigma(1)}-u_1)\dotsm \theta(v_{\sigma(s)}-u_s)\delta^{j_{\sigma(1)}}_{i_1}\dotsm \delta^{j_{\sigma(s)}}_{i_s}.
\]
The normal ordering can be implemented by setting $\theta(0)=0$ which is compatible with the skew-symmetry of $\frac{\dd}{\dd t}$.

\subsection{Expectation values}

Let us consider the observable on $\calM$ which is given by evaluation of a smooth function $f$ on $T^*\R^n$ at some point $u$ in the path. Define
\[
\calO_{f,u}(Q,P):=f(Q(u),P(u)),\quad f\in C^\infty(T^*\R^n),\, u\in S^1\setminus \{\infty\}.
\]
in order to compute the expectation value of this observable, we need to introduce some notation. Let $I=(i_1,\ldots,i_r)$ be a multi-index and set $\vert I\vert=r$, $p_I:=p_{i_1}\dotsm p_{i_r}$, $q^I:=q^{i_1}\dotsm q^{i_r}$ (similarly for $\tilde P_I$ and $\tilde Q^I$). Moreover, define
\[
\de_I:=\frac{\de}{\de q^{i_1}}\dotsm \frac{\de}{\de q^{i_r}},\quad \de^I:=\frac{\de}{\de p_{i_1}}\dotsm \frac{\de}{\de p_{i_r}}.
\]
We extend everything to the case of $\vert I\vert =0$ by setting $p_I=q^I=1$ and $\de_I=\de^I=\mathrm{id}$. If we use the change of variables as before and use the Taylor expansion of $f$, we get
\[
f(Q(u),P(u))=f(q,\tilde Q(u),p+\tilde P(u))=\sum_{r,s\geq 0}\frac{1}{r!s!}\sum_{\vert I\vert=r\atop \vert J\vert=s}\tilde P_I(u)\tilde Q^J(u)\de^I\de_J f(q,p).
\]
Hence 
\[
\langle \calO_{f,u}\rangle_0(q,p)=\left\langle\calO_{f,u}(q+\tilde Q,p+\tilde P)\right\rangle_0^\sim=f(q,p).
\]
If $f$ is integrable we can also define
\[
\langle\calO_{f,u}\rangle_0'=\int_{T^*\R^n}\mu(q,p)f(q,p).
\]

\begin{rem}
Note that these expectation values do not depend on the point $u$.
\end{rem}

Consider now the observable given by 
\begin{equation}
\label{eq:Moyal_observable}
\calO_{f,g;u,v}=f(Q(u),P(u))g(Q(v),P(v)),\quad f,g\in C^\infty(T^*\R^n),\, u,v\in S^1\setminus\{\infty\}\cong \R,\, u<v.
\end{equation}
Its expectation value can then be computed as
\begin{multline*}
  \langle \calO_{f,g;u,v}\rangle_0(q,p)=\left\langle f(q+\tilde Q(u),p+\tilde P(u))g(q+\tilde Q(v),p+\tilde P(v)) \right\rangle^\sim_0\\
  =\sum_{r_1,s_1,r_2,s_2\geq 0}\frac{1}{r_1!s_1!r_2!s_2!}\sum_{\vert I_1\vert=r_1\atop \vert J_1\vert=s_1}\sum_{\vert I_2\vert=r_2\atop \vert J_2\vert=s_2}\left\langle \tilde P_{I_1}(u)\tilde Q^{J_1}(u)\tilde P_{I_2}(v)\tilde Q^{J_2}(v)\right\rangle^\sim_0\de^{I_1}\de_{J_1}f(q,p)\de^{I_2}\de_{J_2}g(q,p)\\
  =\sum_{r,s\geq 0}\frac{1}{r!s!}\left(\frac{\I\hbar}{2}\right)^{r+s}(-1)^s\sum_{\vert I\vert=r\atop \vert J\vert=s}\de^I\de_Jf(q,p)\de^J\de_Ig(q,p)=f\star g(q,p).
\end{multline*}
Here we have denoted by $\star$ the Moyal product induced by the Poisson bivector field $\frac{\de}{\de p_i}\land \frac{\de}{\de q^i}$.

\begin{rem}
Note also here that the expectation is independent of $u$ and $v$.
\end{rem}

If $f\star g$ is integrable, we can compute 
\[
\langle \calO_{f,g;u,v}\rangle_0'=\int_{T^*\R^n}\mu f\star g.
\]
If $\mu$ is given by the Liouville volume $\frac{\omega_0^n}{n!}$, we can define a trace (see also \cite{Fedosov1994,NestTsygan1995,Fedosov1996})
\[
\tr(f\star g):=\int_{T^*\R^n}\frac{\omega_0^n}{n!}f\star g=\int_{T^*\R^n}\frac{\omega_0^n}{n!}fg,
\]
where the second equality follows from using integration by parts in the correction terms to the commutative product (see also \cite{CattaneoFelder2010} for a similar trace formula for the Poisson manifold $\R^d$ endowed with any Poisson structure and \cite{Mosh1} for the formulation on general Poisson manifolds). 
More generally, we can define an observable (see Figure \ref{fig:cyc_Moyal})
\begin{multline*}
\calO_{f_1,\ldots,f_k;u_1,\ldots,u_k}=f_1(Q(u_1),P(u_1))\dotsm f_k(Q(u_k),P(u_k)),\\\ f_1,\ldots,f_k\in C^\infty(T^*\R^n),\, u_1,\ldots,u_k\in S^1\setminus \{\infty\}\cong\R,\, u_1<\dotsm <u_k.
\end{multline*}

\begin{exe}
Show that 
\[
\langle\calO_{f_1,\ldots,f_k;u_1,\ldots,u_k}\rangle_0(q,p)=f_1\star\dotsm \star f_k(q,p).
\]
\end{exe}

\begin{figure}[!ht]
\centering
\tikzset{
particle/.style={thick,draw=black},
particle2/.style={thick,draw=black, postaction={decorate},
    decoration={markings,mark=at position .9 with {\arrow[black]{triangle 45}}}},
gluon/.style={decorate, draw=black,
    decoration={coil,aspect=0}}
 }
\tikzset{Bullet/.style={fill=black,draw,color=#1,circle,minimum size=0.5pt,scale=0.5}}
\begin{tikzpicture}[x=0.04\textwidth, y=0.04\textwidth]
\draw[thick](-2,0) circle (2.5);
\node[Bullet=black, label=left: $u_1$] (u1) at (-4.3,-1) {};
\node[Bullet=black, label=below: $u_2$] (u2) at (-3.3,-2.1) {};
\node[Bullet=black, label=below: $u_3$] (u3) at (-1.5,-2.4) {};
\node[Bullet=black, label=right: $u_4$] (u4) at (0,-1.5) {};
\node[label=below: $\vdots$] (dots) at (1,1) {};
\node[Bullet=black, label=right: $u_k$] (uk) at (0.3,1) {};
\node[Bullet=black, label=above: $\infty$] (infty) at (-2,2.5) {};
\node[] (S1) at (-2,0) {$S^1$};
\end{tikzpicture}
\caption{The points $u_1,\ldots,u_k$ on $S^1$}
\label{fig:cyc_Moyal}
\end{figure}
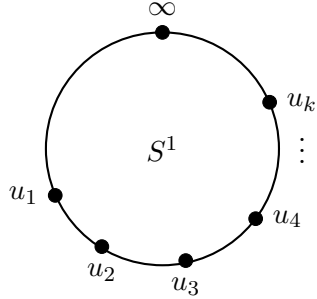

\subsection{Divergence of vector fields}

The normal ordering condition implies that local vector fields are in fact divergence free. Consider a vector at $(\tilde Q,\tilde P)\in \calM(q,p)$ which can be identified with a smooth map $\R\to T^*\R^n$ that vanishes at $\infty$. Note that a vector field $X$ on $\calM(q,p)$ is an assignment of a map $X(\tilde Q,\tilde P)$ to each path $(\tilde Q,\tilde P)$. The vector field is said to be \emph{local} if $X(\tilde Q,\tilde P)(t)$ is a function of $\tilde Q(t)$ and $\tilde P(t)$ for all $t\in\R$.

Formally, we want to express the divergence of a vector field in terms of path integrals similar to the description using ordinary integrals. We would like to have 
\[
\int_{\calM(q,p)}X\left(\exp\left(\I S/\hbar\right)\calO\right)=\int_{\calM(q,p)}\Div X \exp\left(\I S/\hbar\right)\calO,
\]
for all observables $\calO$.

\begin{defn}[Divergence]
If there exists an observable $\Div X$ such that for any observable $\calO$ we have 
\[
\langle X(S)\calO\rangle_0(q,p)=\I\hbar \langle X(\calO)\rangle_0(q,p)-\I\hbar \langle \Div X \calO\rangle_0(q,p),
\]
then we say that $\Div X$ is the \emph{divergence} of $X$.
\end{defn}

\begin{lem}
\label{lem:divergence}
If $X$ is a local vector field, then $\Div X=0$.
\end{lem}

\begin{proof}
Write $X(\tilde Q,\tilde P)(t)=(X_q(\tilde Q(t),\tilde P(t)),X_p(\tilde Q(t),\tilde P(t)))$. Then 
\[
X(S)=\int \left(X_p\frac{\dd}{\dd t}\tilde Q-X_q\frac{\dd}{\dd t}\tilde P\right)\dd t.
\]
By the normal ordering and the locality of $X$, when computing $\langle X(S)\calO\rangle_0(q,p)$ we only have to contract the $\tilde P$'s ($\tilde Q$'s) in $\calO$ with the $\tilde Q$'s ($\tilde P$'s) in $X(S)$ and replace each pair by a propagator. Whenever a $\tilde P$ ($\tilde Q$) in $\calO$ is contracted with the $\frac{\dd}{\dd t}\tilde Q$ ($\frac{\dd}{\dd t}\tilde P$) in $X(S)$, we get the identity operator times $\I\hbar$ ($-\I\hbar$). Summing up the various terms, we get that this is the same as taking the expectation value of $X(\calO)$ multiplied with $\I\hbar$. This completes the proof.
\end{proof}

We can also consider the divergence of vector fields on $\calM$. In fact, a local vector field $X$ on $\calM$ can be uniquely written as the sum of a vector field $X_\infty$ on $T^*\R^n$ and a section $\tilde X$ of local vector fields (i.e. an assignment of a local vector field $\tilde X(q,p)$ on $\calM(q,p)$ to each $(q,p)\in T^*\R^n$). By Lemma \ref{lem:divergence} we get 
\[
\langle X(S)\calO\rangle_0'=\I\hbar\langle X(\calO)\rangle_0'-\I\hbar\langle \Div_\mu X_\infty \calO\rangle_0',
\]
where $\Div_\mu X_\infty$ denotes the ordinary divergence of the vector field $X_\infty$ with respect to the measure $\mu$. One can then define $\Div_\mu X_\infty$ to be the \emph{divergence} of the local vector field $X$ on $\calM$. 

\begin{exe}
Using the local vector field $X(\tilde Q,\tilde P)(t):=(\tilde P(t),0)$, show that
\[
\hbar\frac{\dd}{\dd\hbar} (f\star g)=(p_i\de^i f)\star g+f\star (p_i\de^ig)-p_i\de^i(f\star g).
\]
\end{exe}

\subsection{Independence of evaluation points}

Note that in the previous computations we have considered observables whose definition depends on the choice of points in $S^1\setminus\{\infty\}$. However, computing the expectation values we have seen that they are in fact independent of these points. In particular, we have the following Proposition:
\begin{prop}
The expectation values are invariant under (pointed) diffeomorphism of the source manifold $S^1$ on which the field theory is defined.
\end{prop}

\begin{rem}
A quantum field theory with such a property is usually called \emph{topological quantum field theory (TQFT)}. 
\end{rem}
Given $f\in C^\infty(T^*\R^n)$ and two points $a,b\in \R$, we have 
\begin{multline*}
    f(Q(b),P(b))-f(Q(a),P(a))=\int_a^b\left(\de_i f(Q(t),P(t))\frac{\dd}{\dd t}Q^i(t)+\de^if(Q(t),P(t))\frac{\dd}{\dd t}P_i(t)\right)\dd t\\
    =\lim_{r\to\infty}\int_{-\infty}^{+\infty}\left(\de_if(Q(t),P(t))\frac{\dd}{\dd t}Q^i(t)+\de^if(Q(t),P(t))\frac{\dd}{\dd t}P_i(t)\right)\lambda_r(t)\dd t\\
    =\lim_{r\to\infty}\tilde X_{f,r}(S),
\end{multline*}

where $(\lambda_r)$ is a sequence of smooth, compactly supported functions that converges almost everywhere to the characteristic function of the interval $[a,b]$ and $\tilde X_{f,r}$ is the local vector field
\[
\tilde X_{f,r}(t)=(-\de^i f(Q(t),P(t)),\de_if(Q(t),P(t)))\lambda_r(t).
\]

\begin{rem}
Recall that $\calM$ is a space of maps. So a vector field on the target $T^*\R^n$ generates a local vector field on $\calM$. Let $\tilde X_f$ be the local vector field corresponding to the Hamiltonian vector field $X_f$ of $f$. Then $\tilde X_{f,r}$ is given by multiplying $\tilde X_f$ by $\lambda_r$.
\end{rem}

If we restrict $\tilde X_{f,r}$ to $\calM(q,p)$, we can observe by Lemma \ref{lem:divergence} that it is divergence-free. Hence, we get 
\[
\langle (f(Q(b),P(b))-f(Q(a),P(a)))\calO\rangle_0(q,p)=\lim_{r\to \infty}\left\langle \tilde X_{f,r}(S)\calO\right\rangle_0(q,p)=\I\hbar\lim_{r\to \infty}\left\langle\tilde X_{f,r}(\calO)\right\rangle_0(q,p),
\]
which vanishes if $\calO$ does not depend on $(\tilde Q(t),\tilde P(t))$ for $t\in[a,b]$. This means that we can move the evaluation point at least as long we do not meet another evaluation point.

\subsection{Associativity}

Let us prove the associativity of the Moyal product by using the path integral description considering the expectation value of the observable $\calO_{f,g,h;u,v,w}$. There will be three propagators appearing which correspond to the three different ways of pairing $u,v,w$. Note that the function $\theta$ will not see the difference, since 
\[
\theta(w-u)=\theta(w-v)=\theta(v-u)=\frac{1}{2}.
\]
We group the propagators by considering first only those between $u$ and $v$ and only thereafter the other ones. This will imply that 
\[
\langle\calO_{f,g,h;u,v,w}\rangle_0=(f\star g)\star h.
\]
If we group the propagators such that we first consider those between $v$ and $w$ and then the others, we get 
\[
\langle\calO_{f,g,h;u,v,w}\rangle_0=f\star(g\star h).
\]
This proves associativity of the Moyal product.

\begin{rem}
We can also prove this result by using the independence of evaluation points. This then implies 
\[
\lim_{v\to u^+}\langle \calO_{f,g,h;u,v,w}\rangle_0(q,p)=\lim_{v\to w^-}\langle \calO_{f,g,h;u,v,w}\rangle_0(q,p).
\]
The left hand side corresponds to evaluating first the expectation value of $\calO_{f,g;u,v}$ then placing the result at $u$ and computing the expectation value of $\calO_{\langle\calO_{f,g;u,v}\rangle_0,h;w}$. This will give the result $(f\star g)\star h$. Repeating this computation on the right hand side, we get associativity.
\end{rem}

\subsection{The evolution operator}

Let us consider the deformed action $S_H$ as in \eqref{eq:deformed_action}. We restrict ourselves to a time-independent Hamiltonian $h$ and let it act from the instant $a$ to the instant $b$ with $a<b$. Then we consider the action $S_H$ with 
\begin{equation}
\label{eq:Hamiltonian_function}
H(q,p,t)=h(q,p)\chi_{[a,b]}(t),
\end{equation}
where $\chi_{[a,b]}$ denotes the characteristic function on the interval $[a,b]$. The aim is to compute the \emph{evolution operator}
\[
U(q,p,T):=\frac{\displaystyle{\int}_{\calM(q,p)}\exp\left(\I S_H/\hbar\right)}{\displaystyle{\int}_{\calM(q,p)}\exp\left(\I S/\hbar\right)},
\]
where $T=b-a$. Observe that 
\begin{multline*}
U(q,p,T)=\left\langle \exp\left(\frac{\I}{\hbar}\int H(Q(t),P(t),t)\dd t\right)\right\rangle_0(q,p)\\
=\left\langle \exp\left(\frac{\I}{\hbar}\int_a^bh(Q(t),P(t),t)\dd t\right)\right\rangle_0(q,p).
\end{multline*}
Now we can express the integral in terms of Riemann sums as 
\[
\int_a^bh(Q(t),P(t),t)\dd t=\lim_{N\to\infty}\frac{T}{N}\sum_{1\leq r\leq N}h(Q(a+rT/N),P(a+rT/N)).
\]
Hence, we get 
\begin{multline*}
    U(q,p,T)=\lim_{N\to\infty}\left\langle \prod_{1\leq r\leq N}\exp\left(\frac{\I}{\hbar}\frac{T}{N}h(Q(a+rT/N),P(a+rT/N))\right)\right\rangle_0(q,p)\\
    =\lim_{N\to\infty}\left\langle\calO_{\exp\left(\frac{\I}{\hbar}\frac{T}{N}h\right),\ldots,\exp\left(\frac{\I}{\hbar}\frac{T}{N}h\right);a+\frac{T}{N},a+2\frac{T}{N},\ldots,a+T}\right\rangle_0(q,p)\\
    =\lim_{N\to\infty}\left(\exp\left(\frac{\I}{\hbar}\frac{T}{N}h\right)\right)^{\star N}(q,p)=\exp_\star\left(\frac{\I}{\hbar}h\right)(q,p).
\end{multline*}

\begin{rem}
Note that the previous result also includes negative powers of $\hbar$. However, everything is well-defined since each term in the power series expansion of $\exp_\star$ is actually a \emph{Laurent series} in $\hbar$.
\end{rem}

\subsection{Perturbative evaluation of integrals}

Let $\calM$ be an $n$-dimensional manifold and $S\in C^\infty(\calM)$. We want to understand how we can compute an integral of the form $Z:=\int_\calM \exp(\I S/\hbar)\dd^nx$.
Consider first the case where $S$ has a unique critical point $x_0\in\calM$ which is nondegnerate. Then we can write 
\[
S(x_0+\sqrt{\hbar}x)=S(x_0)+\frac{\hbar}{2}\dd_{x_0}^2S(x)+R_{x_0}(\sqrt{\hbar}x),
\]
where $R_{x_0}$ is a formal powers series given by the Taylor expansion of $S$ starting with the cubic term in $\sqrt{\hbar}x$ with $x\in T_{x_0}\calM$. Let $A$ be the \emph{Hessian} of $S$ at $x_0$ with respect to the Euclidean metric $\langle\enspace,\enspace\rangle$. In particular, we write 
\[
\dd_{x_0}^2S(x)=\langle x,Ax\rangle.
\]

\begin{thm}[Stationary phase expansion]
\label{thm:stationary_phase_expansion}
The \emph{stationary phase expansion} (or also \emph{saddle-point approximation}) of the integral $Z:=\int_\calM\exp(\I S/\hbar)\dd^nx$ is given by 
\begin{multline*}
    Z=\hbar^{\frac{n}{2}}\exp(\I S(x_0)/\hbar)\int_\calM \exp\left(\frac{\I}{2}\langle x,Ax\rangle\right)\sum_{r\geq 0}\frac{1}{r!}R^r_{x_0}(\sqrt{\hbar}x)\dd^nx\\
    =(2\pi\hbar)^\frac{n}{2}\exp(\I S(x_0)/\hbar)\exp\left(\frac{\I\pi\,\mathrm{sign}\, A}{4}\right)\frac{1}{\sqrt{\vert\det A\vert}}\sum_{r\geq 0}\frac{1}{r!}\left\langle R^r_{x_0}(\sqrt{\hbar}x)\right\rangle_0,
\end{multline*}
where $\langle\enspace\rangle_0$ denotes the Gaussian expectation value with respect to the nondegenerate symmetric matrix $A$.
\end{thm}

Expectation values as in Definition \ref{defn:expectation} are then given by
\begin{equation}
\label{eq:expectation_general}
\langle \calO\rangle=\frac{\displaystyle{\sum}_{r\geq 0}\frac{1}{r!}\left\langle \calO(x_0+\sqrt{\hbar}x)R_{x_0}^r(\sqrt{\hbar}x)\right\rangle_0}{\displaystyle{\sum}_{r\geq 0}\frac{1}{r!}\left\langle R^r_{x_0}(\sqrt{\hbar}x)\right\rangle_0}.
\end{equation}
\begin{rem}
Note that the denominator is of the form $1+O(\hbar)$ and thus the expectation value can be computed as a formal power series in $\hbar$. One can consider the computation in a graphical way by associating a vertex of valence $k$ to the term of degree $k$ in $R_{x_0}$ and a vertex of valence $\ell$ to the term of degree $\ell$ in $\calO$. By Wick's theorem (Theorem \ref{thm:Wick}) one has to connect these vertices in pairs in all possible ways the half-edges emanating from each vertex. The result is a collection of \emph{Feynman diagrams} with a weight associated to each of them. The expectation value is then a sum over all Feynman diagrams of their weights. A Feynman diagram with vertices coming only from $R_{x_0}$, i.e. a Feynman diagram appearing only in the denominator of \eqref{eq:expectation_general}, is called a \emph{vacuum diagram}.
A combinatorial fact of \eqref{eq:expectation_general} is that an expectation value is given by the sum over graphs which do not have a connected component which is a vacuum diagram (see \cite{Zinn-Justin1994,P} for more details on Feynman diagrams). 
\end{rem}

\subsection{Infinite dimensions}
Let us look at the situation of computing \eqref{eq:expectation_general} in the case when $\calM$ is infinite-dimensional. We want to recall the following definition.

\begin{defn}[Local function]
\label{defn:local}
A function on a space of fields on some manifold $M$ is \emph{local} if it is the integral on $M$ of a function that depends at each point on finite jets of fields at that point.
\end{defn}

\begin{rem}
If the action is a local function, $A$ will be a differential operator. For our purpose, we will need its Green function (the propagator) and in the Gaussian expectation values of $\calO R^r$ and $R^r$ we will use \eqref{eq:derivative_Green_functions}. Instead of summing over indices we will have to integrate over the Cartesian products of the source manifold. The normal ordering will exclude all the graphs with edges having the same endpoints (these edges are called \emph{tadpoles} or \emph{short loops}) and thus integration is restricted to the configuration space of the source manifold. However, this is usually not enough to make the integrals converge as the Green functions are very singular when the arguments approach each other. The general procedure to get rid of divergencies is called \emph{renormalization}. 
Anyway, for our case (i.e. the case of a TQFT), the configuration-space integrals associated to Feynman diagrams without tadpoles do indeed converge.
\end{rem}

\begin{rem}
If the action $S$ has more critical points and all of them are nondegenerate, the asymptotic expansion is obtained by computing the stationary phase expansion (Theorem \ref{thm:stationary_phase_expansion}) around each critical point and then summing all these terms together. The expression for expectation value is then no longer given by \eqref{eq:expectation_general}. It can happen that one of the critical points dominates the others so that one can forget them. In physical theories this is usually done by regarding $\exp(\I S/\hbar)$ as the analytic continuation of the exponential of minus a poisitive-definite function, called \emph{Euclidean action}, so that the dominating critical point is the absolute minimum. There are still examples where there is no way to do it as all critical points appear as saddle points. In such a case, one can consider another option which consists of selecting one particular critical point (called the \emph{sector}) and expanding only around this point. In this case, \eqref{eq:expectation_general} is again the correct expression for the infinite-dimensional extension. 
\end{rem}

As already mentioned, it can often happen that the critical points are degenerate. A simple case is given when critical points are parametrized by a finite-dimensional manifold $\calM_\mathrm{crit}$. Then, using Fubini's theorem for the finite-dimensional case, we can rewrite the integral similarly as in \eqref{eq:Fubini} as 
\[
\int_\calM(\dotsm)=\int_{x_0\in\calM_\mathrm{crit}}\mu(x_0)\int_{\calM(x_0)}(\dotsm)
\]
where $\int_{\calM(x_0)}$ denotes the asymptotic expansion of the integral in the complement to $T_{x_0}\calM_\mathrm{crit}$ of a formal neighborhood of $x_0$, while $\mu$ denotes a measure on $\calM_\mathrm{crit}$ (which is determined in the finite-dimensional case and has to be chosen in the infinite-dimensional case as part of the definition). If the Hessian is constant on $\calM_\mathrm{crit}$, we can express the expectation value as 
\[
\langle\calO\rangle=\frac{\displaystyle{\int}_{x_0\in\calM_\mathrm{crit}}\mu(x_0)\displaystyle{\sum}_{r\geq 0}\frac{1}{r!}\left\langle \calO(x_0+\sqrt{\hbar}x)R_{x_0}^r(\sqrt{\hbar}x)\right\rangle_0}{\displaystyle{\int}_{x_0\in\calM_\mathrm{crit}}\mu(x_0)\displaystyle{\sum}_{r\geq 0}\frac{1}{r!}\left\langle R_{x_0}^r(\sqrt{\hbar}x)\right\rangle_0},
\]
where $\langle\enspace\rangle_0(x_0)$ denotes the Gaussian expectation value which is computed by expanding around $x_0$ orthogonally to $T_{x_0}\calM_\mathrm{crit}$. A possible choice for the measure $\mu$ is the delta measure peaked at some point $x_0$. In this case, as we already did before, the expectation value will be denoted by $\langle\enspace\rangle(x_0)$. 

\subsection{A simple generalization}

Note that we have only considered the perturbative expansion with formal expansion parameter $\hbar$. It might happen that there is another small expansion parameter, or some coefficient in $S$ which is much smaller than $\hbar$ or in our setting is an element of $\hbar^2\R[\![\hbar]\!]$. In such a case, we can define the Gaussian part by using the quadratic $\hbar$-independent term of $S/\hbar$ and consider all the other terms as the perturbation $R$. It can also happen that the perturbation term $R$ contains quadratic and linear terms as well. In particular, we can have an action function of the form 
\begin{equation}
\label{eq:action_form}
S(y,z)=\langle y,Bz\rangle+f(y,z),
\end{equation}
where $B$ is a nondegenerate matrix and $f$ is a function quadratic in $z$. If we work around the critical point $y=z=0$, we can rescale $z$ by $\hbar$ to get 
\[
S(y,\hbar z)=\hbar\langle y,Bz\rangle+\hbar f(y,z)
\]
and consider $f$ as the perturbation to the Gaussian theory defined by $B$.\\

\subsubsection{Quantum mechanics}

Recall that the topological action \eqref{eq:QM_action} defines the Moyal product in terms of the expectation value of the observable \eqref{eq:Moyal_observable}, which is independent of the evaluation points. Consider a Hamiltonian $H$ as in \eqref{eq:Hamiltonian_function} and define the expectation 
\[
f\hat{\star}_{a,b;u,v}g(q,p):=\frac{\displaystyle{\int}_{\calM(q,p)}\exp(\I S_H)/\hbar)\calO_{f,g;u,v}}{\displaystyle{\int}_{\calM(q,p)}\exp(\I S_H/\hbar)}=\frac{\left\langle \exp\left(\frac{\I}{\hbar}\displaystyle{\int}_a^b h\dd t\right)\calO_{f,g;u,v}\right\rangle_0(q,p)}{\left\langle \exp\left(\frac{\I}{\hbar}\displaystyle{\int}_{a}^b h\dd t\right)\right\rangle_0(q,p)}
\]
which will no longer be independent to $u$ and $v$ ($a<u<v<b$) nor will it define an associative product. By the computations before, we get 
\[
f\hat{\star}_{a,b;u,v}g=\frac{\exp_\star\left(\frac{\I}{\hbar}(u-a)h\right)\star f\star \exp_\star\left(\frac{\I}{\hbar}(v-u)h\right)\star g\star \exp_\star\left(\frac{\I}{\hbar}(b-v)h\right)}{\exp_\star\left(\frac{\I}{\hbar}(b-a)h\right)}
\]
Let us consider the perturbative computation for the Hamiltonian
\[
h(q,p)=\frac{1}{2}G^{ij}(q)p_ip_j,
\]
at $p=0$ and for $f$ and $g$ only depending on $q$. By setting $Q(t)=q+\tilde Q(t)$ and $P(t)=p+\tilde P(t)$, we get 
\[
S_H=\hbar \int_{-\infty}^{+\infty}\tilde P(t)\frac{\dd}{\dd t}\tilde Q(t)\dd t+\frac{\hbar^2}{2}\int_a^bG^{ij}(q+\tilde Q(t))\tilde P_i(t)\tilde P_j(t)\dd t.
\]
Hence, $f\hat{\star}_{a,b;u,v}g(q,0)$ is the ratio of the expectation value of $\exp\left(\frac{\I\hbar}{2}\int_a^b G^{ij}(q+\tilde Q)\tilde P_i\tilde P_j\dd t\right)\calO_{f,g;u,v}$ and the expectation value of $\exp\left(\frac{\I\hbar}{2}\int_a^b G^{ij}(q+\tilde Q)\tilde P_i\tilde P_j\dd t\right)$. The propagator pairs the $\tilde P$'s to $\tilde Q$'s, so it is better to consider oriented graphs. The Feynman diagrams will be oriented graphs with 
\begin{enumerate}
    \item a vertex at $u$ with no outgoing arrows,
    \item a vertex at $v$ with no outgoing arrows,
    \item vertices between $a$ and $b$ with exactly two outgoing arrows.
\end{enumerate}

Moreover, graphs containing tadpoles or vacuum subgraphs will not be allowed. See Figures \ref{fig:Feynman_diagrams_order_one}, \ref{fig:Feynman_diagrams_order_two} and \ref{fig:Feynman_diagrams_not_allowed} for examples of allowed and not allowed Feynman diagrams. Note that to a vertex of first type with $r$ incoming arrows we associate an $r$-th derivative  of $f$, to a vertex of second type with $r$ incoming arrows we associate an $r$-th derivative of $g$ and to a vertex of third type with $r$ incoming arrows we associate an $r$-th derivative of $G$. The order in $\hbar$ is given by the number of vertices of the third type. So, at order zero, we have $f$ placed at $u$ and $g$ placed at $v$. This yields the pointwise product of $f$ and $g$. At order $1$, we have the graphs of Figure \ref{fig:Feynman_diagrams_order_one}. Hence, we get 
\begin{multline*}
   f\hat{\star}_{a,b;u,v}g(q,0)=f(q)g(q)-\frac{\I\hbar}{4}(b-a+2u-2v)G^{ij}(q)\de_i f(g)\de_j g(q)\\
   -\frac{\I\hbar}{4}(b-a)G^{ij}(q)[\de_i\de_jf(q)g(q)+f(q)\de_i\de_jg(q)]+O(\hbar^2).
\end{multline*}

\begin{exe}
Compute some of the further orders.
\end{exe}

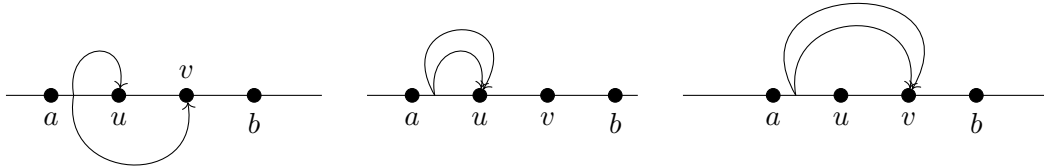
\begin{figure}[ht]
\centering
\begin{tikzpicture}[scale=0.6]
\tikzset{Bullet/.style={fill=black,draw,color=#1,circle,minimum size=0.5pt,scale=0.5}}
\node[Bullet=black, label=below: $a$] (a1) at (1,0) {};
\node[Bullet=black, label=below: $u$] (u1) at (2.5,0) {};
\node[Bullet=black, label=above: $v$] (v1) at (4,0) {};
\node[Bullet=black, label=below: $b$] (b1) at (5.5,0) {};
\node[Bullet=black, label=below: $a$] (a2) at (9,0) {};
\node[Bullet=black, label=below: $u$] (u2) at (10.5,0) {};
\node[Bullet=black, label=below: $v$] (v2) at (12,0) {};
\node[Bullet=black, label=below: $b$] (b2) at (13.5,0) {};
\node[Bullet=black, label=below: $a$] (a3) at (17,0) {};
\node[Bullet=black, label=below: $u$] (u3) at (18.5,0) {};
\node[Bullet=black, label=below: $v$] (v3) at (20,0) {};
\node[Bullet=black, label=below: $b$] (b3) at (21.5,0) {};
\draw (0,0) -- (7,0);
\draw (8,0) -- (14,0);
\draw (15,0) -- (23,0);
\draw[->,looseness=3] (1.5,0) [out=100,in=80] to (u1);
\draw[->,looseness=2] (1.5,0) [out=-100,in=-80] to (v1);
\draw[->,looseness=5] (9.5,0) [out=120,in=60] to (u2);
\draw[->,looseness=3] (9.5,0) [out=100,in=80] to (u2);
\draw[->,looseness=2] (17.5,0) [out=100,in=80] to (v3);
\draw[->,looseness=3] (17.5,0) [out=120,in=60] to (v3);
\end{tikzpicture}
\caption{Feynman diagrams of order one.} 
\label{fig:Feynman_diagrams_order_one}
\end{figure}

\begin{figure}[ht!]
\centering
\begin{tikzpicture}[scale=0.6]
\tikzset{Bullet/.style={fill=black,draw,color=#1,circle,minimum size=0.5pt,scale=0.5}}
\node[Bullet=black, label=below: $a$] (a1) at (1,0) {};
\node[Bullet=black, label=below: $u$] (u1) at (2.5,0) {};
\node[Bullet=black, label=above: $v$] (v1) at (4,0) {};
\node[Bullet=black, label=below: $b$] (b1) at (5.5,0) {};
\node[Bullet=black, label=below: $a$] (a2) at (9,0) {};
\node[Bullet=black, label=below: $u$] (u2) at (10.5,0) {};
\node[Bullet=black, label=below: $v$] (v2) at (12,0) {};
\node[Bullet=black, label=below: $b$] (b2) at (13.5,0) {};
\node[Bullet=black, label=below: $a$] (a3) at (17,0) {};
\node[Bullet=black, label=below: $u$] (u3) at (18.5,0) {};
\node[Bullet=black, label=below: $v$] (v3) at (20,0) {};
\node[Bullet=black, label=below: $b$] (b3) at (21.5,0) {};
\draw (0,0) -- (7,0);
\draw (8,0) -- (14,0);
\draw (15,0) -- (23,0);
\draw[->,looseness=3] (1.5,0) [out=100,in=80] to (u1);
\draw[->,looseness=2] (1.5,0) [out=-100,in=-80] to (v1);
\draw[->,looseness=5] (11,0) [out=90,in=130] to (u2);
\draw[->,looseness=3] (11,0) [out=90,in=80] to (v2);
\draw[->,looseness=3] (17.5,0) [out=-120,in=-60] to (v3);
\draw[->, looseness=3] (4.5,0) [out=90,in=60] to (u1);
\draw[->, looseness=1.2] (4.5,0) [out=-100,in=-60] to (1.5,0);
\draw[->,looseness=3] (13,0) [out=90,in=130] to (v2);
\draw[->,looseness=2.5] (13,0) [out=-90,in=-130] to (11,0);
\draw[->,looseness=3] (17.5,0) [out=100,in=120] to (u3);
\draw[->,looseness=4] (19.5,0) [out=90,in=60] to (u3);
\draw[->,looseness=4] (19.5,0) [out=90,in=120] to (v3);
\end{tikzpicture}
\caption{Feynman diagrams of order two.} 
\label{fig:Feynman_diagrams_order_two}
\end{figure}
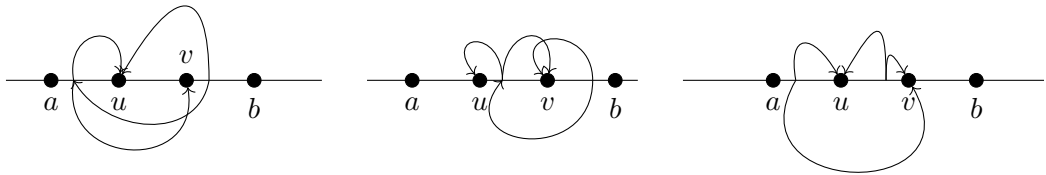

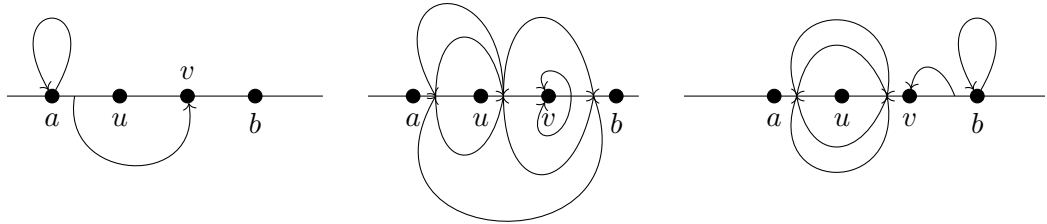
\begin{figure}[ht!]
\centering
\begin{tikzpicture}[scale=0.6]
\tikzset{Bullet/.style={fill=black,draw,color=#1,circle,minimum size=0.5pt,scale=0.5}}
\node[Bullet=black, label=below: $a$] (a1) at (1,0) {};
\node[Bullet=black, label=below: $u$] (u1) at (2.5,0) {};
\node[Bullet=black, label=above: $v$] (v1) at (4,0) {};
\node[Bullet=black, label=below: $b$] (b1) at (5.5,0) {};
\node[Bullet=black, label=below: $a$] (a2) at (9,0) {};
\node[Bullet=black, label=below: $u$] (u2) at (10.5,0) {};
\node[Bullet=black, label=below: $v$] (v2) at (12,0) {};
\node[Bullet=black, label=below: $b$] (b2) at (13.5,0) {};
\node[Bullet=black, label=below: $a$] (a3) at (17,0) {};
\node[Bullet=black, label=below: $u$] (u3) at (18.5,0) {};
\node[Bullet=black, label=below: $v$] (v3) at (20,0) {};
\node[Bullet=black, label=below: $b$] (b3) at (21.5,0) {};
\draw (0,0) -- (7,0);
\draw (8,0) -- (14,0);
\draw (15,0) -- (23,0);
\path[->,every loop/.style={looseness=40}] (a1) edge [in=120,out=60,loop] (a1); 
\path[->,every loop/.style={looseness=40}] (b3) edge [in=120,out=60,loop] (b3); 
\draw[->,looseness=2] (1.5,0) [out=-100,in=-80] to (v1);
\draw[->,looseness=5] (11,0) [out=90,in=120] to (9.5,0);
\draw[->,looseness=3] (9.5,0) [out=90,in=100] to (11,0);
\draw[->,looseness=3] (9.5,0) [out=-90,in=-100] to (11,0);
\draw[->,looseness=3] (11,0) [out=90,in=100] to (13,0);
\draw[->,looseness=3] (11,0) [out=-90,in=-100] to (13,0);
\draw[->,looseness=3] (13,0) [out=-70,in=-120] to (9.5,0);
\draw[->,looseness=3] (12.5,0) [out=90,in=120] to (v2);
\draw[->,looseness=5] (12.5,0) [out=-90,in=-120] to (v2);
\draw[->,looseness=3] (17.5,0) [out=110,in=80] to (19.5,0);
\draw[->,looseness=3] (17.5,0) [out=-110,in=-80] to (19.5,0);
\draw[->,looseness=2] (19.5,0) [out=110,in=80] to (17.5,0);
\draw[->,looseness=2] (19.5,0) [out=-110,in=-80] to (17.5,0);
\draw[->,looseness=2] (21,0) [out=110,in=80] to (v3);
\end{tikzpicture}
\caption{Feynman diagrams which are not allowed: tadpoles, vacuum subdiagrams and both.} 
\label{fig:Feynman_diagrams_not_allowed}
\end{figure}

\section{Symmetries and the BRST formalism}

If the action $S$ has symmetries encoded in a free action of a Lie group $G$ (gauge theory with gauge group $G$), the critical points will never be nondegenerate since they always will appear in $G$-orbits. This is also the case for the Poisson sigma model for linear Poisson structures and hence we need to understand how to deal with these type of theories. 

\subsection{The main construction: Faddeev--Popov ghost method}

Let $\calM$ be a finite-dimensional manifold together with a measure $\mu$. Let $G$ be a compact Lie group which is endowed with an invariant measure (Haar measure) with a measure-preserving free action on $\calM$. Moreover, we assume that $\calM/G$ is a manifold. Then an invariant function $f$ is the pullback of a function $\underline{f}\in\calM/G$ and 
\[
I:=\frac{1}{\mathrm{Vol}(G)}\int_\calM f\mu=\int_{\calM/G}\underline{f}\underline{\mu},
\]
where $\underline{\mu}$ is the measure induced by $\mu$ on the quotient.
If we consider a section of a principal bundle $\pi\colon \calM\to\calM/G$, we can rewrite $I$ as an integral on the image of this section. Assume that this integral is locally given by the zero set of a function $F\colon \calM'\to \mathfrak{g}$ where $\calM'\subset \calM$ is such that $\pi(\calM')=\calM/G$ and $\mathfrak{g}$ denotes the Lie algebra of $G$. In the physics literature the condition $F=0$ is called \emph{gauge fixing} and $F$ is called \emph{gauge-fixing function}. For $x\in\calM'$, let $A(x)$ be the differential $\dd F(x)$ restricted to the vertical tangent space at $x$. This can be identified with the Lie algebra $\mathfrak{g}$ and thus we can regard $A(x)$ as an endomorphism of $\mathfrak{g}$ and denote by $J$ its determinant. In the physics literature $J$ is called \emph{Faddeev--Popov determinant}. Then we get
\[
I=\int_{\calM'}f\delta_0(F)J\mu,
\]
where $\delta_0$ denotes the delta function at $0\in\mathfrak{g}$. Let us now rewrite $I$ such that it is suitable for the stationary phase expansion formula. We need to write $\delta_0(F)$ and $J$ in exponential form. Therefore, we use the Fourier transform of the delta function and the Grassmann (odd) integration as in \ref{subsec:Grassmann_integration}. Denote by $\langle\enspace,\enspace\rangle$ the canonical pairing between $\mathfrak{g}^*$ and $\mathfrak{g}$. Then we have 
\[
I=C\int f(x)\mu(x)\exp\left(\frac{\I}{\hbar}\langle \lambda,F(x)\rangle\right)\omega(\lambda)\exp\left(\frac{\I}{\hbar}\langle \bar c,A(x)c\rangle\right),
\]
where the integral is over $x\in\calM'$, $\lambda\in\mathfrak{g}^*$, $c\in\Pi\mathfrak{g}$, $\bar c\in\Pi\mathfrak{g}^*$ and where $C$ is a constant depending on $\hbar$ and on the choice of the top form $\omega\in \bigwedge^{\mathrm{top}}\mathfrak{g}^*$. Note that we have added the prefactors $\frac{\I}{\hbar}$ for later purposes.

\begin{rem}
In the physics literature $c$ is called \emph{ghost} and $\bar c$ is called \emph{antighost}. $\lambda$ is usually called \emph{Lagrange multiplier}. 
\end{rem}

Note that if we choose a basis $\{c^i\}$ of $\mathfrak{g}^*$, then $\bigwedge\mathfrak{g}^*$ can be identified with the algebra generated by the $c^i$s with the relations
\[
c^ic^j=-c^jc^i,\qquad \forall i,j.
\]
The generators $c^i\in \mathfrak{g}^*$ are often called \emph{ghost variables}. Similarly, one can introduce \emph{antighost variables} $\bar c_i\in \mathfrak{g}$.

\begin{rem}
The determinant of $A$ can also be obtained in terms of an ordinary Gaussian integral over $\mathfrak{g}\times \mathfrak{g}^*$. However, then one needs to put $A^{-1}$ into the exponential. In the infinite-dimensional case we may use the techniques of Feynman diagrams only when dealing with local functions. If the action of the group is local, $A$ will be a differential operator, so $\langle \bar c, Ac\rangle$ will be a local function, in contrary to the quadratic function defined in terms of the Green function $A^{-1}$.
\end{rem}

The expectation value of an invariant function $g$ with respect to a given invariant action $S$ can be written as 
\begin{equation}
\label{eq:FP_expectation}
\langle g\rangle=\frac{\displaystyle{\int}_\calM \exp(\I S/\hbar)g\mu}{\displaystyle{\int}_\calM \exp(\I S/\hbar)\mu}=\frac{\displaystyle{\int}_{\widetilde{\calM}}\exp(\I S_F/\hbar)g\widetilde{\mu}}{\displaystyle{\int}_{\widetilde{\calM}}\exp(\I S/\hbar)\widetilde{\mu}}=:\langle g\rangle_F,
\end{equation}
where 
\begin{align*}
    \widetilde{\calM}&:=\calM'\times \Pi\mathfrak{g}\times\mathfrak{g}^*\times \Pi\mathfrak{g}^*,\\
    S_F&:= S+\langle\lambda,F\rangle-\langle \bar c,A,c\rangle,
\end{align*}
and $\widetilde{\mu}:=\mu\omega$. The function $S_F$ is usually called the \emph{gauge-fixed action}.

\subsection{The BRST formalism}
Note that, by construction, the right hand side of \eqref{eq:FP_expectation} is independent of $F$. Moreover, the assumptions we had are quite restrictive, namely, we want to have a measure-preserving action on $\calM$ by a compact Lie group $G$ and we have to assume that the principal bundle $\calM\to \calM/G$ is trivial. On the other hand, to define 
\begin{equation}
\label{eq:expectation_BRST}
\langle g\rangle_F=\frac{\displaystyle{\int}_{\widetilde{\calM}}\exp(\I S_F/\hbar)g\widetilde{\mu}}{\displaystyle{\int}_{\widetilde{\calM}}\exp(\I S/\hbar)\widetilde{\mu}}
\end{equation}
we only need the infinitesimal action 
\begin{align*}
    X\colon \mathfrak{g}&\to \mathfrak{X}(\calM)\\
    \gamma&\mapsto X_\gamma
\end{align*}
of a Lie algebra $\mathfrak{g}$ on $\calM$. In this case $A$ is simply given by
\[
A_\gamma=L_{X_\gamma}F,\quad \gamma\in\mathfrak{g}.
\]

Moreover, we want to relax the condition that $F^{-1}(0)$ defines a section an rather require that $A(x)$ should be nondegenerate for all $x\in\calM'$. Of course we also require that the integrals are well-defined and that the denominator of \eqref{eq:FP_expectation} does not vanish.

Let $\calF\subset C^\infty(\calM',\mathfrak{g})$ be the space of allowed gauge-fixing functions. As in this case $\langle g\rangle_F$ is well-defined, it makes sense to consider it at more general instances. Observe also that if the Lie group is not compact, an invariant function is not a test function. It is then understood that it is replaced by a test function that in a neighborhood of the zeros of $F$ coincides with the given function. To avoid cumbersome notation, we will never explicitly change the function, so we will write e.g. $\int_\R\delta_0(x)\dd x=1$, without mentioning that the constant function 1 is replaced by a test function which is one in a neighborhood of zero. However, since the gauge-fixing function $F$ is arbitrarily chosen, we want the conditions on $\langle g\rangle_F$ to be locally independent of $F$.

\begin{defn}[Gauge-fixing independence]
A locally constant function on $\calF$ is called \emph{gauge-fixing independent}.
\end{defn}

\begin{thm}
\label{thm:gauge-fixing_independence}
Let $X\colon \mathfrak{g}\to\calM$ be an infinitesimal action of the Lie algebra $\mathfrak{g}$ on the manifold $\calM$. If $S$ and $g$ are invariant functions and 
\begin{equation}
\label{eq:divergence-free_condition}
\Div X_\gamma+\tr \mathrm{ad}_\gamma=0,\qquad \forall \gamma\in \mathfrak{g},
\end{equation}
then $\langle g\rangle_F$ is gauge-fixing independent.
\end{thm}

The condition in Theorem \ref{thm:gauge-fixing_independence} basically tells us that the divergence of $X$ is a constant function. In particular, for the case when the Lie algebra is \emph{unimodular} (i.e. $\tr \mathrm{ad}_\gamma=0$ for all $\gamma$), the condition says that the infinitesimal action must be measure-preserving. The discussions before are covered by Theorem \ref{thm:gauge-fixing_independence} since the Lie algebra of a compact Lie group is indeed unimodular.\\

\subsubsection{The BRST operator and the proof of Theorem \ref{thm:gauge-fixing_independence}}

The infinitesimal action of $\mathfrak{g}$ on $\calM$ gives $C^\infty(\calM)$ the structure of a $\mathfrak{g}$-module. Therefore, we can consider the Lie algebra complex $\bigwedge\mathfrak{g}^*\otimes C^\infty(\calM)$. Note that the Lie algebra differential $\delta$ is in particular a derivation on $\bigwedge\mathfrak{g}^*\otimes C^\infty(\calM)$ and moreover defines a vector field on $\calM\times\Pi\mathfrak{g}$. Recall that a vector field on the superspace $\Pi\mathfrak{g}$ is an element of $\bigwedge\mathfrak{g}\otimes \mathfrak{g}$. If $\mathfrak{g}$ is a Lie algebra, the commutator can be regarded as an element of $\bigwedge^2\mathfrak{g}^*\otimes \mathfrak{g}$ and as such it is a vector field on $\Pi\mathfrak{g}$. Vector fields on $\calM\times \Pi\mathfrak{g}$ can then be identified with elements of $\bigwedge\mathfrak{g}^*\otimes \mathfrak{g}\otimes C^\infty(\calM)\oplus \bigwedge\mathfrak{g}^*\otimes\mathfrak{X}(\calM)$. The commutator tensor the constant function 1 is an element of $\bigwedge^2\mathfrak{g}^*\otimes \mathfrak{g}\otimes C^\infty(\calM)$ while the infinitesimal action of $\mathfrak{g}$ on $\calM$ is an element of $\mathfrak{g}^*\otimes \mathfrak{X}(\calM)$. As such they define vector fields on $\calM\times \Pi\mathfrak{g}$ and $\delta$ is their sum. If we choose a basis $\{e_i\}$ of $\mathfrak{g}$ and denote by $\{f^k_{ij}\}$ the corresponding structure constants, the algebra of functions on $\Pi\mathfrak{g}$ can be identified with the graded commutative algebra with odd generators $c^i$ (ghost variables) and 
\[
\delta c^k=-\frac{1}{2}f^k_{ij}c^ic^j.
\]
On functions $f\in C^\infty(\calM)$, we get that $\delta$ acts by 
\[
\delta f=c^i L_{X_{e_i}}f.
\]
The vector field $\delta$ can be extended to $\widetilde{\calM}$ by adding it to the vector field on $\mathfrak{g}^*\otimes \Pi\mathfrak{g}^*$. Moreover, using the dual basis $\{e^i\}$ of $\mathfrak{g}^*$, the algebra of functions can be identified with the graded commutative algebra with odd generators $\bar c_i$ and even generators $\lambda_i$. There we define
\begin{equation}
\label{eq:FP_cohomological_VF}
\delta \bar c_i=\lambda_i,\qquad \delta\lambda_i=0,\quad \forall i.
\end{equation}
More abstractly, observe that a polynomial vector field on $\mathfrak{g}^*\times \Pi\mathfrak{g}^*$ is an element of $\bigwedge\mathfrak{g}\otimes \Sym(\mathfrak{g})\otimes(\mathfrak{g}^*\oplus\mathfrak{g}^*)$. The identity operator can then be regarded as an element of $\mathfrak{g}\otimes \mathfrak{g}^*$. Using the inclusion map
\begin{align*}
    i\colon \mathfrak{g}\otimes\mathfrak{g}^*&\to \bigwedge\mathfrak{g}\otimes \Sym(\mathfrak{g})\otimes (\mathfrak{g}^*\oplus\mathfrak{g}^*),\\
    a\otimes b&\mapsto 1\otimes a\otimes (b\oplus 0),
\end{align*}
we can regard it as a vector field on $\mathfrak{g}^*\times \Pi\mathfrak{g}^*$. Moreover, note that this vector field corresponds to the de Rham differential on $\Pi\mathfrak{g}^*$.

\begin{lem}
\label{lem:BRST_differential}
The map $\delta$ is a differential (i.e. an odd derivation that squares to zero) on $\bigwedge\mathfrak{g}^*\otimes C^\infty(\calM)\otimes \bigwedge\mathfrak{g}\otimes \Sym(\mathfrak{g})$. It has degree one with respect to the grading 
\[
\deg(\alpha\otimes f\otimes \beta\otimes \gamma):=\deg(\alpha)-\deg(\beta).
\]
\end{lem}

In fact, we can rewrite Lemma \ref{lem:BRST_differential} by saying that $\delta$ is a cohomological vector field on $\widetilde{\calM}$. In the physics literature, $\delta$ is called the \emph{BRST operator}. A function on $\calM$ can also be considered as a function on $\widetilde{\calM}$. By definition of $\delta$, a function $f$ is invariant if and only $\delta f=0$. 
\begin{exe}
\label{exe:divergence_BRST_operator}
Show that 
\[
\Div \delta=\Div X_c+\tr \mathrm{ad}_c.
\]
\end{exe}
Note that gauge-fixing function $F\colon \calM\to \mathfrak{g}$ can be regarded as an element of $C^\infty(\calM)]\otimes \mathfrak{g}$. Using the inclusion
\[
C^\infty(\calM)\otimes\mathfrak{g}\hookrightarrow C^\infty(\calM)\otimes \bigwedge \mathfrak{g}\hookrightarrow \bigwedge\mathfrak{g}^*\otimes C^\infty(\calM)\otimes\bigwedge\mathfrak{g}\otimes \Sym(\mathfrak{g}),
\]
we can associate to $F$ a function $\Psi_F$ on $\widetilde{\calM}$. With the same notations as above, we have 
\[
\Psi_F=\bar c_i F^i.
\]
The odd function $\Psi_F$ is called \emph{gauge-fixing fermion}. 
\begin{exe}
Show that 
\[
S_F=S+\delta \Psi_F.
\]
\end{exe}

Assume now that the action $S$ is invariant, i.e. $\delta S=0$.

\begin{lem}
\label{lem:gauge-fixing_independence}
Let $g$ be a function on $\widetilde{\calM}$. If $\delta g=0$ and $\delta$ is divergence-free, then 
\[
I_F:=\int_{\widetilde{\calM}}\exp(\I S_F/\hbar)g\widetilde{\mu},
\]
is gauge-fixing independent.
\end{lem}

\begin{proof}
Let $(F_t)$ be a family of gauge-fixing functions. Then 
\begin{multline*}
\frac{\dd}{\dd t}I_{F_t}=\frac{\I}{\hbar}\int_{\widetilde{\calM}}\delta\left(\frac{\dd}{\dd t}\Psi_{F_t}\right)\exp(\I S_F/\hbar)g\widetilde{\mu}\\
=\frac{\I}{\hbar}\int_{\widetilde{\calM}}\delta\left(\frac{\dd}{\dd t}\Psi_{F_t}\exp(\I S_F/\hbar)g\right)\widetilde{\mu}=\frac{\I}{\hbar}\int_{\widetilde{\calM}}\Div \delta \exp(\I S_F/\hbar)g\widetilde{\mu}=0.
\end{multline*}
\end{proof}

\begin{proof}[Proof of Theorem \ref{thm:gauge-fixing_independence}]
Using Lemma \ref{lem:gauge-fixing_independence} together with Exercise \ref{exe:divergence_BRST_operator} we immediately get the proof.
\end{proof}

An particular case of a $\delta$-closed function is a $\delta$-exact function. These functions are irrelevant as for computing expectation values. In fact, 
\[
\int_{\widetilde{\calM}}\exp(\I S_F/\hbar)\delta h\widetilde{\mu}=\int_{\widetilde{\calM}}\delta\left(\exp(\I S_F/\hbar)h\right)\widetilde{\mu}=0
\]
if $S$ is invariant and $\delta$ is divergence-free. We can now extend Theorem \ref{thm:gauge-fixing_independence} to the following theorem.
\begin{thm}
\label{thm:BRST_divergence-free}
Let $X\colon \mathfrak{g}\to \calM$ be an infinitesimal action of the Lie algebra $\mathfrak{g}$ on the manifold $\calM$. If the action $S$ is invariant and the BRST operator $\delta$ is divergence-free (i.e. \eqref{eq:divergence-free_condition} holds), then 
\begin{enumerate}
    \item $\langle g\rangle_F$ is gauge-fixing independent for all $g\in \ker \delta$,
    \item $\langle g\rangle_F$ for all $g\in\mathrm{im}\,\delta$.
\end{enumerate}
Hence, the expectation value defines a linear function on the $\delta$-cohomology.
\end{thm}

\begin{rem}[Ward identities]
Note that point (2) produces identities relating expectation values of different quantities. Such identities as called \emph{Ward identities} and usually have nontrivial content. 
\end{rem}

\begin{ex}[Translations]
Let $\calM=\mathfrak{g}=\R$. Moreover, let $\mathfrak{g}$ act by infinitesimal translations. Denoting by $x$ the coordinate on $\calM$, we have $\delta x=c$ and $\delta c=0$. Assume $S$ and $g$ are constant. Then 
\[
\langle g\rangle_F=\frac{\exp(\I S/\hbar)\displaystyle{\int}\delta_0(F(x))F'(x)}{\exp(\I S/\hbar)\displaystyle{\int}\delta_0(F(x))F'(x)}=g
\]
if the denominator does not vanish. Clearly, $\langle g\rangle_F$ is gauge-fixing independent. Similarly, one can treat rotation-invariant functions on $\calM=S^1$. A section here is just a point. We take $\calM'$ to be a neighborhood of this point. After identifying $\calM'$ with $\R$, we can proceed as above with $F$ being any function with a single nondegenerate zero corresponding to the image of the section.
\end{ex}

\begin{ex}[Plane rotations]
Let $\calM=\R^2\setminus\{0\}$ and $\mathfrak{g}=\mathfrak{so}(2)$ acting by infinitesimal rotations. Denoting by $x$ and $y$ the coordinates on $\R^2$, we have 
\[
\delta x=yc,\quad \delta y=-xc,\quad \delta c=0.
\]
Let $\calM'=\{x>0\}$. A possible choice for $F$ is the function $F(x,y)=y$. Then 
\[
S_F(x,y,c,\lambda,\bar c)=S(x,y)+\lambda y+\bar cxc,
\]
where $S$ is the given rotation-invariant action. Then we get 
\[
\langle g\rangle_F=\frac{\displaystyle{\int}_0^{+\infty}\exp(\I S(x,0)/\hbar)g(x,0)x\dd x}{\displaystyle{\int}_0^{+\infty}\exp(\I S(x,0)/\hbar)x\dd x},
\]
which is equal to the expected expression
\[
\frac{\displaystyle{\int}_\calM\exp(\I S(x,y)/\hbar)g(x,y)\dd x\dd y}{\displaystyle{\int}_\calM \exp(\I S(x,y)/\hbar)\dd x\dd y}
\]
for a rotation-invariant function $g$.
\end{ex}

\subsection{Infinite dimensions}

In the infinite-dimensional setting, we want to consider \eqref{eq:expectation_BRST} whenever it makes sense as a perturbative expansion. In this case $\calM$ is an infinite-dimensional manifold, $\mathfrak{g}$ an infinite-dimensional Lie algebra acting freely on $\calM$ and $S$ some local action function on $\calM$. The space $\widetilde{\calM}$ and the BRST operator $\delta$ will be exactly defined as above and $\delta$ will still be a cohomological vector field on $\widetilde{\calM}$. A gauge-fixing function $F$ will be allowed whenever the corresponding $A$ is nondegenerate and the critical point of the action at a zero of $F$ is also nondegenerate. Then, for suitable functions $g$, we will be able to define $\langle g\rangle_F$ as a perturbative expansion using Feynman diagrams. Note that an observable in this setting is a $\delta$-cohomology class $g$ for which $\langle g\rangle_F$ is well-defined. In fact, Theorem \ref{thm:BRST_divergence-free} will hold whenever $\Div \delta=0$. Note that, however, the divergence of $\delta$ is not defined a priori and has to be understood in terms of expectation values. The usual way to proceed is to assume Theorem \ref{thm:BRST_divergence-free} to hold and to derive from it properties of the expectation values. Once they are properly defined in terms of Feynman diagrams, one can check whether the identities hold and call any deviation an \emph{anomaly}.

Note also that if our aims are of mathematical nature, Theorem \ref{thm:BRST_divergence-free} provides a source for a lot of interesting conjectures (which, fortunately, in most cases turn out to be true).\\

\subsubsection{The trivial Poisson sigma model on the plane}

Consider a 2-dimensional generalization of Section \ref{sec:Moyal_product_path_integral} which is also the basis for the study of the Poisson sigma model. Let $\xi$ and $\eta$ be 0-form and 1-form on the plane respectively. We assume that they vanish at infinity sufficiently fast (e.g., as Schwarz functions) so that we may define the action
\begin{equation}
\label{eq:trivial_PSM_action}
S:=\int_\Sigma \eta\,\dd\xi,
\end{equation}
with $\Sigma=\R^2$. Note that we have dropped the wedge product to avoid cumbersome notation and we will stick to this convention for the rest of the discussion.
The space of fields is given by $\calM=\Omega^0_0(\R^2)\oplus \Omega^1_0(\R^2)$. On $\calM$ we have an action of the abelian Lie algebra $\mathfrak{g}=\Omega^0_0(\R^2)$, given by the monomorphism $i\circ \dd$, where 
\[
\mathfrak{g}\xhookrightarrow{\dd}\Omega^1_0(\R^2)\xhookrightarrow{i}\calM.
\]
The action function $S$ is clearly invariant. The BRST differential $\delta$ is given on coordinates by 
\[
\delta\xi=0,\qquad \delta\eta=\dd c.
\]

In order to define a gauge-fixing function, we want to choose a Riemannian metric on $\R^2$. Denote by $*$ the \emph{Hodge star operator} induced by the Riemannian metric. Then, using the Hodge star, we can define the pairing
\begin{equation}
\label{eq:pairing}
\langle\alpha,\beta\rangle_*:=\int_{\R^2}(* \alpha) \beta,\quad \alpha\in \Omega^j(\R^2),\beta\in \Omega^k(\R^2),\quad j,k=0,1,2.
\end{equation}
Let $\dd^*:=*\dd*$ be the formal adjoint of the de Rham differential and choose the gauge-fixing function $F(\xi,\eta)=\dd^*\eta$. Note that different metrics will give different gauge-fixing functions. one can show that the corresponding operator $A$ is then given by the Laplacian on 0-forms which is invertible for the given conditions at infinity. There is also a unique critical point, i.e. a solution to $\dd\xi=\dd\eta=0$, satisfying the gauge-fixing condition $\dd^*\eta=0$, in particular, $\xi=\eta=0$. Hence, this gauge-fixing is indeed allowed.

Using integration, we can identify $\mathfrak{g}^*$ with $\Omega^2(\R^2)$. The corresponding gauge-fixing fermion is 
\[
\Psi_F=\int_{\R^2}\bar c\, \dd^*\eta.
\]
This gives the gauge-fixed action
\begin{equation}
\label{eq:gauge-fixed_action}
S_F=\int_{\R^2}\Big(\eta\,\dd\xi+\lambda\,\dd^*\eta-\bar c\,\dd^*\dd c\Big).
\end{equation}

Using the pairing as in \eqref{eq:pairing}, we can rewrite the gauge-fixed action as
\[
S_F=\frac{1}{2}\langle\phi,\boldsymbol{\dd}\phi\rangle_*-\langle *\bar c,\Delta c\rangle_*,
\]
where $\Delta:=\dd^*\dd+\dd\dd^*$ denotes the \emph{Hodge Laplacian}, 
\[
\phi:=\begin{pmatrix}\xi\\\eta\\\lambda\end{pmatrix}\in\Omega^0_0(\R^2)\oplus \Omega^1_0(\R^2)\otimes \Omega^2_0(\R^2)
\]
and 
\[
\boldsymbol{\dd}:=\begin{pmatrix}0& *\dd& 0\\ *\dd & 0 & \dd*\\0&\dd*&0\end{pmatrix}.
\]
The propagators between $\phi$ and $c$ or $\bar c$ are clearly zero. Hence, we get 
\[
\langle *\bar c(z)c(w)\rangle_0=-\I\hbar G_0(w,z),
\]
where $G_0$ denotes the Green function of the Hodge Laplacian acting on functions. Now in order to get the propagator between two fields $\phi$, we need to invert the symmetric operator $\boldsymbol{\dd}$. For this, we first compute its square 
\[
\boldsymbol{\dd}^2=\begin{pmatrix}\Delta&0&0\\0&\Delta&0\\0&0&\Delta\end{pmatrix}.
\]
Then, we can observe that $\boldsymbol{\dd}^{-1}=\boldsymbol{\dd}\boldsymbol{\dd}^{-2}$ and hence 
\[
\boldsymbol{\dd}^{-1}=\begin{pmatrix}0&*\dd\Delta^{-1}&0\\ *\dd\Delta^{-1}&0&\dd*\Delta^{-1}\\0&\dd*\Delta^{-1}&0\end{pmatrix}.
\]
In particular, we have 
\[
\langle\xi(z)\eta(w)\rangle_0=\I\hbar*_z\dd_z G_1(z,w)=\I\hbar *_w\dd_wG_0(w,z),
\]
where $G_1$ denotes the Green function of the Hodge Laplacian acting on 1-forms. Let us introduce the \emph{superfields}
\begin{align}
\label{eq:superfields}
\begin{split}
    \boldsymbol{\xi}&:=\xi-\dd^*\bar c,\\
    \boldsymbol{\eta}&:=c+\eta.
\end{split}
\end{align}
Then we can define a \emph{superpropagator} 
\begin{multline}
    \label{eq:superpropagator}
    \I\hbar\theta(z,w):=\left\langle\boldsymbol{\xi}(z)\boldsymbol{\eta}\right\rangle_0=\langle\xi(z)\eta(w)\rangle_0-\langle\dd^*\bar c(z)c(w)\rangle_0\\
    =\I\hbar(*_z\dd_z+*_w\dd_w)G_0(w,z)\in \Omega^1(C_2(\R^2)),
\end{multline}
where $C_2(\R^2)$ denotes the configuration space of two points in $\R^2$ as in Remark \ref{rem:conf}.

\begin{lem}
If we choose the Euclidean metric, then 
\[
\theta=\frac{\dd\phi_E}{2\pi},
\]
where $\dd$ denotes the differential on $C_2(\R^2)$ and $\phi_E(z,w)$ the Euclidean angle between a fixed reference line and the line passing through $z$ and $w$.
\end{lem}

\begin{proof}
The Green function for the Euclidean Laplacian in two dimensions is given by 
\[
G_0(z,w)=\frac{1}{2\pi}\log(\vert z-w\vert),
\]
where $\vert\enspace\vert$ denotes the Euclidean norm. In complex coordinates we have 
\[
G_0(z,w)=\frac{1}{4\pi}\log((z-w)(\bar z-\bar w)).
\]
Then 
\[
\dd_zG_0(z,w)=\frac{1}{4\pi}\left(\frac{\dd z}{z-w}+\frac{\dd \bar z}{\bar z-\bar w}\right).
\]
The Euclidean Hodge star operator in complex coordinates gives $*\dd z=\I\dd z$ and $*\dd\bar z=\I\dd\bar z$. Hence, we have 
\[
*_z\dd_zG_0(z,w)=\frac{1}{4\pi\I}\left(\frac{\dd w}{w-z}-\frac{\dd\bar w}{\bar w-\bar z}\right).
\]
Summing everything up, we get 
\[
\theta=\frac{1}{4\pi\I}\left(\frac{\dd z-\dd w}{z-w}-\frac{\dd\bar z-\dd\bar w}{\bar z-\bar w}\right)=\frac{1}{4\pi\I}\dd\log\left(\frac{z-w}{\bar z-\bar w}\right).
\]
One the other hand, we have $z-w=\vert w-z\vert \exp(\I\phi)$, which gives
\[
\phi=\frac{1}{2\I}\log\left(\frac{z-w}{\bar z-\bar w}\right),
\]
and hence we get the claim.
\end{proof}

\begin{rem}
The cohomology class of $\theta$ is in fact the generator of $H^1(C_2(\R^2),\mathbb{Z})$. It is not difficult to see that other choices of metric will still give the same cohomolgy class. One can easily note that $*_w\theta(z,w)$ is the Green function of the operator $P:=*\dd\Delta^{-1}*$ which is a \emph{parametrix} for the de Rham differential on forms that vanish at infinity, in particular, 
\[
\dd P+P\dd=\id.
\]
The convolution relating $P$ to $\theta$ can be written as
\[
P\alpha=-\pi_2(\theta\pi_1^*\alpha),\quad\alpha\in\Omega^j_0(\R^2),\quad j=0,1,2,
\]
with $\pi_1$ and $\pi_2$ being the projections to $\R^2$. Then 
\[
\dd P\alpha-P\dd\alpha=-(\pi_2)_*(\dd\theta\pi_1^*\alpha)+\pi_*^\de(\theta)\alpha,
\]
where $\pi_*^\de(\theta)(w)$ denotes the integral of $\theta$ along a limiting small circle around $w$. Since $P$ is a parametrix and $\alpha$ is arbitrary, we can see that in general $\theta$ is closed and integrates to 1 along the generators of $H_1(C_2(\R^2),\mathbb{Z})$.
\end{rem}

\subsubsection{Expectation values}

Note that any function of $\xi$ is BRST-invariant, i.e. we can consider the evaluation of $\xi$ at some point $u$. A function of $\int_\gamma\eta$ is also invariant for any closed curve $\gamma$. Hence, the expectation value 
\[
\left\langle \xi(u)\int_\gamma\eta\right\rangle_0=:\I\hbar W_\gamma(u),\quad u\not\in \mathrm{im}\,\gamma,
\]
is independent of the gauge-fixing. Since we also have 
\[
\I\hbar W_\gamma(u)=\left\langle \boldsymbol{\xi}(u)\int_\gamma\boldsymbol{\eta}\right\rangle_0,
\]
we can immediately see that $W_\gamma$ is in fact the winding number of $\gamma$ around $u$. This number is in fact invariant under deformations of $\gamma$ or displacements of $u$, which indicates that the theory is topological. For example, let us deform $\gamma$ to $\gamma'$. Denoting by $\sigma$ a 2-chain whose boundary is $\gamma-\gamma'$, we get 
\[
W_\gamma(u)-W_{\gamma'}(u)=\left\langle\xi(u)\int_\sigma\dd\eta\right\rangle_0.
\]
Moreover, introduce the sequence of divergence-free vector fields $X_r(\xi,\eta)=\lambda_r\oplus0$, where $(\lambda_r)$ is a sequence of functions that converges almost everywhere to the characteristic function of the image of $\sigma$. Then we get 
\[
W_\gamma(u)-W_{\gamma'}(u)=\lim_{r\to\infty}\langle\xi(u)X_r(S)\rangle_0=\I\hbar\lim_{r\to\infty}\langle X_r(\xi(u))\rangle_0=0,
\]
under the assumption that $u$ does not belong to $\sigma$.\\

\subsubsection{The trivial Poisson sigma model on the upper half-plane}

Consider now the action \eqref{eq:trivial_PSM_action} where $\Sigma=\mathbb{H}^2$. As a boundary condition, we impose that the 1-form $\eta$ vanishes when restricted to the boundary $\de\mathbb{H}^2=\R\times\{0\}$. The Lie algebra acting on the space of fields is given by 0-forms on $\mathbb{H}^2$ vanishing on $\de\mathbb{H}^2$. The BRST complex can then be defined exactly as before and we can choose the same gauge-fixing function. We define the superpropagator as in \eqref{eq:superpropagator} with the difference that we will denote it by $\vartheta$ instead of $\theta$ in order to avoid any confusion. 

\begin{lem}
\label{lem:propagator_vartheta}
If we choose the Euclidean metric, then 
\[
\vartheta=\frac{\dd\phi_h}{2\pi},
\]
where $\dd$ denotes the differential on $C_2(\mathbb{H}^2)$ and $\phi_h(z,w)$ denotes the angle between the vertical line through $w$ and the geodesic joining $w$ to $z$ in the hyperbolic Poincar\'e metric (recall Figure \ref{fig:angle_map}).
\end{lem}

\begin{proof}
The Green function $G^{\mathbb{H}^2}_0$ of the Laplacian on $\mathbb{H}^2$ is the restriction to $\mathbb{H}^2$ of the Green function of the Laplacian on $\R^2$ plus a harmonic function such that the sum satisfies the boundary conditions. In complex coordinates, we need $G^{\mathbb{H}^2}_0(w,z)=0$ whenever $w$ is real. This can be done by setting 
\[
G^{\mathbb{H}^2}_0(w,z)=G_0(w,z)-G_0(\bar w, z).
\]
Then we get 
\[
\vartheta(z,w)=\theta(z,w)-\theta(z,\bar w).
\]
Since the hyperbolic angle is given by 
\[
\phi_h(z,w)=\frac{1}{2\I}\log\left(\frac{(z-w)(\bar z-w)}{(\bar z-\bar w)(z-\bar w)}\right),
\]
we conclude the claim.
\end{proof}

\begin{rem}
Note that $\vartheta$ is the generator of $H^1(C_2(\mathbb{H}^2),\mathbb{H}^2\times\de\mathbb{H}^2;\mathbb{Z})$.
\end{rem}

\subsubsection{Generalizations}

We can also consider a collection of $n$ 0-forms $\xi^i$ and $n$ 1-forms $\eta_i$ with $i=1,\ldots,n$. Then we want to look at the action 
\[
\int_\Sigma\eta_i\,\dd\xi^i.
\]
We can also think of $\xi$ and $\eta$ as forms taking values in $\R^n$. The Lie algebra $\mathfrak{g}$ of symmetries will then consist of the direct sum of $n$ copies of the previous one; in other words it will be the abelian Lie algebra of $\R^n$-valued 0-forms. Let $c_i$ for $i=1,\ldots,n$ denote the generators of the algebra of functions of $\Pi\mathfrak{g}$. Then we can define the BRST operator through $\delta \xi^i$ and $\delta\eta_i=c_i$. If we choose the gauge-fixing function to be $F_i(\xi,\eta)=\dd^*\eta_i$, everything remains the same as before. In particular, we can again introduce superfields $\boldsymbol{\xi}^i:=\xi^i-\dd^*\bar c^i$ and $\boldsymbol{\eta}_i:=c_i+\eta_i$ and compute the superpropagator 
\begin{equation}
\label{eq:superpropagator2}
\left\langle \boldsymbol{\xi}^i(z)\boldsymbol{\eta}_j(w)\right\rangle_0:=\begin{cases}\I\hbar\theta(z,w)\delta^i_j,&\text{on $\R^2$}\\\I\hbar\vartheta(z,w)\delta^i_j,&\text{on $\mathbb{H}^2$}\end{cases}
\end{equation}
Another generalization can be done by dropping the assumption that the 0-form field vanishes at infinity. More precisely, we denote by $X^i$ a collection of maps to $\R^n$ with no conditions on the boundary or at infinity and consider the action 
\[
S=\int_\Sigma\eta_i\,\dd X^i,
\]
where $\Sigma$ is either $\R^2$ or $\mathbb{H}^2$ and the $\eta_i$s are 1-forms vanishing on the boundary and at infinity. Critical points are then pairs of constant maps together with closed 1-forms. They will be also degenerate modulo the action of the abelian Lie algebra of 0-forms. However, the degeneracy will be of a very simple type as it is parametrized by the finite-dimensional manifold $\R^n$. In fact, it is enough to choose a measure on $\R^n$ and impose Fubini's theorem. We will choose a delta-measure peaked at a point $x\in\R^n$ and require $X$ to map the point $\infty$ to $x$. Note that if we write $X^i=x^i+\xi^i$, the $\xi^i$ vanish at infinity and everything is reduced to the previous case.

A final generalization is to replace $\R^n$ by a manifold $M$. We think of $X^i$ as a local coordinate expression of a map $X\colon \Sigma\to M$. For the action to be covariant, we need to assume that $\eta_i(u)$ for $u\in\Sigma$ is the local coordinate expression of a 1-form on $\Sigma$ taking values in the cotangent space of $M$ at $X(u)$. In particular, we assume $\eta\in\Gamma(\Sigma,T^*\Sigma\otimes X^*T^*M)$. The space of fields $\calM$ can then be identified with the space of vector bundle maps $T\Sigma\to T^*M$ and the action can be invariantly written as 
\[
S=\int_\Sigma\langle\eta,\dd X\rangle,
\]
where $\langle\enspace,\enspace\rangle$ denotes the canonical pairing between the tangent and cotangent bundle of $M$ and $\dd X$ denotes the differential of the map $X$ regarded as a section of $T^*\Sigma\otimes X^*TM$. if we also require $X$ to map the point $\infty$ to a given point $x\in M$, we can expand around critical points by setting $X=x+\xi$ with $\xi\colon \Sigma\to T_xM$ and by regarding $\eta$ as a 1-form taking values in $T_x^*M$. Choosing local coordinates, we can identify $T_xM$ with $\R^n$, where $n=\dim M$, and reduce everything to the previous case.

We also allow $\Sigma$ to be any 2-manifold. The previous discussion will change drastically if $\Sigma$ is not simply connected, as the space of solutions modulo symmetries, with $X=x$, will now be parametrized by $H^1(\Sigma,T^*_xM)$ and one has to choose a measure on this vector space as well.

\section{The Poisson sigma model}

Next we want to look at deformations of the trivial Poisson sigma model as discussed before. Here, we will describe how the action of the Poisson sigma model is expressed in order to derive Kontsevich's star product out of it. 

\subsection{Formulation of the model}

We want to formulate a deformation of the action functional without introducing extra structure on $\Sigma$ (which is either $\R^2$ or $\mathbb{H}^2$). Therefore, the terms we are allowed to add must be 2-forms on $\Sigma$ given in terms of the fields $X^i$ and $\eta_i$. In particular, they need to be linear combinations of terms $\alpha^{ij}(X)\eta_i\eta_j$, $\beta^i_j(X)\eta_i\,\dd x^i$ and $\gamma_{ij}(X)\dd X^i\dd X^j$. Note that we are not considering a term of the form $\phi^i(X)\dd\eta_i$, since integration by parts reduces it to a term of the second type. The second and third terms can be absorbed by a redefinition of $\eta$ adding to it terms linear in $\eta$ and $\dd X$. Hence, modulo field redefinitions, the most general deformation of the action has the form 
\[
S=\int_\Sigma\left(\eta_i\,\dd X^i+\frac{1}{2}\epsilon\alpha^{ij}(X)\eta_i\eta_j\right)+O(\epsilon^2),
\]
where $\epsilon$ is the deformation parameter (here typically $\epsilon=\frac{\I\hbar}{2}$) and $\alpha^{ij}$ is assumed to be skew-symmetric.

\begin{rem}
We want to show that it makes sense to only consider those deformations in which the $\alpha^{ij}$ are the components of a Poisson bivector field and that the BRST formalism is only available if the Poisson structure is affine.
\end{rem}

Recall that the BRST operator before acted by $\delta X^i=0$, $\delta \eta_i=\dd c_i$ and $\delta c_i=0$, with $c\in \Pi\mathfrak{g}$ and $\mathfrak{g}=\Omega^0_0(\Sigma,\R^n)$. We would like to deform the trivial $\delta$ such that $\delta S=O(\epsilon^2)$ for the new $S$. Note that we will only consider the restriction of $\delta$ to $\calM\times\Pi\mathfrak{g}$ as its restriction to $\mathfrak{g}^*\times\Pi\mathfrak{g}$ needs no deformation.

\begin{lem}
\label{lem:acting_BRST_operator_PSM}
Modulo field redefinitions, there is a unique BRST operator deforming the trivial one such that $\delta S=O(\epsilon^2)$ and $\delta^2=O(\epsilon^2)+R$ with $R$ vanishing at critical points. It acts by 
\begin{align*}
    \delta X^i&=-\epsilon \alpha^{ij}(X)c_j+O(\epsilon^2),\\
    \delta\eta_i&=\dd c_i+\epsilon\de_i\alpha^{jk}(X)\eta_j c_k+O(\epsilon^2),\\
    \delta c_i&=-\frac{1}{2}\epsilon \de_i\alpha^{jk}(X)c_jc_k+O(\epsilon^2).
\end{align*}
Moreover, $R$ vanishes on the whole $\calM\times\Pi\mathfrak{g}$ if $\alpha$ is at most linear.
\end{lem}

\begin{proof}
Recall that $\delta$ applied to $X$ or $\eta$ must be linear in the ghost variables $c$. Hence, the most general deformation of $\delta$ (without adding any extra structure on $\Sigma$) is of the form 
\begin{align*}
    \delta X^i&=\epsilon v^{ij}(X)c_j+O(\epsilon^2),\\
    \delta \eta_i&=\dd c_i+\epsilon\Big(a^{jk}_i(X)\eta_jc_k+b^j_i(X)\dd c_j+d^k_{ij}(X)\dd X^jc_k\Big)+O(\epsilon^2),
\end{align*}
for some functions $v^{ij}, a^{jk}_i, b^{j}_i$ and $d^{k}_{ij}$ on $\R^n$. Thus, we get 
\begin{multline*}
\delta S=\epsilon\int_\Sigma\Big((a^{jk}_i(X)\eta_jc_k+b^j_i(X)\dd c_j+d^k_{ij}(X)\dd X^j c_k)\dd X^i\\
+\eta_i(\dd X^r\de_r v^{ij}(X)c_j+v^{ij}(X)\dd c_j)+\alpha^{ij}(X)\dd c_i\eta_j\Big)+O(\epsilon^2).
\end{multline*}
As the identity $\delta S=O(\epsilon^2)$ must hold for any $\eta$, we get the following two equations
\begin{align}
    \label{eq:first_eq}
    a^{jk}_i(X)c_k\dd X^i+\dd X^r\de_rv^{jk}(X)c_k+v^{jk}(X)\dd c_k+\alpha^{kj}(X)\dd c_k&=0,\\
    \label{eq:second_eq}
    \int_\Sigma \Big( b_i^j(X)\dd c_j+d^k_{ij}(X)\dd X^jc_k\Big)\dd X^i&=0.
\end{align}
In particular, if we choose $X$ to be the constant map (with value $x$), we deduce from \eqref{eq:first_eq} that 
\[
v^{jk}(x)\dd c_k+\alpha^{kj}(x)\dd c_k=0,
\]
and since this has to hold for any $c$, we have 
\[
\alpha^{jk}=v^{jk}.
\]
Plugging into \eqref{eq:second_eq}, we get 
\[
a^{jk}_i(X)c_k\dd X^i-\dd X^i\de_i\alpha^{jk}(X)c_k=0,
\]
and since this has to hold for all $c$ and $X$, we get
\[
a^{jk}_i=\de_i\alpha^{jk}.
\]
Using integration by parts, \eqref{eq:second_eq} gives
\[
\int_\Sigma\Big(-\dd X^r\de_r b^j_i(X)c_j+d^k_{ij}(X)\dd X^j c_k\Big)\dd X^i=0,
\]
and since this has to hold for all $X$ and $c$, we finally get
\[
d^k_{ij}=\de_j b^{k}_i.
\]
Hence, we have shown that
\begin{align*}
    \delta X^i&=-\epsilon\alpha^{ij}(X)c_j+O(\epsilon^2),\\
    \delta \eta_i&=\dd c_i+\epsilon\Big(\de_i\alpha^{jk}(X)\eta_j c_k+\dd(b^j_i(X)c_j)\Big)+O(\epsilon^2),
\end{align*}
which, after redefinition $c_i\mapsto c_i-\epsilon b^j_i(X)c_j+O(\epsilon^2)$, gives the first two equations in Lemma \ref{lem:acting_BRST_operator_PSM}. For the last equation, we recall that the BRST operator on $c$ must be quadratic in $c$, so its general form is
\[
\delta c_i=\frac{1}{2}\epsilon f^{jk}_i(X)c_jc_k+O(\epsilon^2).
\]
To determine the \emph{structure} functions $f^{jk}_i$, we can compute $\delta^2$. Note that $\delta^2 X^i=\delta^2 c_i=O(\epsilon^2)$. On the other hand, 
\begin{multline*}
    \delta^2\eta_i=\epsilon\Bigg(\frac{1}{2}\dd \big(f^{jk}_i(X)c_jc_i\big)+\de_i\alpha^{jk}(X)\dd c_j c_k\Bigg)+O(\epsilon^2)\\
    =\epsilon\Bigg(\big(f^{jk}_i(X)+\de_i\alpha^{jk}(X)\big)\dd c_j c_k+\frac{1}{2}\dd X^r\de_r f^{jk}_i(X)c_jc_k\Bigg)+O(\epsilon^2).
\end{multline*}
At a critical point (where $\dd X^i=O(\epsilon^2)$) the third summand of the last equation vanishes. Thus, $\delta^2=O(\epsilon^2)$ at critical points implies that 
\[
f_i^{jk}=-\de_i\alpha^{jk},
\]
which proves the last equation in Lemma \ref{lem:acting_BRST_operator_PSM}, Note also that 
\[
\delta^2\eta_i=-\frac{1}{2}\epsilon\dd x^r\de_r\de_i\alpha^{jk}(X)c_jc_k,
\]
which is zero (not only at critical points) whenever $\alpha$ is at most linear.
\end{proof}

Next, we want to extend deformations beyond the first order in $\epsilon$. Even without knowing the following terms, we can already state the following Lemma.

\begin{lem}
$\delta^2=O(\epsilon^3)$ at critical points only if $\alpha$ is Poisson.
\end{lem}

\begin{proof}
Note that we have 
\[
\delta^2X^i=-\epsilon^2\Bigg(\alpha^{rk}(X)c_k\de_r\alpha^{ij}(X)c_j+\frac{1}{2}\alpha^{ij}(X)\de_i\alpha^{rj}(X)c_rc_j\Bigg)+O(\epsilon^3).
\]
Since this has to hold for all $c$, we get the Jacobi identity for $\alpha$.
\end{proof}

\begin{rem}
It is in fact possible to prove, under the assumption that $\alpha$ is Poisson, that this deformation is not only infinitesimal. In particular, we have the following theorem. 
\end{rem}

\begin{thm}
\label{thm:cohomological_PSM}
Given a Poisson bivector field $\alpha$, the odd vector field
\begin{align*}
    \delta X^i&=-\epsilon\alpha^{ij}(X)c_j,\\
    \delta\eta_i&=\dd c_i+\epsilon\de_i\alpha^{jk}(X)\eta_j c_k,\\
    \delta c_i&=-\frac{1}{2}\epsilon\de_i\alpha^{jk}(X)c_jc_k,
\end{align*}
is cohomological for $\alpha$ at most linear or at critical points for all $\epsilon$. Moreover, 
\[
S:=\int_\Sigma\Bigg(\eta_i\,\dd X^{i}+\frac{1}{2}\epsilon\alpha^{ij}(X)\eta_i\eta_j\Bigg)
\]
is $\delta$-closed for all $\epsilon$.
\end{thm}

\begin{exe}
Prove Theorem \ref{thm:cohomological_PSM}.
\end{exe}

\begin{rem}
The geometrical meaning of Theorem \ref{thm:cohomological_PSM} is that there is a distribution of vector fields on $\calM$ under which the action is invariant. In general, this distribution is involutive only on the submanifold of critical points of $S$. It is involutive on the whole of $\calM$ whenever $\alpha$ is at most linear and in this case it can be regarded as the free, infinitesimal action of a Lie algebra. The action $S$ can be generalized to the case when one wants to consider a Poisson manifold $(M,\alpha)$ instead of $\R^n$. For this, one regards $X$ as a map $\Sigma\to M$ and, for a given map $X$, $\eta$ is taken to be a section of $T^*\Sigma\otimes X^*T^*M$. If $\langle\enspace,\enspace\rangle$ denotes the canonical pairing between the tangent and cotangent bundle of $M$ and by $\alpha^\sharp$ the bundle map $T^*M\to TM$ induced by the Poisson bivector field $\alpha$ (see also Section \ref{subsec:Courant_algebroids}), we can write
\[
S=\int_\Sigma\Bigg(\langle\eta,\dd X\rangle+\frac{1}{2}\epsilon\langle \eta,\alpha^\sharp\eta\rangle\Bigg).
\]
\end{rem}

\subsection{Observables}
\label{subsec:observables}
For the case when $\Sigma=\mathbb{H}^2$, $c$ has to vanish on the boundary. This implies that $\delta X(u)=0$ for $u\in\de\Sigma$. Hence, we get 
\begin{multline*}
\calO_{f_1,\ldots,f_k;u_1,\ldots,u_k}:=f_1(X(u_1))\dotsm f_k(X(u_k)),\\
f_1,\ldots,f_k\in C^\infty(\R^n),\quad u_1,\ldots,u_k\in\de\Sigma\cong\R,\quad u_1<\dotsm <u_k,
\end{multline*}
are observables, i.e. $\delta$-closed functions. In \cite{CF1} it was shown that with the gauge-fixing $\dd^*\eta=0$ for the Euclidean metric on $\Sigma$, one has 
\[
\langle \calO_{f_1,\ldots,f_k;u_1,\ldots,u_k}\rangle(x)=f_1\star\dotsm \star f_k(x),
\]
where $\langle\enspace\rangle(x)$ denotes the expectation value for $X(\infty)=x$ (and expanding only around the trivial critical solution $X=x$, $\eta=0$) while $\star$ denotes Kontsevich's star product for the given Poisson structure. In the next section will derive this result for the case when $\alpha$ is at most linear so that the BRST formalism is available.

\section{Deformation quantization for affine Poisson structures}

An affine Poisson structure on $\R^n$ is given by a Poisson bivector field $\alpha$ which is at most linear. The linear part $\alpha$ gives the dual of $\R^n$ a Lie algebra structure while the constant part is a 2-cocycle in the Lie algebra cohomology with trivial coefficients. Let us denote this Lie algebra by $\mathfrak{h}$. The fields of the Poisson sigma model for an affine Poisson structure are then a map $X\colon \Sigma\to \mathfrak{h}^*$ and a 1-form $\eta$ on $\Sigma$ with values in $\mathfrak{h}$. Since $\mathfrak{h}$ is a Lie algebra, we can regard $\eta$ as a connection 1-form on the trivial principal bundle $P$ over $\Sigma$ (with gauge group any Lie group whose Lie algebra is $\mathfrak{h}$). For definiteness, we will fix $\Sigma$ to be $\mathbb{H}^2$ and we will require $\eta$ to vanish at infinity and on the boundary. The action is then given by 
\[
S=\int_\Sigma\Bigg(\eta_i\,\dd X^i+\frac{1}{2}\alpha^{ij}(X)\eta_i\eta_j\Bigg),
\]
where 
\[
\alpha^{ij}(x)=\chi^{ij}+x^kf^{ij}_k
\]
is the given affine Poisson structure on $\mathfrak{h}^*$. Using integration by parts, we can also rewrite it as 
\[
S=\int_\Sigma\Bigg(\langle X,F_\eta\rangle+\frac{1}{2}\chi(\eta,\eta)\Bigg),
\]
where $\langle\enspace,\enspace\rangle$ denotes the canonical pairing between $\mathfrak{h}$ and $\mathfrak{h}^*$ while 
\[
(F_\eta)_i=\dd \eta_i+\frac{1}{2}f_i^{jk}\eta_j\eta_k
\]
is the curvature 2-form of the connection 1-form $\eta$. Note that there is a Lie algebra $\mathfrak{g}$, which as a vector space consists of functions $\Sigma\to \mathfrak{h}$ vanishing at infinity and on the boundary, that acts on the space of fields $\calM$ and leaves the action invariant. The BRST operator on $\calM\times \Pi\mathfrak{g}$ has the form as in Theorem \ref{thm:cohomological_PSM} with $\epsilon=1$. Geometrically, we can regard $\mathfrak{g}$ as the Lie algebra of infinitesimal gauge transformations of the principal bundle $P$; the field $\eta$ actually transforms as a connection 1-form, while $X$ transforms as a section of the coadjoint bundle in case $\chi=0$. For $\chi\not=0$, we can regard $X\oplus 1$ as a section of the coadjoint bundle for the Lie algebra $\widehat{\mathfrak{h}}\cong\mathfrak{h}\oplus\R$ obtained by central extension of $\mathfrak{h}$ through $\chi$. The BRST operator on $\Pi\mathfrak{g}^*\times\mathfrak{g}^*$ has the usual form \eqref{eq:FP_cohomological_VF}.

\subsection{Gauge-fixing and Feynman diagrams}

Choose a metric on $\Sigma$ and define the gauge-fixing function $F(X,\eta)=\dd^*\eta$. The gauge-fixing fermion is given by 
\[
\Psi_F=\int_\Sigma\langle \bar c,\dd^*\eta\rangle
\]
and the gauge-fixed action is given by 
\[
S_F=\int_\Sigma\Bigg(\eta_i\,\dd X^i+\frac{1}{2}\alpha^{ij}(X)\eta_i\eta_j+\lambda^i\dd^*\eta_i-\bar c^k\dd^*\Big(\dd c_k+\de_k\alpha^{ij}(X)\eta_ic_j\Big)\Bigg).
\]
Fix the value of $X$ at infinity to be given by $x$. We write $X=x+\xi$, where the field $\xi$ has to vanish at infinity. We can observe that $S_F$ has the same form as \eqref{eq:action_form} (with $y$ the collection of $\xi,\bar c,\lambda$ and $z$ the collection of $\eta,c$). Hence, we can write $S_F=S_0+S_1$ with 
\begin{align*}
    S_0&=\int_\Sigma\Bigg(\eta_i\,\dd X^i+\lambda^i\dd^*\eta_i-\bar c^k\dd^*\dd c_k\Bigg),\\
    S_1&=\int_\Sigma\Bigg(\frac{1}{2}\alpha^{ij}(x+\xi)\eta_i\eta_j-\bar c^k\dd^*\Big(\de_k\alpha^{ij}(x+\xi)\eta_ic_j\Big)\Bigg),
\end{align*}
and regard $S_1$ as a perturbation of $S_0$. if we consider superfields $\boldsymbol{\xi}$ and $\boldsymbol{\eta}$ as in \eqref{eq:superfields}, we can write
\[
S_1=\int_\Sigma\frac{1}{2}\alpha^{ij}(x+\boldsymbol{\xi})\boldsymbol{\eta}_i\boldsymbol{\eta}_j,
\]
where integration on $\Sigma$ is understood to select the 2-form component. 

\begin{rem}
This shows that as long as the considered observables can be written as functions of the superfields and expectation values are computed in terms of the superpropagators \eqref{eq:superpropagator2}. If we denote the superpropagator graphically as an arrow from $\boldsymbol{\eta}$ to $\boldsymbol{\xi}$, the perturbation $S_1$ is represented by the two vertices as in Figure \ref{fig:two_vertices} with the bivalent vertex corresponding to $\alpha^{ij}(x)=\chi^{ij}+x^kf^{ij}_k$ and the trivial vertex corresponding to the structure constants.
\end{rem}

\begin{figure}[ht]
\centering
\begin{tikzpicture}[scale=0.7]
\tikzset{Bullet/.style={fill=black,draw,color=#1,circle,minimum size=0.5pt,scale=0.5}}
\node[Bullet=gray] (v1) at (0,0) {};
\node[Bullet=gray] (v2) at (5,0) {};
\draw[->] (v1) -- (-1,-2);
\draw[->] (v1) -- (1,-2);
\draw[->] (v2) -- (4,-2);
\draw[->] (v2) -- (6,-2);
\draw[->] (5,2) -- (v2);
\end{tikzpicture}
\caption{The two vertices} 
\label{fig:two_vertices}
\end{figure}
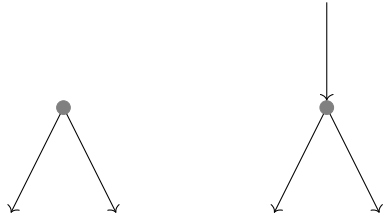

Let us now consider the observable $\calO_{f_1,\ldots,f_k;u_1,\ldots,u_k}$ as in Section \eqref{subsec:observables}. As the evaluation point, i.e., integration along a 0-cycle, is understood to select the 0-form component of a differential form, we can write
\begin{multline*}
\calO_{f_1,\ldots,f_k;u_1,\ldots,u_k}:=f_1(x+\boldsymbol{\xi}(u_1))\dotsm f_k(x+\boldsymbol{\xi}(u_k)),\\
f_1,\ldots,f_k\in C^\infty(\R^n),\quad u_1,\ldots,u_k\in\de\Sigma\cong\R,\quad u_1<\dotsm <u_k.
\end{multline*}

Computing then the expectation value $\langle \calO_{f_1,\ldots,f_k;u_1,\ldots,u_k}\rangle(x)$, we only need the superpropagator. The Feynman diagrams then have three kind of vertices:
\begin{enumerate}
    \item bivalent vertices in the upper half-plane corresponding to $\alpha^{ij}(x)$,
    \item trivalent vertices in the upper half-plane corresponding to $f_k^{ij}$,
    \item $\ell$-valent vertices with $\ell\geq 0$ with only incoming arrows at one of the boundary points $u_i$ corresponding to the $\ell$-th derivative of $f_i$.
\end{enumerate}

Recall that the normal ordering excludes all graphs containing a tadpole (i.e. an edge starting and ending at the same vertex). The combinatorics prevents automatically vacuum subgraphs (see Figure \ref{fig:example_Feynman_diagrams} for examples).

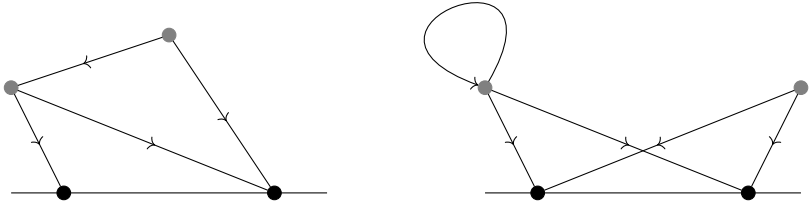
\begin{figure}[ht]
\centering
\begin{tikzpicture}[scale=0.7]
\tikzset{Bullet/.style={fill=black,draw,color=#1,circle,minimum size=0.5pt,scale=0.5}}
\node[Bullet=black] (bar11) at (1,-2) {};
\node[Bullet=black] (bar21) at (5,-2) {};
\node[Bullet=black] (bar12) at (10,-2) {};
\node[Bullet=black] (bar22) at (14,-2) {};
\node[Bullet=gray] (v11) at (0,0) {};
\node[Bullet=gray] (v21) at (3,1) {};
\node[Bullet=gray] (v12) at (9,0) {};
\node[Bullet=gray] (v22) at (15,0) {};
\draw[fermion] (v11) -- (bar11);
\draw[fermion] (v11) -- (bar21);
\draw[fermion] (v21) -- (v11);
\draw[fermion] (v21) -- (bar21);
\draw[fermion] (v12) -- (bar12);
\draw[fermion] (v12) -- (bar22);
\draw[fermion] (v22) -- (bar12);
\draw[fermion] (v22) -- (bar22);
\path[->,every loop/.style={looseness=40}] (v12)
         edge  [in=160,out=60,loop] (v12); 
\draw (0,-2) -- (6,-2);
\draw (9,-2) -- (15,-2);
\end{tikzpicture}
\caption{Example of an allowed graph and a non-allowed graph.} 
\label{fig:example_Feynman_diagrams}
\end{figure}

For the case when $k=2$ and considering the gauge-fixing $\dd^*\eta=0$ with respect to the Euclidean metric on $\mathbb{H}^2$, i.e. with the superpropagator determined by the 1-form $\vartheta$ as in Lemma \ref{lem:propagator_vartheta}, we get 
\begin{equation}
\label{eq:expectation_star_product}
\langle\calO_{f,g;0,1}\rangle(x)=f\star g(x),
\end{equation}
where $\star$ denotes Kontsevich's star product for the given affine Poisson structure.

\begin{exe}
Prove \eqref{eq:expectation_star_product}.
\end{exe}

\subsection{Independence of the evaluation point}

We want to give a formal proof of the independence of the expectation values of $\calO_{f_1,\ldots,f_k;u_1,\ldots,u_k}$ from the points $u_1,\ldots,u_k$. Note first that
\[
f(X(v))-f(X(u))=\int_u^v\dd X^i\de_i f(X)=\int_u^v\Big(\dd X^i+\dd^*\lambda^i\Big)\de_i f(X)-\delta\Phi,
\]
where 
\[
\Phi:=\int_u^v\dd^*\bar c\de_i f(X).
\]
Let $(\omega_r)$ be a sequence of 1-forms on $\Sigma$ vanishing on the boundary and at infinity that converges to the measure concentrated on the interval $(u,v)\in \Sigma$. Denoting by $(a,b)$, with $b\geq 0$, the coordinates on $\Sigma=\mathbb{H}^2$, a possible choice for this sequence is 
\[
\omega_r(a,b)=rb\exp(-rb^2/2)\chi_r(a)\dd a,
\]
where $(\chi_r)$ is a sequence of smooth, compactly supported functions converging almost everywhere to the characteristic function of the interval $(u,v)$. Let $Y_{f,r}$ be the local vector field on $\widetilde{\calM}$ corresponding to the infinitesimal displacement of $\eta_i$ by $\omega_r\de_i f(X)$. Then 
\[
f(X(v))-f(X(u))=\lim_{r\to\infty}Y_{f,r}(S_F)-\delta\Phi.
\]
If $\calO$ is a BRST-observable depending on the fields outside the closed interval $[u,v]$, we get 
\[
\langle(f(X(v))-f(X(u)))\calO\rangle=\I\hbar \lim_{r\to\infty}\langle Y_{f,r}(\calO)\rangle-\langle\delta(\Phi\calO)\rangle=0.
\]

\subsection{Associativity}

The independence of the evaluation points gives us 
\[
\lim_{v\to u^+}\langle \calO_{f,g,h;u,v,w}\rangle_0(x)=\lim_{v\to w^-}\langle\calO_{f,g,h;u,v,w}\rangle_0(x).
\]
The left hand side corresponds intuitively to evaluating first the expectation value of $\calO_{f,g;u,v}$, then replacing the result at $u$ and finally computing the expectation value of $\calO_{\langle \calO_{f,g,;u,v}\rangle_0,h;w}$. The result is then $(f\star g)\star h$. Repeating the computation on the right hand side, we get $f\star (g\star h)$. This rather formal argument explains why one should expect the star product defined by the Poisson sigma model to be associative.

\section{The general construction}

We want to also state the theorem for general Poisson structures on $\R^d$. However, we will not provide a proof of the general construction here since one needs the more general gauge formalism provided by Batalin and Vilkovisky. 

\subsection{The theorem of Cattaneo--Felder}

\begin{thm}[Cattaneo--Felder\cite{CF1}]
Consider the Poisson manifold $(\R^d,\pi)$ with Poisson structure $\pi$ together with the corresponding Poisson sigma model on a disk $D$. The fields are then given by a map $X\colon D\to \R^d$ and a 1-form $\eta\in \Gamma(D,T^*D\otimes X^*T^*\R^d)$ with boundary condition $\eta\big|_{S^1}=0$. Moreover, let $0,1$ and $\infty$ be cyclically ordered points on $\de D=S^1$, i.e. if we start at $0$ and move counterclockwise on $S^1$ we will first meet $1$ and then $\infty$ (see Figure \ref{fig:cyc}), and impose the condition $\eta(\infty)=0$. Then Kontsevich's star product between two functions $f,g\in C^\infty(\R^d)$ evaluated at the point $x=X(\infty)\in \R^d$ is given by 
\[
f\star g(x)=\int_{X(\infty)=x}f(X(1))g(X(0))\exp(\I S(X,\eta)/\hbar)\mathscr{D}[X]\mathscr{D}[\eta],
\]
where the right hand side should be understood by perturbative expansion in terms of Feynman diagrams.
\end{thm}

\begin{rem}
The global field theoretic approach was given in \cite{CMW3} using cutting and gluing methods for manifolds with boundary by introducing a formal global action as more generally developed in \cite{CMW4}.  Moreover, also the case for manifolds with corners was covered. The main gauge formalism that was used is known as the \emph{BV-BFV formalism} \cite{CMR1,CMR2,CattMosh1}, which can bee seen as a extension of the Batalin--Vilkovisky formalism for manifolds with boundary. A first approach to a global formulation of the Poisson sigma model on closed manifolds using the Batalin--Vilkovisky formalism was given in \cite{BCM}.
\end{rem}

\begin{figure}[!ht]
\centering
\tikzset{
particle/.style={thick,draw=black},
particle2/.style={thick,draw=black, postaction={decorate},
    decoration={markings,mark=at position .9 with {\arrow[black]{triangle 45}}}},
gluon/.style={decorate, draw=black,
    decoration={coil,aspect=0}}
 }
\tikzset{Bullet/.style={fill=black,draw,color=#1,circle,minimum size=0.5pt,scale=0.5}}
\begin{tikzpicture}[x=0.04\textwidth, y=0.04\textwidth]
\filldraw[color=black!60, fill=gray!20, very thick](-2,0) circle (2.5);
\node[Bullet=black, label=left: $0$] (0) at (-4.3,-1) {};
\node[Bullet=black, label=right: $1$] (1) at (0.3,-1) {};
\node[Bullet=black, label=above: $\infty$] (infty) at (-2,2.5) {};
\node[] (D) at (-2,0) {$D$};
\end{tikzpicture}
\caption{Cyclically ordered points on $\partial D=S^1$}
\label{fig:cyc}
\end{figure}
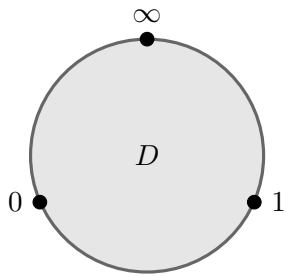

\subsection{Other similar constructions}
The way of using quantum field theory to obtain certain mathematical constructions has been proven to be a very interesting and deep method relating seemingly purely mathematical constructions to physics. Other examples include Witten's famous construction to obtain the 4-manifold invariants described by Donaldson \cite{Donaldson1983} using a certain deformation of the supersymmetric \emph{Yang--Mills} action and a special type of observables \cite{Witten1988}. Computing the expectation value of this observable with respect to the theory formulated by this special action, one can recover the Donaldson polynomials of the given 4-manifold. Another important and surprising result of Witten \cite{Witten1989} was that the expectation value of a \emph{Wilson loop} observable representing a given knot with respect to the \emph{Chern--Simons} action \cite{Chern1974,AS,AS2} will give the \emph{Jones polynomial} \cite{Jones1985} of the knot, which is an important knot invariant. These results have lead to many mathematical conjectures, deeper insights in physics and created a whole new perspective towards the interplay between mathematics and quantum field theory.

\printbibliography
\end{document}